\documentclass[%
aps,
twocolumn,
superscriptaddress,
amsmath,amssymb,
pra,
floatfix,
longbibliography,
10pt
]{revtex4-2}
\usepackage{url}
\usepackage{graphicx}
\usepackage{placeins}
\usepackage{dcolumn}
\usepackage{bm}
\usepackage{amsmath,amsfonts,amssymb}
\usepackage{mathtools}

\usepackage{silence}
\usepackage{physics}
\usepackage{dsfont}
\usepackage{tikz}
\usepackage{pgfplots}
\usepackage{enumerate}
\usepackage{subfigure}
\usepackage{bbold}
\usepackage{makecell}
\usepackage{array}

\usepackage{xcolor}
\usepackage{tabularray}
\usepackage{amssymb}
\usepackage{ytableau}

\usepackage{tikz-qtree}

\usepackage[colorlinks, linkcolor=red, anchorcolor=green, citecolor=blue]{hyperref}

\usepackage{algorithmic}
\usepackage[ruled,linesnumbered]{algorithm2e}

\SetCommentSty{mycommfont}

\makeatletter
\newcommand{\thickhline}{%
    \noalign {\ifnum 0=`}\fi \hrule height 1pt
    \futurelet \reserved@a \@xhline
}
\newcolumntype{"}{@{\hskip\tabcolsep\vrule width 1pt\hskip\tabcolsep}}
\makeatother

\usepackage{multirow}
\usepackage{tabu}
\usepackage{textcomp,booktabs}
\usepackage{colortbl}
\usepackage[T1]{fontenc}
\definecolor{mygray}{gray}{0.9}
\definecolor{mypink}{rgb}{0.99,0.91,0.95}
\definecolor{mycyan}{cmyk}{0.3,0,0,0}

\newcommand{\id}{\text{id}}

\newcommand{\mE}{\mathcal{E}}
\newcommand{\mF}{\mathcal{F}}

\newcommand{\mN}{\mathcal{N}}

\newcommand{\mO}{\mathcal{O}}
\newcommand{\mP}{\mathcal{P}}

\newcommand{\mH}{\mathcal{H}}
\newcommand{\mR}{\mathcal{R}}
\newcommand{\mS}{\mathcal{S}}
\newcommand{\mU}{\mathcal{U}}
\newcommand{\T}{\mathbf{T}}
\newcommand{\0}{\mathbb{0}}
\newcommand{\1}{\mathbb{1}}

\makeatletter
\newcommand*{\rom}[1]{\expandafter\@slowromancap\romannumeral #1@}
\makeatother

\makeatletter
\newcommand*{\Relbarfill@}{\arrowfill@\Relbar\Relbar\Relbar}
\newcommand*{\xeq}[2][]{\ext@arrow 0055\Relbarfill@{#1}{#2}}
\makeatother

\definecolor{mGreen}{RGB}{193, 225, 193}
\definecolor{mRed}{RGB}{255, 204, 203}
\definecolor{mBlue}{RGB}{176, 224, 230}  
\definecolor{mOrange}{RGB}{255, 218, 185} 
\definecolor{mPurple}{HTML}{E5D1FA}

\usepackage{amsthm}
\usepackage{tcolorbox}
\tcbuselibrary{theorems}
\newtheorem{thm}{Theorem}
\newtheorem{lem}{Lemma}

\newtcbtheorem[number within=section, use counter*=thm]{mythm}{Theorem}%
{colback=red!5,colframe=red!35!red,fonttitle=\bfseries}{thm}

\newtcbtheorem[number within=section, use counter*=thm]{mylem}{Lemma}%
{colback=green!5,colframe=green!35!black,fonttitle=\bfseries}{lem}

\newtcbtheorem[number within=section, use counter*=thm]{mydef}{Definition}%
{colback=yellow!5,colframe=yellow!70!black,fonttitle=\bfseries}{def}

\newtcbtheorem[number within=section, use counter*=thm]{mycor}{Corollary}%
{colback=blue!5,colframe=blue!70!black,fonttitle=\bfseries}{cor}

\newtcbtheorem[number within=section, use counter*=thm]{myprop}{Proposition}%
{colback=pink!5,colframe=pink!70!black,fonttitle=\bfseries}{prop}

\pgfplotsset{compat=1.17}

\begin{document}



\title{\texorpdfstring{Forward-Assisted Purification: \\
A Spatiotemporal Framework Beyond Conventional Limits}{Forward-Assisted Purification: A Spatiotemporal Framework Beyond Conventional Limits}}


\author{Fei Meng}
\thanks{These authors contributed equally}
\affiliation{School of Physics and Astronomy, University of Glasgow, Glasgow G12 8QQ, United Kingdom}

\author{Jinge Bao}
\thanks{These authors contributed equally}
\affiliation{School of Informatics, University of Edinburgh, Edinburgh EH8 9AB, United Kingdom}

\author{Yunlong Xiao}
\email{mathxiao123@gmail.com}
\affiliation{Institute of High Performance Computing (IHPC), Agency for Science, Technology and Research (A*STAR), 1 Fusionopolis Way, \#16-16 Connexis, Singapore 138632, Republic of Singapore}


\begin{abstract}
Noise remains the primary obstacle to realizing quantum advantage, continuously degrading the resources that enable quantum technologies. 
Purification aims to reverse this degradation by extracting high-fidelity resources from noisy ensembles, yet its conventional formulation is intrinsically static, acting only after noise has taken effect. 
Here we instead recast purification as a dynamical task, introducing a spatiotemporal framework that distributes interventions across the noise process.
This formulation reveals operational capabilities inaccessible to existing approaches and gives rise to forward-assisted purifications that extend achievable performance. 
In certain regimes, a single-copy protocol already exceeds what can be achieved with up to $50$ copies under conventional purification,  demonstrating a significant overhead in required resources.
Beyond these gains, our framework circumvents no-purification theorems within conventional protocols, including for Bell-state ensembles, thereby enabling purification previously considered impossible and pointing toward an efficient route to mitigating noise in quantum systems.
\end{abstract}

\maketitle


\noindent \textbf{Introduction}---Quantum technologies do not merely extend the reach of existing capabilities; they reshape the very limits of what can be achieved in computation~\cite{365700,10.1145/237814.237866,Arute2019,doi:10.1126/science.abe8770}, communication~\cite{PhysRevLett.70.1895,Bouwmeester1997,Kimble2008,Ren2017,RevModPhys.92.025002}, and sensing~\cite{doi:10.1126/science.1104149,PhysRevLett.96.010401,Giovannetti2011,RevModPhys.89.035002} by exploiting intrinsically quantum features such as superposition~\cite{RevModPhys.85.1103} and entanglement~\cite{RevModPhys.81.865}. 
This promise, however, is inseparable from a fundamental fragility. 
In any realistic setting, interactions with the environment are unavoidable; they introduce noise that degrades coherence, converting ideally pure states into mixed ones and thereby eroding the resources that give quantum protocols their power.
Confronting this challenge has driven the development of a spectrum of approaches, including quantum error correction~\cite{PhysRevA.52.R2493,PhysRevLett.81.2152,RevModPhys.87.307,Acharya2023,Bluvstein2024}, quantum error mitigation~\cite{PhysRevX.7.021050,PhysRevLett.119.180509,PhysRevX.8.031027,doi:10.7566/JPSJ.90.032001,RevModPhys.95.045005}, and quantum state purification~\cite{PhysRevLett.82.4344,Keyl2001,PhysRevA.70.032308,li2025optimalquantumpurityamplification,Childs2025streamingquantum,grier2025streamingquantumstatepurification,Yao_2025}.
Among them, purification occupies a distinctive position: by providing operational means to recover high-quality quantum resources from noisy data~\cite{PhysRevLett.76.722,Pan2001,Pan2003,Reichle2006}, it establishes a concrete bridge between the presence of noise and the restoration of functionality in quantum information processing~\cite{PhysRevA.59.169,Pan2001Entanglement,PhysRevLett.90.067901,PhysRevLett.126.010503,PhysRevLett.128.080504}.

\begin{figure}[t]
    \centering
    \includegraphics[width=\linewidth]{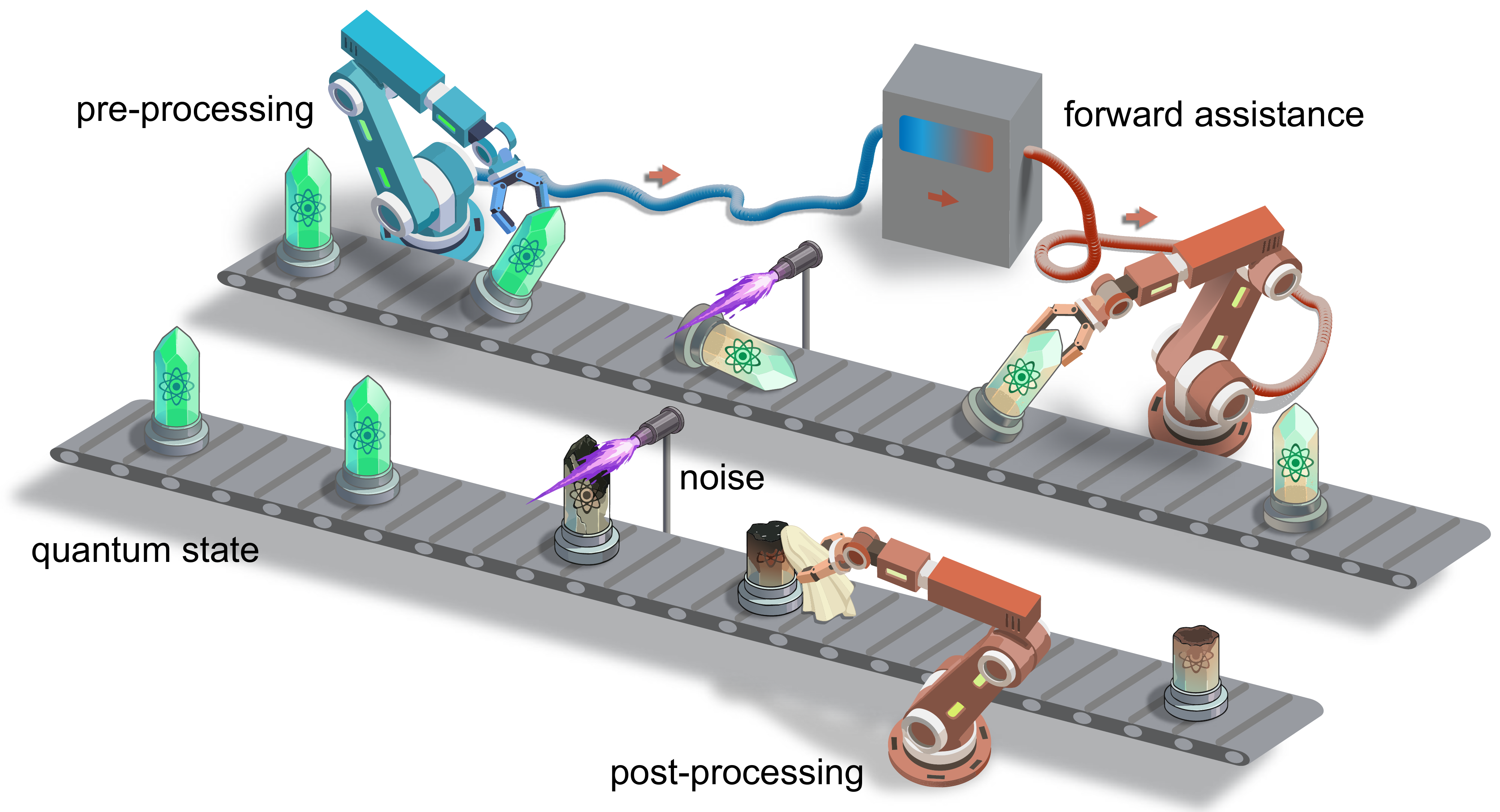}
    \caption{\textbf{Forward-Assisted versus Conventional Purification.} 
    The lower conveyor belt depicts conventional purification, where quantum states, depicted as crystals, are first degraded by environmental noise, represented by a purple flame, and are then recovered only through post-processing. 
    The upper conveyor belt illustrates the forward-assisted setting introduced here: 
    a pre-processing operation is applied before the noise acts, and a forward memory link connects this intervention to the later post-processing stage. 
    Treating noise as a dynamical process in this way enlarges the set of admissible purification strategies and gives access to recovery capabilities that are unavailable in conventional purifications.}
    \label{fig:MT_Sketch}
\end{figure}

Building on this role, the conventional formulation of purification considers the setting in which multiple noisy copies are available and subjects them to collective operations designed to partially reverse the action of noise, producing an output state closer to the original input. 
This paradigm offers a direct operational route from noisy data to restored functionality, yet its practical reach remains sharply constrained.
A central issue is whether existing protocols are genuinely optimal, or whether they reflect only a restricted perspective on what quantum mechanics permits.
Even when improvement is, in principle, achievable, it often comes at a prohibitive cost: in relevant noise regimes, marginal gains in fidelity require an overwhelming number of noisy copies, placing such approaches beyond current experimental capabilities. 
At the same time, the theoretical landscape has been progressively delineated by a series of no-go results~\cite{PhysRevLett.125.060405,PRXQuantum.3.010337,Regula2021,PhysRevLett.134.190803,bdw8-k91v,3bb1-pmtp}, establishing that under a wide range of physically natural conditions, purification is unattainable. 
These features point beyond technical obstacles to a more intrinsic limitation, suggesting that the conventional formulation itself constrains how purification can be realized in practice.

We address these challenges by introducing a spatiotemporal framework that treats noise as an intrinsically dynamical process.
Interpreting noise as a quantum channel promotes its most general manipulation to superchannels~\cite{Chiribella_2008,8678741,xiao2025superchanneltearsgeneralizedoccams}, encompassing pre- and post-processing, with quantum memory linking them, and in doing so reveals a spatiotemporal structure that lies beyond the reach of conventional approaches, as illustrated in Fig.~\ref{fig:MT_Sketch}. 
This perspective gives rise to forward-assisted purifications with operational advantages: 
they surpass standard strategies and, notably, even memory-free single-copy implementations can outperform conventional approaches that rely on many noisy samples.
In extreme cases, conventional protocols remain unable to match the performance of a single-copy forward-assisted scheme even with up to $50$ noisy copies, highlighting a substantial separation in sample complexity.
The framework further enables purification in regimes previously considered unattainable, including Bell states, thereby redefining the limits of quantum purification. 
To enable systematic benchmarking, we develop representation-theoretic algorithms based on Schur-Weyl duality and Clebsch-Gordan recursion that compress the underlying semidefinite programming, extending tractable analysis Hilbert space dimension $\approx2^{8}$ to scales approaching $2^{50}$.


\noindent \textbf{Conventional Purification}---Assume an ideal source prepares a desired pure state $\psi \in \mS$ for a given task, where $\mS$ denotes the set of admissible inputs.
In realistic settings, environmental decoherence and device imperfections transform $\psi$ into a noisy state $\mN(\psi)$ (Fig.~\ref{fig:MT_Purification_1}(a1)), thereby limiting its utility in subsequent applications.
To mitigate this degradation, purification protocols process multiple noisy copies, $\mathcal{N}(\psi)^{\otimes n}$, to recover a state that better approximates $\psi$.
Formally, such protocols are described by a quantum channel $\mE$, producing the output $\mE(\mN(\psi)^{\otimes n})$ (Fig.~\ref{fig:MT_Purification_1}(a2)). 
Taking $\psi$ to be drawn uniformly from $\mathcal{S}$, the optimal performance of purification, quantified by the average fidelity, is then given by
\begin{align}\label{eq:MT_Conventional_Purification}
    F_{\ast}\coloneqq
    \max_{\mE\; \mathrm{CPTP}} \quad 
    \frac{1}{|\mS|}\sum_{\psi\in\mS}
    \Tr[J^{\mE}\cdot \left(\left(\mN(\psi)^{\T}\right)^{\otimes n}\otimes\psi\right)
    ]
    .
\end{align}
Here $J^{\mE}$ stands for the Choi operator of $\mE$~\cite{JAMIOLKOWSKI1972275,CHOI1975285}, defined as $J^{\mE}\coloneqq(\id\otimes\mE)(\Gamma)$, where $\Gamma=\ketbra{\Gamma}{\Gamma}$ and $\ket{\Gamma}\coloneqq\sum_{i}\ket{ii}$ is the (unnormalized) maximally entangled state; system labels are omitted for simplicity.
In Eq.~\eqref{eq:MT_Conventional_Purification}, $\mN(\psi)^{\otimes n}$ is fed into the completely positive trace-preserving (CPTP) map $\mE$, with $\psi$ as the target output, and $\T$ denoting transposition on $\mN(\psi)$; the optimization is over all quantum channels.
As a baseline, retaining a single noisy copy $\mN(\psi)$ without processing yields average fidelity $1/|\mS|\sum_{\psi\in\mS}\Tr[\mN(\psi)\cdot\psi]$, and protocols that strictly surpass this benchmark are deemed efficient purification.

\begin{figure*}[ht]
    \centering   
    \includegraphics[width=1\textwidth]{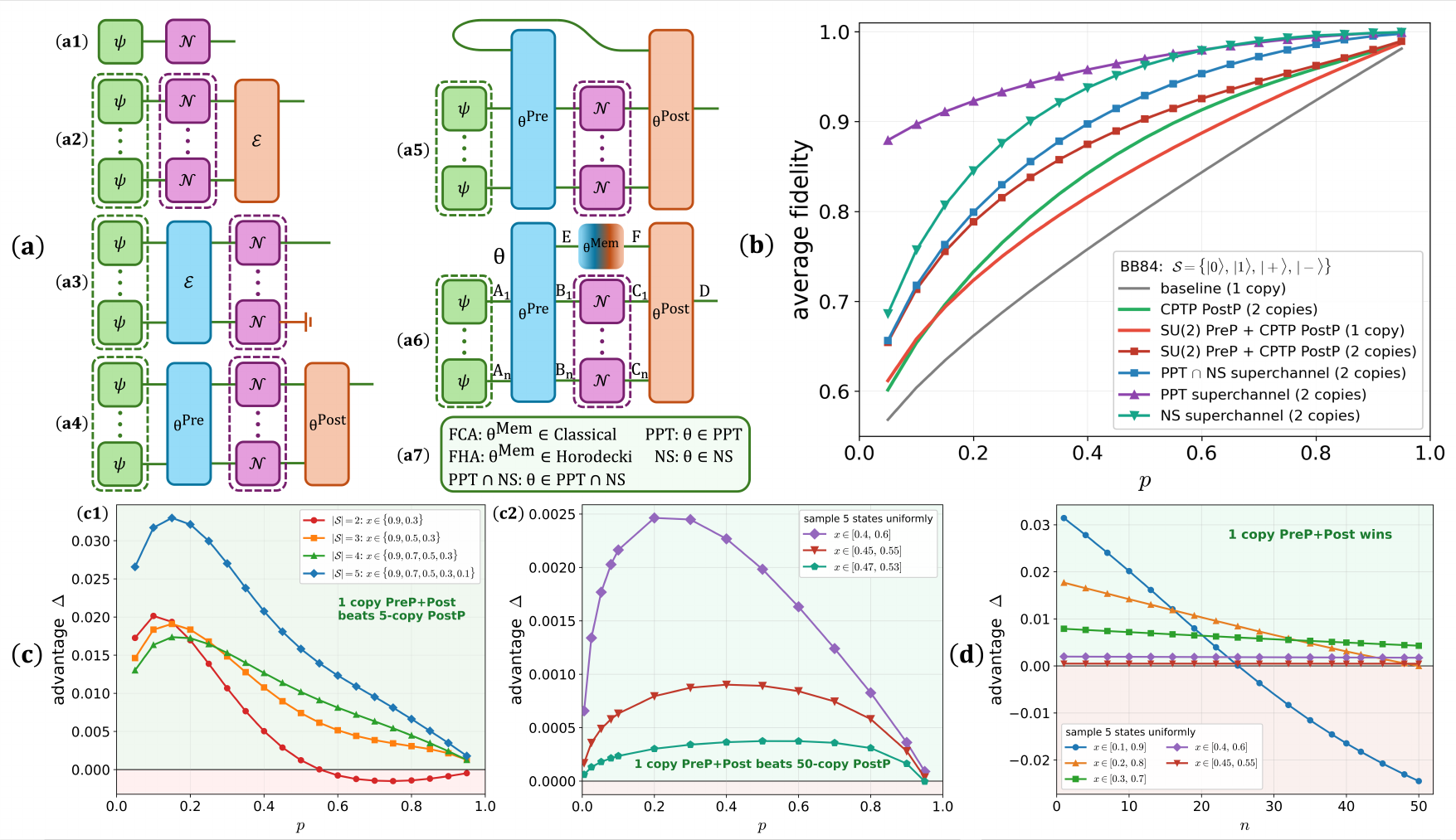}
    \caption{\textbf{Spatiotemporal Purifications}. 
        (a) Purification protocols considered in this work: 
        no purification (a1, Benchmark); 
        conventional post-processing-only purification (a2, PostP); 
        pre-processing-only purification (a3, PreP); 
        unassisted purification combining pre- and post-processing without memory (a4, UA); 
        entanglement-assisted purification with shared entanglement (a5, EA); 
        and forward-assisted purification, where pre-processing $\theta^{\mathrm{Pre}}$ and post-processing $\theta^{\mathrm{Post}}$ are connected by an intermediate memory $\theta^{\mathrm{Mem}}$ (a6, FA).
        Restricted FA classes are defined by the structure of the superchannel $\theta\coloneqq\theta^{\mathrm{Post}}\circ\theta^{\mathrm{Mem}}\circ\theta^{\mathrm{Pre}}$ and its memory $\theta^{\mathrm{Mem}}$, including forward-classical-assisted (FCA) and forward-Horodecki-assisted (FHA) protocols, as well as positive partial transposition (PPT), non-signalling (NS) and PPT $\cap$ NS purifications (a7).
        Across all panels, $\psi$ denotes pure inputs, $\mN$ the noise channel.
        (b) Achievable performance for 2-to-1 purification of BB84 states under amplitude damping noise $\mN=\mN_{\mathrm{AD}}(p)$.
        Forward-assisted protocols enlarge the attainable fidelity, revealing performance beyond conventional purification.
        Here, the pre-processing is restricted to local unitaries.
        (c) Sample advantage of single-copy UA purification. 
        For $\mS\coloneqq\{\sqrt{x}\ket{0}+\sqrt{1-x}\ket{1}\}_{x}$, we plot the fidelity gap $\Delta$ between one-copy UA (a4) and many-copy conventional purification (a2). 
        Panels (c1) and (c2) compare against 5-copy and 50-copy conventional PostP protocols, respectively; 
        green regions mark where one-copy UA achieves higher fidelity.
        (d) Scaling of the sample advantage. 
        At fixed noise strength $p=0.05$, we track $\Delta$ between one-copy UA and $n$-copy conventional protocol as $n$ increases. 
        Green regions show ensembles for which the single-copy UA protocol continues to outperform conventional purification, even when the latter is supplied with up to $50$ noisy copies.
        In (c) and (d), pre-processing uses simple local $y$-axis rotations.
    }
    \label{fig:MT_Purification_1}
\end{figure*}

Despite sustained progress, purification remains intrinsically constrained. 
Across broad noise regimes, even marginal fidelity gains typically require a large number of input copies, imposing an overhead that severely constrains practical implementation. 
More fundamentally, purification is ruled out for ensembles that underpin a wide range of quantum protocols.
A sequence of no-go results shows that, for ensembles such as the Bell states, no positive partial transpose preserving (PPTp) operations, even when probabilistic protocols with arbitrary success probability are permitted, can accomplish efficient distributed purification, thereby establishing a stringent no-purification theorem~\cite{3bb1-pmtp}.
These operational constraints are further compounded by a computational barrier: although the optimal performance $F_{\ast}$ admits a semidefinite programming (SDP) formulation in Eq.~\eqref{eq:MT_Conventional_Purification}, the computational complexity nevertheless grows exponentially with system size $n$.
For local dimension $d$, the associated Choi operator $J^{\mE}$ acts on a Hilbert space of dimension $N=d^{n+1}$; even in the qubit case, this scaling quickly renders the optimization intractable, with memory demands escalating and numerical evaluation already fails near $n\approx 8$, blocking access to larger system sizes.

The interplay between operational constraints and computational complexity brings into focus a set of questions of both foundational and practical importance.
Do the limitations encoded in Eq.~\eqref{eq:MT_Conventional_Purification} reflect intrinsic features of quantum theory, or do they instead delineate the boundaries of the frameworks explored to date? 
Can a broader operational paradigm reduce sample complexity, or even circumvent established no-go results?
And, at a technical level, can the underlying SDP be reformulated so as to remain tractable well beyond the few-copy regime?
Addressing these questions is essential for delineating the limits of quantum theory and for guiding the development of quantum technologies, where managing noise remains a central challenge.
In what follows, we resolve these issues within a unified framework.


\noindent \textbf{Spatiotemporal Framework}---Conventional protocols treat the noisy state as a static resource, confining purification to operations applied only after the noise and thereby masking its fundamentally dynamical nature, which is more naturally described at the level of a quantum channel. 
From this perspective, the most general manipulation of noise is given by a quantum superchannel $\theta$, comprising pre-processing $\theta^{\text{Pre}}$, post-processing $\theta^{\text{Post}}$, and a quantum memory $\theta^{\text{Mem}}$ that links these stages (Fig.~\ref{fig:MT_Purification_1}(a6)). 
Within this broader operational setting, conventional purification appears as a limiting case that allows only post-processing, leaving much of the available structure, especially pre-noise operations and temporal correlations, unused. 
We instead place purification in a fully spatiotemporal architecture, where interventions act both before and after the noise and remain coherently connected across time, so that purification becomes a transformation $\theta$ acting directly on the noise process itself, with output $\theta(\mN^{\otimes n})(\psi^{\otimes n})$.
Unrestricted superchannels, however, are not operationally meaningful: although mathematically admissible, they rely on idealized capabilities and, in the extreme, allow the noise to be bypassed, rendering the task trivial.
We therefore restrict attention to a physically grounded subset, leading to forward-assisted (FA) purifications inspired by quantum Shannon theory~\cite{10.1109/TIT.2015.2439953}, which retain the spatiotemporal structure while remaining experimentally accessible.

In the absence of quantum memory, FA purification reduces to unassisted (UA, Fig.~\ref{fig:MT_Purification_1}(a4)) protocols constructed solely from pre- and post-processing.
This class includes pre-processing-only (PreP, Fig.~\ref{fig:MT_Purification_1}(a3)) and post-processing-only (PostP, Fig.~\ref{fig:MT_Purification_1}(a2)) schemes, the latter coinciding with the conventional formulation. 
Allowing memory expands this structure, with the resulting capabilities governed by its nature: shared entanglement gives rise to entanglement-assisted (EA, Fig.~\ref{fig:MT_Purification_1}(a5)) purification, classical memory yields forward-classical-assisted (FCA) purification, and Horodecki memory defines forward-Horodecki-assisted (FHA) purification. 
Here, $\theta^{\mathrm{Mem}}$ is classical if $\Delta \circ\, \theta^{\mathrm{Mem}} \circ \Delta = \theta^{\mathrm{Mem}}$ where $\Delta$ denotes the completely dephasing map, and Horodecki if its Choi operator is positive under partial transpose (PPT).
This resource-theoretic perspective admits a further refinement at the level of the full spatiotemporal framework by imposing structural constraints on the superchannel $\theta$ (Fig.~\ref{fig:MT_Purification_1}(a7)).
In particular, non-signalling (NS) conditions on $J^{\theta}$, given by $\Tr_{D}[J^{\theta}]=J^{\theta}_{AB}\otimes\1_{C}/d_{C}$ and $\Tr_{B}[J^{\theta}]=J^{\theta}_{CD}\otimes\1_{A}/d_{A}$, define NS purification, while a PPT constraint on the post-processing systems, namely $(J^{\theta})^{\T_{CD}}\geqslant0$, gives PPT purification; their intersection yields PPT $\cap$ NS purification.
Throughout, labels without subscripts denote composite systems, for example $A\coloneqq A_{1}\ldots A_{n}$.
With these structures in place, we proceed to characterize the optimal performance of FA purifications.

\begin{lem}[\textbf{Fundamental Limits of Performance}]
\label{lem:MT_Fundamental_Limits_Performance}
For a forward-assisted purification protocol (Fig.~\ref{fig:MT_Purification_1}(a6)) implemented by a superchannel $\theta$ constrained by the property $\mP$, the optimal performance is characterized by
\begin{align}\label{eq:MT_FA_Fundamental_Limit}
    F_{\mP}
    \coloneqq
    \max \quad 
    & 
    \Tr[J^{\theta}\cdot \left( \Psi_{AD}\otimes\Omega_{BC}
    \right)]
    \\
    \text{s.t.} \quad 
    &J^{\theta}\geqslant0,\; 
    \Tr_{BD}[J^{\theta}]=\1_{AC}, \notag\\
    &\Tr_{D}[J^{\theta}]=
    J^{\theta}_{AB}\otimes\frac{1}{d_C}\1_{C},\;
    \theta\in\mP.\notag
\end{align}
Here, $\Psi_{AD} \coloneqq 1/|\mS|\sum\psi^{\otimes n}_{A}\otimes\psi^{\T}_{D}$ encodes the joint input-target structure induced by the ensemble $\mS$, while $\Omega_{BC} \coloneqq (J^{\mN})^{\otimes n}_{BC}$ denotes the Choi operator of the noise acting independently on each copy.
\end{lem}

The above formulation expresses purification performance under operational constraints $\mP$ as a unified variational characterization, making explicit how physical restrictions are encoded at the level of the superchannel and translate into limits on achievable fidelity. 
From this perspective, different constraints correspond to nested feasible sets of superchannels, thereby inducing an intrinsic ordering of performance. 
This ordering is formalized by the inequality hierarchy in the following theorem.

\begin{thm}[\textbf{Hierarchy of Purification Performance}]
\label{thm:MT_Hierarchy_Purification_Performance}
The fundamental limits of forward-assisted purification, as illustrated in Fig.~\ref{fig:MT_Purification_1}(a), are ordered in terms of average fidelity according to
\begin{align}\label{eq:MT_Hierarchy_NS}
    F_{\mathrm{PostP}} \leqslant 
    F_{\mathrm{UA}} \leqslant
    F_{\mathrm{EA}} \leqslant
    F_{\mathrm{NS}},
\end{align}
and 
\begin{align}\label{eq:MT_Hierarchy_PPT}
    F_{\mathrm{PostP}}\leqslant 
    F_{\mathrm{UA}} \leqslant
    F_{\mathrm{FCA}} \leqslant
    F_{\mathrm{FHA}} \leqslant
    F_{\mathrm{PPT}}.
\end{align}
Here, $\mathrm{PostP}$ denotes post-processing-only purification (Fig.~\ref{fig:MT_Purification_1}(a2)), corresponding to the conventional paradigm and serving as the baseline of the hierarchy. 
Each inclusion enlarges the admissible class of operations and leads to a monotonic increase in achievable fidelity. 
The same ordering holds upon replacing $\mathrm{PostP}$ with $\mathrm{PreP}$.
\end{thm}

Theorem~\ref{thm:MT_Hierarchy_Purification_Performance} establishes two inequality chains that characterize the performance of FA purifications, demonstrating the clear operational advantage of the spatiotemporal framework. 
In these chains, the conventional PostP scheme occupies the leftmost position; consequently, any protocol to its right achieves superior performance.
When constrained to PPT, NS, or their intersection (Fig.~\ref{fig:MT_Purification_1}(a7)), the corresponding optimal fidelity $F_{\mP}$ admits a SDP formulation.
Beyond these analytical bounds, numerical results further confirm that even the experimentally accessible UA class (Fig.~\ref{fig:MT_Purification_1}(a4)) surpasses the conventional paradigm in certain noise regimes, attaining higher fidelity with fewer samples.
Full derivations, together with a systematic comparison across protocols, are deferred to the Appendix.


\noindent \textbf{Numerical Experiments}---The primary challenge for numerical analysis lies in the computational complexity of SDP formulations of purification. 
Although these programs, such as Eq.~\eqref{eq:MT_Conventional_Purification}, scale polynomially in matrix size, this size itself grows as $\mO(d^{n})$ for $n$-to-1 purification on local dimension $d$, leading to an exponential bottleneck that renders direct optimization intractable beyond $n \approx 8$. 
We resolve this by exploiting symmetry: 
a representation-theoretic reduction based on Schur-Weyl duality and Clebsch-Gordan recursion decomposes the SDP into low-dimensional blocks, reducing time and memory costs to $\mO(|\mS|n^4)$ and $\mO(n^3)$, and extending the accessible regime to $n\approx 50$.
Within this tractable formulation, numerical experiments demonstrate two operational advantages of the spatiotemporal framework:
achievable performance is systematically elevated, forward-assisted protocols across UA, PPT $\cap$ NS, PPT, and NS classes surpass conventional limits, and these gains are accompanied by a genuine improvement in efficiency, reaching equal or higher fidelities with fewer noisy copies; 
even when restricted to the memoryless UA purification, single-copy implementations are shown to outperform conventional multi-copy strategies in certain regimes.

Spatiotemporal purification yields a performance hierarchy that is reflected in the inequality chains of Thm.~\ref{thm:MT_Hierarchy_Purification_Performance}, and is validated by the data in Fig.~\ref{fig:MT_Purification_1}(b).
Specifically, FA protocols consistently outperform conventional purification, pushing achievable fidelity beyond standard limits.
We demonstrate this in a 2-to-1 purification using BB84 ensemble $\{\ket{0}, \ket{1}, \ket{+}, \ket{-}\}$~\cite{BENNETT20147}, under amplitude damping noise $\mN=\mN_{\mathrm{AD}}(p)$. 
Within the UA class (Fig.~\ref{fig:MT_Purification_1}(a4)), simple local unitary pre-processing already provides a clear improvement over conventional schemes. 
Integrating quantum memory further amplifies the operational advantage, with PPT $\cap$ NS, PPT, and NS purifications offering additional gains. 
The lack of a universal ordering between PPT and NS purifications across all noise regimes reveals a fundamental incomparability of these operational classes.

\begin{figure*}[ht]
    \centering   
    \includegraphics[width=1\textwidth]{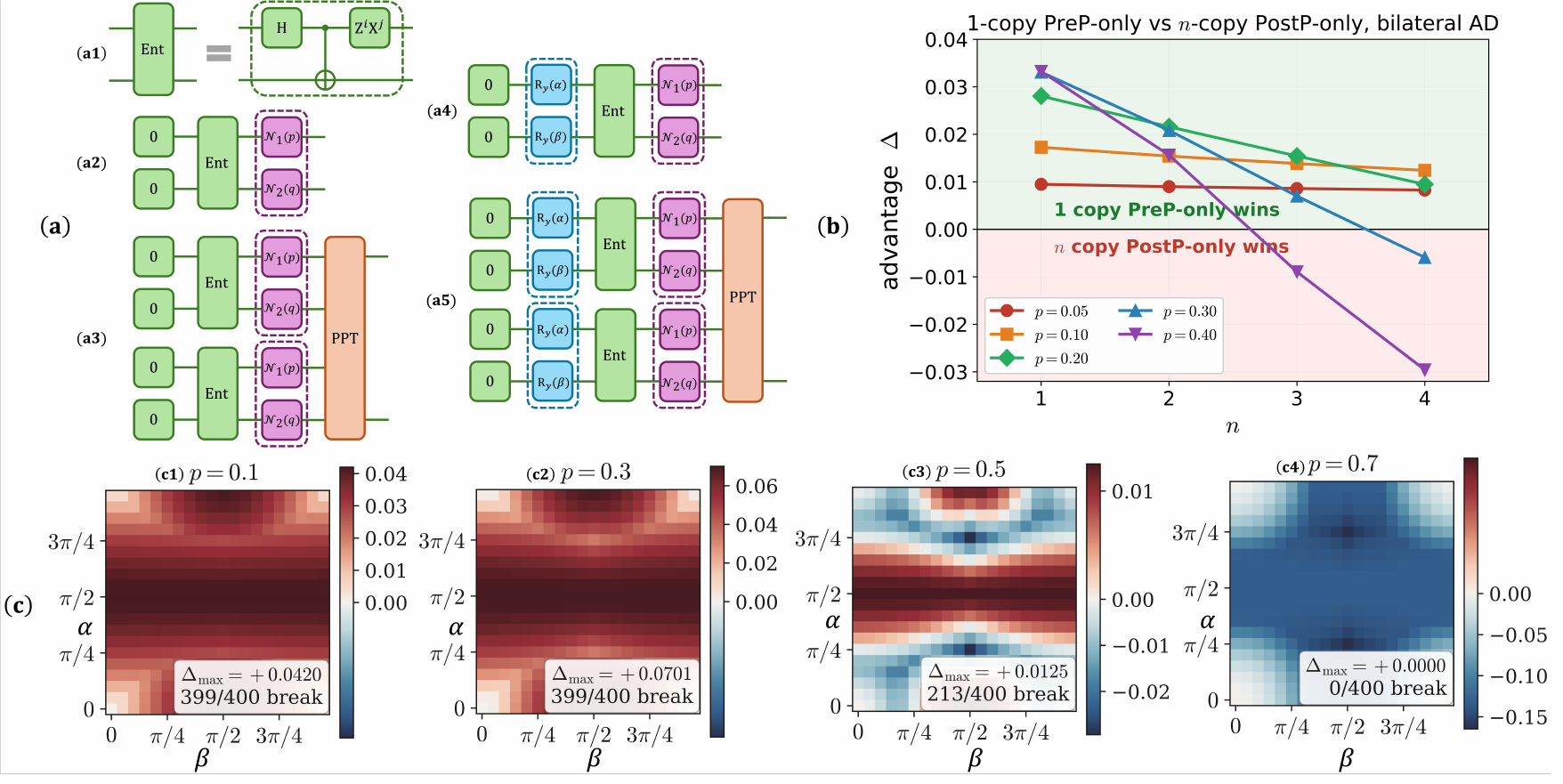}
    \caption{\textbf{Distributed Purifications}.
    (a) Distributed purification architectures: 
    an entangling gate that prepares a maximally entangled state (MES) (a1, Ent); 
    MES generation followed by local noise $\mN_1(p)\otimes\mN_2(q)$ (a2, Benchmark); 
    conventional 2-to-1 purification using PPTp post-processing (a3, PostP); 
    single-copy entanglement purification using local $y$-axis rotations $\mathrm{R}_{y}(\alpha)\otimes\mathrm{R}_{y}(\beta)$ (a4, PreP); 
    and 2-to-1 UA purification that combines the same local pre-processing with PPTp post-processing, without intermediate memory (a5, UA).
    (b) Single-copy advantage. 
    We plot the fidelity gap between the single-copy PreP in (a4) and the multi-copy conventional PostP in (a3), under identical amplitude damping noise, $\mN_1(p)=\mN_2(p)=\mN_{\mathrm{AD}}(p)$. 
    Green regions mark where single-copy PreP outperforms conventional multi-copy purification.
    (c) Distributed Bell states purification. 
    For Bell-state ensemble, recent no-purification results show that PostP cannot increase fidelity under local depolarizing noise, $\mN_{\mathrm{D}}(p)\otimes\mN_{\mathrm{D}}(q)$, even probabilistically.
    Adding the local $y$-axis rotation pre-processing layer in (a5) bypasses this obstruction and enables an efficient purification. 
    We focus on the symmetric noise setting $p=q$ with success probability 0.1. 
    The fidelity gap relative to the conventional purification in (a3) highlights the operational advantage unlocked by spatiotemporal framework.
    }
    \label{fig:MT_Purification_2}
\end{figure*}

Conventional purification is often limited in practice by the large number of noisy samples it requires, raising a natural question: 
can the spatiotemporal framework reduce this cost?
The answer is affirmative. 
Across a broad range of noise parameters and input ensembles, single-copy UA protocols already outperform multi-copy conventional approaches, as illustrated in Fig.~\ref{fig:MT_Purification_1}(c).
For ensembles $\mS\coloneqq\{\sqrt{x}\ket{0}+\sqrt{1-x}\ket{1}\}_{x}$ containing two to five distinct states, single-copy UA purification with local $y$-axis rotations $\mathrm{R}_{y}(\alpha)$ as PreP surpasses conventional PostP supplied with 5 noisy copies (Fig.~\ref{fig:MT_Purification_1}(c1)).
In certain regimes, this advantage persists even when the conventional scheme is given up to 50 noisy copies (Fig.~\ref{fig:MT_Purification_1}(c2)).

Figure~\ref{fig:MT_Purification_1}(d) further tracks this separation as the sample number $n$ increases, showing the fidelity gap between single-copy UA purification and $n$-copy conventional purification under amplitude damping noise with $p=0.05$.
Thus, particularly at small $p$, a simple pre-processing layer can replace tens of noisy copies in conventional purification, revealing a substantially more sample-efficient route to quantum information recovery.
Additional numerical results and technical details are provided in the Appendix, including extrapolations suggesting that conventional purification may require on the order of $10^3$ noisy copies to match the fidelity achieved by single-copy UA protocols, highlighting its practical advantage.


\noindent \textbf{Distributed Purification}---Large-scale quantum systems are intrinsically fragile, making monolithic architectures impractical. 
A more feasible approach is to interconnect smaller quantum processors and distribute tasks across them, establishing the distributed paradigm as the natural framework for quantum information processing under realistic conditions. 
In this setting, purification is performed between remote nodes using only local operations and classical communication (LOCC). 
Protocols defined under these constraints are known as distributed purification, with entanglement distillation as a representative example.
Extending the spatiotemporal framework to this setting reveals advantages over conventional approaches: 
forward-assisted protocols achieve performances beyond the limits of post-processing-only architectures, the inclusion of pre-processing improves efficiency by reducing the number of noisy samples required, and, most importantly, the framework reshapes the underlying physical constraints, enabling purification in regimes that remain inaccessible within conventional formulations.


Analogous to Thm.~\ref{thm:MT_Hierarchy_Purification_Performance}, a corresponding performance hierarchy emerges for distributed purification. 
We defer the formal analysis to the Appendix and focus here on the most experimentally accessible regime: protocols based solely on PreP. 
To isolate the performance gain, we consider a $n$-to-1 purification targeting maximally entangled states under independent amplitude damping noise $\mN_{\mathrm{AD}}(p)\otimes\mN_{\mathrm{AD}}(q)$. 
In the conventional approach, purification is applied after the noise via LOCC; 
to avoid its complexity in both characterization and optimization, we instead relax post-processing to PPTp operations (Fig.~\ref{fig:MT_Purification_2}(a3)). 
From a dynamical-resource-theoretic perspective~\cite{bauml2019resourcetheoryentanglementbipartite,PhysRevLett.125.180505,xing2023fundamentallimitationscommunicationquantum,glc7-xy8t}, the noise can be viewed as a process that degrades an ideal entangling operation into a noisy entangling gate, which naturally motivates the introduction of pre-processing prior to the noisy gate (Fig.~\ref{fig:MT_Purification_2}(a4)). 
Restricting to local $y$-axis rotations $\mathrm{R}_{y}(\alpha)\otimes\mathrm{R}_{y}(\beta)$, even a single-copy protocol surpasses conventional schemes that require multiple copies (Fig.~\ref{fig:MT_Purification_2}(b)), achieving both higher fidelity and reduced sample consumption.
Additional results, including systematic multi-copy comparisons, are provided in the Appendix.

The no-go theorem for distributed purification is revisited to assess whether it reflects a fundamental limitation or can be overcome within the spatiotemporal framework. 
To this end, we recall the no-purification theorem for Bell states in the conventional setting~\cite{3bb1-pmtp}, which states that

\begin{lem}[\textbf{No-Purification for Bell States~\cite{3bb1-pmtp}}]
\label{lem:MT_No_Purification_Bell}
    For Bell states subjected to local depolarizing noise $\mN_{\mathrm{D}}(p)\otimes\mN_{\mathrm{D}}(q)$, no efficient 2-to-1 purification protocol exists under $\mathrm{PPTp}$ post-processing, even when an arbitrary success probability is allowed.
\end{lem}

Restricting to the symmetric regime $p=q$ simplifies the analysis while retaining the essential features.
A pre-processing stage is introduced, consisting of local $y$-axis rotations $\mathrm{R}_{y}(\alpha)\otimes\mathrm{R}_{y}(\beta)$ applied prior to the noisy entangling operation. 
When combined with probabilistic PPTp post-processing, this sequence defines a UA purification (Fig.~\ref{fig:MT_Purification_2}(a5)). 
Varying the pre-processing parameters $\alpha$ and $\beta$ enables a direct comparison with the no-purification baseline (Fig.~\ref{fig:MT_Purification_2}(a2)) in terms of average fidelity. 
Figure~\ref{fig:MT_Purification_2}(c) shows a clear advantage at success probability 0.1, demonstrating that genuine purification, forbidden in the conventional setting, becomes attainable once spatiotemporal structure is incorporated.
This leads directly to the following theorem.

\begin{thm}[\textbf{Efficient FA Purification for Bell States}]
\label{thm:MT_Efficient_Purification_Bell}
Bell states subjected to local depolarizing noise $\mN_{\mathrm{D}}(p)\otimes\mN_{\mathrm{D}}(p)$ can be efficiently purified in the 2-to-1 setting when forward-assisted strategies are employed, as demonstrated by UA purifications in Fig.~\ref{fig:MT_Purification_2}(a5).
\end{thm}

Physical insight into Thm.~\ref{thm:MT_Efficient_Purification_Bell} rests on two key aspects. 
First, the pre-processing is resource-theoretically free~\cite{RevModPhys.91.025001}: 
it consists only of local rotations and therefore introduces no additional resource relative to conventional purification, while nevertheless achieving a strict performance gain. 
Second, its function is not to directly enhance fidelity, but to alter the symmetry inherent to Bell states. 
Breaking these symmetries exposes otherwise inaccessible operational degrees of freedom that subsequent post-processing can exploit. 
As a result, efficient purification becomes achievable in the deterministic regime under the spatiotemporal framework.
The same mechanism also applies to deterministic purification: 
the introduced pre-processing layer continues to bypass the no-purification theorem, with further details provided in the Appendix.


\noindent \textbf{Discussions}---Dealing with noise remains a central challenge in quantum information processing and a key step toward fault-tolerant quantum computing. 
When multiple noisy copies of a state are available, purification protocols seek to recover a state that more closely approximates the original input prior to the action of noise. 
Conventional approaches, however, adopt a static viewpoint, treating noisy states as resources and applying purification only after the noise has acted. 
Here, we develop instead a dynamical perspective, constructing a spatiotemporal framework that captures the inherently dynamical nature of noise, giving rise to forward-assisted purification. 
This framework extends achievable performance beyond conventional limits, enables comparable, or even superior, fidelity to be attained with fewer noisy copies, and renders efficient purification feasible in regimes inaccessible to conventional approaches.
More fundamentally, these results indicate that the commonly perceived limitations of purification do not reflect fundamental constraints of quantum mechanics, but rather stem from the static perspective imposed on the problem. 
Once the spatiotemporal framework is incorporated, new operational pathways emerge, redefining the limits of quantum purification.

Beyond the results established here, several directions remain open and merit further investigation.
Benchmarking against conventional protocols that operate on multiple noisy copies is achieved through algorithms grounded in representation-theoretic methods, which enable substantial reductions of the underlying SDPs in the global setting; however, these techniques do not transfer directly to distributed scenarios, where the analysis therefore remains confined to the few-copy regime, and extending them will require more refined optimization together with symmetry-adapted constructions. 
At a conceptual level, although purification and quantum error correction both confront noise, they are separated by a fundamental structural distinction: quantum error correction is intrinsically built around encoding, whereas such pre-processing is absent in conventional purification. 
This asymmetry not only motivates the systematic incorporation of pre-processing into purification but also reveals a unifying perspective in which quantum error correction emerges naturally as a limiting case of our spatiotemporal framework. 
Viewed from this standpoint, the framework does more than reconcile two established paradigms; it delineates a broader operational landscape in which spatiotemporal structure plays a central role, pointing toward new forms of quantum error corrections, particularly dynamical codes, that remain to be systematically developed.


\noindent \textbf{Algorithmic Framework}---This section gives a compact account of the algorithmic ideas that make the purification semidefinite programming computationally accessible beyond the few-copy regime. 
By exploiting the representation-theoretic structure of the problem, our approach replaces the direct formulation, already limited to roughly eight copies, with a matrix dimension $2^8$, by a symmetry-adapted formulation that reaches fifty copies, corresponding to a dimension of $2^{50}$. 
The workflow is sketched in Fig.~\ref{fig:MT_Complexity}.
Full technical details, together with the general framework of forward-assisted purification, are given in the Supplementary Information.

\begin{figure}[t]
    \centering
    \includegraphics[width=\linewidth]{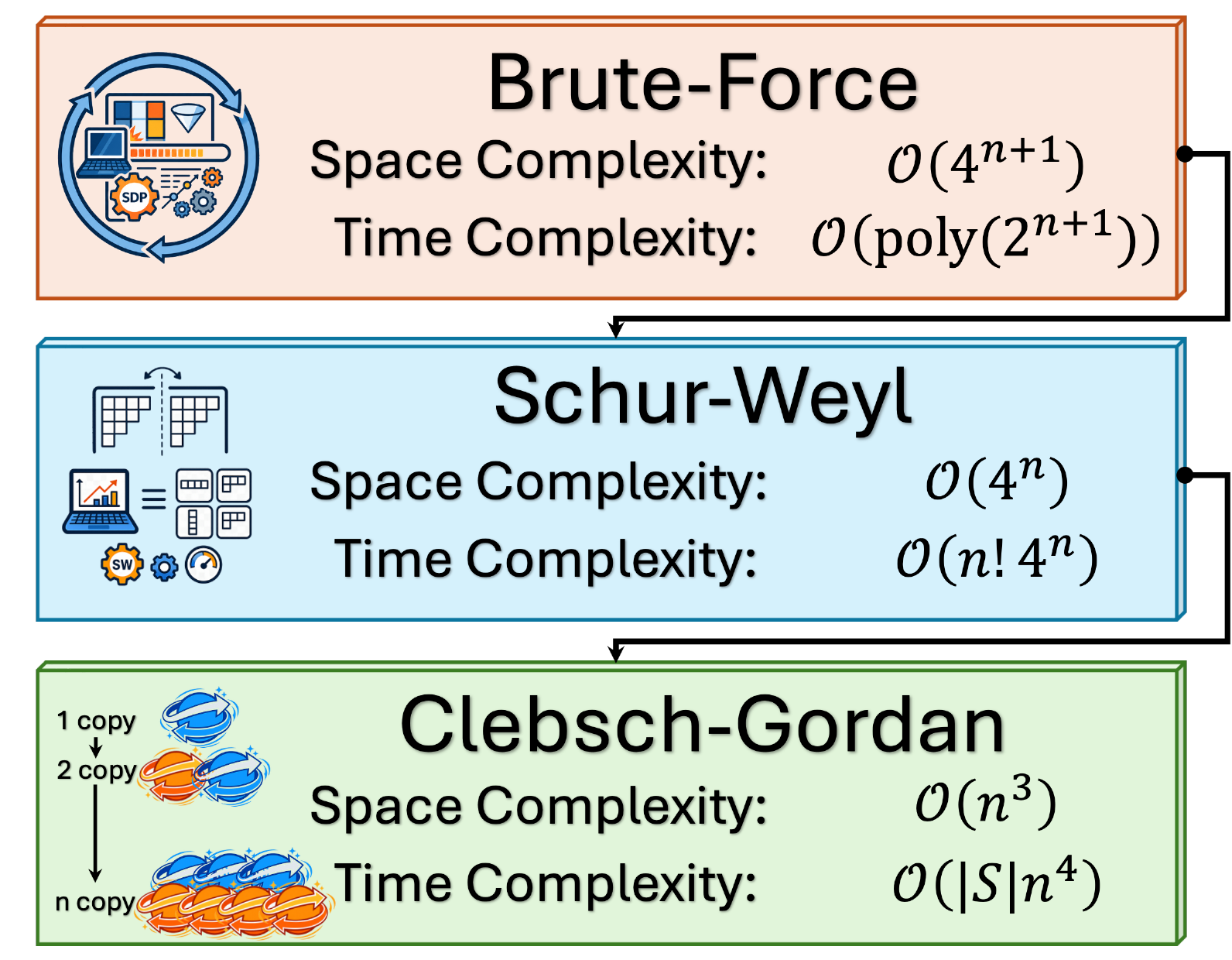}
    \caption{\textbf{Reducing Computational Complexity through Symmetry.} 
    A direct SDP formulation of conventional $n$-to-1 purification (see Eq.~\eqref{eq:Method_Conventional_Purification}) incurs exponential memory and time costs, making brute-force computation rapidly intractable as the number of input copies grows (orange). 
    The first layer of our algorithm exploits Schur-Weyl duality to resolve the permutation symmetry and decompose the problem into smaller representation-theoretic blocks (see Eq.~\eqref{eq:Method_SW_State}), reducing the memory footprint but still requiring an expensive block-construction step (blue). 
    The second layer removes this bottleneck by introducing a Clebsch-Gordan recursion, which constructs the blocks iteratively (see Eq.~\eqref{eq:Method_CG_Recursion}) from one copy to $n$ copies without enumerating the symmetric group or forming the full $2^n$-dimensional operators (green). 
    This two-layer representation-theoretic algorithm yields an exponential acceleration of the block construction and pushes the conventional purification benchmark to approximately 50 copies.}
    \label{fig:MT_Complexity}
\end{figure}


\noindent \textbf{Fundamental Limits of Conventional Purification.} 
We benchmark conventional $n$-to-1 purification, in which recovery is restricted to post-processing (see Fig.~\ref{fig:MT_Sketch}), by optimizing over the Choi operator $J^{\mE}_{\mathrm{i}\mathrm{o}}$ of a completely positive trace-preserving (CPTP) map $\mE$ acting on $n$ noisy input copies. 
For an input ensemble $\mS$ and a noise channel $\mN$, the optimization maximizes the average fidelity (see Eq.~\eqref{eq:MT_Conventional_Purification})
\begin{align}\label{eq:Method_Conventional_Purification}
    F_{\ast}=
    \max \quad 
    & 
    \Tr[J^{\mE}_{\mathrm{i}\mathrm{o}}\cdot 
    \left(\frac{1}{|\mS|}\sum_{\psi\in\mS}\left(\mN(\psi)^{\T}\right)^{\otimes n}_{\mathrm{i}}
    \otimes\psi_{\mathrm{o}}\right)
    ] \nonumber
    \\
    \text{s.t.} \quad 
    &J^{\mE}_{\mathrm{i}\mathrm{o}}\geqslant0,\,\, \Tr_{\mathrm{o}}[J^{\theta^{\mathrm{Post}}}_{\mathrm{i}\mathrm{o}}]=\1_{\mathrm{i}}.    
\end{align}
between the purified output $\mE(\mN(\psi)^{\otimes n})$ on system $\mathrm{o}$ and the corresponding ideal target state $\psi$, where $\mathrm{i}$ denotes the $n$-copy input system. 
A direct solution is memory-intensive, because the Choi operator acts on a Hilbert space of dimension $d^{n+1}$, with $d$ the single-copy Hilbert-space dimension. 
To access the many-copy regime, we exploit the permutation symmetry of $\mN(\psi)^{\otimes n}$.


\noindent \textbf{Schur-Weyl Duality.}
Since the $n$ input copies are prepared independently and identically, operator $\mN(\psi)^{\otimes n}$ commutes with the action of the symmetric group $\mathfrak{S}_{n}$.
We exploit this symmetry through Schur-Weyl duality, which gives
\begin{align}\label{eq:Method_SW}
    (\mathbb{C}^{d})^{\otimes n}
    \simeq\bigoplus_{\lambda\vdash n,\,\ell(\lambda)\leqslant d}
    V_{\lambda}\otimes W_{\lambda}.
\end{align}
Here $\lambda$ labels a Young diagram with at most $d$ rows, $V_{\lambda}$ is an irreducible representation of $\mathrm{SU}(d)$, and $W_{\lambda}$ is an irreducible representation of $\mathfrak{S}_{n}$. 

It follows that any permutation-invariant operator, including the tensor power of each noisy input state, decomposes as
\begin{align}\label{eq:Method_SW_State}
    \left(\mN(\psi)^{\T}\right)^{\otimes n} = \bigoplus_{\lambda} \left(\mN(\psi)^{\T}\right)^{(n)}_{\lambda}\otimes \1_{W_{\lambda}},
\end{align}
where $\left(\mN(\psi)^{\T}\right)^{(n)}_{\lambda}$ acts on the $\mathrm{SU}(d)$ irreducible sector $V_{\lambda}$, and $\1_{W_{\lambda}}$ acts on the corresponding multiplicity space $W_{\lambda}$.
The same symmetry can be imposed on the Choi operator $J^{\mE}$ without loss of optimality: averaging any feasible Choi operator over permutations preserves complete positivity, trace preservation, and the objective value. 

The original semidefinite programming of $F_{\ast}$ therefore decomposes into independent symmetry sectors,
\begin{align}
    F_{\lambda}=
    \max \quad 
    & 
    \Tr\left[J_{\lambda}\Xi_{\lambda}\right] \nonumber
    \\
    \text{s.t.} \quad 
    &J_{\lambda}\geqslant0,\,\, \Tr_{\mathrm{o}}[J_{\lambda}]=\1_{V_{\lambda}},    
\end{align}
with sector objective
\begin{align}
    \Xi_{\lambda} 
    \coloneqq 
    \frac{1}{|\mathcal{S}|} \sum_{\psi\in\mS} \left(\mN(\psi)^{\T}\right)^{(n)}_{\lambda} \otimes 
    \psi_{\mathrm{o}}.
\end{align}
The post-processing fidelity is then recovered by summing over sectors with the appropriate symmetric-group multiplicities,
\begin{align}
    F_{\ast} = \sum_{\lambda} d_{\lambda}F_{\lambda},
\end{align}
with $d_{\lambda}\coloneqq\dim W_{\lambda}$.
Such a symmetry-adapted formulation replaces a semidefinite programming on a space of dimension $d^{n+1}$ by a collection of much smaller block problems, whose largest representation-theoretic sector grows only polynomially, of order $\mO(n^{d-1})$, for fixed local dimension $d$.

This Schur-Weyl reduction is the first layer of our algorithm. 
It reveals the symmetry sectors of the purification problem and removes the redundant multiplicity degrees of freedom carried by the symmetric-group representations. 
The full procedure then works directly within these sectors: instead of constructing exponentially large tensor-power operators and subsequently block diagonalizing them, it builds the required blocks recursively from the representation-theoretic structure itself. 
This sector construction is the key step that avoids the original tensor-product basis and makes the many-copy computation feasible.


\noindent \textbf{Character Formula Construction.}
For a general local dimension $d$, we first constructed the reduced blocks through the projector $\Pi_{\lambda}$ onto each sector,
\begin{align}
    \Pi_{\lambda} = \frac{d_{\lambda}}{n!} \sum_{\pi \in \mathfrak{S}_{n}} \chi_{\lambda}(\pi^{-1})P_{\pi},
\end{align}
where $P_{\pi}$ permutes the $n$ tensor factors and $\chi_{\lambda}$ is the character of the irreducible $\mathfrak{S}_{n}$ representation $W_{\lambda}$.
Diagonalizing $\Pi_{\lambda}$ identifies the $\lambda$-isotypic subspace. 
A canonical copy of $V_{\lambda}$ was then selected using a Young symmetrizer, and the resulting basis was orthonormalized by QR decomposition or singular value decomposition.
This produced an isometry $\Phi_{\lambda}$ into the chosen representation block, from which the reduced inputs $\left(\mN(\psi)^{\T}\right)^{(n)}_{\lambda}$ (see Eq.~\ref{eq:Method_SW_State}) are obtained as
\begin{align}\label{eq:Method_Phi}
    \left(\mN(\psi)^{\T}\right)^{(n)}_{\lambda} = \Phi_{\lambda}^{\dagger} \left(\mN(\psi)^{\T}\right)^{\otimes n} \Phi_{\lambda}.
\end{align}
The block SDPs were solved independently for each $\lambda$, with the final objective obtained by summing the sector contributions weighted by the multiplicities $d_{\lambda}$.
This character formula construction provides a direct symmetry-adapted implementation of the Schur-Weyl duality. 
It removes the redundant symmetric-group degrees of freedom and replaces the original SDP by a collection of smaller sectors. 
Its bottleneck, however, lies in the block construction itself: 
explicitly summing over permutations introduces a factorial construction cost, scaling as $\mO(n!d^{2n})$ in time and $\mO(d^{2n})$ in memory.
Thus, while Schur-Weyl duality reveals the correct reduced structure, the character formula construction remains computationally prohibitive in the many-copy regime. 
This motivates the recursive qubit construction introduced below, which builds the reduced blocks directly in the symmetry-adapted basis and avoids forming the exponentially large tensor product operators.


\noindent \textbf{Clebsch-Gordan Recursion.}
As discussed in the previous section, the first layer of our algorithm, based on Schur-Weyl duality, exposes the symmetry structure of the problem but does not by itself remove the computational bottleneck in many-copy purification: 
a direct character-formula implementation still requires explicit construction of large tensor product operators and has prohibitive scaling. 
We now introduce the second layer of the algorithm, a Clebsch-Gordan recursion specialized to the qubit case. 
The recursion constructs the blocks iteratively, without enumerating permutations or forming the full $2^n$-dimensional operators. 
This recursive algorithm is the core technical ingredient that makes the many-copy SDP computationally accessible.

In particular, for qubit systems, i.e., $d=2$, the $\lambda$-sectors (see Eq.~\ref{eq:Method_SW}) can be labelled equivalently by the total spin $j$. 
We therefore replace the explicit construction of $\Phi_{\lambda}$ (see Eq.~\ref{eq:Method_Phi}) with a Clebsch-Gordan recursion that builds the reduced spin-sector blocks.
Let $R^{(n)}_{j}$ denote the matrix representation of the reduced block $\left(\mN(\psi)^{\T}\right)^{(n)}_{j}$ associated with the spin $j$-sector of $n$-to-1 purification.
The recursion starts from the single qubit case, $n=1$, where the only admissible sector is $j=1/2$.
The corresponding block is given by
\begin{align}
    R^{(1)}_{\tfrac{1}{2}}[m_1,m_2] 
    = 
    \mN(\psi)^{\T}[s_1,s_2],
\end{align}
where $m_1, m_2 \in \{-1/2, +1/2\}$ are the magnetic quantum numbers indexing the rows and columns of the spin $j$-sector block, respectively. 
Here, $\mN(\psi)^{\T}$ is the single qubit input, and $s_1, s_2 \in \{0, 1\}$ are the corresponding row and column indices in the computational basis, such that $\mN(\psi)^{\T}[s_1,s_2] = \langle s_1|\mN(\psi)^{\T}|s_2\rangle$. 
The linear shift $s_k = 1/2 - m_k$ provides a direct mapping from the physical magnetic quantum numbers to the computational basis states.

To integrate the $n$-th qubit into the system, we recursively construct the $n$-qubit coupled block $R^{(n)}_{j'}$ from an admissible parent block $R^{(n-1)}_{j}$. By the standard rules of angular momentum addition, the parent spin must satisfy $j \in \{j'-1/2, j'+1/2\}$. 
Crucially, because the global state $\left(\mN(\psi)^{\T}\right)^{\otimes n}$ is permutation-invariant, Schur's Lemma guarantees that the reduced matrix $R^{(n)}_{j'}$ is perfectly degenerate across the multiplicity space $W_{j'}$. 
Consequently, all copies of the spin-$j'$ multiplet -- regardless of whether they descend from the parent spin $j = j'-1/2$ or $j = j'+1/2$ -- yield the strictly identical matrix $R^{(n)}_{j'}$. We therefore do not sum over the parent spin $j$. Instead, we arbitrarily select any single valid parent $j$ (e.g., $j=j'+1/2$ if $j'=0$, and for other values of $j'$ we can choose either $j = j'-1/2$ or $j = j'+1/2$) and evaluate the recursion.

Incorporating the physical constraint of angular momentum conservation, any new magnetic nmuber $m_k$ can be written from the old magnetic number $m'_k$ and added spin value $s_k$, i.e.,
\begin{align}
    m_k = m'_k - \frac{1}{2} + s_k \, .
\end{align}
the recursion relation for the target matrix element at row $m_1'$ and column $m_2'$, routed through a chosen parent $j$, is given by the recursion relation
\begin{widetext}
\begin{align}\label{eq:Method_CG_Recursion}
    R^{(n)}_{j'}[m_1',m_2'] 
    = 
    \sum_{s_1,s_2\in\{0,1\}} 
    c^{j'\leftarrow j}_{s_1}(m_1') \, c^{j'\leftarrow j}_{s_2}(m_2') \,\mN(\psi)^{\T}[s_1,s_2] 
    \times 
    R^{(n-1)}_{j}\left[m_1'-\frac{1}{2}+s_1, \, m_2'-\frac{1}{2}+s_2\right].
\end{align}
\end{widetext}
Here $c^{j'\leftarrow j}_{s}(m')$ is the Clebsch-Gordan (CG) coefficient for coupling a parent parent spin-$j$ sector with an additional spin-$1/2$ degree of freedom to the total-spin sector $j'$ with magnetic quantum number $m'$. 
These coefficients encode the angular momentum entering the recursive change of basis.
The summation over $s_1$ and $s_2$ essentially traces over the single-qubit degrees of freedom weighted by these CG amplitudes. 
Contributions requiring parent magnetic quantum numbers outside the allowed range
\begin{align}
    \left|m_k' - \frac{1}{2} + s_k\right|> j,
\end{align}
are set to zero, since they lie outside the spin-$j$ representation and therefore do not correspond to admissible angular momentum states.

Although Schur's lemma ensures that any admissible parent $j$-sector gives the same reduced block, boundary sectors restrict the available choices. 
For example, when the target sector is the singlet sector, i.e., $j'=0$, the parent spin is necessarily $j=1/2$. 
To make the recursion well defined for all spin sectors, we therefore record the CG coefficients for both coupling branches, using the standard Condon-Shortley phase convention~\cite{condon1935theory}.
For the branch in which the parent spin is smaller, $j=j'-1/2$, the coefficients are given by
\begin{align}
    c^{j'\leftarrow j}_{0}(m') &= \sqrt{\frac{j'+m'}{2j'}}, 
\end{align}
and
\begin{align}
    c^{j'\leftarrow j}_{1}(m') &= \sqrt{\frac{j'-m'}{2j'}}.
\end{align}
For the complementary branch, in which the parent spin is larger, $j=j'+1/2$, the corresponding coefficients read
\begin{align}
    c^{j'\leftarrow j}_{0}(m') 
    = 
    -\sqrt{\frac{j'-m'+1}{2j'+2}}, 
\end{align}
and
\begin{align}
    c^{j'\leftarrow j}_{1}(m') 
    = 
    \sqrt{\frac{j'+m'+1}{2j'+2}}.
\end{align}

With the recursion specified, the algorithm proceeds directly at the level of spin sectors. 
For each pure input state $\psi \in \mS$, we first compute the corresponding noisy single-copy state $\mN(\psi)$.
The Clebsch-Gordan recursion is then applied iteratively to $\mN(\psi)^{\T}$, generating the full family of $n$-copy spin-sector blocks $\left(\mN(\psi)^{\T}\right)^{(n)}_{j}$.
For each sector $j$, the ensemble-averaged objective operator is constructed by combining these blocks with the corresponding target output states
\begin{align}
    \Xi_j \coloneqq \frac{1}{|\mS|} \sum_{\psi\in\mS} \left(\mN(\psi)^{\T}\right)^{(n)}_{j}
    \otimes 
    \psi_{\mathrm{o}}.
\end{align}

Since the optimization decouples across spin sectors, the SDP characterizing the purification limit for conventional approaches decomposes into independent sectors. 
For each $j$-sector, the optimal contribution is obtained from
\begin{align}
    F_{j}=
    \max \quad 
    & 
    \Tr[J_j\Xi_j] \nonumber
    \\
    \text{s.t.} \quad 
    &J_j\geqslant0,\,\, \Tr_{\mathrm{o}}[J_j]=\1_{V_j},    
\end{align}
where $V_j$ denotes the irreducible representation of $\mathrm{SU}(2)$ associated with total spin $j$.
The fidelity $F_{\ast}$ (see Eq.~\eqref{eq:MT_Conventional_Purification}) is then recovered by summing the sector optima with their multiplicities
\begin{align}
    F_{\ast} = \sum_j d_j F_j,
\end{align}
where $d_j$ is the multiplicity of the spin-$j$ sector.
This completes the second layer of the algorithm, enabling the many-copy purification limit to be evaluated recursively from the spin sectors.


\noindent \textbf{Exponential Speedup.}
The Clebsch-Gordan recursion bypasses the explicit enumeration of the $n!$ elements of $\mathfrak{S}_{n}$ and avoids forming the full $2^n$-dimensional operators.
For each input state, the block construction scales as $\mO(n^4)$ in time and $\mO(n^3)$ in memory, giving a total cost of $\mO(|\mS|n^4)$ for an ensemble $\mS$.
The resulting SDPs decouple across the admissible spin sectors and can therefore be solved independently, allowing a simple parallel implementation. 
This recursive construction is the computational engine behind the many-copy conventional purification benchmarks in this work, making regimes with tens of input copies accessible well beyond the reach of the unreduced SDP formulation.

Although developed here for quantum purification, the approach is more general.
It applies broadly to SDPs in quantum information processing that involve $n$ identically prepared states, where permutation symmetry creates large redundant degrees of freedom. 
By combining symmetry reduction with recursive block construction, the method offers a compact route to many-copy optimization problems, extending the scale at which operational questions in quantum theory can be explored.


\noindent \textbf{Acknowledgments}---This research is supported by A*STAR under its Career Development Fund (C243512002).
Yunlong Xiao thanks Penghui Yao for valuable discussions during his visit to the Quantum Software Lab at the University of Edinburgh.
Fei Meng acknowledges the support from the Quantum Theory Group at the University of Glasgow. 
Jinge Bao is supported by the Quantum Advantage Pathfinder (QAP) project of UKRI Engineering and Physical Sciences Research Council under grant No.~EP/X026167/1.


\bibliography{Bib}

\begin{thebibliography}{96}%
\makeatletter
\providecommand \@ifxundefined [1]{%
 \@ifx{#1\undefined}
}%
\providecommand \@ifnum [1]{%
 \ifnum #1\expandafter \@firstoftwo
 \else \expandafter \@secondoftwo
 \fi
}%
\providecommand \@ifx [1]{%
 \ifx #1\expandafter \@firstoftwo
 \else \expandafter \@secondoftwo
 \fi
}%
\providecommand \natexlab [1]{#1}%
\providecommand \enquote  [1]{``#1''}%
\providecommand \bibnamefont  [1]{#1}%
\providecommand \bibfnamefont [1]{#1}%
\providecommand \citenamefont [1]{#1}%
\providecommand \href@noop [0]{\@secondoftwo}%
\providecommand \href [0]{\begingroup \@sanitize@url \@href}%
\providecommand \@href[1]{\@@startlink{#1}\@@href}%
\providecommand \@@href[1]{\endgroup#1\@@endlink}%
\providecommand \@sanitize@url [0]{\catcode `\\12\catcode `\$12\catcode
  `\&12\catcode `\#12\catcode `\^12\catcode `\_12\catcode `\%12\relax}%
\providecommand \@@startlink[1]{}%
\providecommand \@@endlink[0]{}%
\providecommand \url  [0]{\begingroup\@sanitize@url \@url }%
\providecommand \@url [1]{\endgroup\@href {#1}{\urlprefix }}%
\providecommand \urlprefix  [0]{URL }%
\providecommand \Eprint [0]{\href }%
\providecommand \doibase [0]{https://doi.org/}%
\providecommand \selectlanguage [0]{\@gobble}%
\providecommand \bibinfo  [0]{\@secondoftwo}%
\providecommand \bibfield  [0]{\@secondoftwo}%
\providecommand \translation [1]{[#1]}%
\providecommand \BibitemOpen [0]{}%
\providecommand \bibitemStop [0]{}%
\providecommand \bibitemNoStop [0]{.\EOS\space}%
\providecommand \EOS [0]{\spacefactor3000\relax}%
\providecommand \BibitemShut  [1]{\csname bibitem#1\endcsname}%
\let\auto@bib@innerbib\@empty
\bibitem [{\citenamefont {Shor}(1994)}]{365700}%
  \BibitemOpen
  \bibfield  {author} {\bibinfo {author} {\bibfnamefont {P.~W.}\ \bibnamefont
  {Shor}},\ }\bibfield  {title} {\bibinfo {title} {Algorithms for quantum
  computation: discrete logarithms and factoring},\ }in\ \href
  {https://doi.org/10.1109/SFCS.1994.365700} {\emph {\bibinfo {booktitle}
  {Proceedings of the 35th Annual Symposium on Foundations of Computer
  Science}}}\ (\bibinfo {year} {1994})\ pp.\ \bibinfo {pages}
  {124--134}\BibitemShut {NoStop}%
\bibitem [{\citenamefont {Grover}(1996)}]{10.1145/237814.237866}%
  \BibitemOpen
  \bibfield  {author} {\bibinfo {author} {\bibfnamefont {L.~K.}\ \bibnamefont
  {Grover}},\ }\bibfield  {title} {\bibinfo {title} {A fast quantum mechanical
  algorithm for database search},\ }in\ \href
  {https://doi.org/10.1145/237814.237866} {\emph {\bibinfo {booktitle}
  {Proceedings of the 28th Annual ACM Symposium on Theory of Computing}}}\
  (\bibinfo {year} {1996})\ p.\ \bibinfo {pages} {212–219}\BibitemShut
  {NoStop}%
\bibitem [{\citenamefont {Arute}\ \emph {et~al.}(2019)\citenamefont {Arute},
  \citenamefont {Arya}, \citenamefont {Babbush}, \citenamefont {Bacon},
  \citenamefont {Bardin}, \citenamefont {Barends}, \citenamefont {Biswas},
  \citenamefont {Boixo}, \citenamefont {Brandao}, \citenamefont {Buell},
  \citenamefont {Burkett}, \citenamefont {Chen}, \citenamefont {Chen},
  \citenamefont {Chiaro}, \citenamefont {Collins}, \citenamefont {Courtney},
  \citenamefont {Dunsworth}, \citenamefont {Farhi}, \citenamefont {Foxen},
  \citenamefont {Fowler}, \citenamefont {Gidney}, \citenamefont {Giustina},
  \citenamefont {Graff}, \citenamefont {Guerin}, \citenamefont {Habegger},
  \citenamefont {Harrigan}, \citenamefont {Hartmann}, \citenamefont {Ho},
  \citenamefont {Hoffmann}, \citenamefont {Huang}, \citenamefont {Humble},
  \citenamefont {Isakov}, \citenamefont {Jeffrey}, \citenamefont {Jiang},
  \citenamefont {Kafri}, \citenamefont {Kechedzhi}, \citenamefont {Kelly},
  \citenamefont {Klimov}, \citenamefont {Knysh}, \citenamefont {Korotkov},
  \citenamefont {Kostritsa}, \citenamefont {Landhuis}, \citenamefont
  {Lindmark}, \citenamefont {Lucero}, \citenamefont {Lyakh}, \citenamefont
  {Mandr{\`a}}, \citenamefont {McClean}, \citenamefont {McEwen}, \citenamefont
  {Megrant}, \citenamefont {Mi}, \citenamefont {Michielsen}, \citenamefont
  {Mohseni}, \citenamefont {Mutus}, \citenamefont {Naaman}, \citenamefont
  {Neeley}, \citenamefont {Neill}, \citenamefont {Niu}, \citenamefont {Ostby},
  \citenamefont {Petukhov}, \citenamefont {Platt}, \citenamefont {Quintana},
  \citenamefont {Rieffel}, \citenamefont {Roushan}, \citenamefont {Rubin},
  \citenamefont {Sank}, \citenamefont {Satzinger}, \citenamefont {Smelyanskiy},
  \citenamefont {Sung}, \citenamefont {Trevithick}, \citenamefont
  {Vainsencher}, \citenamefont {Villalonga}, \citenamefont {White},
  \citenamefont {Yao}, \citenamefont {Yeh}, \citenamefont {Zalcman},
  \citenamefont {Neven},\ and\ \citenamefont {Martinis}}]{Arute2019}%
  \BibitemOpen
  \bibfield  {author} {\bibinfo {author} {\bibfnamefont {F.}~\bibnamefont
  {Arute}}, \bibinfo {author} {\bibfnamefont {K.}~\bibnamefont {Arya}},
  \bibinfo {author} {\bibfnamefont {R.}~\bibnamefont {Babbush}}, \bibinfo
  {author} {\bibfnamefont {D.}~\bibnamefont {Bacon}}, \bibinfo {author}
  {\bibfnamefont {J.~C.}\ \bibnamefont {Bardin}}, \bibinfo {author}
  {\bibfnamefont {R.}~\bibnamefont {Barends}}, \bibinfo {author} {\bibfnamefont
  {R.}~\bibnamefont {Biswas}}, \bibinfo {author} {\bibfnamefont
  {S.}~\bibnamefont {Boixo}}, \bibinfo {author} {\bibfnamefont {F.~G. S.~L.}\
  \bibnamefont {Brandao}}, \bibinfo {author} {\bibfnamefont {D.~A.}\
  \bibnamefont {Buell}}, \bibinfo {author} {\bibfnamefont {B.}~\bibnamefont
  {Burkett}}, \bibinfo {author} {\bibfnamefont {Y.}~\bibnamefont {Chen}},
  \bibinfo {author} {\bibfnamefont {Z.}~\bibnamefont {Chen}}, \bibinfo {author}
  {\bibfnamefont {B.}~\bibnamefont {Chiaro}}, \bibinfo {author} {\bibfnamefont
  {R.}~\bibnamefont {Collins}}, \bibinfo {author} {\bibfnamefont
  {W.}~\bibnamefont {Courtney}}, \bibinfo {author} {\bibfnamefont
  {A.}~\bibnamefont {Dunsworth}}, \bibinfo {author} {\bibfnamefont
  {E.}~\bibnamefont {Farhi}}, \bibinfo {author} {\bibfnamefont
  {B.}~\bibnamefont {Foxen}}, \bibinfo {author} {\bibfnamefont
  {A.}~\bibnamefont {Fowler}}, \bibinfo {author} {\bibfnamefont
  {C.}~\bibnamefont {Gidney}}, \bibinfo {author} {\bibfnamefont
  {M.}~\bibnamefont {Giustina}}, \bibinfo {author} {\bibfnamefont
  {R.}~\bibnamefont {Graff}}, \bibinfo {author} {\bibfnamefont
  {K.}~\bibnamefont {Guerin}}, \bibinfo {author} {\bibfnamefont
  {S.}~\bibnamefont {Habegger}}, \bibinfo {author} {\bibfnamefont {M.~P.}\
  \bibnamefont {Harrigan}}, \bibinfo {author} {\bibfnamefont {M.~J.}\
  \bibnamefont {Hartmann}}, \bibinfo {author} {\bibfnamefont {A.}~\bibnamefont
  {Ho}}, \bibinfo {author} {\bibfnamefont {M.}~\bibnamefont {Hoffmann}},
  \bibinfo {author} {\bibfnamefont {T.}~\bibnamefont {Huang}}, \bibinfo
  {author} {\bibfnamefont {T.~S.}\ \bibnamefont {Humble}}, \bibinfo {author}
  {\bibfnamefont {S.~V.}\ \bibnamefont {Isakov}}, \bibinfo {author}
  {\bibfnamefont {E.}~\bibnamefont {Jeffrey}}, \bibinfo {author} {\bibfnamefont
  {Z.}~\bibnamefont {Jiang}}, \bibinfo {author} {\bibfnamefont
  {D.}~\bibnamefont {Kafri}}, \bibinfo {author} {\bibfnamefont
  {K.}~\bibnamefont {Kechedzhi}}, \bibinfo {author} {\bibfnamefont
  {J.}~\bibnamefont {Kelly}}, \bibinfo {author} {\bibfnamefont {P.~V.}\
  \bibnamefont {Klimov}}, \bibinfo {author} {\bibfnamefont {S.}~\bibnamefont
  {Knysh}}, \bibinfo {author} {\bibfnamefont {A.}~\bibnamefont {Korotkov}},
  \bibinfo {author} {\bibfnamefont {F.}~\bibnamefont {Kostritsa}}, \bibinfo
  {author} {\bibfnamefont {D.}~\bibnamefont {Landhuis}}, \bibinfo {author}
  {\bibfnamefont {M.}~\bibnamefont {Lindmark}}, \bibinfo {author}
  {\bibfnamefont {E.}~\bibnamefont {Lucero}}, \bibinfo {author} {\bibfnamefont
  {D.}~\bibnamefont {Lyakh}}, \bibinfo {author} {\bibfnamefont
  {S.}~\bibnamefont {Mandr{\`a}}}, \bibinfo {author} {\bibfnamefont {J.~R.}\
  \bibnamefont {McClean}}, \bibinfo {author} {\bibfnamefont {M.}~\bibnamefont
  {McEwen}}, \bibinfo {author} {\bibfnamefont {A.}~\bibnamefont {Megrant}},
  \bibinfo {author} {\bibfnamefont {X.}~\bibnamefont {Mi}}, \bibinfo {author}
  {\bibfnamefont {K.}~\bibnamefont {Michielsen}}, \bibinfo {author}
  {\bibfnamefont {M.}~\bibnamefont {Mohseni}}, \bibinfo {author} {\bibfnamefont
  {J.}~\bibnamefont {Mutus}}, \bibinfo {author} {\bibfnamefont
  {O.}~\bibnamefont {Naaman}}, \bibinfo {author} {\bibfnamefont
  {M.}~\bibnamefont {Neeley}}, \bibinfo {author} {\bibfnamefont
  {C.}~\bibnamefont {Neill}}, \bibinfo {author} {\bibfnamefont {M.~Y.}\
  \bibnamefont {Niu}}, \bibinfo {author} {\bibfnamefont {E.}~\bibnamefont
  {Ostby}}, \bibinfo {author} {\bibfnamefont {A.}~\bibnamefont {Petukhov}},
  \bibinfo {author} {\bibfnamefont {J.~C.}\ \bibnamefont {Platt}}, \bibinfo
  {author} {\bibfnamefont {C.}~\bibnamefont {Quintana}}, \bibinfo {author}
  {\bibfnamefont {E.~G.}\ \bibnamefont {Rieffel}}, \bibinfo {author}
  {\bibfnamefont {P.}~\bibnamefont {Roushan}}, \bibinfo {author} {\bibfnamefont
  {N.~C.}\ \bibnamefont {Rubin}}, \bibinfo {author} {\bibfnamefont
  {D.}~\bibnamefont {Sank}}, \bibinfo {author} {\bibfnamefont {K.~J.}\
  \bibnamefont {Satzinger}}, \bibinfo {author} {\bibfnamefont {V.}~\bibnamefont
  {Smelyanskiy}}, \bibinfo {author} {\bibfnamefont {K.~J.}\ \bibnamefont
  {Sung}}, \bibinfo {author} {\bibfnamefont {M.~D.}\ \bibnamefont
  {Trevithick}}, \bibinfo {author} {\bibfnamefont {A.}~\bibnamefont
  {Vainsencher}}, \bibinfo {author} {\bibfnamefont {B.}~\bibnamefont
  {Villalonga}}, \bibinfo {author} {\bibfnamefont {T.}~\bibnamefont {White}},
  \bibinfo {author} {\bibfnamefont {Z.~J.}\ \bibnamefont {Yao}}, \bibinfo
  {author} {\bibfnamefont {P.}~\bibnamefont {Yeh}}, \bibinfo {author}
  {\bibfnamefont {A.}~\bibnamefont {Zalcman}}, \bibinfo {author} {\bibfnamefont
  {H.}~\bibnamefont {Neven}},\ and\ \bibinfo {author} {\bibfnamefont {J.~M.}\
  \bibnamefont {Martinis}},\ }\bibfield  {title} {\bibinfo {title} {Quantum
  supremacy using a programmable superconducting processor},\ }\href
  {https://doi.org/10.1038/s41586-019-1666-5} {\bibfield  {journal} {\bibinfo
  {journal} {Nature}\ }\textbf {\bibinfo {volume} {574}},\ \bibinfo {pages}
  {505} (\bibinfo {year} {2019})}\BibitemShut {NoStop}%
\bibitem [{\citenamefont {Zhong}\ \emph {et~al.}(2020)\citenamefont {Zhong},
  \citenamefont {Wang}, \citenamefont {Deng}, \citenamefont {Ming-Cheng},
  \citenamefont {Peng}, \citenamefont {Yi-Han}, \citenamefont {Qian},
  \citenamefont {Wu}, \citenamefont {Ding}, \citenamefont {Hu}, \citenamefont
  {Hu}, \citenamefont {Yang}, \citenamefont {Zhang}, \citenamefont {Li},
  \citenamefont {Li}, \citenamefont {Jiang}, \citenamefont {Gan}, \citenamefont
  {Yang}, \citenamefont {You}, \citenamefont {Wang}, \citenamefont {Li},
  \citenamefont {Liu}, \citenamefont {Lu},\ and\ \citenamefont
  {Pan}}]{doi:10.1126/science.abe8770}%
  \BibitemOpen
  \bibfield  {author} {\bibinfo {author} {\bibfnamefont {H.-S.}\ \bibnamefont
  {Zhong}}, \bibinfo {author} {\bibfnamefont {H.}~\bibnamefont {Wang}},
  \bibinfo {author} {\bibfnamefont {Y.-H.}\ \bibnamefont {Deng}}, \bibinfo
  {author} {\bibfnamefont {C.}~\bibnamefont {Ming-Cheng}}, \bibinfo {author}
  {\bibfnamefont {L.-C.}\ \bibnamefont {Peng}}, \bibinfo {author}
  {\bibfnamefont {L.}~\bibnamefont {Yi-Han}}, \bibinfo {author} {\bibfnamefont
  {J.}~\bibnamefont {Qian}}, \bibinfo {author} {\bibfnamefont {D.}~\bibnamefont
  {Wu}}, \bibinfo {author} {\bibfnamefont {X.}~\bibnamefont {Ding}}, \bibinfo
  {author} {\bibfnamefont {Y.}~\bibnamefont {Hu}}, \bibinfo {author}
  {\bibfnamefont {P.}~\bibnamefont {Hu}}, \bibinfo {author} {\bibfnamefont
  {X.-Y.}\ \bibnamefont {Yang}}, \bibinfo {author} {\bibfnamefont {W.-J.}\
  \bibnamefont {Zhang}}, \bibinfo {author} {\bibfnamefont {H.}~\bibnamefont
  {Li}}, \bibinfo {author} {\bibfnamefont {Y.}~\bibnamefont {Li}}, \bibinfo
  {author} {\bibfnamefont {X.}~\bibnamefont {Jiang}}, \bibinfo {author}
  {\bibfnamefont {L.}~\bibnamefont {Gan}}, \bibinfo {author} {\bibfnamefont
  {G.}~\bibnamefont {Yang}}, \bibinfo {author} {\bibfnamefont {L.}~\bibnamefont
  {You}}, \bibinfo {author} {\bibfnamefont {Z.}~\bibnamefont {Wang}}, \bibinfo
  {author} {\bibfnamefont {L.}~\bibnamefont {Li}}, \bibinfo {author}
  {\bibfnamefont {N.-L.}\ \bibnamefont {Liu}}, \bibinfo {author} {\bibfnamefont
  {C.-Y.}\ \bibnamefont {Lu}},\ and\ \bibinfo {author} {\bibfnamefont {J.-W.}\
  \bibnamefont {Pan}},\ }\bibfield  {title} {\bibinfo {title} {Quantum
  computational advantage using photons},\ }\href
  {https://doi.org/10.1126/science.abe8770} {\bibfield  {journal} {\bibinfo
  {journal} {Science}\ }\textbf {\bibinfo {volume} {370}},\ \bibinfo {pages}
  {1460} (\bibinfo {year} {2020})}\BibitemShut {NoStop}%
\bibitem [{\citenamefont {Bennett}\ \emph {et~al.}(1993)\citenamefont
  {Bennett}, \citenamefont {Brassard}, \citenamefont {Cr\'epeau}, \citenamefont
  {Jozsa}, \citenamefont {Peres},\ and\ \citenamefont
  {Wootters}}]{PhysRevLett.70.1895}%
  \BibitemOpen
  \bibfield  {author} {\bibinfo {author} {\bibfnamefont {C.~H.}\ \bibnamefont
  {Bennett}}, \bibinfo {author} {\bibfnamefont {G.}~\bibnamefont {Brassard}},
  \bibinfo {author} {\bibfnamefont {C.}~\bibnamefont {Cr\'epeau}}, \bibinfo
  {author} {\bibfnamefont {R.}~\bibnamefont {Jozsa}}, \bibinfo {author}
  {\bibfnamefont {A.}~\bibnamefont {Peres}},\ and\ \bibinfo {author}
  {\bibfnamefont {W.~K.}\ \bibnamefont {Wootters}},\ }\bibfield  {title}
  {\bibinfo {title} {Teleporting an unknown quantum state via dual classical
  and einstein-podolsky-rosen channels},\ }\href
  {https://doi.org/10.1103/PhysRevLett.70.1895} {\bibfield  {journal} {\bibinfo
   {journal} {Physical Review Letters}\ }\textbf {\bibinfo {volume} {70}},\
  \bibinfo {pages} {1895} (\bibinfo {year} {1993})}\BibitemShut {NoStop}%
\bibitem [{\citenamefont {Bouwmeester}\ \emph {et~al.}(1997)\citenamefont
  {Bouwmeester}, \citenamefont {Pan}, \citenamefont {Mattle}, \citenamefont
  {Eibl}, \citenamefont {Weinfurter},\ and\ \citenamefont
  {Zeilinger}}]{Bouwmeester1997}%
  \BibitemOpen
  \bibfield  {author} {\bibinfo {author} {\bibfnamefont {D.}~\bibnamefont
  {Bouwmeester}}, \bibinfo {author} {\bibfnamefont {J.-W.}\ \bibnamefont
  {Pan}}, \bibinfo {author} {\bibfnamefont {K.}~\bibnamefont {Mattle}},
  \bibinfo {author} {\bibfnamefont {M.}~\bibnamefont {Eibl}}, \bibinfo {author}
  {\bibfnamefont {H.}~\bibnamefont {Weinfurter}},\ and\ \bibinfo {author}
  {\bibfnamefont {A.}~\bibnamefont {Zeilinger}},\ }\bibfield  {title} {\bibinfo
  {title} {Experimental quantum teleportation},\ }\href
  {https://doi.org/10.1038/37539} {\bibfield  {journal} {\bibinfo  {journal}
  {Nature}\ }\textbf {\bibinfo {volume} {390}},\ \bibinfo {pages} {575}
  (\bibinfo {year} {1997})}\BibitemShut {NoStop}%
\bibitem [{\citenamefont {Kimble}(2008)}]{Kimble2008}%
  \BibitemOpen
  \bibfield  {author} {\bibinfo {author} {\bibfnamefont {H.~J.}\ \bibnamefont
  {Kimble}},\ }\bibfield  {title} {\bibinfo {title} {The quantum internet},\
  }\href {https://doi.org/10.1038/nature07127} {\bibfield  {journal} {\bibinfo
  {journal} {Nature}\ }\textbf {\bibinfo {volume} {453}},\ \bibinfo {pages}
  {1023} (\bibinfo {year} {2008})}\BibitemShut {NoStop}%
\bibitem [{\citenamefont {Ren}\ \emph {et~al.}(2017)\citenamefont {Ren},
  \citenamefont {Xu}, \citenamefont {Yong}, \citenamefont {Zhang},
  \citenamefont {Liao}, \citenamefont {Yin}, \citenamefont {Liu}, \citenamefont
  {Cai}, \citenamefont {Yang}, \citenamefont {Li}, \citenamefont {Yang},
  \citenamefont {Han}, \citenamefont {Yao}, \citenamefont {Li}, \citenamefont
  {Wu}, \citenamefont {Wan}, \citenamefont {Liu}, \citenamefont {Liu},
  \citenamefont {Kuang}, \citenamefont {He}, \citenamefont {Shang},
  \citenamefont {Guo}, \citenamefont {Zheng}, \citenamefont {Tian},
  \citenamefont {Zhu}, \citenamefont {Liu}, \citenamefont {Lu}, \citenamefont
  {Shu}, \citenamefont {Chen}, \citenamefont {Peng}, \citenamefont {Wang},\
  and\ \citenamefont {Pan}}]{Ren2017}%
  \BibitemOpen
  \bibfield  {author} {\bibinfo {author} {\bibfnamefont {J.-G.}\ \bibnamefont
  {Ren}}, \bibinfo {author} {\bibfnamefont {P.}~\bibnamefont {Xu}}, \bibinfo
  {author} {\bibfnamefont {H.-L.}\ \bibnamefont {Yong}}, \bibinfo {author}
  {\bibfnamefont {L.}~\bibnamefont {Zhang}}, \bibinfo {author} {\bibfnamefont
  {S.-K.}\ \bibnamefont {Liao}}, \bibinfo {author} {\bibfnamefont
  {J.}~\bibnamefont {Yin}}, \bibinfo {author} {\bibfnamefont {W.-Y.}\
  \bibnamefont {Liu}}, \bibinfo {author} {\bibfnamefont {W.-Q.}\ \bibnamefont
  {Cai}}, \bibinfo {author} {\bibfnamefont {M.}~\bibnamefont {Yang}}, \bibinfo
  {author} {\bibfnamefont {L.}~\bibnamefont {Li}}, \bibinfo {author}
  {\bibfnamefont {K.-X.}\ \bibnamefont {Yang}}, \bibinfo {author}
  {\bibfnamefont {X.}~\bibnamefont {Han}}, \bibinfo {author} {\bibfnamefont
  {Y.-Q.}\ \bibnamefont {Yao}}, \bibinfo {author} {\bibfnamefont
  {J.}~\bibnamefont {Li}}, \bibinfo {author} {\bibfnamefont {H.-Y.}\
  \bibnamefont {Wu}}, \bibinfo {author} {\bibfnamefont {S.}~\bibnamefont
  {Wan}}, \bibinfo {author} {\bibfnamefont {L.}~\bibnamefont {Liu}}, \bibinfo
  {author} {\bibfnamefont {D.-Q.}\ \bibnamefont {Liu}}, \bibinfo {author}
  {\bibfnamefont {Y.-W.}\ \bibnamefont {Kuang}}, \bibinfo {author}
  {\bibfnamefont {Z.-P.}\ \bibnamefont {He}}, \bibinfo {author} {\bibfnamefont
  {P.}~\bibnamefont {Shang}}, \bibinfo {author} {\bibfnamefont
  {C.}~\bibnamefont {Guo}}, \bibinfo {author} {\bibfnamefont {R.-H.}\
  \bibnamefont {Zheng}}, \bibinfo {author} {\bibfnamefont {K.}~\bibnamefont
  {Tian}}, \bibinfo {author} {\bibfnamefont {Z.-C.}\ \bibnamefont {Zhu}},
  \bibinfo {author} {\bibfnamefont {N.-L.}\ \bibnamefont {Liu}}, \bibinfo
  {author} {\bibfnamefont {C.-Y.}\ \bibnamefont {Lu}}, \bibinfo {author}
  {\bibfnamefont {R.}~\bibnamefont {Shu}}, \bibinfo {author} {\bibfnamefont
  {Y.-A.}\ \bibnamefont {Chen}}, \bibinfo {author} {\bibfnamefont {C.-Z.}\
  \bibnamefont {Peng}}, \bibinfo {author} {\bibfnamefont {J.-Y.}\ \bibnamefont
  {Wang}},\ and\ \bibinfo {author} {\bibfnamefont {J.-W.}\ \bibnamefont
  {Pan}},\ }\bibfield  {title} {\bibinfo {title} {Ground-to-satellite quantum
  teleportation},\ }\href {https://doi.org/10.1038/nature23675} {\bibfield
  {journal} {\bibinfo  {journal} {Nature}\ }\textbf {\bibinfo {volume} {549}},\
  \bibinfo {pages} {70} (\bibinfo {year} {2017})}\BibitemShut {NoStop}%
\bibitem [{\citenamefont {Xu}\ \emph {et~al.}(2020)\citenamefont {Xu},
  \citenamefont {Ma}, \citenamefont {Zhang}, \citenamefont {Lo},\ and\
  \citenamefont {Pan}}]{RevModPhys.92.025002}%
  \BibitemOpen
  \bibfield  {author} {\bibinfo {author} {\bibfnamefont {F.}~\bibnamefont
  {Xu}}, \bibinfo {author} {\bibfnamefont {X.}~\bibnamefont {Ma}}, \bibinfo
  {author} {\bibfnamefont {Q.}~\bibnamefont {Zhang}}, \bibinfo {author}
  {\bibfnamefont {H.-K.}\ \bibnamefont {Lo}},\ and\ \bibinfo {author}
  {\bibfnamefont {J.-W.}\ \bibnamefont {Pan}},\ }\bibfield  {title} {\bibinfo
  {title} {Secure quantum key distribution with realistic devices},\ }\href
  {https://doi.org/10.1103/RevModPhys.92.025002} {\bibfield  {journal}
  {\bibinfo  {journal} {Reviews of Modern Physics}\ }\textbf {\bibinfo {volume}
  {92}},\ \bibinfo {pages} {025002} (\bibinfo {year} {2020})}\BibitemShut
  {NoStop}%
\bibitem [{\citenamefont {Giovannetti}\ \emph {et~al.}(2004)\citenamefont
  {Giovannetti}, \citenamefont {Lloyd},\ and\ \citenamefont
  {Maccone}}]{doi:10.1126/science.1104149}%
  \BibitemOpen
  \bibfield  {author} {\bibinfo {author} {\bibfnamefont {V.}~\bibnamefont
  {Giovannetti}}, \bibinfo {author} {\bibfnamefont {S.}~\bibnamefont {Lloyd}},\
  and\ \bibinfo {author} {\bibfnamefont {L.}~\bibnamefont {Maccone}},\
  }\bibfield  {title} {\bibinfo {title} {Quantum-enhanced measurements: Beating
  the standard quantum limit},\ }\href
  {https://doi.org/10.1126/science.1104149} {\bibfield  {journal} {\bibinfo
  {journal} {Science}\ }\textbf {\bibinfo {volume} {306}},\ \bibinfo {pages}
  {1330} (\bibinfo {year} {2004})}\BibitemShut {NoStop}%
\bibitem [{\citenamefont {Giovannetti}\ \emph {et~al.}(2006)\citenamefont
  {Giovannetti}, \citenamefont {Lloyd},\ and\ \citenamefont
  {Maccone}}]{PhysRevLett.96.010401}%
  \BibitemOpen
  \bibfield  {author} {\bibinfo {author} {\bibfnamefont {V.}~\bibnamefont
  {Giovannetti}}, \bibinfo {author} {\bibfnamefont {S.}~\bibnamefont {Lloyd}},\
  and\ \bibinfo {author} {\bibfnamefont {L.}~\bibnamefont {Maccone}},\
  }\bibfield  {title} {\bibinfo {title} {Quantum metrology},\ }\href
  {https://doi.org/10.1103/PhysRevLett.96.010401} {\bibfield  {journal}
  {\bibinfo  {journal} {Physicsal Review Letters}\ }\textbf {\bibinfo {volume}
  {96}},\ \bibinfo {pages} {010401} (\bibinfo {year} {2006})}\BibitemShut
  {NoStop}%
\bibitem [{\citenamefont {Giovannetti}\ \emph {et~al.}(2011)\citenamefont
  {Giovannetti}, \citenamefont {Lloyd},\ and\ \citenamefont
  {Maccone}}]{Giovannetti2011}%
  \BibitemOpen
  \bibfield  {author} {\bibinfo {author} {\bibfnamefont {V.}~\bibnamefont
  {Giovannetti}}, \bibinfo {author} {\bibfnamefont {S.}~\bibnamefont {Lloyd}},\
  and\ \bibinfo {author} {\bibfnamefont {L.}~\bibnamefont {Maccone}},\
  }\bibfield  {title} {\bibinfo {title} {Advances in quantum metrology},\
  }\href {https://doi.org/10.1038/nphoton.2011.35} {\bibfield  {journal}
  {\bibinfo  {journal} {Nature Photonics}\ }\textbf {\bibinfo {volume} {5}},\
  \bibinfo {pages} {222} (\bibinfo {year} {2011})}\BibitemShut {NoStop}%
\bibitem [{\citenamefont {Degen}\ \emph {et~al.}(2017)\citenamefont {Degen},
  \citenamefont {Reinhard},\ and\ \citenamefont
  {Cappellaro}}]{RevModPhys.89.035002}%
  \BibitemOpen
  \bibfield  {author} {\bibinfo {author} {\bibfnamefont {C.~L.}\ \bibnamefont
  {Degen}}, \bibinfo {author} {\bibfnamefont {F.}~\bibnamefont {Reinhard}},\
  and\ \bibinfo {author} {\bibfnamefont {P.}~\bibnamefont {Cappellaro}},\
  }\bibfield  {title} {\bibinfo {title} {Quantum sensing},\ }\href
  {https://doi.org/10.1103/RevModPhys.89.035002} {\bibfield  {journal}
  {\bibinfo  {journal} {Reviews of Modern Physics}\ }\textbf {\bibinfo {volume}
  {89}},\ \bibinfo {pages} {035002} (\bibinfo {year} {2017})}\BibitemShut
  {NoStop}%
\bibitem [{\citenamefont {Wineland}(2013)}]{RevModPhys.85.1103}%
  \BibitemOpen
  \bibfield  {author} {\bibinfo {author} {\bibfnamefont {D.~J.}\ \bibnamefont
  {Wineland}},\ }\bibfield  {title} {\bibinfo {title} {Nobel lecture:
  Superposition, entanglement, and raising schr\"odinger's cat},\ }\href
  {https://doi.org/10.1103/RevModPhys.85.1103} {\bibfield  {journal} {\bibinfo
  {journal} {Reviews of Modern Physics}\ }\textbf {\bibinfo {volume} {85}},\
  \bibinfo {pages} {1103} (\bibinfo {year} {2013})}\BibitemShut {NoStop}%
\bibitem [{\citenamefont {Horodecki}\ \emph {et~al.}(2009)\citenamefont
  {Horodecki}, \citenamefont {Horodecki}, \citenamefont {Horodecki},\ and\
  \citenamefont {Horodecki}}]{RevModPhys.81.865}%
  \BibitemOpen
  \bibfield  {author} {\bibinfo {author} {\bibfnamefont {R.}~\bibnamefont
  {Horodecki}}, \bibinfo {author} {\bibfnamefont {P.}~\bibnamefont
  {Horodecki}}, \bibinfo {author} {\bibfnamefont {M.}~\bibnamefont
  {Horodecki}},\ and\ \bibinfo {author} {\bibfnamefont {K.}~\bibnamefont
  {Horodecki}},\ }\bibfield  {title} {\bibinfo {title} {Quantum entanglement},\
  }\href {https://doi.org/10.1103/RevModPhys.81.865} {\bibfield  {journal}
  {\bibinfo  {journal} {Review of Modern Physics}\ }\textbf {\bibinfo {volume}
  {81}},\ \bibinfo {pages} {865} (\bibinfo {year} {2009})}\BibitemShut
  {NoStop}%
\bibitem [{\citenamefont {Shor}(1995)}]{PhysRevA.52.R2493}%
  \BibitemOpen
  \bibfield  {author} {\bibinfo {author} {\bibfnamefont {P.~W.}\ \bibnamefont
  {Shor}},\ }\bibfield  {title} {\bibinfo {title} {Scheme for reducing
  decoherence in quantum computer memory},\ }\href
  {https://doi.org/10.1103/PhysRevA.52.R2493} {\bibfield  {journal} {\bibinfo
  {journal} {Physical Review A}\ }\textbf {\bibinfo {volume} {52}},\ \bibinfo
  {pages} {R2493} (\bibinfo {year} {1995})}\BibitemShut {NoStop}%
\bibitem [{\citenamefont {Cory}\ \emph {et~al.}(1998)\citenamefont {Cory},
  \citenamefont {Price}, \citenamefont {Maas}, \citenamefont {Knill},
  \citenamefont {Laflamme}, \citenamefont {Zurek}, \citenamefont {Havel},\ and\
  \citenamefont {Somaroo}}]{PhysRevLett.81.2152}%
  \BibitemOpen
  \bibfield  {author} {\bibinfo {author} {\bibfnamefont {D.~G.}\ \bibnamefont
  {Cory}}, \bibinfo {author} {\bibfnamefont {M.~D.}\ \bibnamefont {Price}},
  \bibinfo {author} {\bibfnamefont {W.}~\bibnamefont {Maas}}, \bibinfo {author}
  {\bibfnamefont {E.}~\bibnamefont {Knill}}, \bibinfo {author} {\bibfnamefont
  {R.}~\bibnamefont {Laflamme}}, \bibinfo {author} {\bibfnamefont {W.~H.}\
  \bibnamefont {Zurek}}, \bibinfo {author} {\bibfnamefont {T.~F.}\ \bibnamefont
  {Havel}},\ and\ \bibinfo {author} {\bibfnamefont {S.~S.}\ \bibnamefont
  {Somaroo}},\ }\bibfield  {title} {\bibinfo {title} {Experimental quantum
  error correction},\ }\href {https://doi.org/10.1103/PhysRevLett.81.2152}
  {\bibfield  {journal} {\bibinfo  {journal} {Physical Review Letters}\
  }\textbf {\bibinfo {volume} {81}},\ \bibinfo {pages} {2152} (\bibinfo {year}
  {1998})}\BibitemShut {NoStop}%
\bibitem [{\citenamefont {Terhal}(2015)}]{RevModPhys.87.307}%
  \BibitemOpen
  \bibfield  {author} {\bibinfo {author} {\bibfnamefont {B.~M.}\ \bibnamefont
  {Terhal}},\ }\bibfield  {title} {\bibinfo {title} {Quantum error correction
  for quantum memories},\ }\href {https://doi.org/10.1103/RevModPhys.87.307}
  {\bibfield  {journal} {\bibinfo  {journal} {Reviews of Modern Physics}\
  }\textbf {\bibinfo {volume} {87}},\ \bibinfo {pages} {307} (\bibinfo {year}
  {2015})}\BibitemShut {NoStop}%
\bibitem [{\citenamefont {Acharya}\ \emph {et~al.}(2023)\citenamefont
  {Acharya}, \citenamefont {Aleiner}, \citenamefont {Allen}, \citenamefont
  {Andersen}, \citenamefont {Ansmann}, \citenamefont {Arute}, \citenamefont
  {Arya}, \citenamefont {Asfaw}, \citenamefont {Atalaya}, \citenamefont
  {Babbush}, \citenamefont {Bacon}, \citenamefont {Bardin}, \citenamefont
  {Basso}, \citenamefont {Bengtsson}, \citenamefont {Boixo}, \citenamefont
  {Bortoli}, \citenamefont {Bourassa}, \citenamefont {Bovaird}, \citenamefont
  {Brill}, \citenamefont {Broughton}, \citenamefont {Buckley}, \citenamefont
  {Buell}, \citenamefont {Burger}, \citenamefont {Burkett}, \citenamefont
  {Bushnell}, \citenamefont {Chen}, \citenamefont {Chen}, \citenamefont
  {Chiaro}, \citenamefont {Cogan}, \citenamefont {Collins}, \citenamefont
  {Conner}, \citenamefont {Courtney}, \citenamefont {Crook}, \citenamefont
  {Curtin}, \citenamefont {Debroy}, \citenamefont {Del Toro~Barba},
  \citenamefont {Demura}, \citenamefont {Dunsworth}, \citenamefont {Eppens},
  \citenamefont {Erickson}, \citenamefont {Faoro}, \citenamefont {Farhi},
  \citenamefont {Fatemi}, \citenamefont {Flores~Burgos}, \citenamefont
  {Forati}, \citenamefont {Fowler}, \citenamefont {Foxen}, \citenamefont
  {Giang}, \citenamefont {Gidney}, \citenamefont {Gilboa}, \citenamefont
  {Giustina}, \citenamefont {Grajales~Dau}, \citenamefont {Gross},
  \citenamefont {Habegger}, \citenamefont {Hamilton}, \citenamefont {Harrigan},
  \citenamefont {Harrington}, \citenamefont {Higgott}, \citenamefont {Hilton},
  \citenamefont {Hoffmann}, \citenamefont {Hong}, \citenamefont {Huang},
  \citenamefont {Huff}, \citenamefont {Huggins}, \citenamefont {Ioffe},
  \citenamefont {Isakov}, \citenamefont {Iveland}, \citenamefont {Jeffrey},
  \citenamefont {Jiang}, \citenamefont {Jones}, \citenamefont {Juhas},
  \citenamefont {Kafri}, \citenamefont {Kechedzhi}, \citenamefont {Kelly},
  \citenamefont {Khattar}, \citenamefont {Khezri}, \citenamefont
  {Kieferov{\'a}}, \citenamefont {Kim}, \citenamefont {Kitaev}, \citenamefont
  {Klimov}, \citenamefont {Klots}, \citenamefont {Korotkov}, \citenamefont
  {Kostritsa}, \citenamefont {Kreikebaum}, \citenamefont {Landhuis},
  \citenamefont {Laptev}, \citenamefont {Lau}, \citenamefont {Laws},
  \citenamefont {Lee}, \citenamefont {Lee}, \citenamefont {Lester},
  \citenamefont {Lill}, \citenamefont {Liu}, \citenamefont {Locharla},
  \citenamefont {Lucero}, \citenamefont {Malone}, \citenamefont {Marshall},
  \citenamefont {Martin}, \citenamefont {McClean}, \citenamefont {McCourt},
  \citenamefont {McEwen}, \citenamefont {Megrant}, \citenamefont
  {Meurer~Costa}, \citenamefont {Mi}, \citenamefont {Miao}, \citenamefont
  {Mohseni}, \citenamefont {Montazeri}, \citenamefont {Morvan}, \citenamefont
  {Mount}, \citenamefont {Mruczkiewicz}, \citenamefont {Naaman}, \citenamefont
  {Neeley}, \citenamefont {Neill}, \citenamefont {Nersisyan}, \citenamefont
  {Neven}, \citenamefont {Newman}, \citenamefont {Ng}, \citenamefont {Nguyen},
  \citenamefont {Nguyen}, \citenamefont {Niu}, \citenamefont {O'Brien},
  \citenamefont {Opremcak}, \citenamefont {Platt}, \citenamefont {Petukhov},
  \citenamefont {Potter}, \citenamefont {Pryadko}, \citenamefont {Quintana},
  \citenamefont {Roushan}, \citenamefont {Rubin}, \citenamefont {Saei},
  \citenamefont {Sank}, \citenamefont {Sankaragomathi}, \citenamefont
  {Satzinger}, \citenamefont {Schurkus}, \citenamefont {Schuster},
  \citenamefont {Shearn}, \citenamefont {Shorter}, \citenamefont {Shvarts},
  \citenamefont {Skruzny}, \citenamefont {Smelyanskiy}, \citenamefont {Smith},
  \citenamefont {Sterling}, \citenamefont {Strain}, \citenamefont {Szalay},
  \citenamefont {Torres}, \citenamefont {Vidal}, \citenamefont {Villalonga},
  \citenamefont {Vollgraff~Heidweiller}, \citenamefont {White}, \citenamefont
  {Xing}, \citenamefont {Yao}, \citenamefont {Yeh}, \citenamefont {Yoo},
  \citenamefont {Young}, \citenamefont {Zalcman}, \citenamefont {Zhang},
  \citenamefont {Zhu},\ and\ \citenamefont {AI}}]{Acharya2023}%
  \BibitemOpen
  \bibfield  {author} {\bibinfo {author} {\bibfnamefont {R.}~\bibnamefont
  {Acharya}}, \bibinfo {author} {\bibfnamefont {I.}~\bibnamefont {Aleiner}},
  \bibinfo {author} {\bibfnamefont {R.}~\bibnamefont {Allen}}, \bibinfo
  {author} {\bibfnamefont {T.~I.}\ \bibnamefont {Andersen}}, \bibinfo {author}
  {\bibfnamefont {M.}~\bibnamefont {Ansmann}}, \bibinfo {author} {\bibfnamefont
  {F.}~\bibnamefont {Arute}}, \bibinfo {author} {\bibfnamefont
  {K.}~\bibnamefont {Arya}}, \bibinfo {author} {\bibfnamefont {A.}~\bibnamefont
  {Asfaw}}, \bibinfo {author} {\bibfnamefont {J.}~\bibnamefont {Atalaya}},
  \bibinfo {author} {\bibfnamefont {R.}~\bibnamefont {Babbush}}, \bibinfo
  {author} {\bibfnamefont {D.}~\bibnamefont {Bacon}}, \bibinfo {author}
  {\bibfnamefont {J.~C.}\ \bibnamefont {Bardin}}, \bibinfo {author}
  {\bibfnamefont {J.}~\bibnamefont {Basso}}, \bibinfo {author} {\bibfnamefont
  {A.}~\bibnamefont {Bengtsson}}, \bibinfo {author} {\bibfnamefont
  {S.}~\bibnamefont {Boixo}}, \bibinfo {author} {\bibfnamefont
  {G.}~\bibnamefont {Bortoli}}, \bibinfo {author} {\bibfnamefont
  {A.}~\bibnamefont {Bourassa}}, \bibinfo {author} {\bibfnamefont
  {J.}~\bibnamefont {Bovaird}}, \bibinfo {author} {\bibfnamefont
  {L.}~\bibnamefont {Brill}}, \bibinfo {author} {\bibfnamefont
  {M.}~\bibnamefont {Broughton}}, \bibinfo {author} {\bibfnamefont {B.~B.}\
  \bibnamefont {Buckley}}, \bibinfo {author} {\bibfnamefont {D.~A.}\
  \bibnamefont {Buell}}, \bibinfo {author} {\bibfnamefont {T.}~\bibnamefont
  {Burger}}, \bibinfo {author} {\bibfnamefont {B.}~\bibnamefont {Burkett}},
  \bibinfo {author} {\bibfnamefont {N.}~\bibnamefont {Bushnell}}, \bibinfo
  {author} {\bibfnamefont {Y.}~\bibnamefont {Chen}}, \bibinfo {author}
  {\bibfnamefont {Z.}~\bibnamefont {Chen}}, \bibinfo {author} {\bibfnamefont
  {B.}~\bibnamefont {Chiaro}}, \bibinfo {author} {\bibfnamefont
  {J.}~\bibnamefont {Cogan}}, \bibinfo {author} {\bibfnamefont
  {R.}~\bibnamefont {Collins}}, \bibinfo {author} {\bibfnamefont
  {P.}~\bibnamefont {Conner}}, \bibinfo {author} {\bibfnamefont
  {W.}~\bibnamefont {Courtney}}, \bibinfo {author} {\bibfnamefont {A.~L.}\
  \bibnamefont {Crook}}, \bibinfo {author} {\bibfnamefont {B.}~\bibnamefont
  {Curtin}}, \bibinfo {author} {\bibfnamefont {D.~M.}\ \bibnamefont {Debroy}},
  \bibinfo {author} {\bibfnamefont {A.}~\bibnamefont {Del Toro~Barba}},
  \bibinfo {author} {\bibfnamefont {S.}~\bibnamefont {Demura}}, \bibinfo
  {author} {\bibfnamefont {A.}~\bibnamefont {Dunsworth}}, \bibinfo {author}
  {\bibfnamefont {D.}~\bibnamefont {Eppens}}, \bibinfo {author} {\bibfnamefont
  {C.}~\bibnamefont {Erickson}}, \bibinfo {author} {\bibfnamefont
  {L.}~\bibnamefont {Faoro}}, \bibinfo {author} {\bibfnamefont
  {E.}~\bibnamefont {Farhi}}, \bibinfo {author} {\bibfnamefont
  {R.}~\bibnamefont {Fatemi}}, \bibinfo {author} {\bibfnamefont
  {L.}~\bibnamefont {Flores~Burgos}}, \bibinfo {author} {\bibfnamefont
  {E.}~\bibnamefont {Forati}}, \bibinfo {author} {\bibfnamefont {A.~G.}\
  \bibnamefont {Fowler}}, \bibinfo {author} {\bibfnamefont {B.}~\bibnamefont
  {Foxen}}, \bibinfo {author} {\bibfnamefont {W.}~\bibnamefont {Giang}},
  \bibinfo {author} {\bibfnamefont {C.}~\bibnamefont {Gidney}}, \bibinfo
  {author} {\bibfnamefont {D.}~\bibnamefont {Gilboa}}, \bibinfo {author}
  {\bibfnamefont {M.}~\bibnamefont {Giustina}}, \bibinfo {author}
  {\bibfnamefont {A.}~\bibnamefont {Grajales~Dau}}, \bibinfo {author}
  {\bibfnamefont {J.~A.}\ \bibnamefont {Gross}}, \bibinfo {author}
  {\bibfnamefont {S.}~\bibnamefont {Habegger}}, \bibinfo {author}
  {\bibfnamefont {M.~C.}\ \bibnamefont {Hamilton}}, \bibinfo {author}
  {\bibfnamefont {M.~P.}\ \bibnamefont {Harrigan}}, \bibinfo {author}
  {\bibfnamefont {S.~D.}\ \bibnamefont {Harrington}}, \bibinfo {author}
  {\bibfnamefont {O.}~\bibnamefont {Higgott}}, \bibinfo {author} {\bibfnamefont
  {J.}~\bibnamefont {Hilton}}, \bibinfo {author} {\bibfnamefont
  {M.}~\bibnamefont {Hoffmann}}, \bibinfo {author} {\bibfnamefont
  {S.}~\bibnamefont {Hong}}, \bibinfo {author} {\bibfnamefont {T.}~\bibnamefont
  {Huang}}, \bibinfo {author} {\bibfnamefont {A.}~\bibnamefont {Huff}},
  \bibinfo {author} {\bibfnamefont {W.~J.}\ \bibnamefont {Huggins}}, \bibinfo
  {author} {\bibfnamefont {L.~B.}\ \bibnamefont {Ioffe}}, \bibinfo {author}
  {\bibfnamefont {S.~V.}\ \bibnamefont {Isakov}}, \bibinfo {author}
  {\bibfnamefont {J.}~\bibnamefont {Iveland}}, \bibinfo {author} {\bibfnamefont
  {E.}~\bibnamefont {Jeffrey}}, \bibinfo {author} {\bibfnamefont
  {Z.}~\bibnamefont {Jiang}}, \bibinfo {author} {\bibfnamefont
  {C.}~\bibnamefont {Jones}}, \bibinfo {author} {\bibfnamefont
  {P.}~\bibnamefont {Juhas}}, \bibinfo {author} {\bibfnamefont
  {D.}~\bibnamefont {Kafri}}, \bibinfo {author} {\bibfnamefont
  {K.}~\bibnamefont {Kechedzhi}}, \bibinfo {author} {\bibfnamefont
  {J.}~\bibnamefont {Kelly}}, \bibinfo {author} {\bibfnamefont
  {T.}~\bibnamefont {Khattar}}, \bibinfo {author} {\bibfnamefont
  {M.}~\bibnamefont {Khezri}}, \bibinfo {author} {\bibfnamefont
  {M.}~\bibnamefont {Kieferov{\'a}}}, \bibinfo {author} {\bibfnamefont
  {S.}~\bibnamefont {Kim}}, \bibinfo {author} {\bibfnamefont {A.}~\bibnamefont
  {Kitaev}}, \bibinfo {author} {\bibfnamefont {P.~V.}\ \bibnamefont {Klimov}},
  \bibinfo {author} {\bibfnamefont {A.~R.}\ \bibnamefont {Klots}}, \bibinfo
  {author} {\bibfnamefont {A.~N.}\ \bibnamefont {Korotkov}}, \bibinfo {author}
  {\bibfnamefont {F.}~\bibnamefont {Kostritsa}}, \bibinfo {author}
  {\bibfnamefont {J.~M.}\ \bibnamefont {Kreikebaum}}, \bibinfo {author}
  {\bibfnamefont {D.}~\bibnamefont {Landhuis}}, \bibinfo {author}
  {\bibfnamefont {P.}~\bibnamefont {Laptev}}, \bibinfo {author} {\bibfnamefont
  {K.-M.}\ \bibnamefont {Lau}}, \bibinfo {author} {\bibfnamefont
  {L.}~\bibnamefont {Laws}}, \bibinfo {author} {\bibfnamefont {J.}~\bibnamefont
  {Lee}}, \bibinfo {author} {\bibfnamefont {K.}~\bibnamefont {Lee}}, \bibinfo
  {author} {\bibfnamefont {B.~J.}\ \bibnamefont {Lester}}, \bibinfo {author}
  {\bibfnamefont {A.}~\bibnamefont {Lill}}, \bibinfo {author} {\bibfnamefont
  {W.}~\bibnamefont {Liu}}, \bibinfo {author} {\bibfnamefont {A.}~\bibnamefont
  {Locharla}}, \bibinfo {author} {\bibfnamefont {E.}~\bibnamefont {Lucero}},
  \bibinfo {author} {\bibfnamefont {F.~D.}\ \bibnamefont {Malone}}, \bibinfo
  {author} {\bibfnamefont {J.}~\bibnamefont {Marshall}}, \bibinfo {author}
  {\bibfnamefont {O.}~\bibnamefont {Martin}}, \bibinfo {author} {\bibfnamefont
  {J.~R.}\ \bibnamefont {McClean}}, \bibinfo {author} {\bibfnamefont
  {T.}~\bibnamefont {McCourt}}, \bibinfo {author} {\bibfnamefont
  {M.}~\bibnamefont {McEwen}}, \bibinfo {author} {\bibfnamefont
  {A.}~\bibnamefont {Megrant}}, \bibinfo {author} {\bibfnamefont
  {B.}~\bibnamefont {Meurer~Costa}}, \bibinfo {author} {\bibfnamefont
  {X.}~\bibnamefont {Mi}}, \bibinfo {author} {\bibfnamefont {K.~C.}\
  \bibnamefont {Miao}}, \bibinfo {author} {\bibfnamefont {M.}~\bibnamefont
  {Mohseni}}, \bibinfo {author} {\bibfnamefont {S.}~\bibnamefont {Montazeri}},
  \bibinfo {author} {\bibfnamefont {A.}~\bibnamefont {Morvan}}, \bibinfo
  {author} {\bibfnamefont {E.}~\bibnamefont {Mount}}, \bibinfo {author}
  {\bibfnamefont {W.}~\bibnamefont {Mruczkiewicz}}, \bibinfo {author}
  {\bibfnamefont {O.}~\bibnamefont {Naaman}}, \bibinfo {author} {\bibfnamefont
  {M.}~\bibnamefont {Neeley}}, \bibinfo {author} {\bibfnamefont
  {C.}~\bibnamefont {Neill}}, \bibinfo {author} {\bibfnamefont
  {A.}~\bibnamefont {Nersisyan}}, \bibinfo {author} {\bibfnamefont
  {H.}~\bibnamefont {Neven}}, \bibinfo {author} {\bibfnamefont
  {M.}~\bibnamefont {Newman}}, \bibinfo {author} {\bibfnamefont {J.~H.}\
  \bibnamefont {Ng}}, \bibinfo {author} {\bibfnamefont {A.}~\bibnamefont
  {Nguyen}}, \bibinfo {author} {\bibfnamefont {M.}~\bibnamefont {Nguyen}},
  \bibinfo {author} {\bibfnamefont {M.~Y.}\ \bibnamefont {Niu}}, \bibinfo
  {author} {\bibfnamefont {T.~E.}\ \bibnamefont {O'Brien}}, \bibinfo {author}
  {\bibfnamefont {A.}~\bibnamefont {Opremcak}}, \bibinfo {author}
  {\bibfnamefont {J.}~\bibnamefont {Platt}}, \bibinfo {author} {\bibfnamefont
  {A.}~\bibnamefont {Petukhov}}, \bibinfo {author} {\bibfnamefont
  {R.}~\bibnamefont {Potter}}, \bibinfo {author} {\bibfnamefont {L.~P.}\
  \bibnamefont {Pryadko}}, \bibinfo {author} {\bibfnamefont {C.}~\bibnamefont
  {Quintana}}, \bibinfo {author} {\bibfnamefont {P.}~\bibnamefont {Roushan}},
  \bibinfo {author} {\bibfnamefont {N.~C.}\ \bibnamefont {Rubin}}, \bibinfo
  {author} {\bibfnamefont {N.}~\bibnamefont {Saei}}, \bibinfo {author}
  {\bibfnamefont {D.}~\bibnamefont {Sank}}, \bibinfo {author} {\bibfnamefont
  {K.}~\bibnamefont {Sankaragomathi}}, \bibinfo {author} {\bibfnamefont
  {K.~J.}\ \bibnamefont {Satzinger}}, \bibinfo {author} {\bibfnamefont {H.~F.}\
  \bibnamefont {Schurkus}}, \bibinfo {author} {\bibfnamefont {C.}~\bibnamefont
  {Schuster}}, \bibinfo {author} {\bibfnamefont {M.~J.}\ \bibnamefont
  {Shearn}}, \bibinfo {author} {\bibfnamefont {A.}~\bibnamefont {Shorter}},
  \bibinfo {author} {\bibfnamefont {V.}~\bibnamefont {Shvarts}}, \bibinfo
  {author} {\bibfnamefont {J.}~\bibnamefont {Skruzny}}, \bibinfo {author}
  {\bibfnamefont {V.}~\bibnamefont {Smelyanskiy}}, \bibinfo {author}
  {\bibfnamefont {W.~C.}\ \bibnamefont {Smith}}, \bibinfo {author}
  {\bibfnamefont {G.}~\bibnamefont {Sterling}}, \bibinfo {author}
  {\bibfnamefont {D.}~\bibnamefont {Strain}}, \bibinfo {author} {\bibfnamefont
  {M.}~\bibnamefont {Szalay}}, \bibinfo {author} {\bibfnamefont
  {A.}~\bibnamefont {Torres}}, \bibinfo {author} {\bibfnamefont
  {G.}~\bibnamefont {Vidal}}, \bibinfo {author} {\bibfnamefont
  {B.}~\bibnamefont {Villalonga}}, \bibinfo {author} {\bibfnamefont
  {C.}~\bibnamefont {Vollgraff~Heidweiller}}, \bibinfo {author} {\bibfnamefont
  {T.}~\bibnamefont {White}}, \bibinfo {author} {\bibfnamefont
  {C.}~\bibnamefont {Xing}}, \bibinfo {author} {\bibfnamefont {Z.~J.}\
  \bibnamefont {Yao}}, \bibinfo {author} {\bibfnamefont {P.}~\bibnamefont
  {Yeh}}, \bibinfo {author} {\bibfnamefont {J.}~\bibnamefont {Yoo}}, \bibinfo
  {author} {\bibfnamefont {G.}~\bibnamefont {Young}}, \bibinfo {author}
  {\bibfnamefont {A.}~\bibnamefont {Zalcman}}, \bibinfo {author} {\bibfnamefont
  {Y.}~\bibnamefont {Zhang}}, \bibinfo {author} {\bibfnamefont
  {N.}~\bibnamefont {Zhu}},\ and\ \bibinfo {author} {\bibfnamefont {G.~Q.}\
  \bibnamefont {AI}},\ }\bibfield  {title} {\bibinfo {title} {Suppressing
  quantum errors by scaling a surface code logical qubit},\ }\href
  {https://doi.org/10.1038/s41586-022-05434-1} {\bibfield  {journal} {\bibinfo
  {journal} {Nature}\ }\textbf {\bibinfo {volume} {614}},\ \bibinfo {pages}
  {676} (\bibinfo {year} {2023})}\BibitemShut {NoStop}%
\bibitem [{\citenamefont {Bluvstein}\ \emph {et~al.}(2024)\citenamefont
  {Bluvstein}, \citenamefont {Evered}, \citenamefont {Geim}, \citenamefont
  {Li}, \citenamefont {Zhou}, \citenamefont {Manovitz}, \citenamefont {Ebadi},
  \citenamefont {Cain}, \citenamefont {Kalinowski}, \citenamefont {Hangleiter},
  \citenamefont {Bonilla~Ataides}, \citenamefont {Maskara}, \citenamefont
  {Cong}, \citenamefont {Gao}, \citenamefont {Sales~Rodriguez}, \citenamefont
  {Karolyshyn}, \citenamefont {Semeghini}, \citenamefont {Gullans},
  \citenamefont {Greiner}, \citenamefont {Vuleti{\'c}},\ and\ \citenamefont
  {Lukin}}]{Bluvstein2024}%
  \BibitemOpen
  \bibfield  {author} {\bibinfo {author} {\bibfnamefont {D.}~\bibnamefont
  {Bluvstein}}, \bibinfo {author} {\bibfnamefont {S.~J.}\ \bibnamefont
  {Evered}}, \bibinfo {author} {\bibfnamefont {A.~A.}\ \bibnamefont {Geim}},
  \bibinfo {author} {\bibfnamefont {S.~H.}\ \bibnamefont {Li}}, \bibinfo
  {author} {\bibfnamefont {H.}~\bibnamefont {Zhou}}, \bibinfo {author}
  {\bibfnamefont {T.}~\bibnamefont {Manovitz}}, \bibinfo {author}
  {\bibfnamefont {S.}~\bibnamefont {Ebadi}}, \bibinfo {author} {\bibfnamefont
  {M.}~\bibnamefont {Cain}}, \bibinfo {author} {\bibfnamefont {M.}~\bibnamefont
  {Kalinowski}}, \bibinfo {author} {\bibfnamefont {D.}~\bibnamefont
  {Hangleiter}}, \bibinfo {author} {\bibfnamefont {J.~P.}\ \bibnamefont
  {Bonilla~Ataides}}, \bibinfo {author} {\bibfnamefont {N.}~\bibnamefont
  {Maskara}}, \bibinfo {author} {\bibfnamefont {I.}~\bibnamefont {Cong}},
  \bibinfo {author} {\bibfnamefont {X.}~\bibnamefont {Gao}}, \bibinfo {author}
  {\bibfnamefont {P.}~\bibnamefont {Sales~Rodriguez}}, \bibinfo {author}
  {\bibfnamefont {T.}~\bibnamefont {Karolyshyn}}, \bibinfo {author}
  {\bibfnamefont {G.}~\bibnamefont {Semeghini}}, \bibinfo {author}
  {\bibfnamefont {M.~J.}\ \bibnamefont {Gullans}}, \bibinfo {author}
  {\bibfnamefont {M.}~\bibnamefont {Greiner}}, \bibinfo {author} {\bibfnamefont
  {V.}~\bibnamefont {Vuleti{\'c}}},\ and\ \bibinfo {author} {\bibfnamefont
  {M.~D.}\ \bibnamefont {Lukin}},\ }\bibfield  {title} {\bibinfo {title}
  {Logical quantum processor based on reconfigurable atom arrays},\ }\href
  {https://doi.org/10.1038/s41586-023-06927-3} {\bibfield  {journal} {\bibinfo
  {journal} {Nature}\ }\textbf {\bibinfo {volume} {626}},\ \bibinfo {pages}
  {58} (\bibinfo {year} {2024})}\BibitemShut {NoStop}%
\bibitem [{\citenamefont {Li}\ and\ \citenamefont
  {Benjamin}(2017)}]{PhysRevX.7.021050}%
  \BibitemOpen
  \bibfield  {author} {\bibinfo {author} {\bibfnamefont {Y.}~\bibnamefont
  {Li}}\ and\ \bibinfo {author} {\bibfnamefont {S.~C.}\ \bibnamefont
  {Benjamin}},\ }\bibfield  {title} {\bibinfo {title} {Efficient variational
  quantum simulator incorporating active error minimization},\ }\href
  {https://doi.org/10.1103/PhysRevX.7.021050} {\bibfield  {journal} {\bibinfo
  {journal} {Physical Review X}\ }\textbf {\bibinfo {volume} {7}},\ \bibinfo
  {pages} {021050} (\bibinfo {year} {2017})}\BibitemShut {NoStop}%
\bibitem [{\citenamefont {Temme}\ \emph {et~al.}(2017)\citenamefont {Temme},
  \citenamefont {Bravyi},\ and\ \citenamefont
  {Gambetta}}]{PhysRevLett.119.180509}%
  \BibitemOpen
  \bibfield  {author} {\bibinfo {author} {\bibfnamefont {K.}~\bibnamefont
  {Temme}}, \bibinfo {author} {\bibfnamefont {S.}~\bibnamefont {Bravyi}},\ and\
  \bibinfo {author} {\bibfnamefont {J.~M.}\ \bibnamefont {Gambetta}},\
  }\bibfield  {title} {\bibinfo {title} {Error mitigation for short-depth
  quantum circuits},\ }\href {https://doi.org/10.1103/PhysRevLett.119.180509}
  {\bibfield  {journal} {\bibinfo  {journal} {Physical Review Letters}\
  }\textbf {\bibinfo {volume} {119}},\ \bibinfo {pages} {180509} (\bibinfo
  {year} {2017})}\BibitemShut {NoStop}%
\bibitem [{\citenamefont {Endo}\ \emph {et~al.}(2018)\citenamefont {Endo},
  \citenamefont {Benjamin},\ and\ \citenamefont {Li}}]{PhysRevX.8.031027}%
  \BibitemOpen
  \bibfield  {author} {\bibinfo {author} {\bibfnamefont {S.}~\bibnamefont
  {Endo}}, \bibinfo {author} {\bibfnamefont {S.~C.}\ \bibnamefont {Benjamin}},\
  and\ \bibinfo {author} {\bibfnamefont {Y.}~\bibnamefont {Li}},\ }\bibfield
  {title} {\bibinfo {title} {Practical quantum error mitigation for near-future
  applications},\ }\href {https://doi.org/10.1103/PhysRevX.8.031027} {\bibfield
   {journal} {\bibinfo  {journal} {Physical Review X}\ }\textbf {\bibinfo
  {volume} {8}},\ \bibinfo {pages} {031027} (\bibinfo {year}
  {2018})}\BibitemShut {NoStop}%
\bibitem [{\citenamefont {Endo}\ \emph {et~al.}(2021)\citenamefont {Endo},
  \citenamefont {Cai}, \citenamefont {Benjamin},\ and\ \citenamefont
  {Yuan}}]{doi:10.7566/JPSJ.90.032001}%
  \BibitemOpen
  \bibfield  {author} {\bibinfo {author} {\bibfnamefont {S.}~\bibnamefont
  {Endo}}, \bibinfo {author} {\bibfnamefont {Z.}~\bibnamefont {Cai}}, \bibinfo
  {author} {\bibfnamefont {S.~C.}\ \bibnamefont {Benjamin}},\ and\ \bibinfo
  {author} {\bibfnamefont {X.}~\bibnamefont {Yuan}},\ }\bibfield  {title}
  {\bibinfo {title} {Hybrid quantum-classical algorithms and quantum error
  mitigation},\ }\href {https://doi.org/10.7566/JPSJ.90.032001} {\bibfield
  {journal} {\bibinfo  {journal} {Journal of the Physical Society of Japan}\
  }\textbf {\bibinfo {volume} {90}},\ \bibinfo {pages} {032001} (\bibinfo
  {year} {2021})}\BibitemShut {NoStop}%
\bibitem [{\citenamefont {Cai}\ \emph {et~al.}(2023)\citenamefont {Cai},
  \citenamefont {Babbush}, \citenamefont {Benjamin}, \citenamefont {Endo},
  \citenamefont {Huggins}, \citenamefont {Li}, \citenamefont {McClean},\ and\
  \citenamefont {O'Brien}}]{RevModPhys.95.045005}%
  \BibitemOpen
  \bibfield  {author} {\bibinfo {author} {\bibfnamefont {Z.}~\bibnamefont
  {Cai}}, \bibinfo {author} {\bibfnamefont {R.}~\bibnamefont {Babbush}},
  \bibinfo {author} {\bibfnamefont {S.~C.}\ \bibnamefont {Benjamin}}, \bibinfo
  {author} {\bibfnamefont {S.}~\bibnamefont {Endo}}, \bibinfo {author}
  {\bibfnamefont {W.~J.}\ \bibnamefont {Huggins}}, \bibinfo {author}
  {\bibfnamefont {Y.}~\bibnamefont {Li}}, \bibinfo {author} {\bibfnamefont
  {J.~R.}\ \bibnamefont {McClean}},\ and\ \bibinfo {author} {\bibfnamefont
  {T.~E.}\ \bibnamefont {O'Brien}},\ }\bibfield  {title} {\bibinfo {title}
  {Quantum error mitigation},\ }\href
  {https://doi.org/10.1103/RevModPhys.95.045005} {\bibfield  {journal}
  {\bibinfo  {journal} {Reviews of Modern Physics}\ }\textbf {\bibinfo {volume}
  {95}},\ \bibinfo {pages} {045005} (\bibinfo {year} {2023})}\BibitemShut
  {NoStop}%
\bibitem [{\citenamefont {Cirac}\ \emph {et~al.}(1999)\citenamefont {Cirac},
  \citenamefont {Ekert},\ and\ \citenamefont
  {Macchiavello}}]{PhysRevLett.82.4344}%
  \BibitemOpen
  \bibfield  {author} {\bibinfo {author} {\bibfnamefont {J.~I.}\ \bibnamefont
  {Cirac}}, \bibinfo {author} {\bibfnamefont {A.~K.}\ \bibnamefont {Ekert}},\
  and\ \bibinfo {author} {\bibfnamefont {C.}~\bibnamefont {Macchiavello}},\
  }\bibfield  {title} {\bibinfo {title} {Optimal purification of single
  qubits},\ }\href {https://doi.org/10.1103/PhysRevLett.82.4344} {\bibfield
  {journal} {\bibinfo  {journal} {Physical Review Letters}\ }\textbf {\bibinfo
  {volume} {82}},\ \bibinfo {pages} {4344} (\bibinfo {year}
  {1999})}\BibitemShut {NoStop}%
\bibitem [{\citenamefont {Keyl}\ and\ \citenamefont {Werner}(2001)}]{Keyl2001}%
  \BibitemOpen
  \bibfield  {author} {\bibinfo {author} {\bibfnamefont {M.}~\bibnamefont
  {Keyl}}\ and\ \bibinfo {author} {\bibfnamefont {R.~F.}\ \bibnamefont
  {Werner}},\ }\bibfield  {title} {\bibinfo {title} {The rate of optimal
  purification procedures},\ }\href {https://doi.org/10.1007/PL00001027}
  {\bibfield  {journal} {\bibinfo  {journal} {Annales Henri Poincar{\'e}}\
  }\textbf {\bibinfo {volume} {2}},\ \bibinfo {pages} {1} (\bibinfo {year}
  {2001})}\BibitemShut {NoStop}%
\bibitem [{\citenamefont {Fiur\'a\ifmmode~\check{s}\else
  \v{s}\fi{}ek}(2004)}]{PhysRevA.70.032308}%
  \BibitemOpen
  \bibfield  {author} {\bibinfo {author} {\bibfnamefont {J.}~\bibnamefont
  {Fiur\'a\ifmmode~\check{s}\else \v{s}\fi{}ek}},\ }\bibfield  {title}
  {\bibinfo {title} {Optimal probabilistic cloning and purification of quantum
  states},\ }\href {https://doi.org/10.1103/PhysRevA.70.032308} {\bibfield
  {journal} {\bibinfo  {journal} {Physical Review A}\ }\textbf {\bibinfo
  {volume} {70}},\ \bibinfo {pages} {032308} (\bibinfo {year}
  {2004})}\BibitemShut {NoStop}%
\bibitem [{\citenamefont {Li}\ \emph {et~al.}(2025{\natexlab{a}})\citenamefont
  {Li}, \citenamefont {Fu}, \citenamefont {Isogawa}, \citenamefont {Silva},\
  and\ \citenamefont {Chuang}}]{li2025optimalquantumpurityamplification}%
  \BibitemOpen
  \bibfield  {author} {\bibinfo {author} {\bibfnamefont {Z.}~\bibnamefont
  {Li}}, \bibinfo {author} {\bibfnamefont {H.}~\bibnamefont {Fu}}, \bibinfo
  {author} {\bibfnamefont {T.}~\bibnamefont {Isogawa}}, \bibinfo {author}
  {\bibfnamefont {C.}~\bibnamefont {Silva}},\ and\ \bibinfo {author}
  {\bibfnamefont {I.}~\bibnamefont {Chuang}},\ }\href@noop {} {\bibinfo {title}
  {Optimal quantum purity amplification}},\ \bibinfo {howpublished} {ArXiv
  preprints} (\bibinfo {year} {2025}{\natexlab{a}}),\ \Eprint
  {https://arxiv.org/abs/2409.18167} {arXiv:2409.18167 [quant-ph]} \BibitemShut
  {NoStop}%
\bibitem [{\citenamefont {Childs}\ \emph {et~al.}(2025)\citenamefont {Childs},
  \citenamefont {Fu}, \citenamefont {Leung}, \citenamefont {Li}, \citenamefont
  {Ozols},\ and\ \citenamefont {Vyas}}]{Childs2025streamingquantum}%
  \BibitemOpen
  \bibfield  {author} {\bibinfo {author} {\bibfnamefont {A.~M.}\ \bibnamefont
  {Childs}}, \bibinfo {author} {\bibfnamefont {H.}~\bibnamefont {Fu}}, \bibinfo
  {author} {\bibfnamefont {D.}~\bibnamefont {Leung}}, \bibinfo {author}
  {\bibfnamefont {Z.}~\bibnamefont {Li}}, \bibinfo {author} {\bibfnamefont
  {M.}~\bibnamefont {Ozols}},\ and\ \bibinfo {author} {\bibfnamefont
  {V.}~\bibnamefont {Vyas}},\ }\bibfield  {title} {\bibinfo {title} {Streaming
  quantum state purification},\ }\href
  {https://doi.org/10.22331/q-2025-01-21-1603} {\bibfield  {journal} {\bibinfo
  {journal} {Quantum}\ }\textbf {\bibinfo {volume} {9}},\ \bibinfo {pages}
  {1603} (\bibinfo {year} {2025})}\BibitemShut {NoStop}%
\bibitem [{\citenamefont {Grier}\ \emph {et~al.}(2025)\citenamefont {Grier},
  \citenamefont {Leung}, \citenamefont {Li}, \citenamefont {Pashayan},\ and\
  \citenamefont {Schaeffer}}]{grier2025streamingquantumstatepurification}%
  \BibitemOpen
  \bibfield  {author} {\bibinfo {author} {\bibfnamefont {D.}~\bibnamefont
  {Grier}}, \bibinfo {author} {\bibfnamefont {D.}~\bibnamefont {Leung}},
  \bibinfo {author} {\bibfnamefont {Z.}~\bibnamefont {Li}}, \bibinfo {author}
  {\bibfnamefont {H.}~\bibnamefont {Pashayan}},\ and\ \bibinfo {author}
  {\bibfnamefont {L.}~\bibnamefont {Schaeffer}},\ }\href@noop {} {\bibinfo
  {title} {Streaming quantum state purification for general mixed states}},\
  \bibinfo {howpublished} {ArXiv preprints} (\bibinfo {year} {2025}),\ \Eprint
  {https://arxiv.org/abs/2503.22644} {arXiv:2503.22644 [quant-ph]} \BibitemShut
  {NoStop}%
\bibitem [{\citenamefont {Yao}\ \emph {et~al.}(2025)\citenamefont {Yao},
  \citenamefont {Chen}, \citenamefont {Huang}, \citenamefont {Chen},
  \citenamefont {Fu},\ and\ \citenamefont {Wang}}]{Yao_2025}%
  \BibitemOpen
  \bibfield  {author} {\bibinfo {author} {\bibfnamefont {H.}~\bibnamefont
  {Yao}}, \bibinfo {author} {\bibfnamefont {Y.-A.}\ \bibnamefont {Chen}},
  \bibinfo {author} {\bibfnamefont {E.}~\bibnamefont {Huang}}, \bibinfo
  {author} {\bibfnamefont {K.}~\bibnamefont {Chen}}, \bibinfo {author}
  {\bibfnamefont {H.}~\bibnamefont {Fu}},\ and\ \bibinfo {author}
  {\bibfnamefont {X.}~\bibnamefont {Wang}},\ }\bibfield  {title} {\bibinfo
  {title} {Protocols and trade-offs of quantum state purification},\ }\href
  {https://doi.org/10.1088/2058-9565/add17e} {\bibfield  {journal} {\bibinfo
  {journal} {Quantum Science and Technology}\ }\textbf {\bibinfo {volume}
  {10}},\ \bibinfo {pages} {035020} (\bibinfo {year} {2025})}\BibitemShut
  {NoStop}%
\bibitem [{\citenamefont {Bennett}\ \emph {et~al.}(1996)\citenamefont
  {Bennett}, \citenamefont {Brassard}, \citenamefont {Popescu}, \citenamefont
  {Schumacher}, \citenamefont {Smolin},\ and\ \citenamefont
  {Wootters}}]{PhysRevLett.76.722}%
  \BibitemOpen
  \bibfield  {author} {\bibinfo {author} {\bibfnamefont {C.~H.}\ \bibnamefont
  {Bennett}}, \bibinfo {author} {\bibfnamefont {G.}~\bibnamefont {Brassard}},
  \bibinfo {author} {\bibfnamefont {S.}~\bibnamefont {Popescu}}, \bibinfo
  {author} {\bibfnamefont {B.}~\bibnamefont {Schumacher}}, \bibinfo {author}
  {\bibfnamefont {J.~A.}\ \bibnamefont {Smolin}},\ and\ \bibinfo {author}
  {\bibfnamefont {W.~K.}\ \bibnamefont {Wootters}},\ }\bibfield  {title}
  {\bibinfo {title} {Purification of noisy entanglement and faithful
  teleportation via noisy channels},\ }\href
  {https://doi.org/10.1103/PhysRevLett.76.722} {\bibfield  {journal} {\bibinfo
  {journal} {Physical Review Letters}\ }\textbf {\bibinfo {volume} {76}},\
  \bibinfo {pages} {722} (\bibinfo {year} {1996})}\BibitemShut {NoStop}%
\bibitem [{\citenamefont {Pan}\ \emph {et~al.}(2001{\natexlab{a}})\citenamefont
  {Pan}, \citenamefont {Simon}, \citenamefont {Brukner},\ and\ \citenamefont
  {Zeilinger}}]{Pan2001}%
  \BibitemOpen
  \bibfield  {author} {\bibinfo {author} {\bibfnamefont {J.-W.}\ \bibnamefont
  {Pan}}, \bibinfo {author} {\bibfnamefont {C.}~\bibnamefont {Simon}}, \bibinfo
  {author} {\bibfnamefont {{\v C}.}~\bibnamefont {Brukner}},\ and\ \bibinfo
  {author} {\bibfnamefont {A.}~\bibnamefont {Zeilinger}},\ }\bibfield  {title}
  {\bibinfo {title} {Entanglement purification for quantum communication},\
  }\href {https://doi.org/10.1038/35074041} {\bibfield  {journal} {\bibinfo
  {journal} {Nature}\ }\textbf {\bibinfo {volume} {410}},\ \bibinfo {pages}
  {1067} (\bibinfo {year} {2001}{\natexlab{a}})}\BibitemShut {NoStop}%
\bibitem [{\citenamefont {Pan}\ \emph {et~al.}(2003)\citenamefont {Pan},
  \citenamefont {Gasparoni}, \citenamefont {Ursin}, \citenamefont {Weihs},\
  and\ \citenamefont {Zeilinger}}]{Pan2003}%
  \BibitemOpen
  \bibfield  {author} {\bibinfo {author} {\bibfnamefont {J.-W.}\ \bibnamefont
  {Pan}}, \bibinfo {author} {\bibfnamefont {S.}~\bibnamefont {Gasparoni}},
  \bibinfo {author} {\bibfnamefont {R.}~\bibnamefont {Ursin}}, \bibinfo
  {author} {\bibfnamefont {G.}~\bibnamefont {Weihs}},\ and\ \bibinfo {author}
  {\bibfnamefont {A.}~\bibnamefont {Zeilinger}},\ }\bibfield  {title} {\bibinfo
  {title} {Experimental entanglement purification of arbitrary unknown
  states},\ }\href {https://doi.org/10.1038/nature01623} {\bibfield  {journal}
  {\bibinfo  {journal} {Nature}\ }\textbf {\bibinfo {volume} {423}},\ \bibinfo
  {pages} {417} (\bibinfo {year} {2003})}\BibitemShut {NoStop}%
\bibitem [{\citenamefont {Reichle}\ \emph {et~al.}(2006)\citenamefont
  {Reichle}, \citenamefont {Leibfried}, \citenamefont {Knill}, \citenamefont
  {Britton}, \citenamefont {Blakestad}, \citenamefont {Jost}, \citenamefont
  {Langer}, \citenamefont {Ozeri}, \citenamefont {Seidelin},\ and\
  \citenamefont {Wineland}}]{Reichle2006}%
  \BibitemOpen
  \bibfield  {author} {\bibinfo {author} {\bibfnamefont {R.}~\bibnamefont
  {Reichle}}, \bibinfo {author} {\bibfnamefont {D.}~\bibnamefont {Leibfried}},
  \bibinfo {author} {\bibfnamefont {E.}~\bibnamefont {Knill}}, \bibinfo
  {author} {\bibfnamefont {J.}~\bibnamefont {Britton}}, \bibinfo {author}
  {\bibfnamefont {R.~B.}\ \bibnamefont {Blakestad}}, \bibinfo {author}
  {\bibfnamefont {J.~D.}\ \bibnamefont {Jost}}, \bibinfo {author}
  {\bibfnamefont {C.}~\bibnamefont {Langer}}, \bibinfo {author} {\bibfnamefont
  {R.}~\bibnamefont {Ozeri}}, \bibinfo {author} {\bibfnamefont
  {S.}~\bibnamefont {Seidelin}},\ and\ \bibinfo {author} {\bibfnamefont
  {D.~J.}\ \bibnamefont {Wineland}},\ }\bibfield  {title} {\bibinfo {title}
  {Experimental purification of two-atom entanglement},\ }\href
  {https://doi.org/10.1038/nature05146} {\bibfield  {journal} {\bibinfo
  {journal} {Nature}\ }\textbf {\bibinfo {volume} {443}},\ \bibinfo {pages}
  {838} (\bibinfo {year} {2006})}\BibitemShut {NoStop}%
\bibitem [{\citenamefont {D\"ur}\ \emph {et~al.}(1999)\citenamefont {D\"ur},
  \citenamefont {Briegel}, \citenamefont {Cirac},\ and\ \citenamefont
  {Zoller}}]{PhysRevA.59.169}%
  \BibitemOpen
  \bibfield  {author} {\bibinfo {author} {\bibfnamefont {W.}~\bibnamefont
  {D\"ur}}, \bibinfo {author} {\bibfnamefont {H.-J.}\ \bibnamefont {Briegel}},
  \bibinfo {author} {\bibfnamefont {J.~I.}\ \bibnamefont {Cirac}},\ and\
  \bibinfo {author} {\bibfnamefont {P.}~\bibnamefont {Zoller}},\ }\bibfield
  {title} {\bibinfo {title} {Quantum repeaters based on entanglement
  purification},\ }\href {https://doi.org/10.1103/PhysRevA.59.169} {\bibfield
  {journal} {\bibinfo  {journal} {Physical Review A}\ }\textbf {\bibinfo
  {volume} {59}},\ \bibinfo {pages} {169} (\bibinfo {year} {1999})}\BibitemShut
  {NoStop}%
\bibitem [{\citenamefont {Pan}\ \emph {et~al.}(2001{\natexlab{b}})\citenamefont
  {Pan}, \citenamefont {Simon}, \citenamefont {Brukner},\ and\ \citenamefont
  {Zeilinger}}]{Pan2001Entanglement}%
  \BibitemOpen
  \bibfield  {author} {\bibinfo {author} {\bibfnamefont {J.-W.}\ \bibnamefont
  {Pan}}, \bibinfo {author} {\bibfnamefont {C.}~\bibnamefont {Simon}}, \bibinfo
  {author} {\bibfnamefont {{\v C}.}~\bibnamefont {Brukner}},\ and\ \bibinfo
  {author} {\bibfnamefont {A.}~\bibnamefont {Zeilinger}},\ }\bibfield  {title}
  {\bibinfo {title} {Entanglement purification for quantum communication},\
  }\href {https://doi.org/10.1038/35074041} {\bibfield  {journal} {\bibinfo
  {journal} {Nature}\ }\textbf {\bibinfo {volume} {410}},\ \bibinfo {pages}
  {1067} (\bibinfo {year} {2001}{\natexlab{b}})}\BibitemShut {NoStop}%
\bibitem [{\citenamefont {D\"ur}\ and\ \citenamefont
  {Briegel}(2003)}]{PhysRevLett.90.067901}%
  \BibitemOpen
  \bibfield  {author} {\bibinfo {author} {\bibfnamefont {W.}~\bibnamefont
  {D\"ur}}\ and\ \bibinfo {author} {\bibfnamefont {H.-J.}\ \bibnamefont
  {Briegel}},\ }\bibfield  {title} {\bibinfo {title} {Entanglement purification
  for quantum computation},\ }\href
  {https://doi.org/10.1103/PhysRevLett.90.067901} {\bibfield  {journal}
  {\bibinfo  {journal} {Physical Review Letters}\ }\textbf {\bibinfo {volume}
  {90}},\ \bibinfo {pages} {067901} (\bibinfo {year} {2003})}\BibitemShut
  {NoStop}%
\bibitem [{\citenamefont {Hu}\ \emph {et~al.}(2021)\citenamefont {Hu},
  \citenamefont {Huang}, \citenamefont {Sheng}, \citenamefont {Zhou},
  \citenamefont {Liu}, \citenamefont {Guo}, \citenamefont {Zhang},
  \citenamefont {Xing}, \citenamefont {Huang}, \citenamefont {Li},\ and\
  \citenamefont {Guo}}]{PhysRevLett.126.010503}%
  \BibitemOpen
  \bibfield  {author} {\bibinfo {author} {\bibfnamefont {X.-M.}\ \bibnamefont
  {Hu}}, \bibinfo {author} {\bibfnamefont {C.-X.}\ \bibnamefont {Huang}},
  \bibinfo {author} {\bibfnamefont {Y.-B.}\ \bibnamefont {Sheng}}, \bibinfo
  {author} {\bibfnamefont {L.}~\bibnamefont {Zhou}}, \bibinfo {author}
  {\bibfnamefont {B.-H.}\ \bibnamefont {Liu}}, \bibinfo {author} {\bibfnamefont
  {Y.}~\bibnamefont {Guo}}, \bibinfo {author} {\bibfnamefont {C.}~\bibnamefont
  {Zhang}}, \bibinfo {author} {\bibfnamefont {W.-B.}\ \bibnamefont {Xing}},
  \bibinfo {author} {\bibfnamefont {Y.-F.}\ \bibnamefont {Huang}}, \bibinfo
  {author} {\bibfnamefont {C.-F.}\ \bibnamefont {Li}},\ and\ \bibinfo {author}
  {\bibfnamefont {G.-C.}\ \bibnamefont {Guo}},\ }\bibfield  {title} {\bibinfo
  {title} {Long-distance entanglement purification for quantum communication},\
  }\href {https://doi.org/10.1103/PhysRevLett.126.010503} {\bibfield  {journal}
  {\bibinfo  {journal} {Physical Review Letters}\ }\textbf {\bibinfo {volume}
  {126}},\ \bibinfo {pages} {010503} (\bibinfo {year} {2021})}\BibitemShut
  {NoStop}%
\bibitem [{\citenamefont {Yan}\ \emph {et~al.}(2022)\citenamefont {Yan},
  \citenamefont {Zhong}, \citenamefont {Chang}, \citenamefont {Bienfait},
  \citenamefont {Chou}, \citenamefont {Conner}, \citenamefont {Dumur},
  \citenamefont {Grebel}, \citenamefont {Povey},\ and\ \citenamefont
  {Cleland}}]{PhysRevLett.128.080504}%
  \BibitemOpen
  \bibfield  {author} {\bibinfo {author} {\bibfnamefont {H.}~\bibnamefont
  {Yan}}, \bibinfo {author} {\bibfnamefont {Y.}~\bibnamefont {Zhong}}, \bibinfo
  {author} {\bibfnamefont {H.-S.}\ \bibnamefont {Chang}}, \bibinfo {author}
  {\bibfnamefont {A.}~\bibnamefont {Bienfait}}, \bibinfo {author}
  {\bibfnamefont {M.-H.}\ \bibnamefont {Chou}}, \bibinfo {author}
  {\bibfnamefont {C.~R.}\ \bibnamefont {Conner}}, \bibinfo {author}
  {\bibfnamefont {E.}~\bibnamefont {Dumur}}, \bibinfo {author} {\bibfnamefont
  {J.}~\bibnamefont {Grebel}}, \bibinfo {author} {\bibfnamefont {R.~G.}\
  \bibnamefont {Povey}},\ and\ \bibinfo {author} {\bibfnamefont {A.~N.}\
  \bibnamefont {Cleland}},\ }\bibfield  {title} {\bibinfo {title} {Entanglement
  purification and protection in a superconducting quantum network},\ }\href
  {https://doi.org/10.1103/PhysRevLett.128.080504} {\bibfield  {journal}
  {\bibinfo  {journal} {Physical Review Letters}\ }\textbf {\bibinfo {volume}
  {128}},\ \bibinfo {pages} {080504} (\bibinfo {year} {2022})}\BibitemShut
  {NoStop}%
\bibitem [{\citenamefont {Fang}\ and\ \citenamefont
  {Liu}(2020)}]{PhysRevLett.125.060405}%
  \BibitemOpen
  \bibfield  {author} {\bibinfo {author} {\bibfnamefont {K.}~\bibnamefont
  {Fang}}\ and\ \bibinfo {author} {\bibfnamefont {Z.-W.}\ \bibnamefont {Liu}},\
  }\bibfield  {title} {\bibinfo {title} {No-go theorems for quantum resource
  purification},\ }\href {https://doi.org/10.1103/PhysRevLett.125.060405}
  {\bibfield  {journal} {\bibinfo  {journal} {Physical Review Letters}\
  }\textbf {\bibinfo {volume} {125}},\ \bibinfo {pages} {060405} (\bibinfo
  {year} {2020})}\BibitemShut {NoStop}%
\bibitem [{\citenamefont {Fang}\ and\ \citenamefont
  {Liu}(2022)}]{PRXQuantum.3.010337}%
  \BibitemOpen
  \bibfield  {author} {\bibinfo {author} {\bibfnamefont {K.}~\bibnamefont
  {Fang}}\ and\ \bibinfo {author} {\bibfnamefont {Z.-W.}\ \bibnamefont {Liu}},\
  }\bibfield  {title} {\bibinfo {title} {No-go theorems for quantum resource
  purification: New approach and channel theory},\ }\href
  {https://doi.org/10.1103/PRXQuantum.3.010337} {\bibfield  {journal} {\bibinfo
   {journal} {PRX Quantum}\ }\textbf {\bibinfo {volume} {3}},\ \bibinfo {pages}
  {010337} (\bibinfo {year} {2022})}\BibitemShut {NoStop}%
\bibitem [{\citenamefont {Regula}\ and\ \citenamefont
  {Takagi}(2021)}]{Regula2021}%
  \BibitemOpen
  \bibfield  {author} {\bibinfo {author} {\bibfnamefont {B.}~\bibnamefont
  {Regula}}\ and\ \bibinfo {author} {\bibfnamefont {R.}~\bibnamefont
  {Takagi}},\ }\bibfield  {title} {\bibinfo {title} {Fundamental limitations on
  distillation of quantum channel resources},\ }\href
  {https://doi.org/10.1038/s41467-021-24699-0} {\bibfield  {journal} {\bibinfo
  {journal} {Nature Communications}\ }\textbf {\bibinfo {volume} {12}},\
  \bibinfo {pages} {4411} (\bibinfo {year} {2021})}\BibitemShut {NoStop}%
\bibitem [{\citenamefont {Zang}\ \emph {et~al.}(2025)\citenamefont {Zang},
  \citenamefont {Chen}, \citenamefont {Chitambar}, \citenamefont {Suchara},\
  and\ \citenamefont {Zhong}}]{PhysRevLett.134.190803}%
  \BibitemOpen
  \bibfield  {author} {\bibinfo {author} {\bibfnamefont {A.}~\bibnamefont
  {Zang}}, \bibinfo {author} {\bibfnamefont {X.}~\bibnamefont {Chen}}, \bibinfo
  {author} {\bibfnamefont {E.}~\bibnamefont {Chitambar}}, \bibinfo {author}
  {\bibfnamefont {M.}~\bibnamefont {Suchara}},\ and\ \bibinfo {author}
  {\bibfnamefont {T.}~\bibnamefont {Zhong}},\ }\bibfield  {title} {\bibinfo
  {title} {No-go theorems for universal entanglement purification},\ }\href
  {https://doi.org/10.1103/PhysRevLett.134.190803} {\bibfield  {journal}
  {\bibinfo  {journal} {Physical Review Letters}\ }\textbf {\bibinfo {volume}
  {134}},\ \bibinfo {pages} {190803} (\bibinfo {year} {2025})}\BibitemShut
  {NoStop}%
\bibitem [{\citenamefont {He}\ \emph {et~al.}(2026)\citenamefont {He},
  \citenamefont {Zhu}, \citenamefont {Yao}, \citenamefont {Liu}, \citenamefont
  {Li},\ and\ \citenamefont {Wang}}]{bdw8-k91v}%
  \BibitemOpen
  \bibfield  {author} {\bibinfo {author} {\bibfnamefont {K.}~\bibnamefont
  {He}}, \bibinfo {author} {\bibfnamefont {C.}~\bibnamefont {Zhu}}, \bibinfo
  {author} {\bibfnamefont {H.}~\bibnamefont {Yao}}, \bibinfo {author}
  {\bibfnamefont {J.}~\bibnamefont {Liu}}, \bibinfo {author} {\bibfnamefont
  {Y.}~\bibnamefont {Li}},\ and\ \bibinfo {author} {\bibfnamefont
  {X.}~\bibnamefont {Wang}},\ }\bibfield  {title} {\bibinfo {title} {No-go
  theorems for universal quantum state purification via classically simulable
  operations},\ }\href {https://doi.org/10.1103/bdw8-k91v} {\bibfield
  {journal} {\bibinfo  {journal} {Physical Review Letters}\ }\textbf {\bibinfo
  {volume} {136}},\ \bibinfo {pages} {090204} (\bibinfo {year}
  {2026})}\BibitemShut {NoStop}%
\bibitem [{\citenamefont {Zhao}\ \emph {et~al.}(2026)\citenamefont {Zhao},
  \citenamefont {Chen}, \citenamefont {Zhao}, \citenamefont {Zhu},
  \citenamefont {Chiribella},\ and\ \citenamefont {Wang}}]{3bb1-pmtp}%
  \BibitemOpen
  \bibfield  {author} {\bibinfo {author} {\bibfnamefont {B.}~\bibnamefont
  {Zhao}}, \bibinfo {author} {\bibfnamefont {Y.-A.}\ \bibnamefont {Chen}},
  \bibinfo {author} {\bibfnamefont {X.}~\bibnamefont {Zhao}}, \bibinfo {author}
  {\bibfnamefont {C.}~\bibnamefont {Zhu}}, \bibinfo {author} {\bibfnamefont
  {G.}~\bibnamefont {Chiribella}},\ and\ \bibinfo {author} {\bibfnamefont
  {X.}~\bibnamefont {Wang}},\ }\bibfield  {title} {\bibinfo {title} {Power and
  limitations of distributed quantum state purification},\ }\href
  {https://doi.org/10.1103/3bb1-pmtp} {\bibfield  {journal} {\bibinfo
  {journal} {Physical Review Letters}\ }\textbf {\bibinfo {volume} {136}},\
  \bibinfo {pages} {090203} (\bibinfo {year} {2026})}\BibitemShut {NoStop}%
\bibitem [{\citenamefont {Chiribella}\ \emph
  {et~al.}(2008{\natexlab{a}})\citenamefont {Chiribella}, \citenamefont
  {D'Ariano},\ and\ \citenamefont {Perinotti}}]{Chiribella_2008}%
  \BibitemOpen
  \bibfield  {author} {\bibinfo {author} {\bibfnamefont {G.}~\bibnamefont
  {Chiribella}}, \bibinfo {author} {\bibfnamefont {G.~M.}\ \bibnamefont
  {D'Ariano}},\ and\ \bibinfo {author} {\bibfnamefont {P.}~\bibnamefont
  {Perinotti}},\ }\bibfield  {title} {\bibinfo {title} {Transforming quantum
  operations: Quantum supermaps},\ }\href
  {https://doi.org/10.1209/0295-5075/83/30004} {\bibfield  {journal} {\bibinfo
  {journal} {Europhysics Letters}\ }\textbf {\bibinfo {volume} {83}},\ \bibinfo
  {pages} {30004} (\bibinfo {year} {2008}{\natexlab{a}})}\BibitemShut {NoStop}%
\bibitem [{\citenamefont {Gour}(2019)}]{8678741}%
  \BibitemOpen
  \bibfield  {author} {\bibinfo {author} {\bibfnamefont {G.}~\bibnamefont
  {Gour}},\ }\bibfield  {title} {\bibinfo {title} {Comparison of quantum
  channels by superchannels},\ }\href
  {https://doi.org/10.1109/TIT.2019.2907989} {\bibfield  {journal} {\bibinfo
  {journal} {IEEE Transactions on Information Theory}\ }\textbf {\bibinfo
  {volume} {65}},\ \bibinfo {pages} {5880} (\bibinfo {year}
  {2019})}\BibitemShut {NoStop}%
\bibitem [{\citenamefont
  {Xiao}(2025)}]{xiao2025superchanneltearsgeneralizedoccams}%
  \BibitemOpen
  \bibfield  {author} {\bibinfo {author} {\bibfnamefont {Y.}~\bibnamefont
  {Xiao}},\ }\href@noop {} {\bibinfo {title} {Superchannel without tears: A
  generalized occam's razor for quantum processes}},\ \bibinfo {howpublished}
  {ArXiv preprints} (\bibinfo {year} {2025}),\ \Eprint
  {https://arxiv.org/abs/2512.02493} {arXiv:2512.02493 [quant-ph]} \BibitemShut
  {NoStop}%
\bibitem [{\citenamefont {Jamio\l{}kowski}(1972)}]{JAMIOLKOWSKI1972275}%
  \BibitemOpen
  \bibfield  {author} {\bibinfo {author} {\bibfnamefont {A.}~\bibnamefont
  {Jamio\l{}kowski}},\ }\bibfield  {title} {\bibinfo {title} {Linear
  transformations which preserve trace and positive semidefiniteness of
  operators},\ }\href
  {https://doi.org/https://doi.org/10.1016/0034-4877(72)90011-0} {\bibfield
  {journal} {\bibinfo  {journal} {Reports on Mathematical Physics}\ }\textbf
  {\bibinfo {volume} {3}},\ \bibinfo {pages} {275} (\bibinfo {year}
  {1972})}\BibitemShut {NoStop}%
\bibitem [{\citenamefont {Choi}(1975)}]{CHOI1975285}%
  \BibitemOpen
  \bibfield  {author} {\bibinfo {author} {\bibfnamefont {M.-D.}\ \bibnamefont
  {Choi}},\ }\bibfield  {title} {\bibinfo {title} {Completely positive linear
  maps on complex matrices},\ }\href
  {https://doi.org/10.1016/0024-3795(75)90075-0} {\bibfield  {journal}
  {\bibinfo  {journal} {Linear Algebra and its Applications}\ }\textbf
  {\bibinfo {volume} {10}},\ \bibinfo {pages} {285} (\bibinfo {year}
  {1975})}\BibitemShut {NoStop}%
\bibitem [{\citenamefont {Leung}\ and\ \citenamefont
  {Matthews}(2015)}]{10.1109/TIT.2015.2439953}%
  \BibitemOpen
  \bibfield  {author} {\bibinfo {author} {\bibfnamefont {D.}~\bibnamefont
  {Leung}}\ and\ \bibinfo {author} {\bibfnamefont {W.}~\bibnamefont
  {Matthews}},\ }\bibfield  {title} {\bibinfo {title} {On the power of
  ppt-preserving and non-signalling codes},\ }\href
  {https://doi.org/10.1109/TIT.2015.2439953} {\bibfield  {journal} {\bibinfo
  {journal} {IEEE Transactions on Information Theory}\ }\textbf {\bibinfo
  {volume} {61}},\ \bibinfo {pages} {4486–4499} (\bibinfo {year}
  {2015})}\BibitemShut {NoStop}%
\bibitem [{\citenamefont {Bennett}\ and\ \citenamefont
  {Brassard}(2014)}]{BENNETT20147}%
  \BibitemOpen
  \bibfield  {author} {\bibinfo {author} {\bibfnamefont {C.~H.}\ \bibnamefont
  {Bennett}}\ and\ \bibinfo {author} {\bibfnamefont {G.}~\bibnamefont
  {Brassard}},\ }\bibfield  {title} {\bibinfo {title} {Quantum cryptography:
  Public key distribution and coin tossing},\ }\href
  {https://doi.org/10.1016/j.tcs.2014.05.025} {\bibfield  {journal} {\bibinfo
  {journal} {Theoretical Computer Science}\ }\textbf {\bibinfo {volume}
  {560}},\ \bibinfo {pages} {7} (\bibinfo {year} {2014})}\BibitemShut {NoStop}%
\bibitem [{\citenamefont {Bäuml}\ \emph {et~al.}(2019)\citenamefont {Bäuml},
  \citenamefont {Das}, \citenamefont {Wang},\ and\ \citenamefont
  {Wilde}}]{bauml2019resourcetheoryentanglementbipartite}%
  \BibitemOpen
  \bibfield  {author} {\bibinfo {author} {\bibfnamefont {S.}~\bibnamefont
  {Bäuml}}, \bibinfo {author} {\bibfnamefont {S.}~\bibnamefont {Das}},
  \bibinfo {author} {\bibfnamefont {X.}~\bibnamefont {Wang}},\ and\ \bibinfo
  {author} {\bibfnamefont {M.~M.}\ \bibnamefont {Wilde}},\ }\href@noop {}
  {\bibinfo {title} {Resource theory of entanglement for bipartite quantum
  channels}},\ \bibinfo {howpublished} {ArXiv preprints} (\bibinfo {year}
  {2019}),\ \Eprint {https://arxiv.org/abs/1907.04181} {arXiv:1907.04181
  [quant-ph]} \BibitemShut {NoStop}%
\bibitem [{\citenamefont {Gour}\ and\ \citenamefont
  {Scandolo}(2020)}]{PhysRevLett.125.180505}%
  \BibitemOpen
  \bibfield  {author} {\bibinfo {author} {\bibfnamefont {G.}~\bibnamefont
  {Gour}}\ and\ \bibinfo {author} {\bibfnamefont {C.~M.}\ \bibnamefont
  {Scandolo}},\ }\bibfield  {title} {\bibinfo {title} {Dynamical
  entanglement},\ }\href {https://doi.org/10.1103/PhysRevLett.125.180505}
  {\bibfield  {journal} {\bibinfo  {journal} {Physical Review Letters}\
  }\textbf {\bibinfo {volume} {125}},\ \bibinfo {pages} {180505} (\bibinfo
  {year} {2020})}\BibitemShut {NoStop}%
\bibitem [{\citenamefont {Xing}\ \emph {et~al.}(2023)\citenamefont {Xing},
  \citenamefont {Feng}, \citenamefont {Fan}, \citenamefont {Ma}, \citenamefont
  {Bharti}, \citenamefont {Koh},\ and\ \citenamefont
  {Xiao}}]{xing2023fundamentallimitationscommunicationquantum}%
  \BibitemOpen
  \bibfield  {author} {\bibinfo {author} {\bibfnamefont {J.}~\bibnamefont
  {Xing}}, \bibinfo {author} {\bibfnamefont {T.}~\bibnamefont {Feng}}, \bibinfo
  {author} {\bibfnamefont {Z.}~\bibnamefont {Fan}}, \bibinfo {author}
  {\bibfnamefont {H.}~\bibnamefont {Ma}}, \bibinfo {author} {\bibfnamefont
  {K.}~\bibnamefont {Bharti}}, \bibinfo {author} {\bibfnamefont {D.~E.}\
  \bibnamefont {Koh}},\ and\ \bibinfo {author} {\bibfnamefont {Y.}~\bibnamefont
  {Xiao}},\ }\href@noop {} {\bibinfo {title} {Fundamental limitations on
  communication over a quantum network}},\ \bibinfo {howpublished} {ArXiv
  preprints} (\bibinfo {year} {2023}),\ \Eprint
  {https://arxiv.org/abs/2306.04983} {arXiv:2306.04983 [quant-ph]} \BibitemShut
  {NoStop}%
\bibitem [{\citenamefont {Li}\ \emph {et~al.}(2025{\natexlab{b}})\citenamefont
  {Li}, \citenamefont {Xing}, \citenamefont {Qu}, \citenamefont {Gao},
  \citenamefont {Xiao}, \citenamefont {Liu}, \citenamefont {Xiao},\ and\
  \citenamefont {Xue}}]{glc7-xy8t}%
  \BibitemOpen
  \bibfield  {author} {\bibinfo {author} {\bibfnamefont {Y.}~\bibnamefont
  {Li}}, \bibinfo {author} {\bibfnamefont {J.}~\bibnamefont {Xing}}, \bibinfo
  {author} {\bibfnamefont {D.}~\bibnamefont {Qu}}, \bibinfo {author}
  {\bibfnamefont {H.}~\bibnamefont {Gao}}, \bibinfo {author} {\bibfnamefont
  {L.}~\bibnamefont {Xiao}}, \bibinfo {author} {\bibfnamefont {J.-M.}\
  \bibnamefont {Liu}}, \bibinfo {author} {\bibfnamefont {Y.}~\bibnamefont
  {Xiao}},\ and\ \bibinfo {author} {\bibfnamefont {P.}~\bibnamefont {Xue}},\
  }\bibfield  {title} {\bibinfo {title} {Temporal asymmetry in entanglement
  distillation},\ }\href {https://doi.org/10.1103/glc7-xy8t} {\bibfield
  {journal} {\bibinfo  {journal} {Physical Review Letters}\ }\textbf {\bibinfo
  {volume} {135}},\ \bibinfo {pages} {170801} (\bibinfo {year}
  {2025}{\natexlab{b}})}\BibitemShut {NoStop}%
\bibitem [{\citenamefont {Chitambar}\ and\ \citenamefont
  {Gour}(2019)}]{RevModPhys.91.025001}%
  \BibitemOpen
  \bibfield  {author} {\bibinfo {author} {\bibfnamefont {E.}~\bibnamefont
  {Chitambar}}\ and\ \bibinfo {author} {\bibfnamefont {G.}~\bibnamefont
  {Gour}},\ }\bibfield  {title} {\bibinfo {title} {Quantum resource theories},\
  }\href {https://doi.org/10.1103/RevModPhys.91.025001} {\bibfield  {journal}
  {\bibinfo  {journal} {Review of Modern Physics}\ }\textbf {\bibinfo {volume}
  {91}},\ \bibinfo {pages} {025001} (\bibinfo {year} {2019})}\BibitemShut
  {NoStop}%
\bibitem [{\citenamefont {Condon}\ and\ \citenamefont
  {Shortley}(1935)}]{condon1935theory}%
  \BibitemOpen
  \bibfield  {author} {\bibinfo {author} {\bibfnamefont {E.~U.}\ \bibnamefont
  {Condon}}\ and\ \bibinfo {author} {\bibfnamefont {G.~H.}\ \bibnamefont
  {Shortley}},\ }\href
  {https://www.cambridge.org/sg/universitypress/subjects/physics/atomic-physics-molecular-physics-and-chemical-physics/theory-atomic-spectra?format=PB&isbn=9780521092098}
  {\emph {\bibinfo {title} {The Theory of Atomic Spectra}}}\ (\bibinfo
  {publisher} {Cambridge University Press},\ \bibinfo {year}
  {1935})\BibitemShut {NoStop}%
\bibitem [{\citenamefont {Wood}\ \emph {et~al.}(2015)\citenamefont {Wood},
  \citenamefont {Biamonte},\ and\ \citenamefont
  {Cory}}]{wood2015tensornetworksgraphicalcalculus}%
  \BibitemOpen
  \bibfield  {author} {\bibinfo {author} {\bibfnamefont {C.~J.}\ \bibnamefont
  {Wood}}, \bibinfo {author} {\bibfnamefont {J.~D.}\ \bibnamefont {Biamonte}},\
  and\ \bibinfo {author} {\bibfnamefont {D.~G.}\ \bibnamefont {Cory}},\
  }\href@noop {} {\bibinfo {title} {Tensor networks and graphical calculus for
  open quantum systems}},\ \bibinfo {howpublished} {ArXiv preprints} (\bibinfo
  {year} {2015}),\ \Eprint {https://arxiv.org/abs/1111.6950} {arXiv:1111.6950
  [quant-ph]} \BibitemShut {NoStop}%
\bibitem [{\citenamefont {Coecke}\ and\ \citenamefont
  {Kissinger}(2017)}]{Coecke_Kissinger_2017}%
  \BibitemOpen
  \bibfield  {author} {\bibinfo {author} {\bibfnamefont {B.}~\bibnamefont
  {Coecke}}\ and\ \bibinfo {author} {\bibfnamefont {A.}~\bibnamefont
  {Kissinger}},\ }\href {https://doi.org/10.1017/9781316219317} {\emph
  {\bibinfo {title} {Picturing Quantum Processes: A First Course in Quantum
  Theory and Diagrammatic Reasoning}}}\ (\bibinfo  {publisher} {Cambridge
  University Press},\ \bibinfo {year} {2017})\BibitemShut {NoStop}%
\bibitem [{\citenamefont {Bridgeman}\ and\ \citenamefont
  {Chubb}(2017)}]{Bridgeman_2017}%
  \BibitemOpen
  \bibfield  {author} {\bibinfo {author} {\bibfnamefont {J.~C.}\ \bibnamefont
  {Bridgeman}}\ and\ \bibinfo {author} {\bibfnamefont {C.~T.}\ \bibnamefont
  {Chubb}},\ }\bibfield  {title} {\bibinfo {title} {Hand-waving and
  interpretive dance: An introductory course on tensor networks},\ }\href
  {https://doi.org/10.1088/1751-8121/aa6dc3} {\bibfield  {journal} {\bibinfo
  {journal} {Journal of Physics A: Mathematical and Theoretical}\ }\textbf
  {\bibinfo {volume} {50}},\ \bibinfo {pages} {223001} (\bibinfo {year}
  {2017})}\BibitemShut {NoStop}%
\bibitem [{\citenamefont
  {Biamonte}(2020)}]{biamonte2020lecturesquantumtensornetworks}%
  \BibitemOpen
  \bibfield  {author} {\bibinfo {author} {\bibfnamefont {J.~D.}\ \bibnamefont
  {Biamonte}},\ }\href@noop {} {\bibinfo {title} {Lectures on quantum tensor
  networks}},\ \bibinfo {howpublished} {ArXiv preprints} (\bibinfo {year}
  {2020}),\ \Eprint {https://arxiv.org/abs/1912.10049} {arXiv:1912.10049
  [quant-ph]} \BibitemShut {NoStop}%
\bibitem [{\citenamefont {Collura}\ \emph {et~al.}(2024)\citenamefont
  {Collura}, \citenamefont {Lami}, \citenamefont {Ranabhat},\ and\
  \citenamefont {Santini}}]{Collura_2024}%
  \BibitemOpen
  \bibfield  {author} {\bibinfo {author} {\bibfnamefont {M.}~\bibnamefont
  {Collura}}, \bibinfo {author} {\bibfnamefont {G.}~\bibnamefont {Lami}},
  \bibinfo {author} {\bibfnamefont {N.}~\bibnamefont {Ranabhat}},\ and\
  \bibinfo {author} {\bibfnamefont {A.}~\bibnamefont {Santini}},\ }\href
  {https://doi.org/10.22323/9788898587049} {\emph {\bibinfo {title} {Tensor
  Network Techniques for Quantum Computation}}}\ (\bibinfo  {publisher} {SISSA
  Medialab s.r.l.},\ \bibinfo {year} {2024})\BibitemShut {NoStop}%
\bibitem [{\citenamefont {Nielsen}\ and\ \citenamefont
  {Chuang}(2010)}]{Nielsen_Chuang_2010}%
  \BibitemOpen
  \bibfield  {author} {\bibinfo {author} {\bibfnamefont {M.~A.}\ \bibnamefont
  {Nielsen}}\ and\ \bibinfo {author} {\bibfnamefont {I.~L.}\ \bibnamefont
  {Chuang}},\ }\href {https://doi.org/10.1017/CBO9780511976667} {\emph
  {\bibinfo {title} {Quantum Computation and Quantum Information}}}\ (\bibinfo
  {publisher} {Cambridge University Press},\ \bibinfo {year}
  {2010})\BibitemShut {NoStop}%
\bibitem [{\citenamefont {Caruso}\ \emph {et~al.}(2014)\citenamefont {Caruso},
  \citenamefont {Giovannetti}, \citenamefont {Lupo},\ and\ \citenamefont
  {Mancini}}]{RevModPhys.86.1203}%
  \BibitemOpen
  \bibfield  {author} {\bibinfo {author} {\bibfnamefont {F.}~\bibnamefont
  {Caruso}}, \bibinfo {author} {\bibfnamefont {V.}~\bibnamefont {Giovannetti}},
  \bibinfo {author} {\bibfnamefont {C.}~\bibnamefont {Lupo}},\ and\ \bibinfo
  {author} {\bibfnamefont {S.}~\bibnamefont {Mancini}},\ }\bibfield  {title}
  {\bibinfo {title} {Quantum channels and memory effects},\ }\href
  {https://doi.org/10.1103/RevModPhys.86.1203} {\bibfield  {journal} {\bibinfo
  {journal} {Reviews of Modern Physics}\ }\textbf {\bibinfo {volume} {86}},\
  \bibinfo {pages} {1203} (\bibinfo {year} {2014})}\BibitemShut {NoStop}%
\bibitem [{\citenamefont {Wilde}(2017)}]{Wilde_2017}%
  \BibitemOpen
  \bibfield  {author} {\bibinfo {author} {\bibfnamefont {M.~M.}\ \bibnamefont
  {Wilde}},\ }\href {https://doi.org/10.1017/9781316809976} {\emph {\bibinfo
  {title} {Quantum Information Theory}}}\ (\bibinfo  {publisher} {Cambridge
  University Press},\ \bibinfo {year} {2017})\BibitemShut {NoStop}%
\bibitem [{\citenamefont {Watrous}(2018)}]{Watrous_2018}%
  \BibitemOpen
  \bibfield  {author} {\bibinfo {author} {\bibfnamefont {J.}~\bibnamefont
  {Watrous}},\ }\href {https://doi.org/10.1017/9781316848142} {\emph {\bibinfo
  {title} {The Theory of Quantum Information}}}\ (\bibinfo  {publisher}
  {Cambridge University Press},\ \bibinfo {year} {2018})\BibitemShut {NoStop}%
\bibitem [{\citenamefont {Khatri}\ and\ \citenamefont
  {Wilde}(2024)}]{khatri2024principlesquantumcommunicationtheory}%
  \BibitemOpen
  \bibfield  {author} {\bibinfo {author} {\bibfnamefont {S.}~\bibnamefont
  {Khatri}}\ and\ \bibinfo {author} {\bibfnamefont {M.~M.}\ \bibnamefont
  {Wilde}},\ }\href@noop {} {\bibinfo {title} {Principles of quantum
  communication theory: A modern approach}},\ \bibinfo {howpublished} {Arxiv
  preprints} (\bibinfo {year} {2024}),\ \Eprint
  {https://arxiv.org/abs/2011.04672} {arXiv:2011.04672 [quant-ph]} \BibitemShut
  {NoStop}%
\bibitem [{\citenamefont {Rains}(2001)}]{959270}%
  \BibitemOpen
  \bibfield  {author} {\bibinfo {author} {\bibfnamefont {E.~M.}\ \bibnamefont
  {Rains}},\ }\bibfield  {title} {\bibinfo {title} {A semidefinite program for
  distillable entanglement},\ }\href {https://doi.org/10.1109/18.959270}
  {\bibfield  {journal} {\bibinfo  {journal} {IEEE Transactions on Information
  Theory}\ }\textbf {\bibinfo {volume} {47}},\ \bibinfo {pages} {2921}
  (\bibinfo {year} {2001})}\BibitemShut {NoStop}%
\bibitem [{\citenamefont {Chiribella}\ \emph
  {et~al.}(2008{\natexlab{b}})\citenamefont {Chiribella}, \citenamefont
  {D'Ariano},\ and\ \citenamefont {Perinotti}}]{PhysRevLett.101.060401}%
  \BibitemOpen
  \bibfield  {author} {\bibinfo {author} {\bibfnamefont {G.}~\bibnamefont
  {Chiribella}}, \bibinfo {author} {\bibfnamefont {G.~M.}\ \bibnamefont
  {D'Ariano}},\ and\ \bibinfo {author} {\bibfnamefont {P.}~\bibnamefont
  {Perinotti}},\ }\bibfield  {title} {\bibinfo {title} {Quantum circuit
  architecture},\ }\href {https://doi.org/10.1103/PhysRevLett.101.060401}
  {\bibfield  {journal} {\bibinfo  {journal} {Physical Review Letters}\
  }\textbf {\bibinfo {volume} {101}},\ \bibinfo {pages} {060401} (\bibinfo
  {year} {2008}{\natexlab{b}})}\BibitemShut {NoStop}%
\bibitem [{\citenamefont {Chiribella}\ \emph {et~al.}(2009)\citenamefont
  {Chiribella}, \citenamefont {D'Ariano},\ and\ \citenamefont
  {Perinotti}}]{PhysRevA.80.022339}%
  \BibitemOpen
  \bibfield  {author} {\bibinfo {author} {\bibfnamefont {G.}~\bibnamefont
  {Chiribella}}, \bibinfo {author} {\bibfnamefont {G.~M.}\ \bibnamefont
  {D'Ariano}},\ and\ \bibinfo {author} {\bibfnamefont {P.}~\bibnamefont
  {Perinotti}},\ }\bibfield  {title} {\bibinfo {title} {Theoretical framework
  for quantum networks},\ }\href {https://doi.org/10.1103/PhysRevA.80.022339}
  {\bibfield  {journal} {\bibinfo  {journal} {Physical Review A}\ }\textbf
  {\bibinfo {volume} {80}},\ \bibinfo {pages} {022339} (\bibinfo {year}
  {2009})}\BibitemShut {NoStop}%
\bibitem [{\citenamefont {Beckman}\ \emph {et~al.}(2001)\citenamefont
  {Beckman}, \citenamefont {Gottesman}, \citenamefont {Nielsen},\ and\
  \citenamefont {Preskill}}]{PhysRevA.64.052309}%
  \BibitemOpen
  \bibfield  {author} {\bibinfo {author} {\bibfnamefont {D.}~\bibnamefont
  {Beckman}}, \bibinfo {author} {\bibfnamefont {D.}~\bibnamefont {Gottesman}},
  \bibinfo {author} {\bibfnamefont {M.~A.}\ \bibnamefont {Nielsen}},\ and\
  \bibinfo {author} {\bibfnamefont {J.}~\bibnamefont {Preskill}},\ }\bibfield
  {title} {\bibinfo {title} {Causal and localizable quantum operations},\
  }\href {https://doi.org/10.1103/PhysRevA.64.052309} {\bibfield  {journal}
  {\bibinfo  {journal} {Physical Review A}\ }\textbf {\bibinfo {volume} {64}},\
  \bibinfo {pages} {052309} (\bibinfo {year} {2001})}\BibitemShut {NoStop}%
\bibitem [{\citenamefont {Eggeling}\ \emph {et~al.}(2002)\citenamefont
  {Eggeling}, \citenamefont {Schlingemann},\ and\ \citenamefont
  {Werner}}]{Eggeling_2002}%
  \BibitemOpen
  \bibfield  {author} {\bibinfo {author} {\bibfnamefont {T.}~\bibnamefont
  {Eggeling}}, \bibinfo {author} {\bibfnamefont {D.}~\bibnamefont
  {Schlingemann}},\ and\ \bibinfo {author} {\bibfnamefont {R.~F.}\ \bibnamefont
  {Werner}},\ }\bibfield  {title} {\bibinfo {title} {Semicausal operations are
  semilocalizable},\ }\href {https://doi.org/10.1209/epl/i2002-00579-4}
  {\bibfield  {journal} {\bibinfo  {journal} {Europhysics Letters}\ }\textbf
  {\bibinfo {volume} {57}},\ \bibinfo {pages} {782} (\bibinfo {year}
  {2002})}\BibitemShut {NoStop}%
\bibitem [{\citenamefont {Piani}\ \emph {et~al.}(2006)\citenamefont {Piani},
  \citenamefont {Horodecki}, \citenamefont {Horodecki},\ and\ \citenamefont
  {Horodecki}}]{PhysRevA.74.012305}%
  \BibitemOpen
  \bibfield  {author} {\bibinfo {author} {\bibfnamefont {M.}~\bibnamefont
  {Piani}}, \bibinfo {author} {\bibfnamefont {M.}~\bibnamefont {Horodecki}},
  \bibinfo {author} {\bibfnamefont {P.}~\bibnamefont {Horodecki}},\ and\
  \bibinfo {author} {\bibfnamefont {R.}~\bibnamefont {Horodecki}},\ }\bibfield
  {title} {\bibinfo {title} {Properties of quantum nonsignaling boxes},\ }\href
  {https://doi.org/10.1103/PhysRevA.74.012305} {\bibfield  {journal} {\bibinfo
  {journal} {Physical Review A}\ }\textbf {\bibinfo {volume} {74}},\ \bibinfo
  {pages} {012305} (\bibinfo {year} {2006})}\BibitemShut {NoStop}%
\bibitem [{\citenamefont {Horodecki}\ \emph {et~al.}(2000)\citenamefont
  {Horodecki}, \citenamefont {Horodecki},\ and\ \citenamefont
  {Horodecki}}]{Horodecki01022000}%
  \BibitemOpen
  \bibfield  {author} {\bibinfo {author} {\bibfnamefont {P.}~\bibnamefont
  {Horodecki}}, \bibinfo {author} {\bibfnamefont {M.}~\bibnamefont
  {Horodecki}},\ and\ \bibinfo {author} {\bibfnamefont {R.}~\bibnamefont
  {Horodecki}},\ }\bibfield  {title} {\bibinfo {title} {Binding entanglement
  channels},\ }\href {https://doi.org/10.1080/09500340008244047} {\bibfield
  {journal} {\bibinfo  {journal} {Journal of Modern Optics}\ }\textbf {\bibinfo
  {volume} {47}},\ \bibinfo {pages} {347} (\bibinfo {year} {2000})}\BibitemShut
  {NoStop}%
\bibitem [{\citenamefont {Uhlmann}(1976)}]{UHLMANN1976273}%
  \BibitemOpen
  \bibfield  {author} {\bibinfo {author} {\bibfnamefont {A.}~\bibnamefont
  {Uhlmann}},\ }\bibfield  {title} {\bibinfo {title} {The “transition
  probability” in the state space of a $\ast$-algebra},\ }\href
  {https://doi.org/10.1016/0034-4877(76)90060-4} {\bibfield  {journal}
  {\bibinfo  {journal} {Reports on Mathematical Physics}\ }\textbf {\bibinfo
  {volume} {9}},\ \bibinfo {pages} {273} (\bibinfo {year} {1976})}\BibitemShut
  {NoStop}%
\bibitem [{\citenamefont {Jorza}(1994)}]{Jozsa01121994}%
  \BibitemOpen
  \bibfield  {author} {\bibinfo {author} {\bibfnamefont {R.}~\bibnamefont
  {Jorza}},\ }\bibfield  {title} {\bibinfo {title} {Fidelity for mixed quantum
  states},\ }\href {https://doi.org/10.1080/09500349414552171} {\bibfield
  {journal} {\bibinfo  {journal} {Journal of Modern Optics}\ }\textbf {\bibinfo
  {volume} {41}},\ \bibinfo {pages} {2315} (\bibinfo {year}
  {1994})}\BibitemShut {NoStop}%
\bibitem [{\citenamefont {Shor}(1977)}]{Shor1977}%
  \BibitemOpen
  \bibfield  {author} {\bibinfo {author} {\bibfnamefont {N.~Z.}\ \bibnamefont
  {Shor}},\ }\bibfield  {title} {\bibinfo {title} {Cut-off method with space
  extension in convex programming problems},\ }\href
  {https://doi.org/10.1007/BF01071394} {\bibfield  {journal} {\bibinfo
  {journal} {Cybernetics}\ }\textbf {\bibinfo {volume} {13}},\ \bibinfo {pages}
  {94} (\bibinfo {year} {1977})}\BibitemShut {NoStop}%
\bibitem [{\citenamefont {G.}(1980)}]{KHACHIYAN198053}%
  \BibitemOpen
  \bibfield  {author} {\bibinfo {author} {\bibfnamefont {K.~L.}\ \bibnamefont
  {G.}},\ }\bibfield  {title} {\bibinfo {title} {Polynomial algorithms in
  linear programming},\ }\href {https://doi.org/10.1016/0041-5553(80)90061-0}
  {\bibfield  {journal} {\bibinfo  {journal} {USSR Computational Mathematics
  and Mathematical Physics}\ }\textbf {\bibinfo {volume} {20}},\ \bibinfo
  {pages} {53} (\bibinfo {year} {1980})}\BibitemShut {NoStop}%
\bibitem [{\citenamefont {Nemirovski}\ and\ \citenamefont
  {Yudin}(1983)}]{nemirovskij1983problem}%
  \BibitemOpen
  \bibfield  {author} {\bibinfo {author} {\bibfnamefont {A.~S.}\ \bibnamefont
  {Nemirovski}}\ and\ \bibinfo {author} {\bibfnamefont {D.~B.}\ \bibnamefont
  {Yudin}},\ }\href@noop {} {\emph {\bibinfo {title} {Problem complexity and
  method efficiency in optimization}}}\ (\bibinfo  {publisher}
  {Wiley-Interscience},\ \bibinfo {year} {1983})\BibitemShut {NoStop}%
\bibitem [{\citenamefont {Tarasov}\ \emph {et~al.}(1988)\citenamefont
  {Tarasov}, \citenamefont {Khachiyan},\ and\ \citenamefont
  {{\`E}rlikh}}]{TarKhaErl88}%
  \BibitemOpen
  \bibfield  {author} {\bibinfo {author} {\bibfnamefont {S.~P.}\ \bibnamefont
  {Tarasov}}, \bibinfo {author} {\bibfnamefont {L.~G.}\ \bibnamefont
  {Khachiyan}},\ and\ \bibinfo {author} {\bibfnamefont {I.~I.}\ \bibnamefont
  {{\`E}rlikh}},\ }\bibfield  {title} {\bibinfo {title} {The method of
  inscribed ellipsoids},\ }\href {http://mi.mathnet.ru/dan7782} {\bibfield
  {journal} {\bibinfo  {journal} {Soviet Mathematics Doklady}\ }\textbf
  {\bibinfo {volume} {37}},\ \bibinfo {pages} {226} (\bibinfo {year}
  {1988})}\BibitemShut {NoStop}%
\bibitem [{\citenamefont {Vaidya}(1989)}]{63500}%
  \BibitemOpen
  \bibfield  {author} {\bibinfo {author} {\bibfnamefont {P.~M.}\ \bibnamefont
  {Vaidya}},\ }\bibfield  {title} {\bibinfo {title} {A new algorithm for
  minimizing convex functions over convex sets},\ }in\ \href
  {https://doi.org/10.1109/SFCS.1989.63500} {\emph {\bibinfo {booktitle}
  {Proceedings of the 30th Annual Symposium on Foundations of Computer
  Science}}}\ (\bibinfo {year} {1989})\ pp.\ \bibinfo {pages}
  {338--343}\BibitemShut {NoStop}%
\bibitem [{\citenamefont {Nesterov}\ and\ \citenamefont
  {Nemirovski}(1992)}]{Nesterov01011992}%
  \BibitemOpen
  \bibfield  {author} {\bibinfo {author} {\bibfnamefont {Y.}~\bibnamefont
  {Nesterov}}\ and\ \bibinfo {author} {\bibfnamefont {A.}~\bibnamefont
  {Nemirovski}},\ }\bibfield  {title} {\bibinfo {title} {Conic formulation of a
  convex programming problem and duality †},\ }\href
  {https://doi.org/10.1080/10556789208805510} {\bibfield  {journal} {\bibinfo
  {journal} {Optimization Methods and Software}\ }\textbf {\bibinfo {volume}
  {1}},\ \bibinfo {pages} {95} (\bibinfo {year} {1992})}\BibitemShut {NoStop}%
\bibitem [{\citenamefont {Nesterov}\ and\ \citenamefont
  {Nemirovski}(1994)}]{nesterov1994interior}%
  \BibitemOpen
  \bibfield  {author} {\bibinfo {author} {\bibfnamefont {Y.}~\bibnamefont
  {Nesterov}}\ and\ \bibinfo {author} {\bibfnamefont {A.}~\bibnamefont
  {Nemirovski}},\ }\href {https://doi.org/10.1137/1.9781611970791} {\emph
  {\bibinfo {title} {Interior-Point Polynomial Algorithms in Convex
  Programming}}}\ (\bibinfo  {publisher} {Society for Industrial and Applied
  Mathematics},\ \bibinfo {year} {1994})\BibitemShut {NoStop}%
\bibitem [{\citenamefont {Anstreicher}(2000)}]{10.1287/moor.25.3.365.12212}%
  \BibitemOpen
  \bibfield  {author} {\bibinfo {author} {\bibfnamefont {K.~M.}\ \bibnamefont
  {Anstreicher}},\ }\bibfield  {title} {\bibinfo {title} {The volumetric
  barrier for semidefinite programming},\ }\href
  {https://doi.org/10.1287/moor.25.3.365.12212} {\bibfield  {journal} {\bibinfo
   {journal} {Mathematics of Operations ResearchVol}\ }\textbf {\bibinfo
  {volume} {25}},\ \bibinfo {pages} {365–380} (\bibinfo {year}
  {2000})}\BibitemShut {NoStop}%
\bibitem [{\citenamefont {Sivaramakrishnan}\ and\ \citenamefont
  {Mitchell}(2007)}]{sivaramakrishnan2007properties}%
  \BibitemOpen
  \bibfield  {author} {\bibinfo {author} {\bibfnamefont {K.~K.}\ \bibnamefont
  {Sivaramakrishnan}}\ and\ \bibinfo {author} {\bibfnamefont {J.~E.}\
  \bibnamefont {Mitchell}},\ }\bibfield  {title} {\bibinfo {title} {Properties
  of a cutting plane method for semidefinite programming},\ }\href
  {http://www.yokohamapublishers.jp/online2/pjov9.html} {\bibfield  {journal}
  {\bibinfo  {journal} {Pacific Journal of Optimization}\ }\textbf {\bibinfo
  {volume} {8}},\ \bibinfo {pages} {779} (\bibinfo {year} {2007})}\BibitemShut
  {NoStop}%
\bibitem [{\citenamefont {Lee}\ \emph {et~al.}(2015)\citenamefont {Lee},
  \citenamefont {Sidford},\ and\ \citenamefont {Wong}}]{7354442}%
  \BibitemOpen
  \bibfield  {author} {\bibinfo {author} {\bibfnamefont {Y.~T.}\ \bibnamefont
  {Lee}}, \bibinfo {author} {\bibfnamefont {A.}~\bibnamefont {Sidford}},\ and\
  \bibinfo {author} {\bibfnamefont {S.~C.-W.}\ \bibnamefont {Wong}},\
  }\bibfield  {title} {\bibinfo {title} {A faster cutting plane method and its
  implications for combinatorial and convex optimization},\ }in\ \href
  {https://doi.org/10.1109/FOCS.2015.68} {\emph {\bibinfo {booktitle}
  {Proceedings of the 56th Annual Symposium on Foundations of Computer
  Science}}}\ (\bibinfo {year} {2015})\ pp.\ \bibinfo {pages}
  {1049--1065}\BibitemShut {NoStop}%
\bibitem [{\citenamefont {Jiang}\ \emph
  {et~al.}(2020{\natexlab{a}})\citenamefont {Jiang}, \citenamefont {Lee},
  \citenamefont {Song},\ and\ \citenamefont {Wong}}]{10.1145/3357713.3384284}%
  \BibitemOpen
  \bibfield  {author} {\bibinfo {author} {\bibfnamefont {H.}~\bibnamefont
  {Jiang}}, \bibinfo {author} {\bibfnamefont {Y.~T.}\ \bibnamefont {Lee}},
  \bibinfo {author} {\bibfnamefont {Z.}~\bibnamefont {Song}},\ and\ \bibinfo
  {author} {\bibfnamefont {S.~C.-W.}\ \bibnamefont {Wong}},\ }\bibfield
  {title} {\bibinfo {title} {An improved cutting plane method for convex
  optimization, convex-concave games, and its applications},\ }in\ \href
  {https://doi.org/10.1145/3357713.3384284} {\emph {\bibinfo {booktitle}
  {Proceedings of the 52nd Annual ACM SIGACT Symposium on Theory of
  Computing}}}\ (\bibinfo {year} {2020})\ p.\ \bibinfo {pages}
  {944–953}\BibitemShut {NoStop}%
\bibitem [{\citenamefont {Jiang}\ \emph
  {et~al.}(2020{\natexlab{b}})\citenamefont {Jiang}, \citenamefont {Kathuria},
  \citenamefont {Lee}, \citenamefont {Padmanabhan},\ and\ \citenamefont
  {Song}}]{9317892}%
  \BibitemOpen
  \bibfield  {author} {\bibinfo {author} {\bibfnamefont {H.}~\bibnamefont
  {Jiang}}, \bibinfo {author} {\bibfnamefont {T.}~\bibnamefont {Kathuria}},
  \bibinfo {author} {\bibfnamefont {Y.~T.}\ \bibnamefont {Lee}}, \bibinfo
  {author} {\bibfnamefont {S.}~\bibnamefont {Padmanabhan}},\ and\ \bibinfo
  {author} {\bibfnamefont {Z.}~\bibnamefont {Song}},\ }\bibfield  {title}
  {\bibinfo {title} {A faster interior point method for semidefinite
  programming},\ }in\ \href {https://doi.org/10.1109/FOCS46700.2020.00089}
  {\emph {\bibinfo {booktitle} {Proceedings of the 61st Annual Symposium on
  Foundations of Computer Science}}}\ (\bibinfo {year} {2020})\ pp.\ \bibinfo
  {pages} {910--918}\BibitemShut {NoStop}%
\bibitem [{\citenamefont {Parzygnat}\ \emph {et~al.}(2024)\citenamefont
  {Parzygnat}, \citenamefont {Fullwood}, \citenamefont {Buscemi},\ and\
  \citenamefont {Chiribella}}]{PhysRevLett.132.110203}%
  \BibitemOpen
  \bibfield  {author} {\bibinfo {author} {\bibfnamefont {A.~J.}\ \bibnamefont
  {Parzygnat}}, \bibinfo {author} {\bibfnamefont {J.}~\bibnamefont {Fullwood}},
  \bibinfo {author} {\bibfnamefont {F.}~\bibnamefont {Buscemi}},\ and\ \bibinfo
  {author} {\bibfnamefont {G.}~\bibnamefont {Chiribella}},\ }\bibfield  {title}
  {\bibinfo {title} {Virtual quantum broadcasting},\ }\href
  {https://doi.org/10.1103/PhysRevLett.132.110203} {\bibfield  {journal}
  {\bibinfo  {journal} {Physical Review Letters}\ }\textbf {\bibinfo {volume}
  {132}},\ \bibinfo {pages} {110203} (\bibinfo {year} {2024})}\BibitemShut
  {NoStop}%
\bibitem [{\citenamefont {Xiao}\ \emph
  {et~al.}(2025{\natexlab{a}})\citenamefont {Xiao}, \citenamefont {Liu},\ and\
  \citenamefont {Liu}}]{z2pr-zbwl}%
  \BibitemOpen
  \bibfield  {author} {\bibinfo {author} {\bibfnamefont {Y.}~\bibnamefont
  {Xiao}}, \bibinfo {author} {\bibfnamefont {X.}~\bibnamefont {Liu}},\ and\
  \bibinfo {author} {\bibfnamefont {Z.}~\bibnamefont {Liu}},\ }\bibfield
  {title} {\bibinfo {title} {No practical quantum broadcasting: Even
  virtually},\ }\href {https://doi.org/10.1103/z2pr-zbwl} {\bibfield  {journal}
  {\bibinfo  {journal} {Physical Review Letters}\ }\textbf {\bibinfo {volume}
  {135}},\ \bibinfo {pages} {090202} (\bibinfo {year}
  {2025}{\natexlab{a}})}\BibitemShut {NoStop}%
\bibitem [{\citenamefont {Xiao}\ \emph
  {et~al.}(2025{\natexlab{b}})\citenamefont {Xiao}, \citenamefont {Liu},\ and\
  \citenamefont {Liu}}]{8g6j-w7ld}%
  \BibitemOpen
  \bibfield  {author} {\bibinfo {author} {\bibfnamefont {Y.}~\bibnamefont
  {Xiao}}, \bibinfo {author} {\bibfnamefont {X.}~\bibnamefont {Liu}},\ and\
  \bibinfo {author} {\bibfnamefont {Z.}~\bibnamefont {Liu}},\ }\bibfield
  {title} {\bibinfo {title} {No practical quantum broadcasting: General
  framework},\ }\href {https://doi.org/10.1103/8g6j-w7ld} {\bibfield  {journal}
  {\bibinfo  {journal} {Physical Review Research}\ }\textbf {\bibinfo {volume}
  {7}},\ \bibinfo {pages} {033194} (\bibinfo {year}
  {2025}{\natexlab{b}})}\BibitemShut {NoStop}%
\bibitem [{\citenamefont {Okada}\ and\ \citenamefont
  {Buscemi}(2025)}]{okada2025virtualphasecovariantquantumbroadcasting}%
  \BibitemOpen
  \bibfield  {author} {\bibinfo {author} {\bibfnamefont {R.}~\bibnamefont
  {Okada}}\ and\ \bibinfo {author} {\bibfnamefont {F.}~\bibnamefont
  {Buscemi}},\ }\href@noop {} {\bibinfo {title} {Virtual phase-covariant
  quantum broadcasting for qubits}},\ \bibinfo {howpublished} {ArXiv preprints}
  (\bibinfo {year} {2025}),\ \Eprint {https://arxiv.org/abs/2511.20014}
  {arXiv:2511.20014 [quant-ph]} \BibitemShut {NoStop}%
\bibitem [{\citenamefont {Wang}\ and\ \citenamefont
  {Xiao}(2026)}]{wang2026practicalquantumbroadcasting}%
  \BibitemOpen
  \bibfield  {author} {\bibinfo {author} {\bibfnamefont {X.}~\bibnamefont
  {Wang}}\ and\ \bibinfo {author} {\bibfnamefont {Y.}~\bibnamefont {Xiao}},\
  }\href@noop {} {\bibinfo {title} {Practical quantum broadcasting}},\ \bibinfo
  {howpublished} {ArXiv preprints} (\bibinfo {year} {2026}),\ \Eprint
  {https://arxiv.org/abs/2603.19089} {arXiv:2603.19089 [quant-ph]} \BibitemShut
  {NoStop}%
\end{thebibliography}%


\newpage

\onecolumngrid

\appendix

\tableofcontents

\section*{Appendices}

The appendices develop, in a self-contained manner, the conceptual and technical foundations underlying forward-assisted quantum state purification, and trace their consequences across both global and distributed settings.
Section~\ref{sec:Preliminaries} establishes the formal groundwork, introducing the language of quantum channels and superchannels, the noise models considered, and the structural role of Bell states. 
Section~\ref{sec:Forward-Assisted_Purification} then formulates purification within a genuinely spatiotemporal framework, casting forward-assisted protocols as optimizations over superchannels and characterizing their fundamental limits via semidefinite programming. 
Building on this foundation, Sec.~\ref{sec:Advantages_GQSP} demonstrates the resulting operational advantages in the global setting, revealing a clear hierarchy of protocols and identifying regimes where pre-processing and quantum memory yield substantial gains. 
Section~\ref{sec:Advantages_DQSP} extends these insights to distributed scenarios, where pre-processing alone can already activate performance beyond conventional limits, and where forward-assisted protocols overcome established no-go constraints. 
Against this backdrop, Sec.~\ref{sec:Outline} provides a concise guide to the logical structure of the appendix and the progression of ideas developed throughout.


\section{Outline}\label{sec:Outline}

This section provides a concise overview of the appendix, distilling its structure and central results into a unified narrative that guides the reader through the logical progression of the analysis. 
Figure~\ref{fig:Outline} presents a schematic overview of this architecture. 
The exposition is organized into four interlocking parts, each centred on a distinct theme and arranged to build cumulatively toward the full framework.

\begin{figure}[htbp]
    \centering   
    \includegraphics[width=1\textwidth]{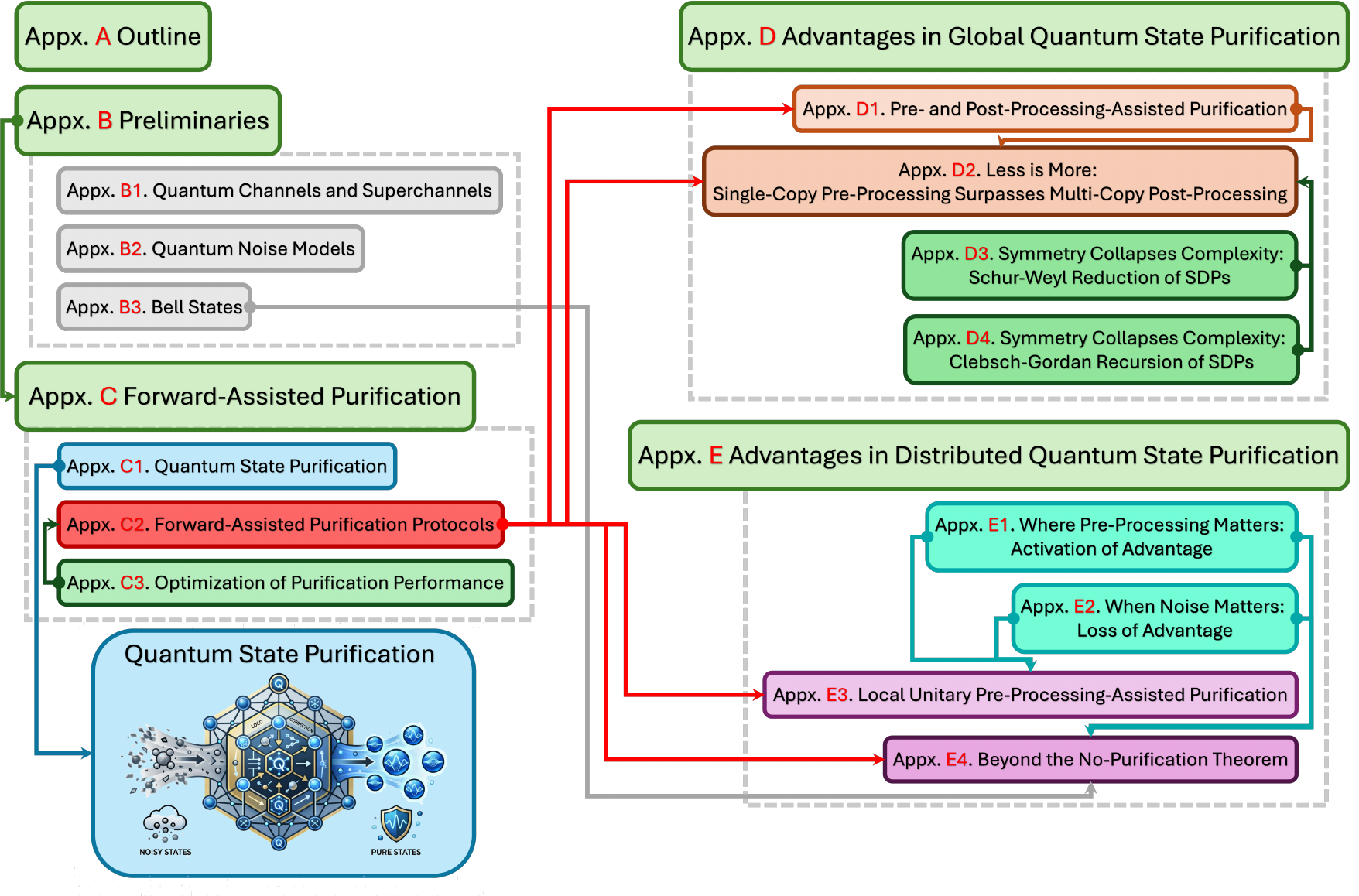}
    \caption{\textbf{Schematic Overview}. 
        The diagram traces the conceptual progression from foundational definitions to the operational advantages in forward-assisted (FA) purifications. 
        Section~\ref{sec:Preliminaries} establishes the underlying formalism, including quantum channels, superchannels, noise models, and Bell states, which supports Sec.~\ref{sec:Forward-Assisted_Purification}, where the FA framework is developed within the superchannel formalism and its performance characterized via semidefinite programs (SDPs). 
        On this basis, Secs.~\ref{sec:Advantages_GQSP} and \ref{sec:Advantages_DQSP} demonstrate the resulting advantages in global and distributed settings, respectively, identifying regimes where pre-processing yields performance beyond conventional limits and where symmetry reduces computational complexity. 
        Arrows encode the logical dependencies and propagation of ideas across the appendix.
    }
    \label{fig:Outline}
\end{figure}

\begin{itemize}
  \item \textbf{Foundational Preparations.}
  The first part Sec.~\ref{sec:Preliminaries} lays the formal groundwork for the framework by fixing the notation and core structures that support the subsequent analysis. 
  Subsection~\ref{subsec:Channels_Superchannels} introduces the language of quantum channels and their higher-order extension to superchannels, providing the natural setting for spatiotemporal purification via forward-assisted protocols. Subsection~\ref{subsec:Quantum_Noise_Models} specifies the noise models considered throughout, capturing the dominant error mechanisms relevant to realistic implementations. 
  Subsection~\ref{subsec:Bell_States} introduces the Bell basis, which serves as the central object for the analysis of the no-purification theorem in distributed scenarios.

  \item \textbf{Spatiotemporal Purifications.}
  The second part Sec.~\ref{sec:Forward-Assisted_Purification} revisits the foundations of quantum state purification and recasts them within a genuinely spatiotemporal setting, moving beyond the conventional single-time-point paradigm to a framework that explicitly accounts for the dynamical structure of noise.
  Subsection~\ref{subsec:Quantum_State_Purification} recalls the standard approach, in which noisy states are treated as static resources and purification is applied only after noise has taken place.
  While conceptually straightforward, this viewpoint neglects the temporal structure intrinsic to realistic noise processes. 
  In fact, noise is naturally described at the level of quantum channels, and its most general manipulation is captured by superchannels. 
  To incorporate this missing layer, we elevate the description to the level of superchannels --- the most general transformations of quantum dynamics --- and, within this formalism, formulate forward-assisted purification as a unified spatiotemporal framework for noisy state purification. 
  As developed in Subsec.~\ref{subsec:Forward_Assisted_Purification_Protocols}, this approach encompasses a broad family of protocols shaped by distinct physical constraints, allowing their operational features to be compared on equal footing. 
  Subsection~\ref{subsec:Optimization_Purification_Performance} then shows that the fundamental limits of these protocols admit a concise characterization in terms of semidefinite programs (SDPs), yielding a computationally tractable formulation that enables systematic comparison with conventional purification schemes.

  \item \textbf{Global Purifications.}
  The third part Sec.~\ref{sec:Advantages_GQSP}, brings the spatiotemporal framework to the setting of global quantum state purification, demonstrating that forward-assisted (FA) protocols can outperform conventional schemes restricted to post-processing alone. 
  Subsection~\ref{subsec:PreP_Post} places the various FA strategies, including pre-processing-augmented (PreP), $\mathrm{PPT}\cap \mathrm{NS}$, $\mathrm{PPT}$, and $\mathrm{NS}$, within a unified benchmark against conventional purification. 
  A clear hierarchy emerges. 
  Even the simplest FA protocol, namely PreP, which operates without quantum memory, already exceeds the conventional baseline. 
  Incorporating quantum memory, through $\mathrm{PPT}\cap \mathrm{NS}$, $\mathrm{PPT}$, or $\mathrm{NS}$, yields progressively stronger improvements.
  This superior performance stems directly from the enlarged operational flexibility afforded by our spatiotemporal framework.
  Subsection~\ref{subsec:Less_is_More} isolates the role of pre-processing and exposes a sharp separation. 
  In some noise regimes, a single-copy PreP already surpasses conventional purification using more than 50 copies.
  Extrapolation sharpens this contrast: 
  matching the PreP performance would require on the order of $10^{3}$ noisy copies;
  approximately 1918 under a linear fit and 2431 under a log-exponential fit.
  This reflects a substantial reduction in sample consumption enabled by the proposed protocol and signals a distinct scaling behaviour arising from our spatiotemporal framework.
  The optimal performance of each protocol is obtained by solving an associated semidefinite program (SDP). 
  For qubit systems, brute-force semidefinite programming implementations become prohibitive beyond roughly 8 copies. 
  Subsections~\ref{subsec:SDP_Symmetry} and \ref{subsec:SDP_CG} overcome this barrier by combining Schur-Weyl duality with Clebsch-Gordan recursion, reducing the optimization to symmetry-adapted blocks. 
  This reduction extends the accessible regime to around 50 copies, enabling the comparisons analyzed above.

  \item \textbf{Distributed Purifications.}
  The fourth part, Sec.~\ref{sec:Advantages_DQSP}, extends the spatiotemporal framework to distributed quantum state purification and demonstrates that forward-assisted (FA) protocols can outperform schemes restricted to post-processing alone. 
  In contrast to the global setting analyzed in Sec.~\ref{sec:Advantages_GQSP}, this advantage can already be realized solely through pre-processing, without any subsequent post-processing stage. 
  Subsection~\ref{subsec:Where_PreP_Matters} isolates the role of pre-processing under amplitude damping noise, while Subsec.~\ref{subsec:When_Noise_Matters} examines the corresponding behaviour under depolarizing channels.
  Building on these results, Subsec.~\ref{subsec:LU_Pre_Processing} formulates the full spatiotemporal framework in the distributed setting and establishes that forward-assisted protocols attain performance beyond the reach of conventional approaches. 
  In particular, a single-copy pre-processing scheme is shown to surpass conventional purification requiring 4 copies. 
  Finally, Subsec.~\ref{subsec:Beyond_NP_Theorem} demonstrates that forward-assisted protocols can circumvent previously established no-go results, most notably, the absence of purification for Bell states, and introduces efficient purification schemes for this fundamental class of states.
\end{itemize}

This appendix develops a unified spatiotemporal perspective on quantum state purification, tracing a coherent progression from foundational structures to concrete operational advantages of our forward-assisted purification protocols.
Elevating the description from static states to dynamical transformations through superchannels brings a broad class of protocols into a single optimization-based framework. 
This unification enables systematic comparison with conventional approaches and exposes regimes in which pre-processing, temporal structure, and symmetry emerge as decisive resources. 
The resulting performance hierarchies, together with symmetry-driven reductions in computational complexity, provide both conceptual novelty and practical tractability. 
These elements collectively establish the foundation for understanding and leveraging the advantages of forward-assisted purification in both global and distributed settings.


\section{Preliminaries}\label{sec:Preliminaries}

This section collects the notations, conventions, and standard constructions underlying the main text. 
It establishes a unified formal framework for quantum channels and superchannels in Subsec.~\ref{subsec:Channels_Superchannels}, summarizes the canonical noise models relevant to realistic implementations in Subsec.~\ref{subsec:Quantum_Noise_Models}, and introduces the Bell states in Subsec.~\ref{subsec:Bell_States} that serve as a central object of study in the distributed purification protocols considered here. 
These elements establish a coherent technical foundation for the spatiotemporal framework of purification --- termed forward-assisted strategies --- developed in the main sections.


\subsection{Quantum Channels and Superchannels}\label{subsec:Channels_Superchannels}

We begin by establishing the formal framework underlying quantum channels and their higher-order generalization, quantum superchannels. 
This structure provides a unified and compositional description of quantum dynamics, in which physical transformations are represented at the level of processes rather than states. 
Central to this formulation is the Choi–Jamio{\l}kowski isomorphism and the associated link product, which together enable a compact operator representation of sequential and higher-order maps. 
Within this setting, superchannels naturally decompose into pre-processing, memory, and post-processing stages, making explicit the temporal and causal structure of quantum operations. 
This perspective underpins our subsequent development of a spatiotemporal framework for forward-assisted purification, in which structural constraints --- classicality, non-signalling (NS), and positive partial transpose preservation (PPTp) --- govern the admissible transformations.

Throughout this work, we adopt the following notational conventions.
System $A$ is associated with a Hilbert space $\mH_A$, and is referred to simply as $A$ when no ambiguity arises. 
Quantum states are represented by positive semidefinite operators, i.e., $\rho \geqslant 0$, with unit trace, i.e., $\Tr[\rho]=1$.
General physical transformations are described by quantum channels, i.e., completely positive and trace-preserving (CPTP) linear maps. 
It is convenient to represent such maps via the Choi–Jamio{\l}kowski isomorphism~\cite{JAMIOLKOWSKI1972275,CHOI1975285}, which associates each channel with an operator on a bipartite space. 
For a channel $\mE \colon A \to B$, the corresponding Choi operator $J^{\mE}_{AB}$ is defined as
\begin{align}\label{eq:Choi}
    J^{\mE}_{AB} \coloneqq \id_{A} \otimes \mE_{A' \to B} \left(\Gamma_{AA'}\right),
\end{align}
where the system $A'$ is isomorphic to $A$, and $\Gamma \coloneqq \ketbra{\Gamma}{\Gamma}$ denotes the unnormalized maximally entangled state (UMES), with
\begin{align}\label{eq:UMES}
    \ket{\Gamma} \coloneqq \sum_i \ket{ii}.
\end{align}
When $\dim A = d$, the corresponding maximally entangled state (MES) $\phi^{+} \coloneqq \ketbra{\phi^{+}}{\phi^{+}}$ is given by
\begin{align}\label{eq:MES}
    \ket{\phi^{+}} \coloneqq \frac{1}{\sqrt{d}} \sum_i \ket{ii} = \frac{1}{\sqrt{d}} \ket{\Gamma}.
\end{align}
Subsystem labels are omitted whenever they are clear from context.

Within the Choi–Jamio{\l}kowski representation, the defining constraints of a quantum channel take a particularly compact operator form. 
CP is encoded as
\begin{align}\label{eq:CP}
    J^{\mE}_{AB}\geqslant0,
\end{align}
while TP imposes the normalization condition
\begin{align}\label{eq:TP}
    \Tr_{B}[J^{\mE}_{AB}] = \1_{A}.
\end{align}
Expressed in this way, dynamical constraints are recast as algebraic properties of the Choi operator $J^{\mE}$, a reformulation that is especially amenable to tensor-network representations, where positivity and normalization acquire an immediate graphical interpretation~\cite{wood2015tensornetworksgraphicalcalculus,Coecke_Kissinger_2017,Bridgeman_2017,biamonte2020lecturesquantumtensornetworks,Collura_2024}.
Comprehensive treatments of quantum channels can be found in~\cite{Nielsen_Chuang_2010,RevModPhys.86.1203,Wilde_2017,Watrous_2018,khatri2024principlesquantumcommunicationtheory}.

We next introduce the formalism of quantum superchannels, together with the composition of quantum channels via the link product, beginning with its definition.

\begin{mydef}{Link Product~\cite{959270,PhysRevLett.101.060401,PhysRevA.80.022339}}{Link_Product}
    Given two operators $M$ and $N$ sharing a common subsystem $A$, their link product, denoted by $M \star N$, is defined as
    \begin{align}\label{eq:Link_Product}
        M\star N \coloneqq \Tr_{A}[M\cdot N^{\T_{A}}],
    \end{align}
    where $^{\T_{A}}$ denotes the partial transpose taken with respect to subsystem $A$.
\end{mydef}

As a direct application of the link product, we consider the Choi operator associated with the sequential composition of quantum channels. 
In particular, given channels $\mE \colon A \to B$ and $\mF \colon B \to C$, with Choi operators $J^{\mE}$ and $J^{\mF}$, respectively, the Choi operator of the composed channel $\mF \circ \mE$ is then given by
\begin{align}
    J^{\mF \circ \mE} = J^{\mE}\star J^{\mF}.
\end{align}
Here, the ordering of physical systems is fixed throughout. 
Under this convention, the link product is invariant under subsystem permutations. 
When subsystems are written explicitly, inputs and outputs are arranged from left to right, with multiple systems ordered according to their temporal sequence. Earlier systems appear to the left.

Given a quantum channel $\mE \colon B \to C$, the most general transformation of quantum channels is described by a superchannel $\theta \colon AC\to BD$, comprising three components~\cite{Chiribella_2008,8678741,xiao2025superchanneltearsgeneralizedoccams}: a pre-processing (PreP) map $\theta^{\mathrm{Pre}}$, a post-processing (PostP) map $\theta^{\mathrm{Post}}$, and an auxiliary channel (memory or simply MeM) $\theta^{\mathrm{Mem}}$ that mediates correlations between them. 
While the memory can often be absorbed into either the pre- or post-processing, we retain all three components here for clarity. Their physical implementation is illustrated in Fig.~\ref{fig:Superchannel}.

\begin{figure}[htbp]
    \centering   
    \includegraphics[width=0.5\textwidth]{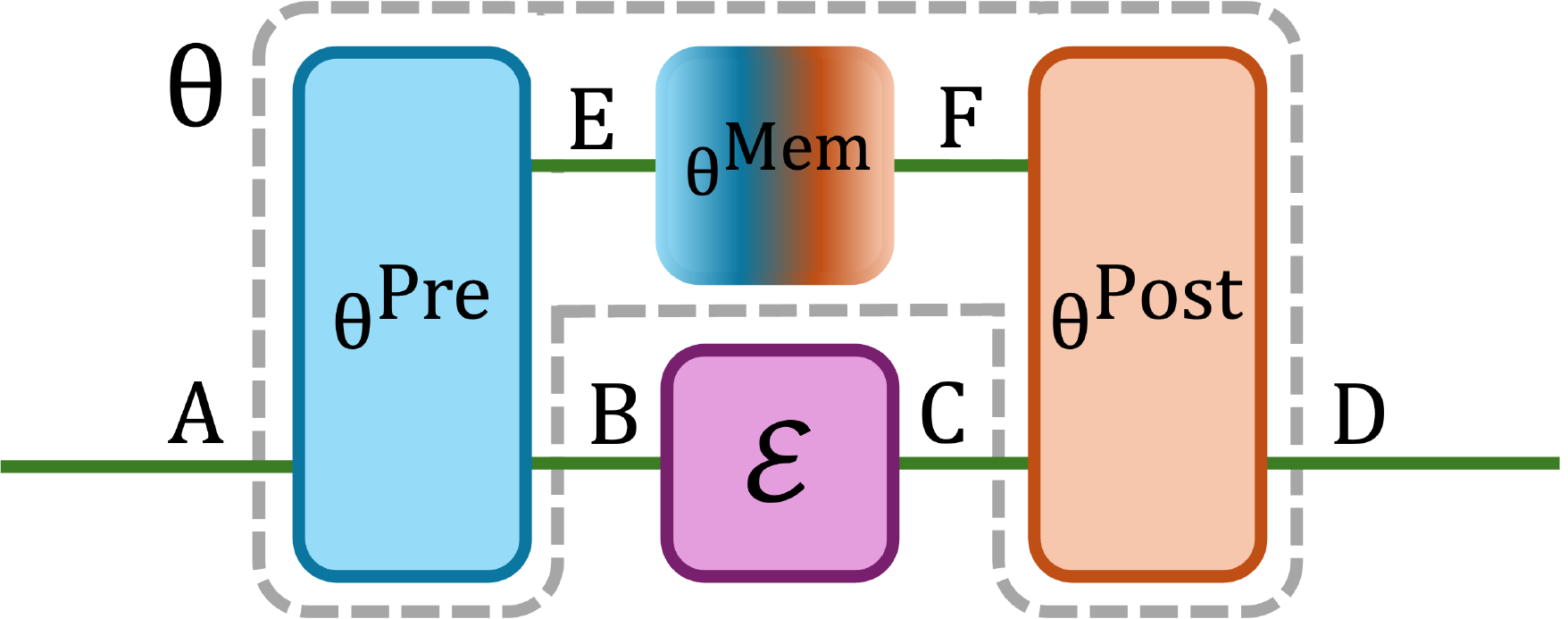}
    \caption{\textbf{General Structure of a Quantum Superchannel}. 
        A quantum superchannel $\theta$ describes the most general transformation of a quantum channel $\mE$, comprising a pre-processing map $\theta^{\mathrm{Pre}}$, a post-processing map $\theta^{\mathrm{Post}}$, and an intermediate memory $\theta^{\mathrm{Mem}}$ that correlates the two. 
        While this memory can be absorbed into either stage --- yielding the familiar two-part representation --- here all three components are retained explicitly to expose the structure of forward-assisted purification in Subsec.~\ref{subsec:Forward_Assisted_Purification_Protocols}.
    }
    \label{fig:Superchannel}
\end{figure}

Superchannels, in their most general form, are both theoretically unconstrained and experimentally demanding, rendering them of limited direct relevance for practical implementations. 
Indeed, access to arbitrary superchannels would trivialize many nontrivial tasks in quantum information processing. 
This motivates restricting attention to physically realizable and mathematically tractable subclasses. 
In this work, we focus on structural constraints that are central to forward-assisted (FA) purification, namely classicality (CL), non-signalling (NS), and positive partial transpose preservation (PPTp, or simply PPT). These properties are introduced in turn.

\begin{mydef}{Classical}{Classical_Channel}
    A channel $\mE$ is termed classical if it remains invariant under completely dephasing $\Delta$ on all its subsystems; 
    that is
    \begin{align}
        \Delta\circ\mE\circ\Delta = \mE.
    \end{align}
    Here, the completely dephasing operation is defined as $\Delta(\cdot) = \sum_i\bra{i}\cdot\ket{i}\ketbra{i}{i}$ for a fixed basis $\{\ket{i}\}$.
\end{mydef}

While classicality in Def.~\ref{def:Classical_Channel} constrains the internal structure of a channel through invariance under dephasing, it does not impose restrictions on the flow of information between processes. 
To capture such directional constraints, one is led to the notion of non-signalling (NS), which formalizes the absence of causal influence from one process to another. 
In particular, one-way NS specifies that the output on a given process remains insensitive to inputs on another, thereby encoding a form of causal independence at the level of quantum channels. 
This condition is made precise below in terms of the associated Choi operator.

\begin{mydef}{One-Way Non-Signalling}{One_Way_NS}
    A bipartite quantum channel $\theta \colon AC \to BD$ is one-way non-signalling from $(C \to D)$ to $(A \to B)$ if the marginal output on $B$ is independent of the input state on $C$.
    In terms of the Choi operator $J^{\theta}$, this condition reads
    \begin{align}
        \Tr_{D}[J^{\theta}] = J^{\theta}_{AB}\otimes\frac{1}{d_{C}}\1_{C},
    \end{align}
    where $J^{\theta}_{AB} \coloneqq \Tr_{CD}[J^{\theta}]$ denotes the reduced operator on systems $AB$, and $d_C$ is the dimension of system $C$.
    Channels satisfying this condition are also referred to as semi-causal in the early literature~\cite{PhysRevA.64.052309,Eggeling_2002,PhysRevA.74.012305}.
\end{mydef}

One-way NS of Def.~\ref{def:One_Way_NS} encodes a directional notion of causal independence, precluding information flow from one process to another while permitting it in the reverse direction. 
A stronger constraint is obtained by imposing this condition symmetrically, thereby excluding signalling in both directions. 
This gives rise to the notion of NS, which characterizes bipartite channels whose outputs are independent of inputs across processes in either causal direction.

\begin{mydef}{Non-Signalling}{NS}
    A bipartite quantum channel $\theta \colon AC \to BD$ is said to be non-signalling if it is both one-way non-signalling from $(C \to D)$ to $(A \to B)$ and from $(A \to B)$ to $(C \to D)$.
\end{mydef}

Non-signalling in Def.~\ref{def:NS} constrains the causal structure of bipartite channels by restricting the flow of information between processes, but does not directly limit the nature of the correlations they generate. 
A complementary constraint is provided by positivity under partial transposition, which acts at the level of correlations themselves. 
This gives rise to the notion of PPT, first introduced for bipartite states and subsequently extended to channels and superchannels.

\begin{mydef}{PPT State}{PPT_State}
    A bipartite quantum state $\rho_{AB}$ is said to be PPT if it remains positive semidefinite under partial transposition with respect to one subsystem, that is, under the action of $\T_B$ (or equivalently $\T_A$):
    \begin{align}
        \rho_{AB}^{\T_{B}} \geqslant 0.
    \end{align}
\end{mydef}

The PPT criterion in Def.~\ref{def:PPT_State} extends naturally from states to dynamical maps through the Choi–Jamio\l kowski representation, allowing structural constraints on correlations to be lifted to the level of quantum channels. 
In this setting, a particularly relevant class is formed by PPT-binding channels, also known as Horodecki channels~\cite{Horodecki01022000}.
This notion is formalized below.

\begin{mydef}{PPT-Binding Channel}{PPT_Binding_Channel}
    A quantum channel $\mE \colon A \to B$ is said to be PPT-binding, or simply PPT, if its Choi operator is PPT with respect to the partial transpose on one subsystem, that is,
    \begin{align}
        (J^{\mE})^{\T_{B}} \geqslant 0.
    \end{align}
\end{mydef}

The PPT-binding condition in Def.~\ref{def:PPT_Binding_Channel} constrains the structure of a channel itself. 
A complementary viewpoint focuses instead on how channels act on inputs, leading to the notion of PPT preservation. 
Rather than imposing a condition on the Choi operator alone, this requirement enforces that the channel maps PPT states to PPT states, thereby extending the PPT constraint from static correlations to bipartite channels.

\begin{mydef}{PPT Preserving (PPTp)}{PPT_p}
    Consider a bipartite quantum channel $\mE \colon A_1 B_1 \to A_2 B_2$.
    The channel $\mE$ is said to be positive-partial-transpose preserving (PPT-preserving, or simply PPT) with respect to the $B$ subsystems if the following condition holds
    \begin{align}
        (J^{\mE})^{\T_{B_1B_2}} \geqslant 0.
    \end{align}
\end{mydef}

In what follows, we do not distinguish between PPT states (see Def.~\ref{def:PPT_State}) and PPTp channels (see Def.~\ref{def:PPT_p}), as the notion extends naturally to settings with multiple inputs and outputs by specifying the subsystems with respect to which the partial transpose is taken.
Within this general framework, two closely related concepts --- often conflated, particularly for superchannels acting on bipartite systems at each time step --- require careful distinction. 
The first concerns superchannels that are PPT with respect to the post-processing stage. 
For instance, for the class of superchannels illustrated in Fig.~\ref{fig:Bipartite_Superchannel}, we define the PPT superchannel accordingly.

\begin{figure}[htbp]
    \centering   
    \includegraphics[width=0.4\textwidth]{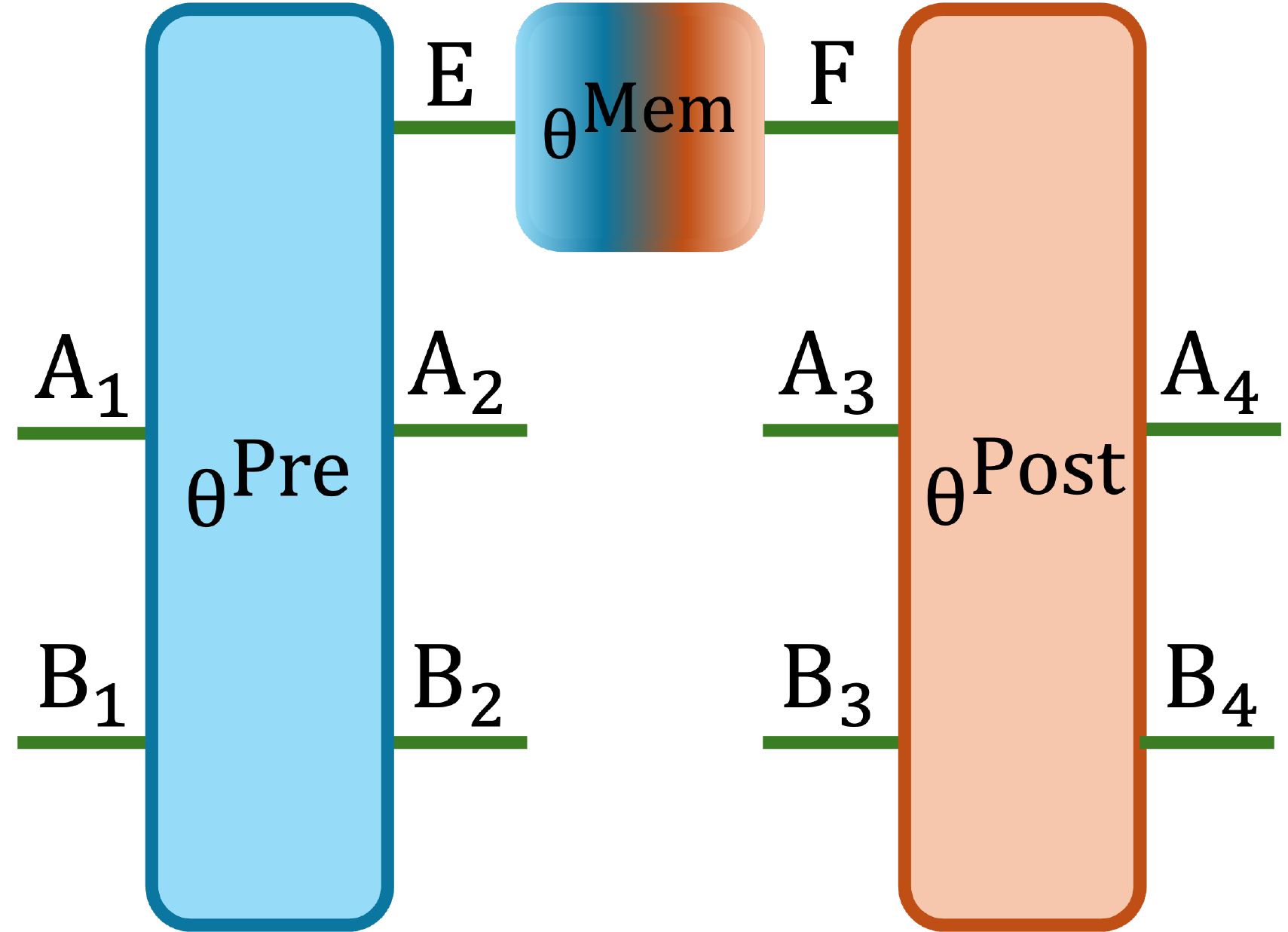}
    \caption{\textbf{Bipartite Superchannel}. 
        The superchannel $\theta$ is decomposed into a pre-processing map $\theta^{\mathrm{Pre}}$, a memory channel $\theta^{\mathrm{Mem}}$, and a post-processing map $\theta^{\mathrm{Post}}$. 
        The pre-processing transforms the input systems $A_1B_1$ into $A_2B_2$ while generating a memory system $E$, which is transmitted through $\theta^{\mathrm{Mem}}$ to $F$. 
        The post-processing then combines $A_3B_3$ with the memory $F$ to produce the final outputs $A_4B_4$.
    }
    \label{fig:Bipartite_Superchannel}
\end{figure}

\begin{mydef}{PPT Superchannel with respect to Post-Processing}{PPT_PostP}
    For the superchannel shown in Fig.~\ref{fig:Bipartite_Superchannel}, we define it to be PPT if its associated Choi operator remains positive under partial transposition on the subsystems corresponding to the post-processing stage; that is,
    \begin{align}
        (J^{\theta})^{\T_{\mathrm{Post}}} = (J^{\theta})^{\T_{A_3B_3A_4B_4}} \geqslant 0.
    \end{align}
    Here, the memory systems $E$ and $F$ are treated as latent degrees of freedom and are not explicitly taken into account.
\end{mydef}

The second notion concerns superchannels (see Fig.~\ref{fig:Bipartite_Superchannel}) whose post-processing map is PPT. 
In this case, the PPT condition on the post-processing stage coincides with the definition of PPT-preserving (see Def.~\ref{def:PPT_p}), upon replacing 
$B_1B_2$ with $B_3B_4$.

Collectively, this subsection establishes a unified and operationally transparent description of quantum dynamics at the level of channels and superchannels.
The Choi-Jamio{\l}kowski isomorphism and link product provide a compositional calculus in which both sequential and higher-order transformations can be expressed and analyzed on equal footing. 
Within this structure, superchannels emerge as the most general maps acting on quantum processes, with their decomposition into pre-processing, memory, and post-processing stages making explicit the temporal and causal organization of quantum operations. 
Imposing physically motivated constraints --- classical, NS, and PPT --- then delineates tractable and experimentally relevant subclasses, each capturing a distinct facet of admissible dynamics. 
The distinction between PPT conditions applied at the level of the entire superchannel and those restricted to the post-processing stage further refines this landscape, clarifying how structural constraints propagate across different layers of the transformation. 
This formalism not only unifies these notions within a single mathematical language, but also sets the stage for systematically quantifying the capabilities and limitations of forward-assisted purification protocols developed in Subsec.~\ref{subsec:Forward_Assisted_Purification_Protocols}.


\subsection{Quantum Noise Models}\label{subsec:Quantum_Noise_Models}

Realistic quantum information processing is inevitably affected by noise, arising from imperfect control, environmental interactions, and measurement limitations. 
Capturing these effects in a tractable yet physically meaningful manner is essential for both analysis and protocol design. 
In this subsection, we focus on a set of typical local noise models (see Tab.~\ref{tab:Quantum_Noise_Models}), parameterized by a single noise strength, which encapsulate the dominant error mechanisms encountered in practice. 
These models admit concise Kraus representations, enabling a unified operator description that will serve as the basis for our subsequent investigation of quantum state purification protocols.

Mathematically, each noise model $\mN$ admits a Kraus decomposition~\cite{xiao2025superchanneltearsgeneralizedoccams}, namely
\begin{align}
    \mN(\cdot) = \sum_{i}K_i(\cdot)K_i^{\dagger}.
\end{align}
In particular, we consider the following canonical noise models widely used in quantum information processing:

\begin{table}[htbp]
    \centering
    \begin{tblr}{
      colspec = {l || c | c | c | c},
      row{1,2} = {bg=gray!50, font=\bfseries}, 
      column{1} = {bg=gray!10},                
      hlines,                                  
      vlines,                                  
      cells = {m},                             
      row{1} = {c},                            
    }
      \SetCell[c=5]{c} Characterization of Quantum Noise Models & & & & \\
      Types of Quantum Noise & Operator $K_0$ & Operator $K_1$ & Operator $K_2$ & Operator $K_3$ \\
      Bit Flip Channel $\mN_{\mathrm{BF}}$ & $\sqrt{p} \ \1$ & $\sqrt{1-p} \ \mathrm{X}$ & $\0$ & $\0$ \\
      Phase Flip Channel $\mN_{\mathrm{PF}}$ & $\sqrt{p} \ \1$ & $\sqrt{1-p} \ \mathrm{Z}$ & $\0$ & $\0$ \\
      Depolarizing Channel $\mN_{\mathrm{D}}$ & $\sqrt{\frac{1+3p}{4}} \ \1$ & $\sqrt{\frac{1-p}{4}} \ \mathrm{X}$ & $\sqrt{\frac{1-p}{4}} \ \mathrm{Y}$ & $\sqrt{\frac{1-p}{4}} \ \mathrm{Z}$ \\
      Amplitude Damping Channel $\mN_{\mathrm{AD}}$ & $\begin{bmatrix} 1 & 0 \\ 0 & \sqrt{p} \ \end{bmatrix}$ & $\begin{bmatrix} 0 & \sqrt{1-p} \,\, \\ 0 & 0 \end{bmatrix}$ & $\0$ & $\0$ \\
    \end{tblr}
    \caption{\textbf{Kraus Decompositions of Quantum Noise Models}.
        We consider four typical quantum noise channels: the bit flip channel $\mN_{\mathrm{BF}}$, phase flip channel $\mN_{\mathrm{PF}}$, depolarizing channel $\mN_{\mathrm{D}}$, and amplitude damping channel $\mN_{\mathrm{AD}}$. 
        The operators $\1, \mathrm{X}, \mathrm{Y}$, and $\mathrm{Z}$ denote the Pauli matrices.
    }
    \label{tab:Quantum_Noise_Models}
\end{table}

The bit flip (BF) channel $\mN_{\mathrm{BF}}$ models stochastic transitions between computational basis states, capturing classical bit errors arising from imperfect control or readout. 
The phase flip (PF) channel $\mN_{\mathrm{PF}}$ describes random phase inversions, which preserve the diagonal elements of the density matrix in the computational basis while degrading the off-diagonal coherence terms, as typically induced by environmental dephasing. 
The depolarizing (D) channel $\mN_{\mathrm{D}}$ represents isotropic noise, in which the state is randomly subjected to Pauli errors, driving it towards the maximally mixed state and serving as a canonical model of unbiased decoherence. 
In contrast, the amplitude damping (AD) channel $\mN_{\mathrm{AD}}$ captures energy relaxation processes, such as spontaneous emission, where excitations decay irreversibly to the ground state, altering both the diagonal incoherences and off-diagonal coherences.


\subsection{Bell States}\label{subsec:Bell_States}

Bell states constitute a canonical set of maximally entangled states and serve as a fundamental resource in quantum information processing. 
Their symmetry and extremal correlation properties make them a natural benchmark for tasks involving entanglement manipulation, particularly in distributed purification protocols. 
In this subsection, we introduce the standard Bell basis for two qubits and fix the notation that will be used throughout, setting the stage for the analysis of purification strategies, including entanglement distillations, and their performance under noise.

Mathematically, in the bipartite qubit setting, the four Bell states are defined as follows.
\begin{align}\label{eq:Bell}
    \ket{\Phi^{+}} &= \frac{1}{\sqrt{2}} \left( \ket{0}_A \otimes \ket{0}_B + \ket{1}_A \otimes \ket{1}_B \right), \\
    \ket{\Phi^{-}} &= \frac{1}{\sqrt{2}} \left( \ket{0}_A \otimes \ket{0}_B - \ket{1}_A \otimes \ket{1}_B \right), \\
    \ket{\Psi^{+}} &= \frac{1}{\sqrt{2}} \left( \ket{0}_A \otimes \ket{1}_B + \ket{1}_A \otimes \ket{0}_B \right), \\
    \ket{\Psi^{-}} &= \frac{1}{\sqrt{2}} \left( \ket{0}_A \otimes \ket{1}_B - \ket{1}_A \otimes \ket{0}_B \right),
\end{align}
which collectively form an orthonormal basis for the two-qubit Hilbert space, with their collection written as 
\begin{align}\label{eq:Bell_Set}
    \mS_{\mathrm{Bell}} \coloneqq \{\Phi^{+}, \Phi^{-}, \Psi^{+}, \Psi^{-}\}.
\end{align}
This work demonstrates that previously established no-go results for Bell states purification can be circumvented within the forward-assisted (FA) framework, specifically through pre-processing (PreP). 
Details are provided in Subsec.~\ref{subsec:Beyond_NP_Theorem}.


\section{Forward-Assisted Purification}\label{sec:Forward-Assisted_Purification}

This section develops the framework of forward-assisted (FA) quantum state purification, providing the formal and operational foundations underlying the main results. 
The presentation proceeds in three stages. 
It begins with a concise formulation of conventional purification protocols in Subsec.~\ref{subsec:Quantum_State_Purification}, establishing the baseline against which improvements are assessed. 
It then introduces the general FA framework within the language of superchannels in Subsec.~\ref{subsec:Forward_Assisted_Purification_Protocols}, capturing the full spatiotemporal structure of noise manipulation through pre-processing, memory, and post-processing. 
Finally, it formulates the fundamental limits of purification performance as semidefinite programs (SDPs) in Subsec.~\ref{subsec:Optimization_Purification_Performance}, enabling a systematic and computable characterization of different protocol classes and the constraints that define them.


\subsection{Quantum State Purification}\label{subsec:Quantum_State_Purification}

Quantum state purification addresses the fundamental task of recovering high-quality quantum states from noisy preparations, a central challenge in realistic quantum information processing. 
Imperfections arising from environmental interactions and device limitations inevitably degrade ideal states, motivating strategies that leverage multiple noisy copies to reconstruct states of higher fidelity. 
In this subsection, we formalize the purification problem, introduce the conventional paradigm --- referred to here as post-processing purification --- and establish the performance benchmarks that will serve as a reference for assessing the advantages of the spatiotemporal framework, namely forward-assisted purification, developed later.

Consider a quantum state $\psi \in \mS$. 
In practice, before being deployed in any task, it is typically degraded by environmental noise or device imperfections, resulting in a transformed state $\mN(\psi)$; that is,
\begin{align}
    \psi \xrightarrow{\text{Noise}} \mN(\psi),
\end{align}
where $\mN$ denotes a noisy quantum channel. 
The goal of quantum state purification is to determine whether, given multiple copies of the noisy state $\mN(\psi)^{\otimes n}$, one can recover an output state that more closely approximates the original state $\psi$. 

In the conventional approach, a quantum channel $\mE$ acts on $n$ copies of the noisy state $\mN(\psi)$, yielding an output state of the form
\begin{align}\label{eq:Conventional_Purification}
    \mE(\left(\mN(\psi)\right)^{\otimes n}).
\end{align}
This process is illustrated in Fig.~\ref{fig:Conventional_Purification}.
Performance is quantified by the average fidelity defined below, 
\begin{align}
    F(\mE,\mN,\mS):=\frac{1}{|\mS|}\sum_{\psi\in\mS}F(\mE(\left(\mN(\psi)\right)^{\otimes n}), \psi),
\end{align}
where $F$ on the right-hand side denotes the fidelity between quantum states~\cite{UHLMANN1976273,Jozsa01121994}, defined as $F(\rho,\sigma)\coloneqq(\Tr[\sqrt{\sqrt{\rho}\sigma\sqrt{\rho}}])^2$.
When one of the states is pure, say $\sigma=\psi$, this reduce to $F(\rho,\psi)\equiv\Tr[\rho\cdot\psi]$.
Throughout, we use the same symbol $F$ to denote both state fidelity and average fidelity; 
the intended meaning is always clear from context.

\begin{figure}[htbp]
    \centering   
    \includegraphics[width=0.37\textwidth]{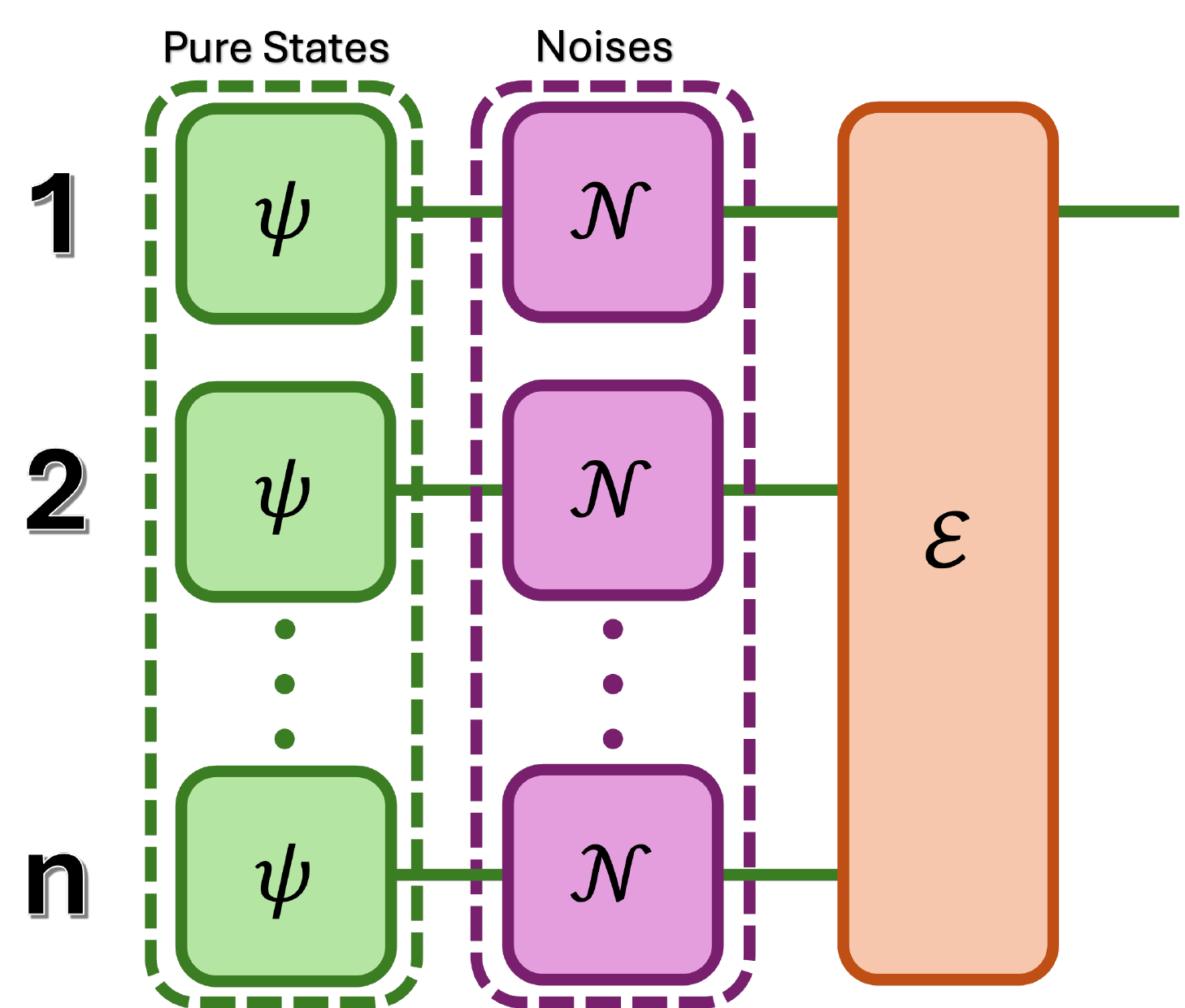}
    \caption{\textbf{Conventional Quantum State Purification}. 
        In conventional quantum state purification, $n$ copies of the noisy state $\mN(\psi)$ are prepared and processed collectively by a purification channel $\mE$. 
        Within the superchannel framework, this conventional protocol (orange) can be interpreted as a post-processing map (see Fig.~\ref{fig:Superchannel}) acting on $n$ parallel uses of the noise channel $\mN^{\otimes n}$ (purple), which acts on the input pure states (green).
    }
    \label{fig:Conventional_Purification}
\end{figure}

The benchmark is given by the average fidelity in the absence of any purification, written as
\begin{align}\label{eq:Purification_Benchmarking}
    F(\mN,\mS):=\frac{1}{|\mS|}\sum_{\psi\in\mS}F(\mN(\psi), \psi).
\end{align}
With these two protocols --- one with purification and one without --- in place, we can now formally introduce the notion of efficient purification.

\begin{mydef}{Efficient Purification}{Efficient_Purification}
    A purification protocol $\mE$ is said to be efficient if it satisfies
    \begin{align}
        F(\mE,\mN,\mS) > F(\mN,\mS),
    \end{align}
    that is, if it yields a nontrivial improvement in fidelity.
\end{mydef}

\noindent
The focus is on the regime where the inequality is strict, as equality can always be achieved by the trivial protocol that discards all but a single copy (see Fig.~\ref{fig:Trivial_Purification}).

\begin{figure}[htbp]
    \centering   
    \includegraphics[width=0.37\textwidth]{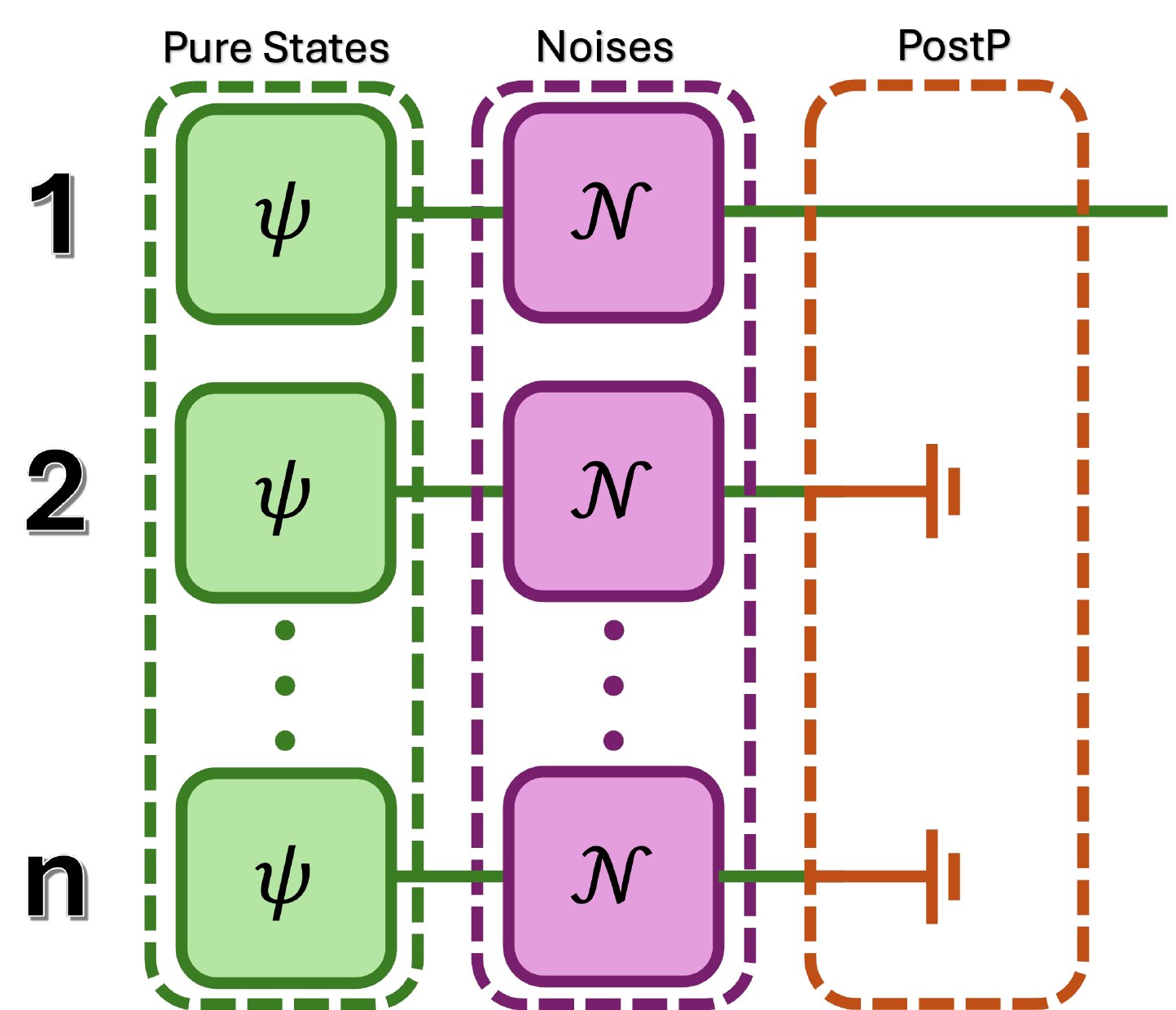}
    \caption{\textbf{Trivial Purification}. 
        $n$ identical copies of the noisy state $\mN(\psi)$ are available, i.e., $\mN(\psi)^{\otimes n}$, but only a single copy is retained while the remaining $n-1$ copies are discarded. 
        In this case, no collective processing is performed, and the resulting fidelity reduces to that of a single noisy copy, as given in Eq.~\eqref{eq:Purification_Benchmarking}.
    }
    \label{fig:Trivial_Purification}
\end{figure}


\subsection{Forward-Assisted Purification Protocols}\label{subsec:Forward_Assisted_Purification_Protocols}

This subsection develops the central framework of this work: forward-assisted (FA) purification. 
Moving beyond conventional approaches that treat noisy states as static objects, we adopt a dynamical perspective in which noise is modeled as a quantum process and manipulated at the level of channels~\cite{10.1109/TIT.2015.2439953}. 
This viewpoint naturally leads to a formulation in terms of superchannels~\cite{Chiribella_2008,8678741,xiao2025superchanneltearsgeneralizedoccams}, where pre-processing, memory, and post-processing are treated on equal footing. 
Within this spatiotemporal framework, purification protocols acquire a richer structure, enabling transformations that are inaccessible to standard post-processing schemes. 
We formalize this setting and introduce the corresponding classes of FA protocols, which will serve as the basis for the performance analysis and structural results developed in the following sections.

The conventional approach treats the noisy state $\mN(\psi)$ as a static object, thereby overlooking the intrinsically dynamical nature of noise $\mN$ (see Fig.~\ref{fig:Conventional_Purification}). 
In reality, noise is described by quantum channels, which are inherently dynamical processes.
To fully exploit the possibilities for addressing noise --- specifically, for purifying it --- one must consider the most general transformations of quantum channels, described by superchannels~\cite{Chiribella_2008,8678741,xiao2025superchanneltearsgeneralizedoccams}.
These comprise pre-processing and post-processing stages connected by a memory channel, with pre-processing applied before the noise and post-processing after, as demonstrated in Fig.~\ref{fig:Superchannel}.

In quantum error correction, these stages are analogous to encoding and decoding. 
A key distinction, however, is that superchannels may incorporate quantum memory $\theta^{\mathrm{Mem}}$ linking the pre- and post-processing stages. 
More precisely, given noise channels $\mN^{\otimes n}$, its manipulation is described by a superchannel $\Theta$, whose action can be expressed as
\begin{align}\label{eq:Superchannel}
    \theta(\mN^{\otimes n})
    =
    \theta^{\text{Post}}
    \circ\left(\theta^{\mathrm{Mem}}\otimes\mN^{\otimes n}\right)\circ
    \theta^{\text{Pre}},
\end{align}
where $\theta^{\text{Pre}}$ and $\theta^{\text{Post}}$ represent the pre-processing and post-processing operations, respectively. 
The resulting purified state takes the form
\begin{align}
    \theta(\mN^{\otimes n})(\psi^{\otimes n})
    =
    \theta^{\text{Post}}
    \circ\left(\theta^{\mathrm{Mem}}\otimes\mN^{\otimes n}\right)\circ
    \theta^{\text{Pre}}
    (\psi^{\otimes n}),
\end{align}
as illustrated in Fig.~\ref{fig:FA_Purification}.

\begin{figure}[htbp]
    \centering   
    \includegraphics[width=0.5\textwidth]{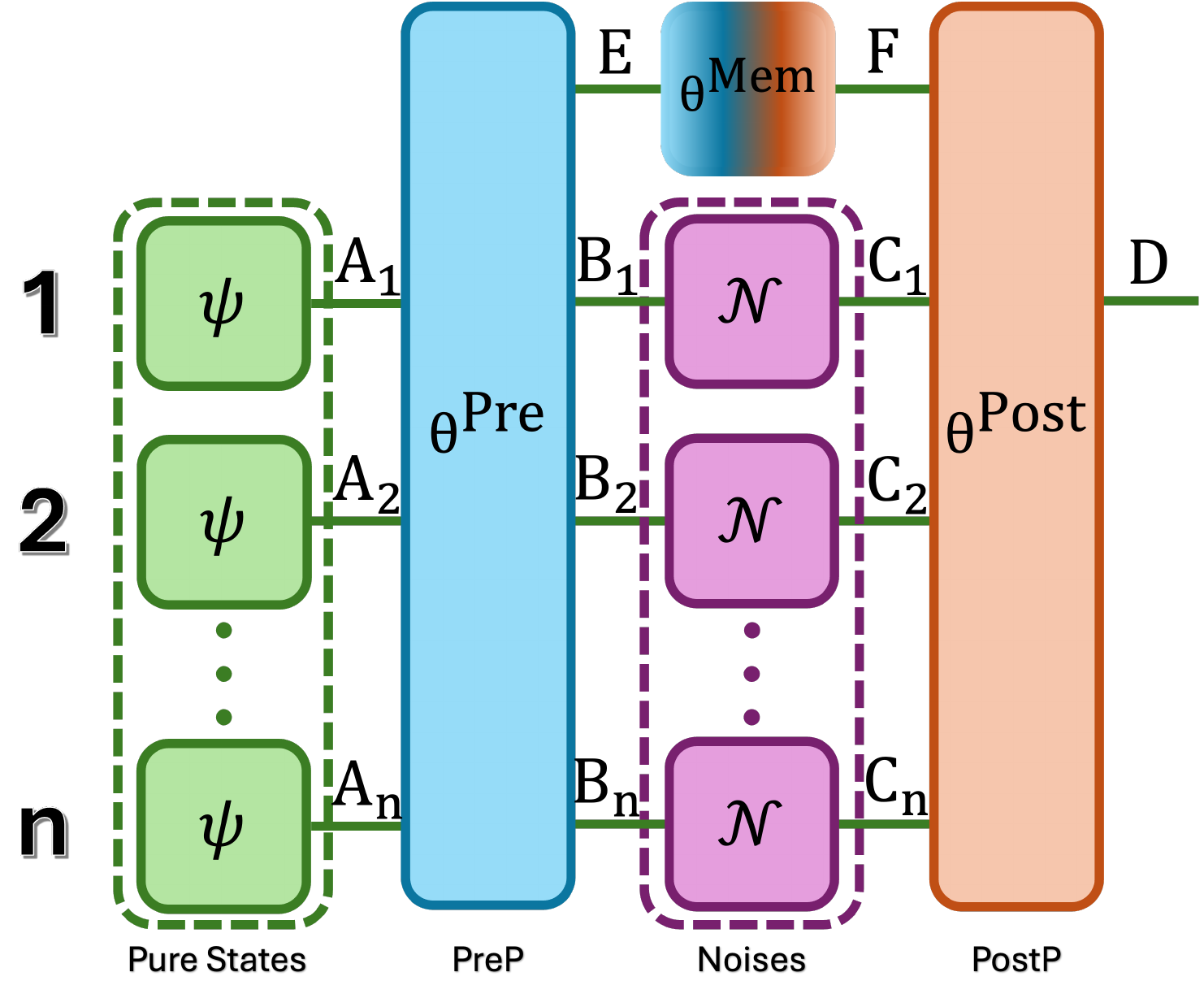}
    \caption{\textbf{Forward-Assisted Purification}. 
        $n$ identical copies of a pure state $\psi$ are first processed by a pre-processing map $\theta^{\mathrm{Pre}}$, whose outputs are subjected to independent noisy channels $\mathcal{N}$. 
        A memory channel $\theta^{\mathrm{Mem}}$ correlates the pre- and post-processing stages (see Fig.~\ref{fig:Superchannel}). 
        The final state is obtained via a post-processing map $\theta^{\mathrm{Post}}$, yielding the purified output on system $D$. This framework captures general spatiotemporal manipulations of noisy quantum states, with conventional purification (see Fig.~\ref{fig:Conventional_Purification}) recovered as the special case without pre-processing or memory.
    }
    \label{fig:FA_Purification}
\end{figure}

Allowing for pre-processing within quantum state purification fundamentally distinguishes the present framework from existing approaches. 
In the absence of pre-processing, that is, by setting $\theta^{\text{Post}} = \mE$ (see Eq.~\eqref{eq:Conventional_Purification}) and neglecting the pre-processing stage $\theta^{\text{Pre}}$, the framework reduces to the conventional purification protocol in Fig.~\ref{fig:Conventional_Purification}. 
In general, however, the inclusion of pre-processing $\theta^{\text{Pre}}$ leads to qualitatively different behavior.
Further details are provided in the subsequent subsections.

\begin{table}[htbp]
    \centering
    \begin{tblr}{
      colspec = {l || c | c | c | c},
      row{1,2} = {bg=gray!50, font=\bfseries}, 
      column{1} = {bg=gray!10},                
      hlines,                                  
      vlines,                                  
      cells = {m},                             
      row{1} = {c},                            
      cell{3}{3} = {bg=mGreen}, cell{3}{4,5} = {bg=mRed},
      cell{4}{3,4} = {bg=mRed}, cell{4}{5} = {bg=mGreen},
      cell{5}{3,5} = {bg=mGreen}, cell{5}{4} = {bg=mRed},
      cell{6-10}{3-5} = {bg=mGreen},
    }
      \SetCell[c=5]{c} Forward-Assisted Purification Protocols & & & & \\
      Types of Purification Protocols & Superchannel & PreP & Memory & PostP \\
      Pre-Processing (PreP) & $\theta=\theta^{\mathrm{Pre}}$ & $\checkmark$ & $\times$ & $\times$ \\
      Post-Processing (PostP) & $\theta=\theta^{\mathrm{Post}}$ & $\times$ & $\times$ & $\checkmark$ \\
      Unassisted (UA) & $\theta=\theta^{\mathrm{Post}}\circ\theta^{\mathrm{Pre}}$ & $\checkmark$ & $\times$ & $\checkmark$ \\
      Entanglement-Assisted (EA) & $\theta=\theta^{\mathrm{Post}}\circ\theta^{\mathrm{Pre}}(\phi^{+})$ & $\checkmark$ & 
      $\checkmark$ & 
      $\checkmark$ \\
      Non-Signalling (NS) & $\theta\in\mathrm{NS}$ (see Def.~\ref{def:NS}) & $\checkmark$ & 
      $\checkmark$ & 
      $\checkmark$ \\
      Forward-Classical-Assisted (FCA) & $\theta^{\mathrm{Mem}}\in\mathrm{CL}$ (see Def.~\ref{def:Classical_Channel}) & $\checkmark$ & 
      $\checkmark$ & 
      $\checkmark$ \\
      Forward-Horodecki-Assisted (FHA) & $\theta^{\mathrm{Mem}}\in\mathrm{CL}$ (see Def.~\ref{def:PPT_Binding_Channel}) & $\checkmark$ & 
      $\checkmark$ & 
      $\checkmark$ \\
      Positive Partial Transpose (PPT) & $(J^{\theta})^{\T_{\mathrm{Post}}}\geqslant0$ & $\checkmark$ & 
      $\checkmark$ & 
      $\checkmark$ \\
    \end{tblr}
    \caption{\textbf{Forward-Assisted Purification Protocols}.
    Generalizations of conventional purification within the post-processing (PostP) setting are considered, all falling under the framework of forward-assisted purification. 
    A checkmark indicates the presence of the corresponding quantum process, while a cross denotes its absence. 
    In the last row, $^{\T_{\mathrm{post}}}$ denotes partial transposition with respect to all subsystems associated with post-processing.
    }
    \label{tab:Forward_Assisted_Purification}
\end{table}

Section~\ref{subsec:Channels_Superchannels} introduced several classes of quantum channels, including classical (see Def.~\ref{def:Classical_Channel}), PPT (see Def.~\ref{def:PPT_Binding_Channel}), and NS (see Def.~\ref{def:NS}) channels. 
Building on this framework, different types of forward-assisted (FA) purification protocols can now be defined, as summarized in Tab.~\ref{tab:Forward_Assisted_Purification}.
Among these, entanglement-assisted (EA) purification is distinct, and its realization is illustrated in Fig.~\ref{fig:EA}.

\begin{figure}[htbp]
    \centering   
    \includegraphics[width=0.5\textwidth]{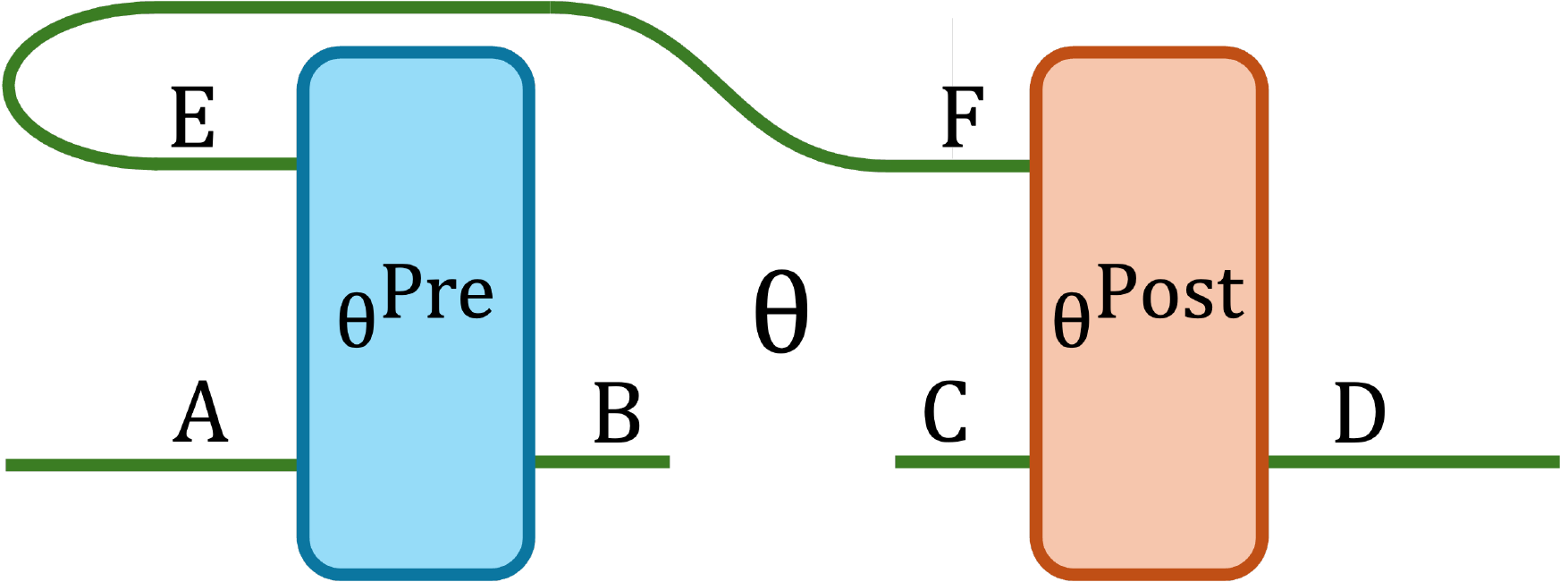}
    \caption{\textbf{Entanglement-Assisted (EA) Purification}. 
        The pre- and post-processing stages of the superchannel are connected via a maximally entangled state. 
        More generally, the maximally entangled state $\phi^{+}$ (see Eq.~\eqref{eq:MES}) can be replaced by an arbitrary shared entangled state between pre- and post-processing.
    }
    \label{fig:EA}
\end{figure}

The inclusion relations among the FA purification protocols listed in Tab.~\ref{tab:Forward_Assisted_Purification} are shown in the Venn diagram of Fig.~\ref{fig:Venn}.
Here, unassisted (UA) purification refers to protocols that do not employ quantum memory, encompassing, for example, PreP and PostP purification schemes.

\begin{figure}[htbp]
    \centering   
    \includegraphics[width=0.7\textwidth]{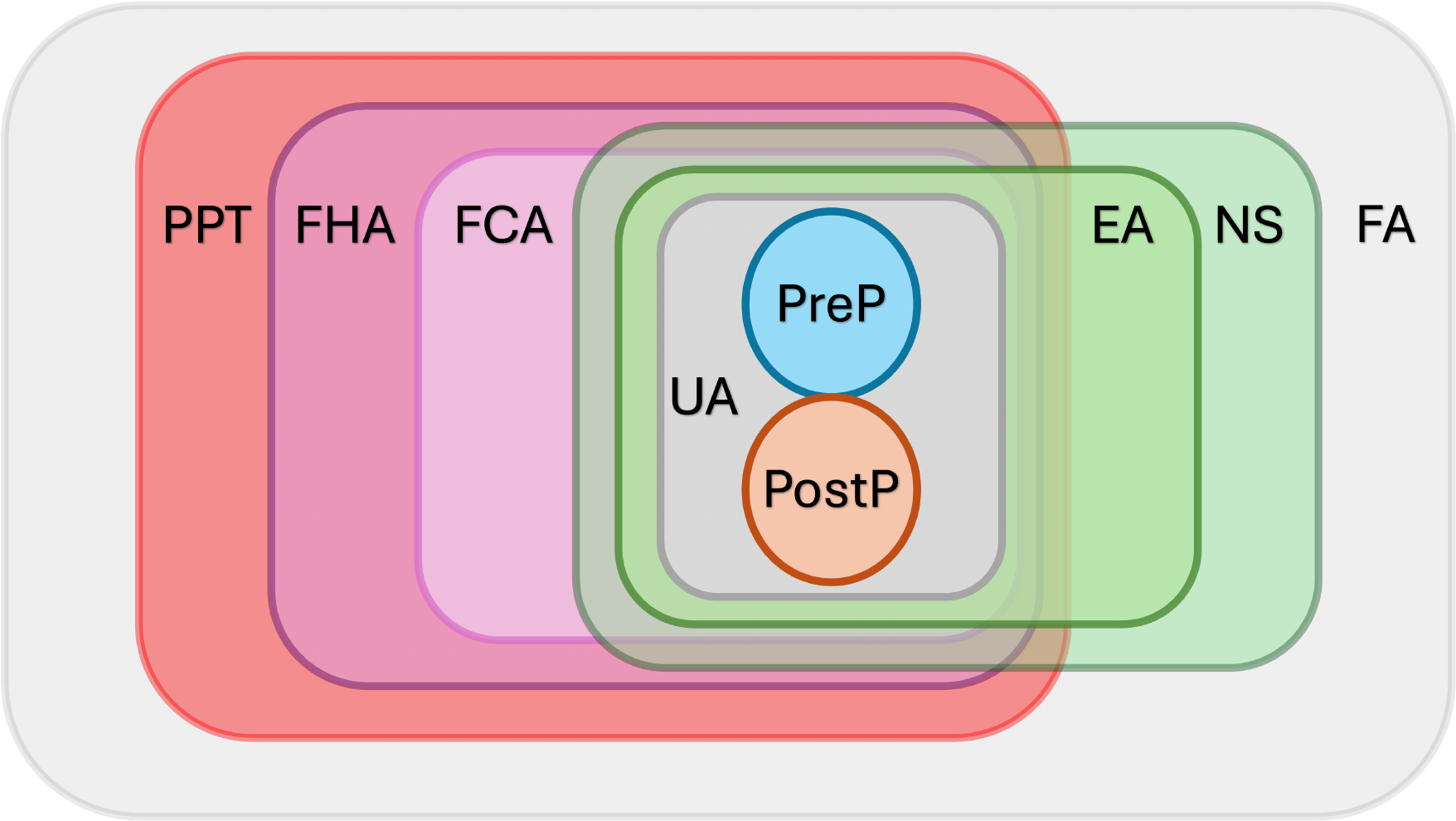}
    \caption{\textbf{Hierarchy of Forward-Assisted Purification Protocols}. 
        Venn diagram illustrating the inclusion relations among different classes of FA purification protocols. 
        The outermost region represents the full FA set. Subclasses defined by structural constraints, non-signalling (NS), entanglement-assisted (EA), and positive partial transpose (PPT), form nested and overlapping regions, with further refinements such as FHA and FCA indicated accordingly. 
        The inner region highlights Unassisted (UA) purifications, composed of pre-processing (PreP) and post-processing (PostP) protocols. 
        The diagram makes explicit the hierarchical organization and intersections among these protocol classes listed in Tab.~\ref{tab:Forward_Assisted_Purification}.
    }
    \label{fig:Venn}
\end{figure}

Since classical channels (see Def.~\ref{def:Classical_Channel}) form a special case of PPT-binding channels (see Def.~\ref{def:PPT_Binding_Channel}), FCA is naturally a subset of FHA (see Fig.~\ref{fig:Venn}). 
We now establish two further inclusion relations: entanglement-assisted (EA) purification is a subset of non-signalling (NS), and FHA is a subset of PPT. These are formalized in the following lemmas.

\begin{mylem}{Hierarchical Inclusion Relation}{Hierarchical_Inclusion_Relation_1}
    Entanglement-assisted (EA) purification protocols are a subset of non-signalling (NS) purification protocols, that is,
    \begin{align}
        \mathrm{EA} \subset \mathrm{NS}.
    \end{align}
\end{mylem}

\begin{proof}
Consider an EA protocol $\theta = \theta^{\mathrm{Post}}\circ\theta^{\mathrm{Pre}}(\phi^{+})$ (see Fig.~\ref{fig:EA}), whose Choi operator is given by    
\begin{align}
    J^{\theta}=
    (J^{\mathrm{Pre}}_{EAB}\otimes J^{\mathrm{Post}}_{FCD})\star\phi^{+}_{EF},
\end{align}
where $\star$ denotes the link product (see Def.~\ref{def:Link_Product}) between quantum processes. 
Here, $J^{\mathrm{Pre}}$ and $J^{\mathrm{Post}}$ are the Choi operators associated with the pre-processing and post-processing channels, respectively.
It suffices to show that $\theta$ is one-way non-signalling from $(A \to B)$ to $(C \to D)$ (see Def.~\ref{def:One_Way_NS}). 
To this end, we trace out the output system $B$, which yields
\begin{align}
    \Tr_{B}[J^{\theta}]=
    &\Tr_{BEF}[(J^{\mathrm{Pre}}_{EAB}\otimes J^{\mathrm{Post}}_{FCD})\cdot
    \phi^{+}_{EF}]\\
    =
    &\Tr_{EF}[(\1_{EA}\otimes J^{\mathrm{Post}}_{FCD})\cdot
    \phi^{+}_{EF}]\\
    =
    &\1_{A}\otimes\frac{1}{d_{F}}
    \Tr_{F}[J^{\mathrm{Post}}_{FCD}],
\end{align}
where $d_{F}$ denotes the dimension of subsystem $F$.
It is also straightforward to verify that the remaining part on systems $C$ and $D$ constitutes the Choi operator of a quantum channel. 
This can be confirmed by tracing over the output system $D$, which leads to
\begin{align}
    \Tr_{D}[\frac{1}{d_{F}}
    \Tr_{F}[J^{\mathrm{Post}}_{FCD}]]
    =
    \frac{1}{d_{F}}\Tr_{F}[\1_{FC}]
    =
    \1_{C}.
\end{align}
This establishes TP, while CP is satisfied by construction. 
Taken together, these conditions imply that the operator is a valid Choi operator of a quantum channel, completing the proof.
\end{proof}

Having established the inclusion of entanglement-assisted (EA) protocols within the non-signalling (NS) class, we next turn to a complementary structural relation arising from positivity constraints. 
In particular, we show that forward-Horodecki-assisted (FHA) protocols are naturally contained within the broader class of PPT protocols, thereby further refining the hierarchical organization of forward-assisted (FA) purification strategies.

\begin{mylem}{Hierarchical Inclusion Relation}{Hierarchical_Inclusion_Relation_2}
    Every forward-Horodecki-assisted (FHA) purification protocol is a special case of a Positive Partial Transpose (PPT) purification protocol, being PPT with respect to the post-processing systems, i.e.,
    \begin{align}
        \mathrm{FHA}\subset\mathrm{PPT}.
    \end{align}
\end{mylem}

\begin{proof}
If a purification protocol is FHA, then its Choi operator admits the form
\begin{align}
    J^{\theta}=
    (J^{\mathrm{Pre}}_{AEB}\otimes J^{\mathrm{Post}}_{FCD})\star
    J^{\mathrm{Mem}}_{EF},
\end{align}
where $J^{\mathrm{Mem}}$ denotes the Choi operator of the memory channel $\theta^{\mathrm{Mem}}$ linking the pre-processing and post-processing stages, as demonstrated in Fig.~\ref{fig:Superchannel}.   
To verify the PPT condition, we take the partial transpose over systems $C$ and $D$, yielding
\begin{align}
    (J^{\theta})^{\T_{CD}}
    =
    &\left(\Tr_{EF}[\left(J^{\mathrm{Pre}}_{AEB}\otimes J^{\mathrm{Post}}_{FCD}\right)
    \cdot
    \left(J^{\mathrm{Mem}}_{EF}\right)^{\T_{EF}}]
    \right)^{\T_{CD}}\\
    =
    &\left(\Tr_{EF}[\left(J^{\mathrm{Pre}}_{AEB}\otimes \left(J^{\mathrm{Post}}_{FCD}\right)^{\T_{FCD}}\right)
    \cdot
    \left(J^{\mathrm{Mem}}_{EF}\right)^{\T_{E}}]
    \right).
\end{align}
As channel $\theta^{\mathrm{Mem}}$ is Horodecki (see Def.~\ref{def:PPT_Binding_Channel}), namely PPT, we have
\begin{align}
    \left(J^{\mathrm{Mem}}_{EF}\right)^{\T_{E}}\geqslant0.
\end{align}
This in turn implies that $(J^{\theta})^{\T_{CD}}\geqslant0$, meaning superchannel $\theta$ is PPT (see Tab.~\ref{tab:Forward_Assisted_Purification}), which completes our proof.
\end{proof}

\begin{figure}[htbp]
    \centering   
    \includegraphics[width=1\textwidth]{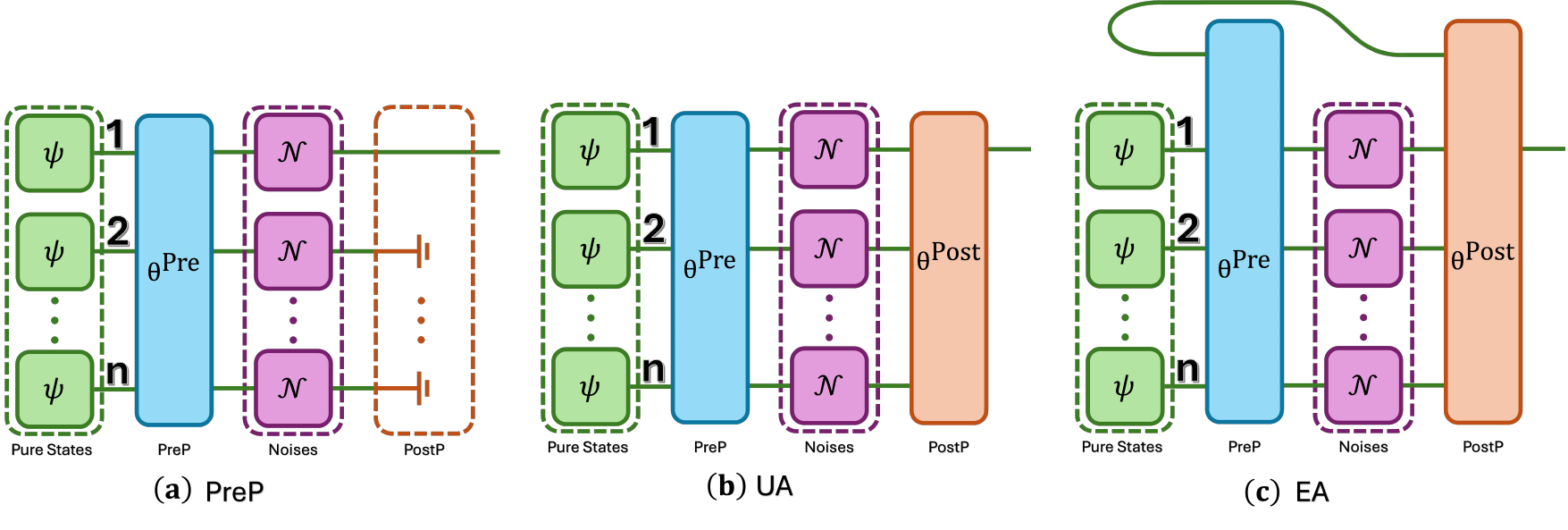}
    \caption{\textbf{Representative Forward-Assisted (FA) Purifications}. 
        (a) Pre-processing (PreP) purification, where only pre-processing is applied prior to the noise channels. 
        In the absence of post-processing, the protocol cannot map $n$ inputs to a single output; consequently, $n-1$ systems are traced out, and only one is retained. 
        For comparison, the conventional post-processing–only purification is shown in Fig.~\ref{fig:Conventional_Purification}.
        (b) Unassisted (UA) purification, where both pre-processing and post-processing are implemented without any quantum memory connecting them, analogous to standard quantum error correction settings.
        (c) Entanglement-assisted (EA) purification, where a maximally entangled state is shared between the pre-processing and post-processing stages, enhancing the purification performance.
    }
    \label{fig:PreP_UA_EA}
\end{figure}

To build intuition for the different types of forward-assisted (FA) purification, representative realizations are provided and illustrated in Fig.~\ref{fig:PreP_UA_EA}.

Within the framework of forward-assisted (FA) purification (see Fig.~\ref{fig:Venn}), both PPT and NS classes are of broad interest and play a central role in the study of quantum information processing, being well motivated from both physical and operational perspectives. 
Importantly, their performance can be formulated as semidefinite programs (SDPs), making them computationally tractable. 
Beyond these classes, particular attention is given to the role of pre-processing (PreP), represented by the blue circle in Fig.~\ref{fig:Venn} (see also Fig.~\ref{fig:PreP_UA_EA}(a)), which is more readily implementable in experimental settings. 
Its performance is contrasted with that of conventional purification protocols, which rely solely on post-processing (PostP), shown as the orange circle in Fig.~\ref{fig:Venn} (see also Fig.~\ref{fig:Conventional_Purification}).
As will become clear, incorporating PreP is not merely advantageous, but, in certain regimes, essential for achieving improved purification performance.


\subsection{Optimization of Purification Performance}
\label{subsec:Optimization_Purification_Performance}

A unified optimization framework is established for quantifying the ultimate performance of forward-assisted (FA) purification protocols. 
Building on the superchannel formulation introduced earlier, the achievable fidelity is cast as a semidefinite program (SDP) in which the Choi operator of the underlying superchannel serves as the central optimization variable, and operational constraints --- such as CP, TP, and structural conditions dictated by the protocol class (see Tab.~\ref{tab:Forward_Assisted_Purification}) --- define the feasible set. 
The formulation is presented at the level of global purification with $n$ copies, thereby encompassing the most general setting of interest. Within this framework, conventional post-processing schemes (see Fig.~\ref{fig:Conventional_Purification}) emerge as special cases, while more general FA protocols, including those constrained by PPT (see Def.~\ref{def:PPT_PostP}) and NS (see Def.~\ref{def:NS}) conditions, are treated on equal footing, making explicit how distinct operational restrictions govern the fundamental limits of purification performance.

For notational simplicity, we denote by $A$ the collection of all input systems $A_1, \ldots, A_n$, namely $A:=A_1, \ldots, A_n$, and adopt the same convention for $B$ and $C$. 
Under this shorthand, general FA purification of Fig.~\ref{fig:FA_Purification} admits the following simplified representation in Fig.~\ref{fig:FA_Purification_Simplified}.

\begin{figure}[htbp]
    \centering   
    \includegraphics[width=0.5\textwidth]{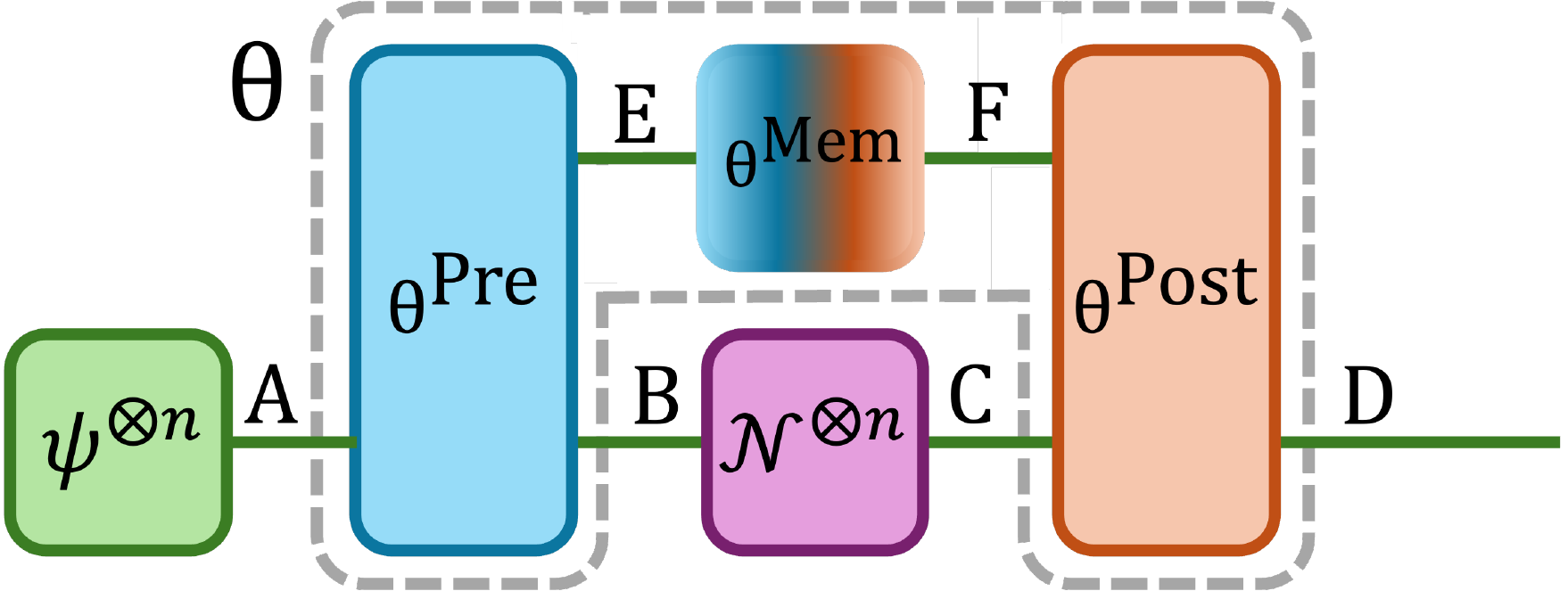}
    \caption{\textbf{Compact Representation of Forward-Assisted Purification}. 
        Schematic representation of forward-assisted (FA) purification in compact form. 
        The $n$-copy input state $\psi^{\otimes n}$ on system $A$ is processed by a pre-processing map $\theta^{\mathrm{Pre}}$, followed by $n$ parallel noisy channels $\mathcal{N}^{\otimes n}$ acting on system $B$. 
        A memory channel $\theta^{\mathrm{Mem}}$ links the pre- and post-processing stages, enabling temporal correlations. 
        The final output on system $D$ is obtained via a post-processing map $\theta^{\mathrm{Post}}$. 
        The dashed boundary denotes the overall superchannel $\theta$, capturing the full spatiotemporal transformation from input to output.
    }
    \label{fig:FA_Purification_Simplified}
\end{figure}

When input states are drawn uniformly from a set $\mS$, and each copy undergoes noise described by $\mN$, an FA purification protocol is characterized by a superchannel $\theta(\cdot)=\theta^{\mathrm{Post}}\circ\left(\theta^{\mathrm{Mem}}\otimes\cdot\right)\circ\theta^{\mathrm{Pre}}$ satisfying a prescribed property $\mP$, as illustrated in Fig.~\ref{fig:FA_Purification_Simplified}, the corresponding purification limit is then given by
\begin{align}\label{eq:FA_Fundamental_Limit}
    F_{\mP}
    \coloneqq
    \max \quad 
    & 
    \Tr[J^{\theta}_{ABCD}\cdot \left( \Psi_{AD}\otimes\Omega_{BC}
    \right)]
    \\
    \text{s.t.} \quad 
    &J^{\theta}\geqslant0, \Tr_{BD}[J^{\theta}]=\1_{AC}, 
    \Tr_{D}[J^{\theta}]=
    J^{\theta}_{AB}\otimes\frac{1}{d_C}\1_{C},
    \theta\in\mP,
\end{align}
where the operators $\Psi$ and $\Omega$ appearing in Eq.~\eqref{eq:FA_Fundamental_Limit} are defined as follows
\begin{align}\label{eq:Psi}
    \Psi_{AD} \coloneqq \frac{1}{|\mS|}\sum\psi^{\otimes n}_{A}\otimes\psi^{\T}_{D},
\end{align}
and
\begin{align}\label{eq:Omega}
    \Omega_{BC} \coloneqq (J^{\mN})^{\otimes n}_{BC}.
\end{align}

Conventional purification protocols are restricted to post-processing alone (see Fig.~\ref{fig:Conventional_Purification}), with pre-processing reduced to identity maps $\id$ from $A_i$ to $B_i$ on each subsystem. Under this constraint, the optimal achievable performance admits the following characterization.
\begin{align}\label{eq:PostP_Fundamental_Limit}
    F_{\mathrm{PostP}}
    \coloneqq
    \max \quad 
    & 
    \Tr[J^{\theta^{\mathrm{Post}}}_{CD}\cdot \left(\frac{1}{|\mS|}\sum_{\psi\in\mS}\psi^{\otimes n}_{A}\otimes\psi_{D}\right)\cdot
    \left((J^{\mN})^{\otimes n}_{AC}\right)^{\T_{AC}}
    ]
    \\
    \text{s.t.} \quad 
    &J^{\theta^{\mathrm{Post}}}\geqslant0, \Tr_{D}[J^{\theta^{\mathrm{Post}}}]=\1_{C}.
\end{align}

In the setting where only pre-processing is permitted, post-processing reduces to an identity map $\id$ from $C_1$ to $D$, while the remaining systems $C_2,\ldots, C_n$ are traced out, as illustrated in Fig.~\ref{fig:PreP_UA_EA}(a). 
Under this constraint, the optimal achievable performance is characterized by the following expression.
\begin{align}\label{eq:PreP_Fundamental_Limit}
    F_{\mathrm{PreP}}
    \coloneqq
    \max \quad 
    & 
    \Tr[J^{\theta^{\mathrm{PreP}}}_{AB}\cdot \left(\frac{1}{|\mS|}\sum_{\psi\in\mS}\psi^{\otimes n}_{A}\otimes\psi_{C_1}\right)\cdot
    \left((J^{\mN})^{\otimes n}_{BC}\right)]
    \\
    \text{s.t.} \quad 
    &J^{\theta^{\mathrm{PreP}}}\geqslant0, \Tr_{B}[J^{\theta^{\mathrm{PreP}}}]=\1_{A}.
\end{align}

Returning to the forward-assisted (FA) setting, attention is first restricted to protocols constrained by PPT operations with respect to the post-processing systems, for which the achievable performance takes the form
\begin{align}\label{eq:PPT_Fundamental_Limit}
    F_{\mathrm{PPT}}
    \coloneqq
    \max \quad 
    & 
    \Tr[J^{\theta}_{ABCD}\cdot \left( \Psi_{AD}\otimes\Omega_{BC}
    \right)]
    \\
    \text{s.t.} \quad 
    &J^{\theta}\geqslant0, \Tr_{BD}[J^{\theta}]=\1_{AC}, 
    \Tr_{D}[J^{\theta}]=
    J^{\theta}_{AB}\otimes\frac{1}{d_C}\1_{C},
    (J^{\theta})^{\T_{CD}}\geqslant0.
\end{align}
The second FA protocol concerns NS purification, for which the achievable performance is given by
\begin{align}\label{eq:NS_Fundamental_Limit}
    F_{\mathrm{NS}}
    \coloneqq
    \max \quad 
    & 
    \Tr[J^{\theta}_{ABCD}\cdot \left( \Psi_{AD}\otimes\Omega_{BC}
    \right)]
    \\
    \text{s.t.} \quad 
    &J^{\theta}\geqslant0, \Tr_{BD}[J^{\theta}]=\1_{AC}, 
    \Tr_{D}[J^{\theta}]=
    J^{\theta}_{AB}\otimes\frac{1}{d_C}\1_{C},
    \Tr_{B}[J^{\theta}]=\frac{1}{d_A}\1_{A}\otimes
    J^{\theta}_{CD}.
\end{align}
If the FA purification protocol satisfies both PPT and NS constraints, the corresponding performance is expressed as
\begin{align}\label{eq:NSPPT_Fundamental_Limit}
    F_{\mathrm{PPT}\cap\mathrm{NS}}
    \coloneqq
    \max \quad 
    & 
    \Tr[J^{\theta}_{ABCD}\cdot \left( \Psi_{AD}\otimes\Omega_{BC}
    \right)]
    \\
    \text{s.t.} \quad 
    &J^{\theta}\geqslant0, \Tr_{BD}[J^{\theta}]=\1_{AC}, 
    \Tr_{D}[J^{\theta}]=
    J^{\theta}_{AB}\otimes\frac{1}{d_C}\1_{C},
    \Tr_{B}[J^{\theta}]=\frac{1}{d_A}\1_{A}\otimes
    J^{\theta}_{CD},
    (J^{\theta})^{\T_{CD}}\geqslant0.
\end{align}

The analysis in this subsection shows that forward-assisted purification can be cast as a unified optimization over superchannels, with operational constraints directly shaping the admissible transformations. 
These constraints do not form a simple hierarchy, but instead define partially overlapping classes, such as PPT and NS protocols, together with their subclasses FHA, FCA, and EA, as illustrated in Fig.~\ref{fig:Venn}. 
Each constraint targets a distinct aspect of the dynamics: PPT restricts the structure of correlations, NS enforces causal independence, and entanglement assistance expands the available resources in a complementary manner. 
Consequently, purification performance does not admit a universal ordering across these classes, but is determined by how the constraints intersect. 
The resulting semidefinite programs provide a unified and computable description of these regimes, making explicit the trade-off between physical constraints and achievable fidelity, and identifying the regimes in which forward-assisted strategies surpass post-processing alone.
Further analysis and numerical results are presented in the following sections.


\section{Advantages in Global Quantum State Purification}\label{sec:Advantages_GQSP}

In this section, we investigate forward-assisted (FA) purification in the global setting, focusing on how temporal structure and pre-processing reshape the attainable limits of purification.
Subsection~\ref{subsec:PreP_Post} benchmarks FA purifications against conventional schemes restricted to post-processing, establishing that the inclusion of pre-processing already leads to a pronounced performance advantage. 
Subsection~\ref{subsec:Less_is_More} sharpens this insight by showing that a single-copy protocol --- integrating pre-processing (PreP) with post-processing --- can surpass multi-copy purification strategies based solely on post-processing (PostP). 
Notably, we identify noise regimes in which a single-copy PreP outperforms PostP schemes operating on up to 1905 copies, revealing a striking gain in sample efficiency.
Subsection~\ref{subsec:SDP_Symmetry} develops the technical framework underpinning our numerical experiments. 
By leveraging Schur-Weyl duality together with Clebsch-Gordan recursion, the underlying SDPs can be reduced to small blocks, rendering previously inaccessible regimes tractable.
Without this structure, direct SDP implementations become prohibitive beyond roughly 8 copies even for qubit systems. 
Further technical details are provided in the corresponding subsections.


\subsection{Pre- and Post-Processing Purification}\label{subsec:PreP_Post}

From a mathematical perspective, forward-assisted (FA) purification can, in principle, achieve higher performance than conventional protocols based solely on post-processing. 
From a physical standpoint, however, the central question is how such advantages can be realized experimentally. 
In particular, it is important to ask whether performance gains persist under practically accessible restrictions of FA protocols.
In this subsection, we compare purification schemes based solely on pre-processing with those relying only on post-processing. 
This comparison highlights the essential role of pre-processing and examines its simplest experimentally viable form --- local unitary operations --- in enabling enhanced purification.

We consider input states $\psi$ drawn from a set $\mS$, subject to amplitude-damping noise $\mN_{\mathrm{AD}}(p)$, which maps $\psi$ to its noisy counterpart $\mN_{\mathrm{AD}}(p)(\psi)$.
The protocols under consideration are illustrated in Fig.~\ref{fig:Global_Purifications_Pre_Post}.

\begin{figure}[htbp]
    \centering   
    \includegraphics[width=1\textwidth]{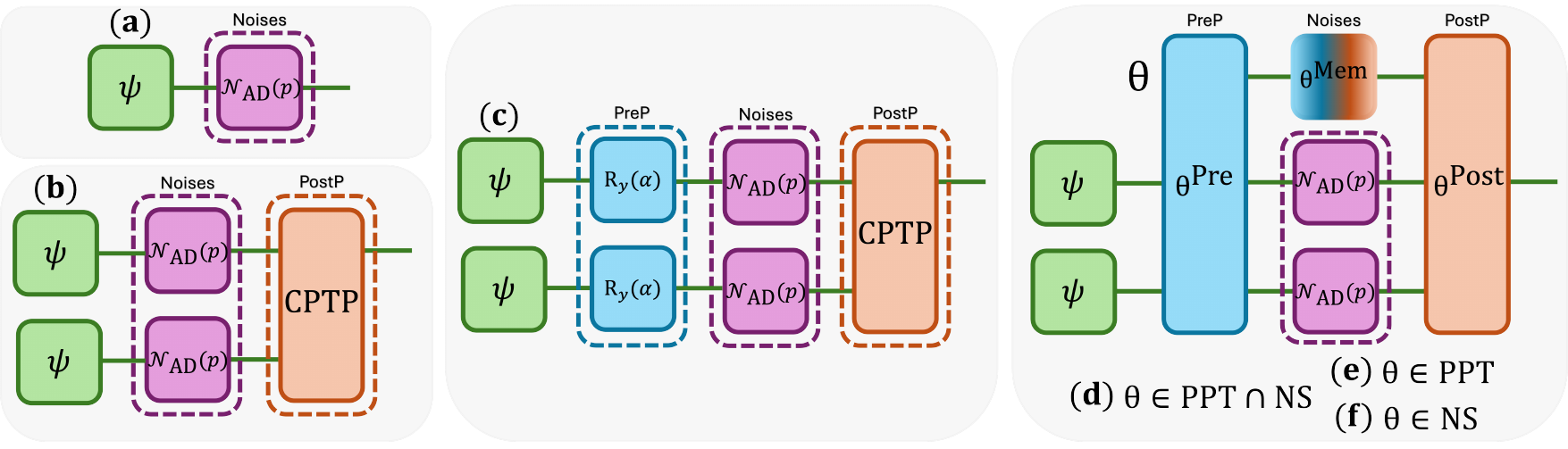}
    \caption{\textbf{Forward-Assisted Global Purifications}. 
        The input is a quantum state $\psi\in\mS$, and noise is modeled by an amplitude damping channel $\mN_{\mathrm{AD}}(p)$. 
        (a) Single-copy baseline without purification, where the input state $\psi$ undergoes amplitude damping noise $\mN_{\mathrm{AD}}(p)$.
        (b) Two-copy purification with CPTP post-processing applied after noise.
        (c) Forward-assisted scheme with local unitary (LU) pre-processing $\mathrm{R}_{y}(\alpha)$ prior to noise, followed by CPTP post-processing.
        (d)–(f) General forward-assisted (FA) purification protocols described by a superchannel $\theta$ (see Fig.~\ref{fig:Superchannel}), incorporating pre-processing $\theta^{\mathrm{Pre}}$, a memory channel $\theta^{\mathrm{Mem}}$, and post-processing $\theta^{\mathrm{Post}}$, subject to different constraints: (d) $\theta\in \mathrm{PPT}\cap \mathrm{NS}$, (e) $\theta\in \mathrm{PPT}$, and (f) $\theta\in \mathrm{NS}$.
        The performance of all protocols is quantified by the average fidelity (see Eq.~\eqref{eq:FA_Fundamental_Limit}) with respect to the target $\psi\in\mS$.
    }
    \label{fig:Global_Purifications_Pre_Post}
\end{figure}

Consider first the simplest setting and evaluate the performance of the protocols introduced in Fig.~\ref{fig:Global_Purifications_Pre_Post}. 
Specifically, we define a set of states $\{\psi(x_i)\}$ as
\begin{align}\label{eq:psi_alpha_i}
    \ket{\psi(x_i)}:=\sqrt{x_i}\ket{0}+\sqrt{1-x_i}\ket{1},
\end{align}
and consider the ensemble
\begin{align}\label{eq:psi_03_09}
    \{\psi(0.3), \psi(0.9)\},
\end{align}
where each state is sampled with equal probability $1/2$. 
Let $\varphi_x$ denote the output state produced from the initial state $\psi$ under protocol $x\in\{a,b,c,d,e,f\}$ in Fig.~\ref{fig:Global_Purifications_Pre_Post}, namely
\begin{align}
    \psi\xrightarrow{\text{protocol}\,\, x\,\,\text{in Fig.~\ref{fig:Global_Purifications_Pre_Post}}}
    \varphi_x,
\end{align}
and define the corresponding performance measure $G_x$ as the average fidelity, given by
\begin{align}\label{eq:G_x}
    G_x \coloneqq \max_{\text{protocol}\,\, x}\frac{1}{|\mS|}\sum_{\psi\in\mS}F(\varphi_x,\psi).
\end{align}
where $F(\varphi_x,\psi) = \Tr[\varphi_x\cdot\psi]$ denotes the quantum fidelity between states.
The optimization under each protocol can be formulated as a semidefinite program (SDP), enabling efficient evaluation; further details are provided in Subsec.~\ref{subsec:Optimization_Purification_Performance}.
The resulting performance of the protocols is presented in Fig.~\ref{fig:Global_Num_Point1}.

\begin{figure}[htbp]
    \centering   
    \includegraphics[width=1\textwidth]{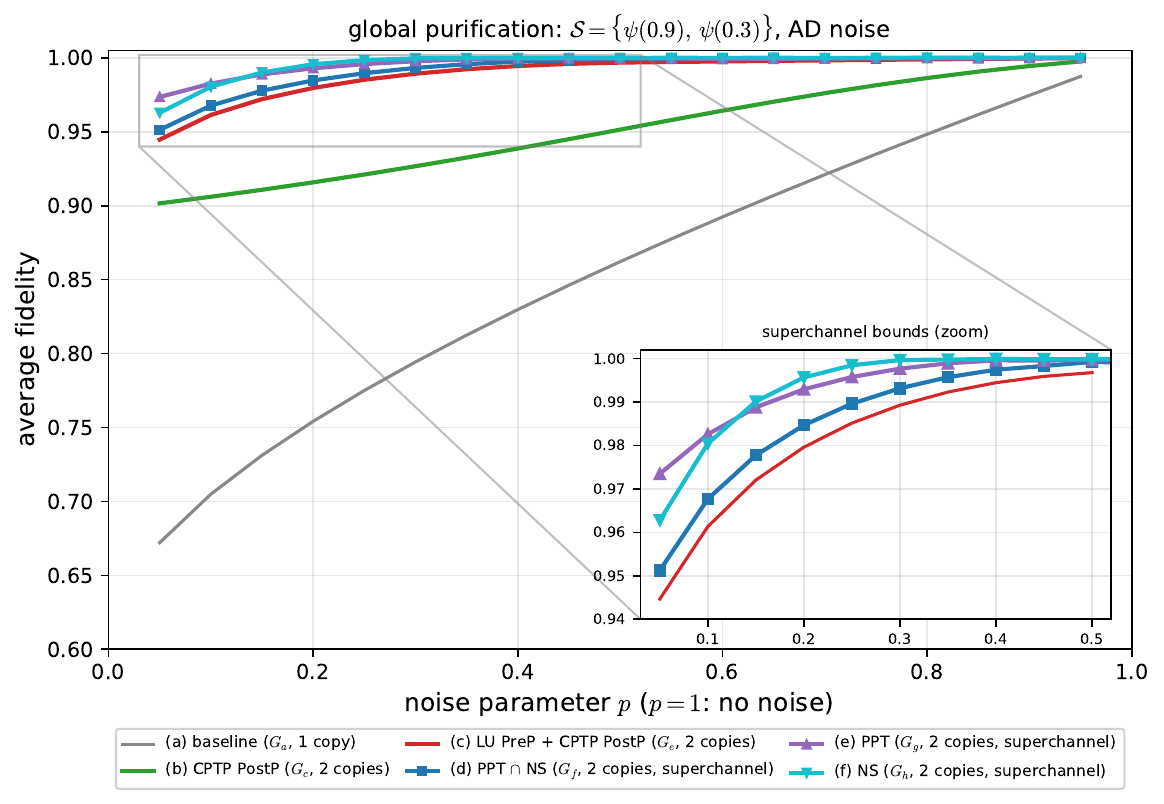}
    \caption{\textbf{Comparison Across Global Purification Protocols}. 
        Average fidelity $G_x$ (see Eq.~\eqref{eq:G_x}) as a function of the noise parameter $p$ (see Fig.~\ref{fig:Global_Purifications_Pre_Post}) for global purification with input ensemble $\mS=\{\psi(0.9),\psi(0.3)\}$ under amplitude damping (AD) noise $\mN_{\mathrm{AD}}(p)$. 
        The baseline (a) corresponds to single-copy transmission without purification, while (b)–(f) compare two-copy strategies including CPTP post-processing, LU pre-processing followed by CPTP post-processing, and forward-assisted (FA) protocols implemented via superchannels subject to PPT, NS, or combined PPT $\cap$ NS constraints. 
        The inset provides a magnified view of the high-fidelity regime, highlighting the separation between different FA purification protocols. 
        Across all noise regimes, FA protocols exhibit enhanced performance.
    }
    \label{fig:Global_Num_Point1}
\end{figure}

Numerical simulations, shown in Fig.~\ref{fig:Global_Num_Point1}, demonstrate that all forward-assisted (FA) purification protocols including $\mathrm{PPT}\cap \mathrm{NS}$ (see Fig.~\ref{fig:Global_Purifications_Pre_Post}(d)), $\mathrm{PPT}$ (see Fig.~\ref{fig:Global_Purifications_Pre_Post}(e)), and $\mathrm{NS}$ (see Fig.~\ref{fig:Global_Purifications_Pre_Post}(f)) classes or even pre-processing-augmented schemes (see Fig.~\ref{fig:Global_Purifications_Pre_Post}(c)), surpass the fundamental limits previously established for schemes restricted to post-processing alone, thereby exceeding the performance of the conventional approach. 
These results highlight a key omission in conventional approaches: the roles of pre-processing and memory (see Fig.~\ref{fig:FA_Purification_Simplified}) --- both across spatially separated systems and over temporal evolution --- have largely been neglected. 
Once these elements are incorporated within the unified framework of FA purification, a marked performance enhancement emerges.
A detailed comparison arises when restricting FA protocols to specific structural classes. 
When FA protocols are simultaneously constrained to be PPT and NS, their performance is strictly inferior to that achieved under either constraint individually (see Fig.~\ref{fig:Global_Num_Point1}), namely
\begin{align}
    G_d\leqslant G_e,
\end{align}
and
\begin{align}
    G_d\leqslant G_f.
\end{align}
As illustrated in Fig.~\ref{fig:Venn} and substantiated by Fig.~\ref{fig:Global_Num_Point1}, neither PPT (whose performance is captured by the purple curve in Fig.~\ref{fig:Global_Num_Point1}) nor NS (represented by the cyan curve in Fig.~\ref{fig:Global_Num_Point1}) forms a subset of the other, and their relative performance depends sensitively on the parameter regime. 
In terms of average fidelity (see Eq.~\eqref{eq:FA_Fundamental_Limit}), no uniform hierarchy emerges: each class outperforms the other in complementary regions of the parameter space.

Despite these advantages over conventional post-processing schemes (see Fig.~\ref{fig:Global_Purifications_Pre_Post}(b)), both PPT- and NS-assisted protocols may require complex operational structures and sustained quantum memory (see Fig.~\ref{fig:Global_Purifications_Pre_Post}(d)-(f)), posing significant challenges for near-term experimental implementation. 
This raises a natural question: can one retain an advantage without relying on quantum memory?
The answer is affirmative. 
Even in the absence of quantum memory, a simple strategy (see Fig.~\ref{fig:Global_Purifications_Pre_Post}(c)) --- combining local unitary pre-processing with standard post-processing --- already exceeds the performance of conventional purification protocols (see Fig.~\ref{fig:Global_Num_Point1}). 
This observation underscores the operational power of pre-processing and points to a practically accessible route toward enhanced quantum state purification.

Before closing this subsection, we summarise the hierarchical relations among the different forward-assisted purification protocols in the following theorem.

\begin{mythm}{Hierarchy of Forward-Assisted Purification in the Global Setting}{Hierarchy_Global}
In the global setting (see Fig.~\ref{fig:Global_Purifications_Pre_Post}), the achievable fidelities of forward-assisted purification protocols organize into a well-defined hierarchy, given by
\begin{align}
    G_a\leqslant G_b\leqslant G_c\leqslant G_d\leqslant G_e,
\end{align}
which, under an alternative admissible constraint, extends to
\begin{align}
    G_a\leqslant G_b\leqslant G_c\leqslant G_d\leqslant G_f.
\end{align}
\end{mythm}

This subsection delineates the structure of forward-assisted purification in the global setting. By placing all protocols within a unified optimisation framework and resolving their performance, a clear hierarchy of achievable fidelities emerges. 
Forward-assisted schemes systematically outperform post-processing-only strategies (see Fig.~\ref{fig:Global_Purifications_Pre_Post}(b)), reflecting an intrinsic enlargement of the operational landscape once pre-processing and spatiotemporal framework are admitted. 
Strikingly, this advantage already manifests at its most minimal level: local-unitary pre-processing, without any additional quantum memory, suffices to surpass the conventional limit. 
Beyond this baseline, the comparison between PPT- (see Fig.~\ref{fig:Global_Purifications_Pre_Post}(e)) and NS-constrained (see Fig.~\ref{fig:Global_Purifications_Pre_Post}(f)) protocols reveals a more intricate structure, with no universal ordering and a pronounced dependence on the noise regime. 
Taken together, these results isolate the mechanisms that drive performance enhancement and clarify the operational role of the spatiotemporal framework, thereby establishing a principled foundation for the distributed setting analyzed next.


\subsection{Less is More: 
Single-Copy Pre-Processing Surpasses Multi-Copy Post-Processing}\label{subsec:Less_is_More}

\begin{figure}[htbp]
    \centering   
    \includegraphics[width=1\textwidth]{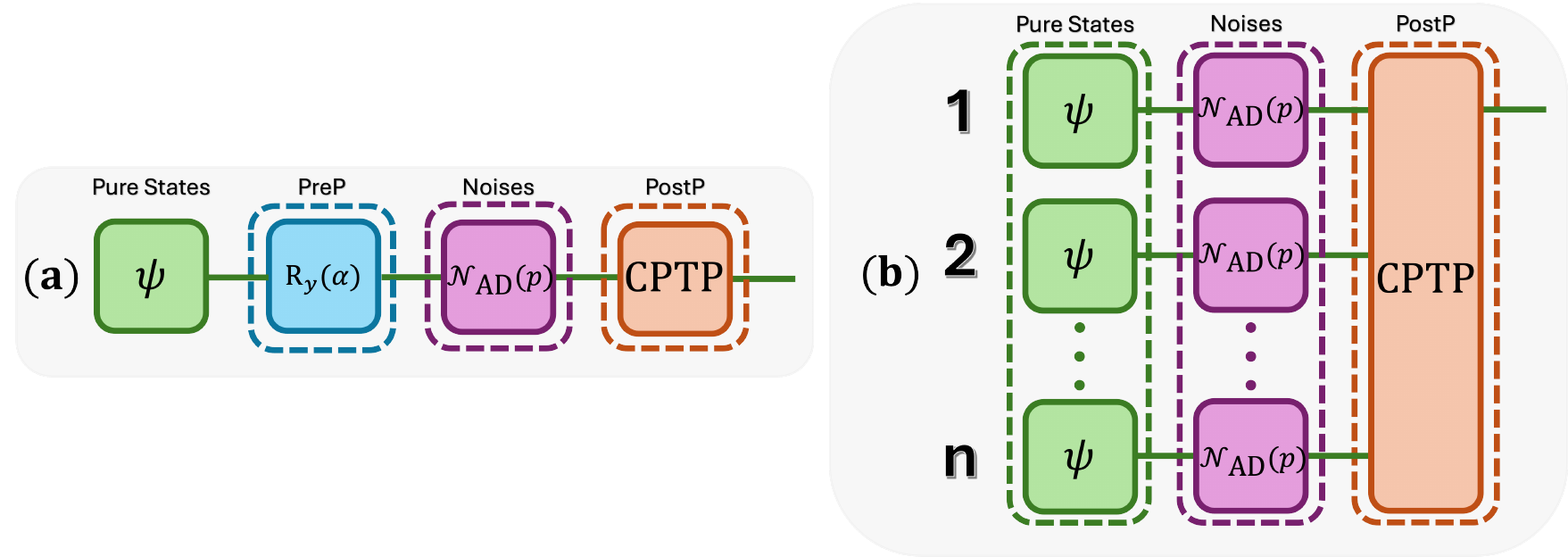}
    \caption{\textbf{Multi-Copy Global Purification}. 
        The input is a quantum state $\psi\in\mS$, and noise is modeled by an amplitude damping channel $\mN_{\mathrm{AD}}(p)$. 
        (a) Single-copy purification protocol incorporating local unitary pre-processing $\mathrm{R}_{y}(\alpha)$ prior to the action of amplitude damping noise $\mN_{\mathrm{AD}}(p)$, followed by a general CPTP post-processing map.
        (b) Conventional multi-copy purification protocol, where $n$ identical noisy copies --- each obtained by applying $\mN_{\mathrm{AD}}(p)$ to the input state $\psi$ --- are collectively processed via a global CPTP operation without pre-processing.
    }
    \label{fig:Global_Multiple_Copies}
\end{figure}

In the previous subsection, we demonstrated that incorporating pre-processing into quantum state purification enables performance beyond what is achievable with conventional protocols that rely solely on post-processing (see Fig.~\ref{fig:Global_Num_Point1}). 
A central practical consideration, however, concerns the number of noisy copies required. 
Beyond improvements in average fidelity, it is essential to ask whether pre-processing can also offer an advantage in terms of sample efficiency.
More specifically, can an $n$-copy protocol with pre-processing and post-processing outperform a conventional protocol that relies on post-processing alone using a strictly larger number $m>n$ of copies?
We show that this is indeed the case. 
Remarkably, even a single-copy pre-processing strategy, when combined with post-processing, can surpass multi-copy post-processing purification.
The protocols under consideration are illustrated in Fig.~\ref{fig:Global_Multiple_Copies}.
In particular, we identify initial states and regimes of noise parameters in which a single-copy pre-processing-augmented protocol (see Fig.~\ref{fig:Global_Multiple_Copies}(a)) outperforms conventional purification (see Fig.~\ref{fig:Global_Multiple_Copies}(b)) that employs 1905 copies; that is
\begin{align}
    \text{1-copy}\,\, \mathrm{PreP}\,\,\text{Purification}
    >
    \text{1905-copy}\,\, \mathrm{PostP}\,\,\text{Purification}.
\end{align}
Here, $>$ denotes superior performance in terms of average fidelity.
These findings establish that incorporating pre-processing is not merely beneficial but essential: it simultaneously enhances achievable performance and reduces the number of noisy copies required, thereby improving the overall sample efficiency of quantum state purification.

We begin with the minimal setting of quantum state purification, in which the ensemble of initial states comprises only two elements.
For the ensembles specified in Eq.~\eqref{eq:psi_03_09}, i.e., $\{\psi(0.3), \psi(0.9)\}$, Fig.~\ref{fig:Global_Num_Scaling} presents the average fidelity as a function of the noise parameter $p$.
Numerical results in Fig.~\ref{fig:Global_Num_Scaling} show that even when local-unitary pre-processing is applied to a single noisy copy, the resulting protocol consistently outperforms conventional purification schemes based solely on post-processing using up to three copies (see Fig.~\ref{fig:Global_Multiple_Copies}(b)). 
For larger numbers of copies ($n=4,5,6$), the pre-processing-augmented protocol remains competitive and continues to exhibit superior performance over a broad range of noise strengths.

\begin{figure}[htbp]
    \centering   
    \includegraphics[width=1\textwidth]{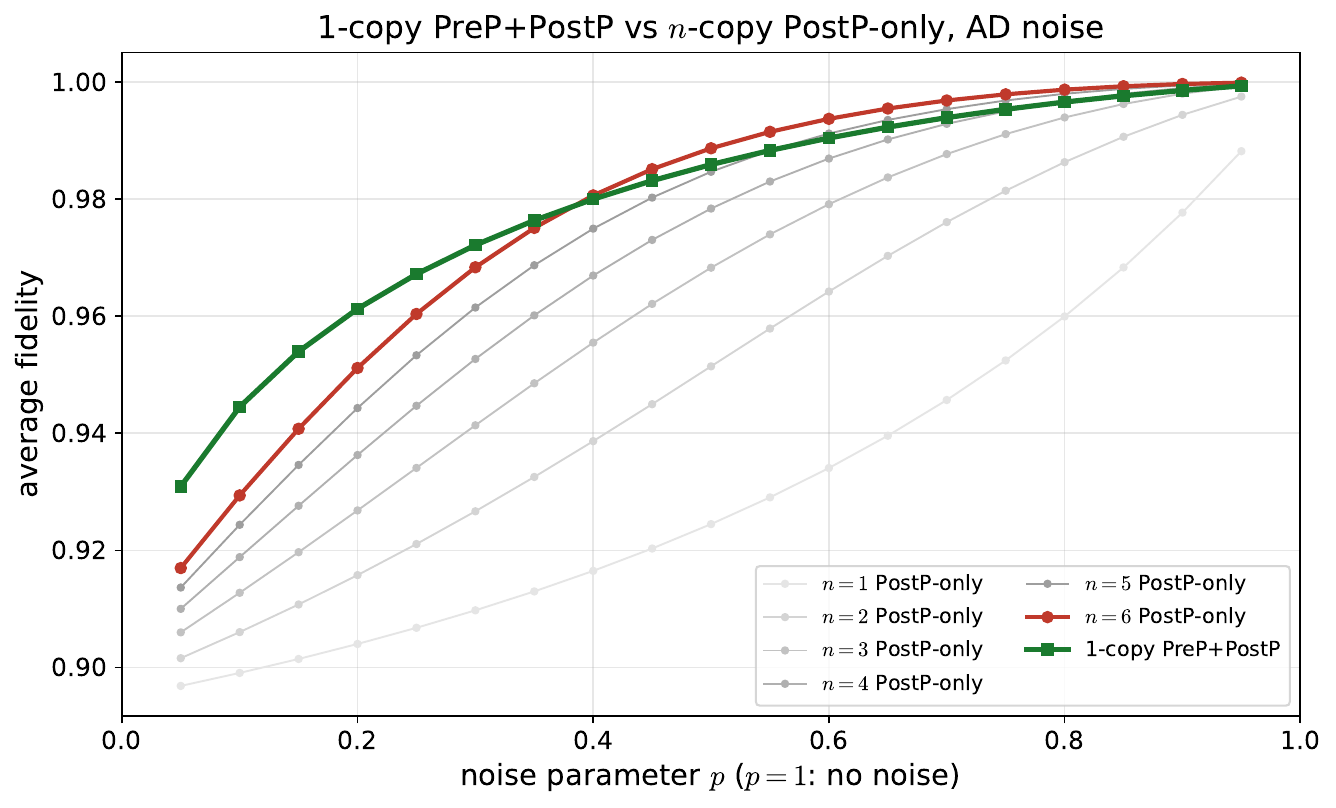}
    \caption{\textbf{Single-Copy LU Pre-Processing vs Multi-Copy Post-Processing}. 
        Average fidelity as a function of the noise parameter $p$ for the ensembles defined in Eq.~\eqref{eq:psi_03_09}. 
        The green curve shows the single-copy protocol combining local unitary (LU) pre-processing with post-processing (see Fig.~\ref{fig:Global_Multiple_Copies}(a)), while the grey curves correspond to conventional post-processing-only protocols (see Fig.~\ref{fig:Global_Multiple_Copies}(b)) using $n=1$ to $5$ copies, and the red curve to $n=6$.
        The results show that single-copy pre-processing can outperform multi-copy post-processing across a broad range of noise strengths, demonstrating a clear advantage in both performance and sample efficiency.
    }
    \label{fig:Global_Num_Scaling}
\end{figure}

\begin{figure}[htbp]
    \centering   
    \includegraphics[width=1\textwidth]{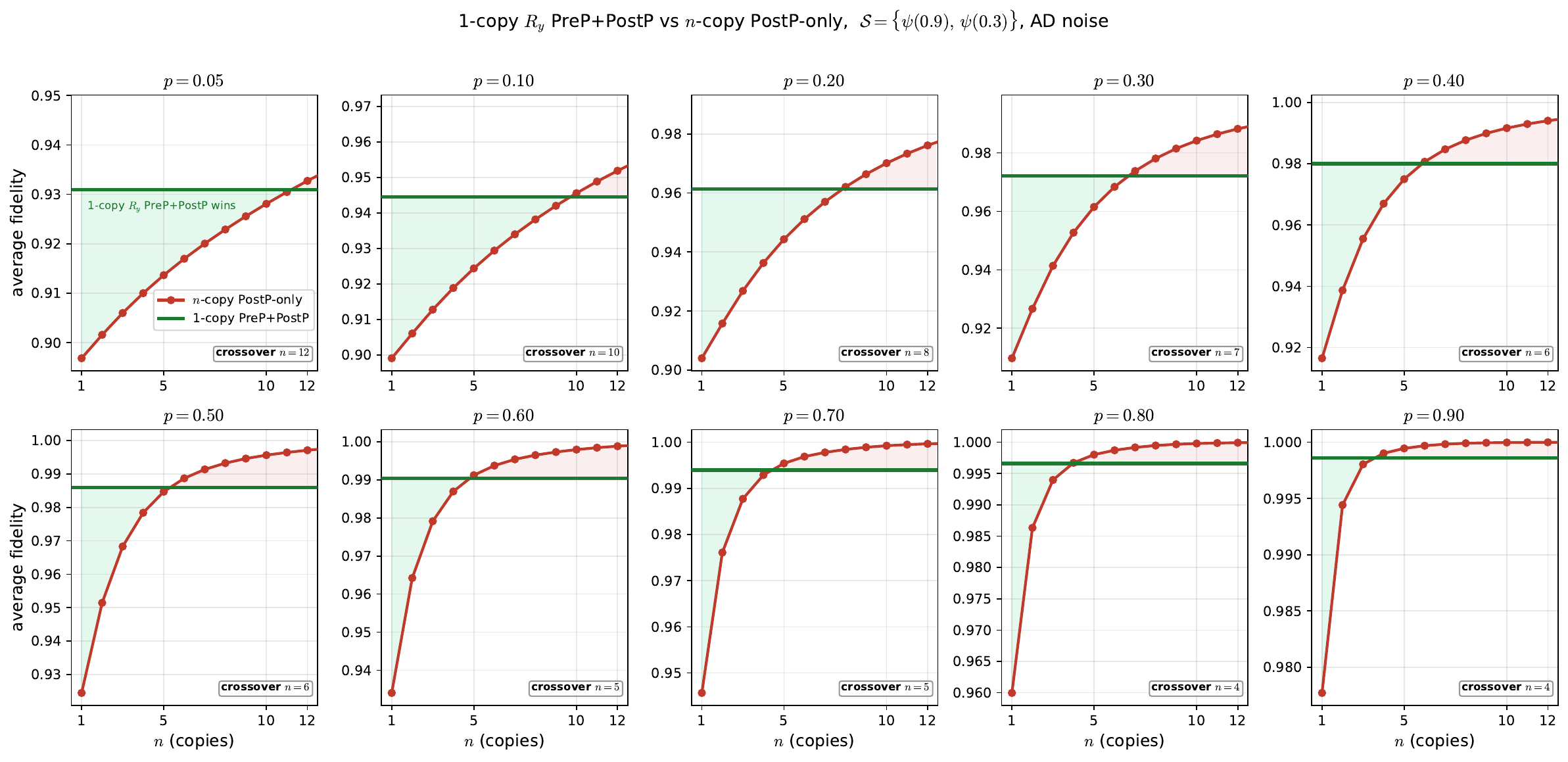}
    \caption{\textbf{Single-Copy PreP vs Multi-Copy PostP across Noise Regimes}. 
        Average fidelity as a function of the number of copies $n$ for fixed amplitude damping noise strengths $p$. 
        Each panel corresponds to a different value of $p$, with red curves representing conventional post-processing-only purification using $n$ copies (see Fig.~\ref{fig:Global_Multiple_Copies}(b)), and the green line indicating the performance of the single-copy protocol combining local-unitary pre-processing with post-processing (see Fig.~\ref{fig:Global_Multiple_Copies}(a)). 
        The shaded region highlights the parameter regime where the single-copy pre-processing protocol exceeds the performance of multi-copy post-processing. 
        Across a broad range of noise strengths, single-copy pre-processing surpasses the performance of protocols using multiple copies, demonstrating a clear advantage in sample efficiency.
    }
    \label{fig:Global_Num_Point2_1}
\end{figure}

\begin{figure}[htbp]
    \centering   
    \includegraphics[width=1\textwidth]{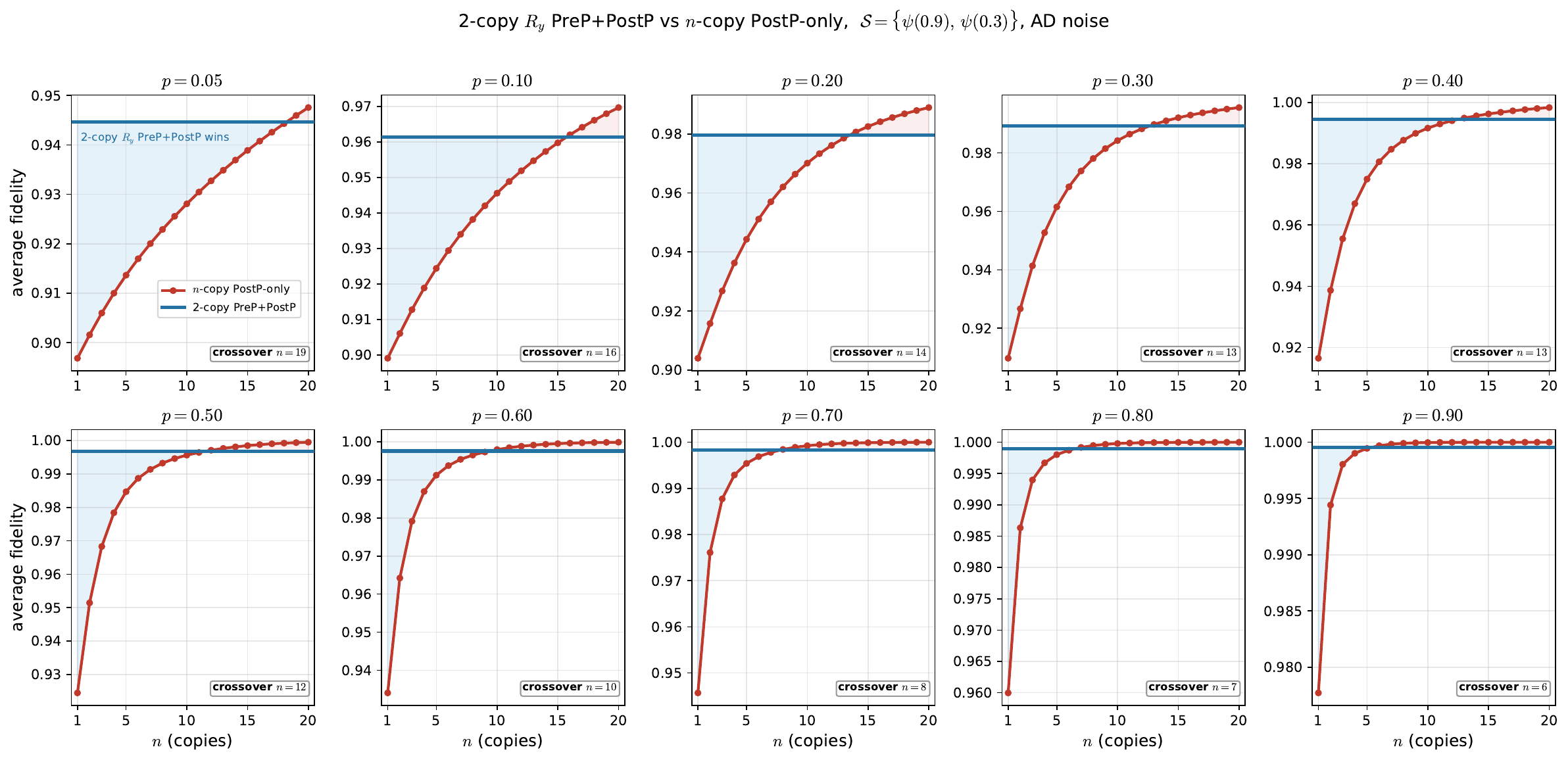}
    \caption{\textbf{Two-Copy PreP vs Multi-Copy PostP across Noise Regimes}. 
        Average fidelity as a function of the number of copies $n$ for fixed amplitude damping noise strengths $p$.
        Each panel corresponds to a different noise level. 
        The red curves depict conventional purification based solely on post-processing with $n$ copies (see Fig.~\ref{fig:Global_Multiple_Copies}(b)), whereas the blue line represents the performance of the two-copy protocol augmented with local-unitary pre-processing.
        The shaded region delineates the regime in which the two-copy pre-processing–augmented strategy surpasses all post-processing-only schemes, even those supplied with a larger number of copies. 
        Over a wide range of noise parameters, this separation persists, showing that incorporating pre-processing reduces the number of required copies while maintaining higher fidelity, thereby significantly improving sample efficiency relative to conventional purification protocols.
    }
    \label{fig:Global_Num_Point2_2}
\end{figure}

Such behavior highlights the pivotal role of pre-processing in determining purification performance. 
Even in the most elementary setting, access to a pre-processing stage enables a protocol operating on a single noisy copy to conventional --- and in broad regimes surpass --- schemes that rely on multiple copies yet lack temporal structure.
More detailed comparisons between single-copy pre-processing with post-processing (see Fig.~\ref{fig:Global_Multiple_Copies}(a)) and multi-copy post-processing (see Fig.~\ref{fig:Global_Multiple_Copies}(b)) are presented in Fig.~\ref{fig:Global_Num_Point2_1}.

Figure~\ref{fig:Global_Num_Point2_1} reveals a striking role of pre-processing in reshaping the resource requirements of purification. 
Even when restricted to a single noisy copy, the inclusion of local-unitary pre-processing elevates performance to a level that rivals, and in a wide parameter regime surpasses, that of conventional protocols relying on multiple copies and post-processing alone. 
The crossover behavior observed across different noise strengths shows that increasing the number of copies cannot compensate for the absence of an appropriate pre-processing stage. 
This identifies pre-processing not as a minor refinement, but as a fundamentally distinct operational ingredient that unlocks otherwise inaccessible performance gains. 
In particular, it demonstrates that tailoring the input prior to noise can be more powerful than accumulating additional noisy resources, thereby establishing pre-processing as a key mechanism for achieving enhanced sample efficiency in realistic quantum information tasks.

A natural next step is to examine how this advantage evolves when two pre-processing-augmented noisy states are available for purification. 
Extending the numerical analysis to this regime, we compare protocols that combine pre-processing with post-processing on two copies against conventional strategies that rely solely on post-processing, even when supplied with a larger number of copies. 
This setting isolates the role of pre-processing in the multi-copy regime and tests whether its benefit persists beyond the single-copy scenario. The resulting comparison, presented in Fig.~\ref{fig:Global_Num_Point2_2}, shows that the advantage not only survives but becomes more pronounced, indicating that appropriately structured pre-processing continues to reshape the effective resource landscape even when additional copies are available.

The analysis is extended to ensembles comprising a larger set of initial pure states, thereby moving beyond minimal instances to more representative configurations. 
To make direct contact with quantum cryptographic applications, attention is then restricted to the BB84 protocol~\cite{BENNETT20147}, for which the relevant ensemble consists of the four qubit states
\begin{align}\label{eq:BB84}
    \{\ketbra{0}{0}, \ketbra{1}{1}, \ketbra{+}{+}, \ketbra{-}{-}\}.
\end{align}
The performance of the single-copy protocol, combining pre-processing with subsequent post-processing, is evaluated and benchmarked against conventional strategies that rely solely on post-processing. 
The numerical results, illustrated in Fig.~\ref{fig:BB84}, demonstrate a clear and systematic advantage of the forward-assisted approach across the parameter regime considered.

\begin{figure}[htbp]
    \centering   
    \includegraphics[width=1\textwidth]{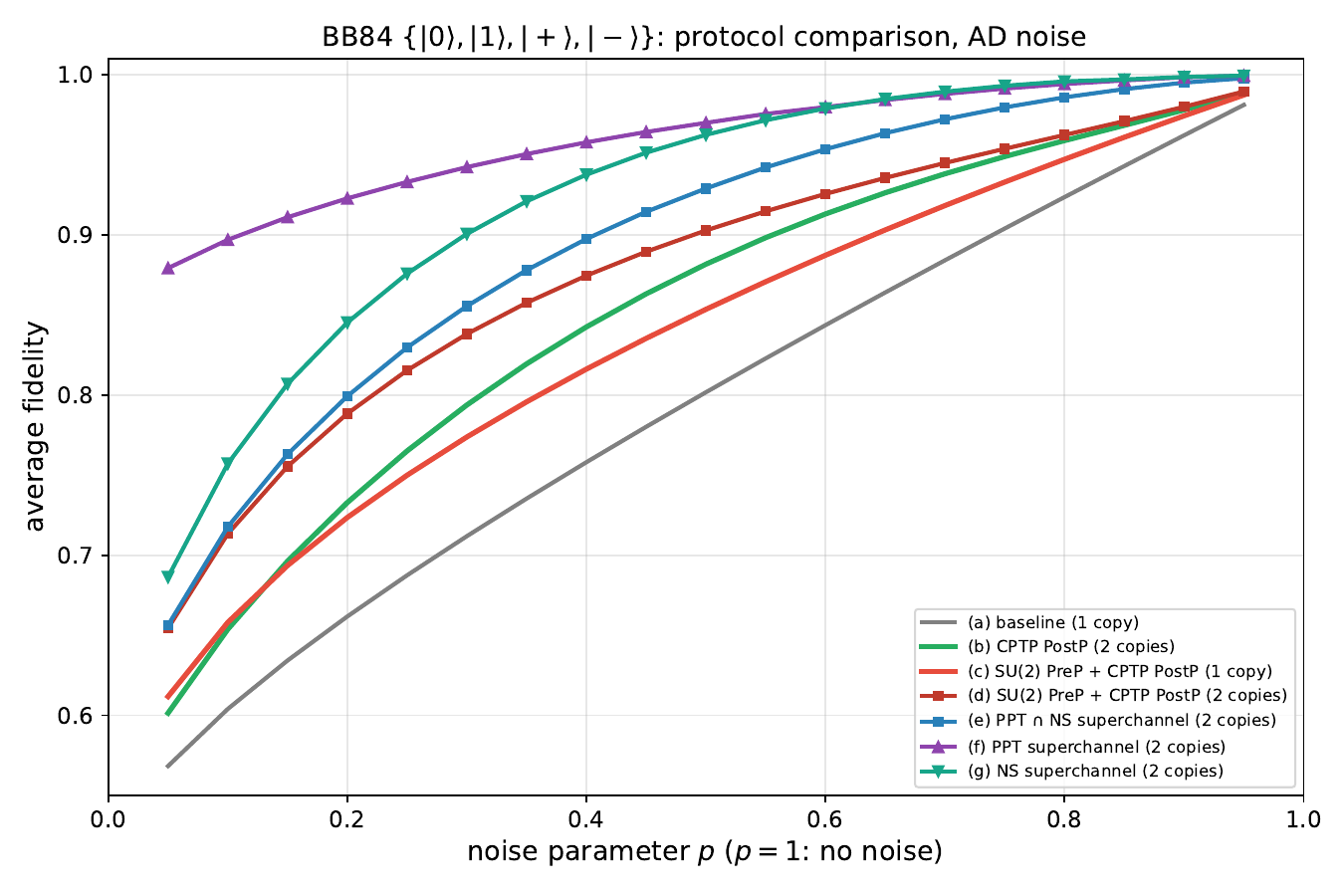}
    \caption{\textbf{Forward-Assisted Purification in BB84}. 
        Average fidelity for the BB84 ensemble $\{\ket{0}, \ket{1}, \ket{+}, \ket{-}\}$ as a function of the noise parameter $p$ under amplitude damping. 
        The comparison spans baseline (a), conventional CPTP post-processing with two copies (b), single- and two-copy protocols augmented by SU(2) pre-processing (c)–(d), and forward-assisted strategies implemented via superchannels under PPT (f), NS (g), and their intersection (e). 
        Across the full noise regime, for the two-copy case, protocols incorporating pre-processing or forward-assisted structures consistently outperform post-processing-only approaches. 
        Among these, generic forward-assisted purification schemes --- particularly those constrained by $\mathrm{PPT}\cap\mathrm{NS}$, as well as the individual $\mathrm{PPT}$ and $\mathrm{NS}$ classes --- achieve the highest fidelities.
    }
    \label{fig:BB84}
\end{figure}

Figure~\ref{fig:BB84} highlights the advantage of forward-assisted purification in a cryptographic setting through the BB84 ensemble under amplitude damping noise. 
Even at the single-copy level, the inclusion of pre-processing enables performance that surpasses conventional post-processing strategies operating on two copies in certain parameter regimes. 
When two copies are available, protocols augmented with pre-processing already outperform post-processing-only approaches, while general forward-assisted implementations based on superchannels provide a further, systematic enhancement across the entire noise range. 
In particular, schemes constrained by PPT, NS, and especially their intersection achieve the highest fidelities, maintaining a clear separation from standard methods even in moderate- and high-noise regimes. 
These results indicate that exploiting structured dynamical resources across multiple temporal stages, rather than restricting to operations at a single time point, leads to a qualitative performance gain, directly strengthening the state purification and the robustness of quantum key distribution under realistic noise.

The scope of the analysis broadens to ensembles comprising a wider variety of initial states, enabling a more stringent and comprehensive assessment of the protocol under realistic conditions. 
In particular, we consider several distinct ensembles, including the following examples
\begin{align}
    &\left\{\psi(0.3), \psi(0.5), \psi(0.9)\right\},\\
    &\left\{\psi(0.3), \psi(0.5), \psi(0.7), \psi(0.9)\right\},\\
    &\left\{\psi(0.1), \psi(0.3), \psi(0.5), \psi(0.7), \psi(0.9)\right\},
\end{align}
where the state $\psi(x_i)$ is defined in Eq.~\eqref{eq:psi_alpha_i}.
For the ensembles of quantum states considered here, the performance comparison between pre-processing-augmented purification and conventional post-processing-only purification is presented in Fig.~\ref{fig:Global_Num_Point3_1}.

\begin{figure}[htbp]
    \centering   
    \includegraphics[width=1\textwidth]{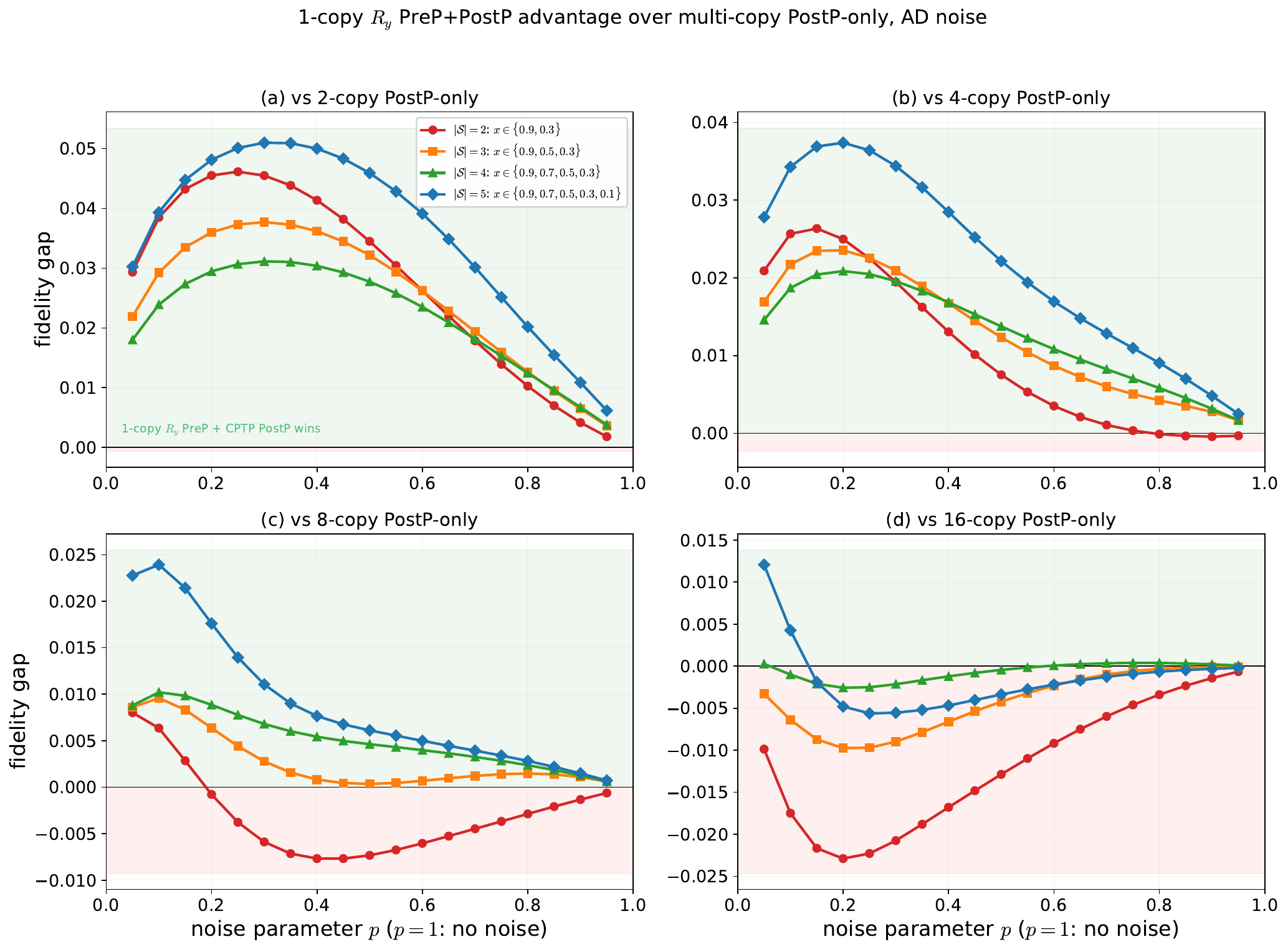}
    \caption{\textbf{Single-Copy Pre-Processing vs Multi-Copy Post-Processing across State Ensembles}. 
        Fidelity gap between single-copy pre-processing-augmented purification (local $\mathrm{R}_{y}$ rotation followed by CPTP post-processing, see Fig.~\ref{fig:Global_Multiple_Copies}(a)) and conventional multi-copy post-processing-only protocols (see Fig.~\ref{fig:Global_Multiple_Copies}(b)) under amplitude damping noise. 
        Panels (a)–(d) compare against 2-, 4-, 8-, and 16-copy post-processing schemes, respectively, for different state ensembles. 
        Positive values (shaded green) indicate regimes where the single-copy pre-processing-augmented protocol outperforms the multi-copy conventional purification, while negative values (shaded red) mark the opposite. 
        Across a broad range of noise parameters, a single pre-processed copy achieves a clear performance advantage over multi-copy post-processing strategies, highlighting a gain in sample efficiency.
    }
    \label{fig:Global_Num_Point3_1}
\end{figure}

Figure~\ref{fig:Global_Num_Point3_1} presents a systematic comparison between single-copy pre-processing-augmented purification (see Fig.~\ref{fig:Global_Multiple_Copies}(a)) and conventional multi-copy post-processing-only protocols (see Fig.~\ref{fig:Global_Multiple_Copies}(b)) across different state ensembles under amplitude damping noise. 
The fidelity gap is plotted as a function of the noise parameter $p$, with panels Fig.~\ref{fig:Global_Num_Point3_1}(a)–(d) benchmarking against 2-, 4-, 8-, and 16-copy post-processing strategies, respectively. 
For small to moderate copy numbers, the single-copy protocol consistently achieves a positive fidelity gap over a broad range of noise strengths, demonstrating a clear performance advantage despite using fewer resources. 
This advantage is particularly pronounced for larger ensembles, indicating that pre-processing effectively tailors the input states to better align with the structure of the noise. 
These results highlight that properly designed pre-processing can significantly enhance purification performance and achieve superior sample efficiency compared to conventional approaches.
A more detailed comparison across diverse state ensembles is presented in Fig.~\ref{fig:Global_Num_Point3_2}, where the advantage of pre-processing is made explicit: higher purification performance is achieved while consuming fewer noisy-state copies.

\begin{figure}[htbp]
    \centering   
    \includegraphics[width=1\textwidth]{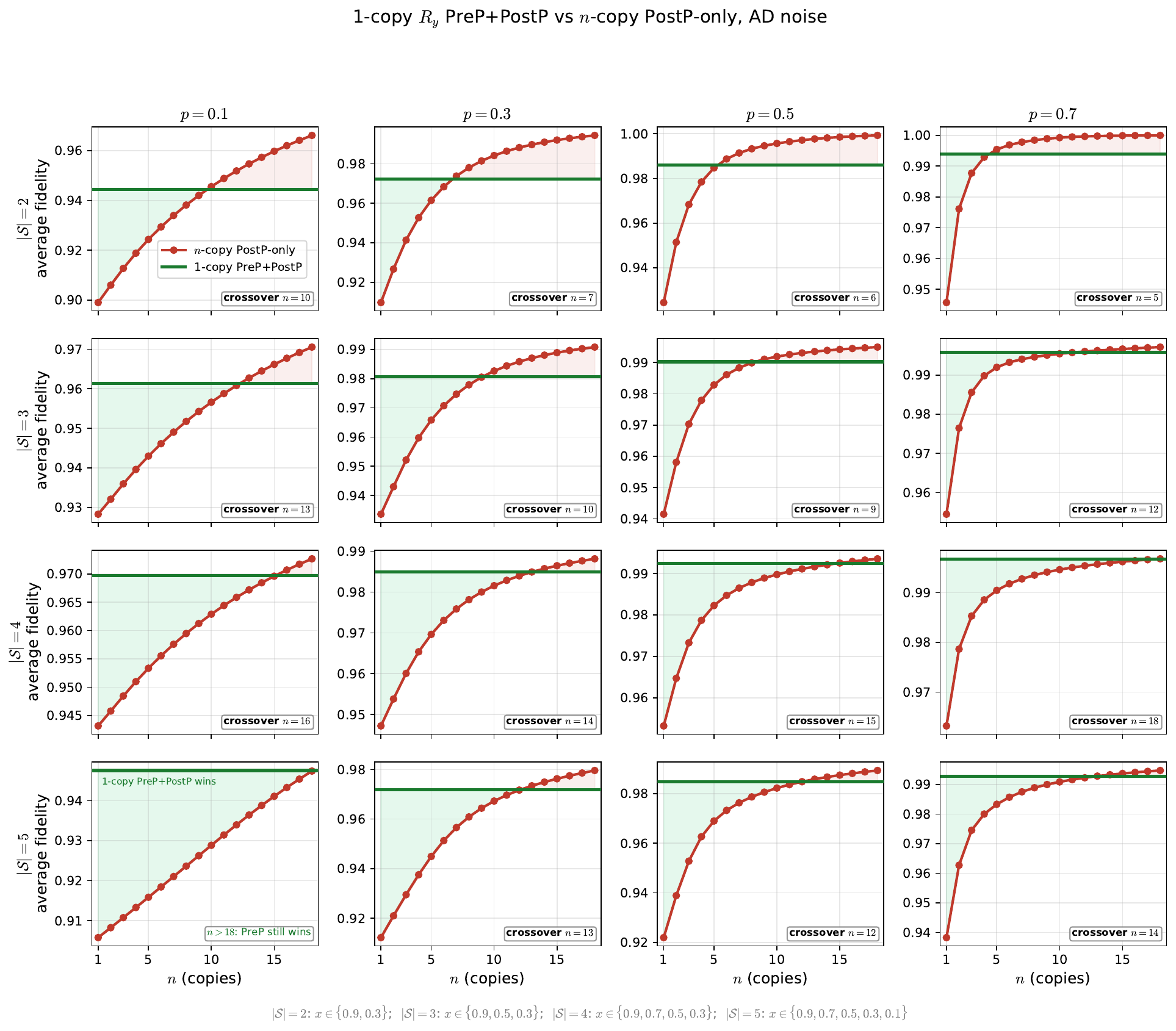}
    \caption{\textbf{Crossover Behavior across Ensembles and Noise}. 
        Average fidelity as a function of the number of copies $n$ for conventional multi-copy post-processing-only protocols (red curves, see Fig.~\ref{fig:Global_Multiple_Copies}(b)), benchmarked against single-copy pre-processing-augmented purification (green lines, see Fig.~\ref{fig:Global_Multiple_Copies}(a)), under amplitude damping noise $\mN_{\mathrm{AD}}(p)$. 
        Columns correspond to different noise strengths $p$, while rows represent state ensembles of increasing size $|\mS|$. 
        The shaded region highlights the regime where single-copy pre-processing outperforms $n$-copy post-processing. 
        The annotated crossover point indicates the minimum number of copies required for post-processing to surpass the single-copy pre-processing-augmented strategy. 
        Across a wide range of ensembles and noise parameters, the crossover occurs at relatively large $n$, and in some regimes is not observed within the considered range, demonstrating that pre-processing achieves competitive and superior performance with substantially fewer copies and thereby provides a clear advantage in sample efficiency.
    }
    \label{fig:Global_Num_Point3_2}
\end{figure}

One may then ask how far single-copy pre-processing-augmented purification (see Fig.~\ref{fig:Global_Multiple_Copies}(a)) can surpass conventional approaches (see Fig.~\ref{fig:Global_Multiple_Copies}(b)).
To address this, we further lower the noise parameter $p$ and consider ensembles of initial states comprising four and five copies. 
The comparison between the pre-processing-augmented protocol and conventional post-processing-only schemes is shown in Fig.~\ref{fig:Global_Num_Point3_3}.
In this setting, a single-copy pre-processing-augmented protocol already surpasses the performance of conventional schemes that rely on up to 25 noisy copies. 
This separation exposes a qualitative advantage of forward-assisted purification, demonstrating that suitable pre-processing can compensate for a substantial increase in resource consumption. 
More broadly, it points to the power of protocols that distribute quantum operations across multiple stages, highlighting the promise of a spatiotemporal framework for quantum information processing.

\begin{figure}[htbp]
    \centering   
    \includegraphics[width=1\textwidth]{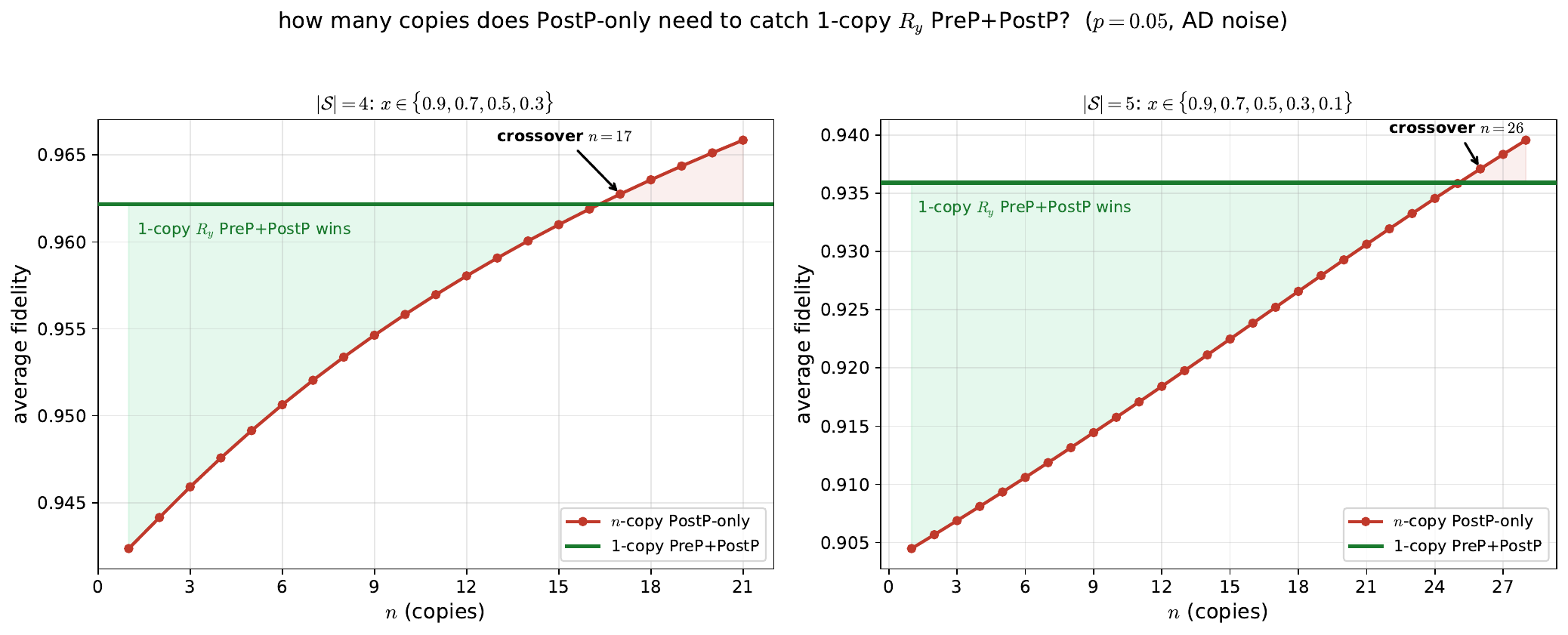}
    \caption{\textbf{Sample Efficiency in Quantum State Purification}. 
        Average fidelity as a function of the number of noisy copies $n$ under amplitude damping noise with $p=0.05$, shown for ensembles of size $|\mS|=4$ (left) and $|\mS|=5$ (right). 
        The green line indicates the performance of the single-copy pre-processing–augmented protocol (see Fig.~\ref{fig:Global_Multiple_Copies}(a)), while the red curve corresponds to conventional post-processing–only purification using $n$ copies (see Fig.~\ref{fig:Global_Multiple_Copies}(b)).
        For small to intermediate $n$, the single-copy protocol consistently achieves higher fidelity, as highlighted by the shaded region. 
        The crossover points identify the number of copies required for post-processing-only strategies to reach the same performance, occurring at $n=17$ and $n=26$, respectively. 
        These results reveal that suitable pre-processing can offset a substantial increase in resource consumption, providing a clear advantage in sample efficiency and pointing towards the potential of spatiotemporal protocol design.
    }
    \label{fig:Global_Num_Point3_3}
\end{figure}

A closer inspection of Fig.~\ref{fig:Global_Num_Point3_3} reveals that the advantage of pre-processing is not only quantitative but also structurally robust across ensembles. 
In both cases, the fidelity achieved by the single-copy protocol remains essentially constant, while the performance of post-processing-only strategies improves only gradually with the number of copies. 
This mismatch in scaling gives rise to a pronounced crossover, indicating that substantial resource overhead is required for conventional approaches to compensate for the absence of pre-processing. 
The location of the crossover shifts to larger $n$ as the ensemble size increases, suggesting that the benefit of pre-processing becomes more pronounced in more complex input settings. 
These trends indicate that pre-processing effectively reshapes the input before noise acts, enabling a more favorable alignment with the noise structure and thereby achieving performance levels that would otherwise require a significantly larger number of copies.

Further numerical results, shown in Fig.~\ref{fig:Global_Num_Point3_4}, push the comparison into a more demanding regime by enlarging the ensemble and extending the post-processing benchmark to the limit of current computational capability. 
The ensembles considered take the form
\begin{align}
    &\left\{\psi(0.2), \psi(0.2571), \psi(0.3143), \psi(0.3714), \psi(0.4286), \psi(0.4857), \psi(0.5429), \psi(0.6) \right\},\\
    &\left\{\psi(0.4), \psi(0.4111), \psi(0.4222), \psi(0.4333), \psi(0.4444), \psi(0.4556), \psi(0.4667), \psi(0.4778), \psi(0.4889), \psi(0.5)\right\}.
\end{align}
For the first ensemble, the states are defined as $\psi(x)$ (see Eq.~\eqref{eq:psi_alpha_i}) with $x$ sampled uniformly over the interval $[0.2, 0.6]$, yielding eight equally spaced instances. 
The second ensemble is constructed analogously, with $\alpha$ uniformly sampled over $[0.4, 0.5]$ to produce ten states.
Even when the number of copies is increased to nearly $n \approx 50$ in post-processing-only purification (see Fig.~\ref{fig:Global_Num_Point3_4}), no crossover is observed: the single-copy pre-processing-augmented protocol (see Fig.~\ref{fig:Global_Multiple_Copies}(a)) continues to outperform post-processing-only strategies (see Fig.~\ref{fig:Global_Multiple_Copies}(b)) throughout the accessible range. 
Linear fit indicates that any potential crossover would occur far beyond this regime. 
For the first ensemble, linear extrapolation places the crossover at approximately $n \approx 97$, while a log-exponential fit refines this estimate to $n \approx 112$. 
For the second ensemble, the corresponding thresholds are pushed dramatically higher, with linear fit yielding $n \approx 1918$ and the log-exponential fit further increasing the estimate to $n \approx 2431$. 
These results highlight the rapidly escalating sample complexity required for post-processing-only strategies (see Fig.~\ref{fig:Global_Multiple_Copies}(b)) to match the performance of the pre-processing-augmented protocol (see Fig.~\ref{fig:Global_Multiple_Copies}(a)). 
Reaching even the present scale of $n \approx 50$ already relies critically on symmetry-reduced SDPs in Subsec.~\ref{subsec:SDP_Symmetry}; without the Schur-Weyl-based reduction and Clebsch-Gordan recursion, the optimization would become intractable at substantially smaller $n\approx8$.

\begin{figure}[htbp]
    \centering   
    \includegraphics[width=1\textwidth]{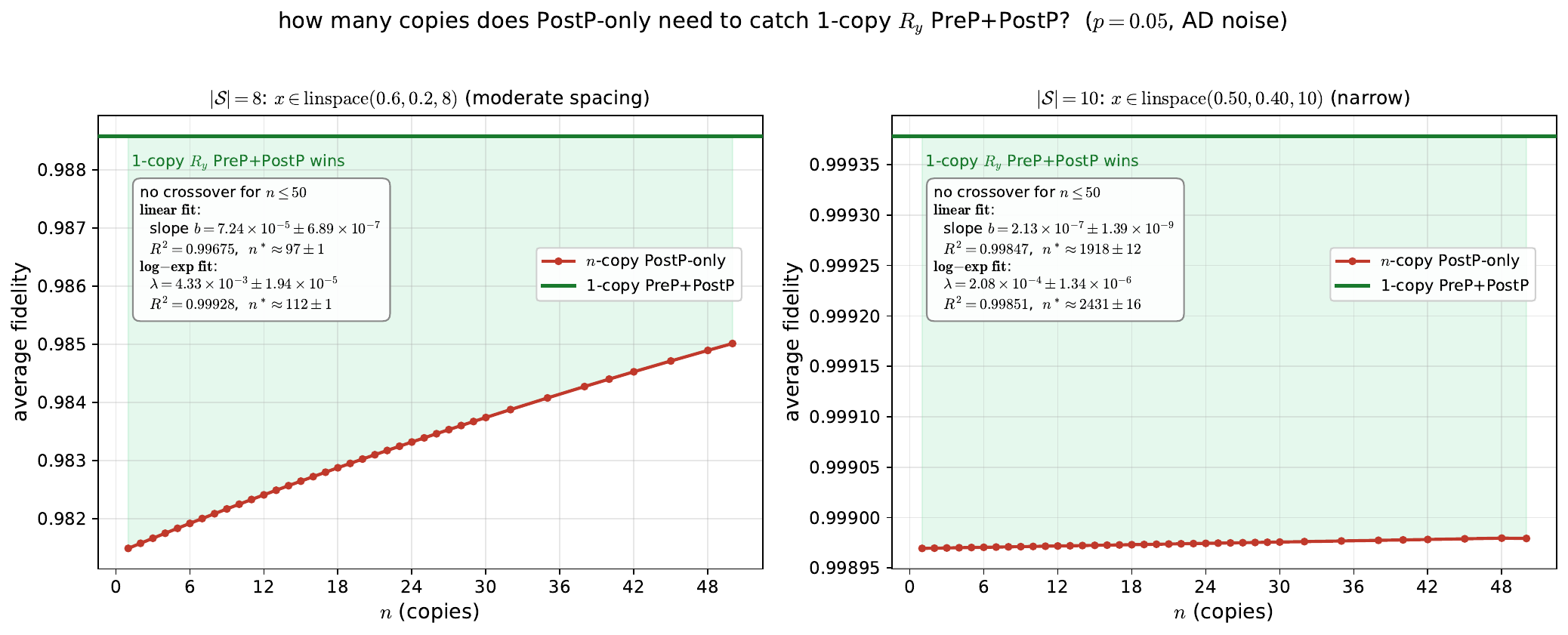}
    \caption{\textbf{Extreme Sample Efficiency in Quantum State Purification}. 
        Performance comparison between single-copy pre-processing–augmented purification (see Fig.~\ref{fig:Global_Multiple_Copies}(a)) and post-processing–only strategies (see Fig.~\ref{fig:Global_Multiple_Copies}(b)) as the number of noisy copies $n$ increases, under amplitude damping noise ($p=0.05$). 
        Results are shown for ensembles with $|\mS|=8$ (left) and $|\mS|=10$ (right), constructed from uniformly spaced parameters. 
        The pre-processing-augmented protocol (green) remains effectively constant, whereas post-processing–only performance (red) improves gradually with $n$, yet fails to reach parity within the explored range $n \leqslant 50$. Extrapolation of the post-processing trend indicates that substantially larger copy numbers are required to match the single-copy performance, with estimates differing markedly between the two ensembles. 
        The shaded region highlights the parameter regime in which pre-processing confers a clear advantage, illustrating a persistent separation in sample efficiency.
    }
    \label{fig:Global_Num_Point3_4}
\end{figure}

It is worth noting that, at the outset of this subsection, we highlighted that a single-copy pre-processing-augmented purification protocol (see Fig.~\ref{fig:Global_Multiple_Copies}(a)) can outperform post-processing-only schemes (see Fig.~\ref{fig:Global_Multiple_Copies}(b)) using as many as 1905 copies, based on a linear extrapolation (see Fig.~\ref{fig:Global_Num_Point3_4}). 
However, the true optimal performance of purification is determined by an optimization over all admissible post-processing quantum channels, formulated as an SDP and therefore intrinsically nonlinear.
As a result, linear extrapolation cannot capture the full structure of the optimization landscape. In realistic settings, it is thus expected that even 1905 copies are insufficient for conventional post-processing purification to attain the performance achieved by the pre-processing-augmented protocol introduced here, further highlighting the advantage conferred by incorporating temporal structure in quantum state purification.

Taken together, these results challenge the standard view that improved purification necessarily relies on access to multiple noisy copies. 
A different picture emerges: performance is determined not by copy number alone, but by how quantum dynamics are structured and exploited. 
By acting prior to noise, pre-processing reshapes the effective channel in a manner that cannot be reproduced by post-processing alone, allowing a single-copy protocol to surpass conventional multi-copy strategies. 
This establishes a clear operational separation between control before and after noise, and identifies pre-processing as a distinct resource for purification. 
The resulting gains in sample efficiency are not merely incremental but qualitative, enabling high-fidelity operation in regimes where preparing and maintaining many copies is prohibitively costly. 
More broadly, these findings point to a shift in how quantum information tasks can be optimized under realistic conditions: rather than compensating for noise after it occurs, one can actively steer its impact, opening new directions for quantum communication and distributed processing in the presence of unavoidable imperfections.


\subsection{Symmetry Collapses Complexity: Schur-Weyl Reduction of SDPs}\label{subsec:SDP_Symmetry}

The semidefinite programming underlying $n$-to-1 purification rapidly encounters a computational wall: its dimension grows exponentially with the number of input copies $n$, rendering even modest instances prohibitively computational-intensive.
This apparent intractability, however, is not intrinsic to the task but reflects a naive formulation that obscures the problem's underlying structure.
The independent and identically distributed (i.i.d.) nature of the inputs endows the problem with a permutation symmetry, and resolving this symmetry compresses the effective optimization to polynomial scale.
In this subsection, we develop a symmetry-resolved formulation of the SDP used in Subsec.~\ref{subsec:Less_is_More}. 
While a direct implementation becomes impractical beyond a few copies ($n\lesssim8$) , the resulting reduction enables efficient optimization deep into the many-copy regime, reaching system sizes exceeding 50 on a 128 GB memory workstation. Here $\lesssim$ means less than or approximately, which is an estimate of the upper bound given the engineering constraint that may vary due to different computational resources available.

We focus on the SDP formulation for $n$-to-1 post-processing purification, as introduced in Eq.~\eqref{eq:PostP_Fundamental_Limit}, which can be cast into the following simplified form.
\begin{align}\label{eq:PostP_Fundamental_Limit_Simplified}
    F_{\mathrm{PostP}}
    \coloneqq
    \max \quad 
    & 
    \Tr[J^{\theta^{\mathrm{Post}}}_{CD}\cdot \left(\frac{1}{|\mS|}\sum_{\psi\in\mS}\left(\mN(\psi)^{\T}\right)^{\otimes n}_{C}\otimes\psi_{D}\right)
    ]
    \\
    \text{s.t.} \quad 
    &J^{\theta^{\mathrm{Post}}}\geqslant0, \Tr_{D}[J^{\theta^{\mathrm{Post}}}]=\1_{C}.
\end{align}
Here, the noisy channel $\mN$ is absorbed into the input state, so that $\mN(\psi)$ acting on system $C$ is treated as an effective input state (see Fig.~\ref{fig:FA_Purification_Simplified}).
The $n$-to-1 post-processing SDP optimizes over a Choi operator $J^{\theta^{\mathrm{Post}}}\in\mathbb{C}^{N\times N}$ with
\begin{align}
    N \coloneqq d^{n+1},
\end{align}
where $d$ denotes the single-copy Hilbert space dimension.
In the absence of symmetry reduction, the problem size grows exponentially with the number of input copies, and the memory cost of standard SDP solvers scales as $\sim N^4$.
This $N^4$ scaling arises because the SDP variable has $\mO(N^2)$ degrees of freedom, and interior-point solvers require storing and factorizing matrices of comparable size, leading to an overall memory cost that scales as $N^4$.
As a result, the computational cost escalates rapidly, rendering even modest instances prohibitively memory-intensive.
For qubits, i.e., $d=2$, $n=7$ already corresponds to $N=256$ and requires $\sim 52\mathrm{GB}$ of memory measured in running; increasing to $n=8$ raises this demand by an additional order of magnitude to $\sim 0.83\mathrm{TB}$ (estimated).
For higher local dimensions, the situation is even more severe: in the distributed setting with $d=4$, $n=4$ yields 
$N=1024$, with a estimated memory requirement of $\sim 13\mathrm{TB}$.
These estimates make clear that, without exploiting additional structure, the SDP formulation is rapidly driven beyond practical computational limits.
This motivates our analysis of the permutation invariance of $(\mN(\psi)^{\T})^{\otimes n}$ as a means to reduce the computational cost of evaluating the SDP optimum $F_{\mathrm{PostP}}$ (see Eq.~\eqref{eq:PostP_Fundamental_Limit_Simplified}), ultimately leading to an equivalent formulation whose effective dimension scales only polynomially with $n$, thereby restoring computational tractability.

General complexity bounds for solving SDPs lay bare the computational bottleneck, highlighting the necessity of structurally simplified formulations in quantum information theory.
Consider an SDP in which the optimization variable is an $N\times N$ matrix subject to $m$ constraints. Achieving a solution with accuracy $\epsilon$ requires computational resources that scale as
\begin{align}
    \mO\left(\mathrm{poly}(N,m)\log(\frac{1}{\epsilon})\right).
\end{align}
Cutting-plane methods (CPMs) operate by maintaining and progressively shrinking a convex region guaranteed to contain the optimal solution. 
At each iteration, a separation oracle identifies a hyperplane that excludes a portion of the current feasible set while preserving the optimal point, and the procedure continues until the region becomes sufficiently small to certify near-optimality.
Interior-point methods (IPMs), by contrast, introduce a barrier function into the objective and solve a sequence of modified optimization problems. 
The resulting iterates follow a well-defined central path and converge to the optimum with high precision.
The computational complexities of these approaches are compared in Tab.~\ref{tab:SDP_Complexity}.

\begin{table}[htbp]
\centering
\begin{tblr}{
  colspec = {c | l || c | c | l},
  row{1} = {bg=gray!50, font=\bfseries}, 
  row{2} = {bg=gray!50, font=\bfseries},
  column{1} = {bg=gray!10, font=\bfseries},
  cell{2}{2} = {bg=gray!50}, 
  column{2} = {bg=white},
  hlines,
  vlines,
  cells = {m, c},
  cell{3,4,5,9,10,11}{3-5} = {bg=mOrange},
  cell{6,7,8,12}{3-5} = {bg=mBlue},
}
  \SetCell[c=5]{c} Computational Complexity of Solving Semidefinite Programming & & & & \\
  Year & References & Method & \#Iters & Cost per Iter \\
  1979 & \cite{Shor1977,KHACHIYAN198053,nemirovskij1983problem} & Cutting-Plane Method & $m^2$ & $mN^2 + m^2 + N^{\omega}$ \\
  1988 & \cite{TarKhaErl88} & Cutting-Plane Method & $m$ & $mN^2 + m^{3.5} + N^{\omega}$ \\
  1989 & \cite{63500} & Cutting-Plane Method & $m$ & $mN^2 + m^{\omega} + N^{\omega}$ \\
  1992 & \cite{Nesterov01011992} & Interior-Point Method & $N^{1/2}$ & $m^2N^2 + mN^{\omega} + m^{\omega}$ \\
  1994 & \cite{nesterov1994interior} & Interior-Point Method & $(mN)^{1/4}$ & $m^4N^2 + m^3N^{\omega}$ \\
  2000 & \cite{10.1287/moor.25.3.365.12212} & Interior-Point Method & $(mN)^{1/4}$ & $m^4N^2 + m^3N^{\omega}$ \\
  2003 & \cite{sivaramakrishnan2007properties} & Cutting-Plane Method & $m$ & $mN^2 + m^{\omega} + N^{\omega}$ \\
  2015 & \cite{7354442} & Cutting-Plane Method & $m$ & $mN^2 + m^2 + N^{\omega}$ \\
  2020 & \cite{10.1145/3357713.3384284} & Cutting-Plane Method & $m$ & $mN^2 + m^2 + N^{\omega}$ \\
  2020 & \cite{9317892} & Interior-Point Method & $N^{1/2}$ & $mN^2 + m^{\omega} + N^{\omega}$ \\
\end{tblr}
\caption{\textbf{Comparison of Computational Costs}.
The size of the SDP variable and the number of constraints are denoted by $N$ and $m$, respectively, while the exponent $\omega$ characterizes the cost of matrix multiplication. 
Accordingly, $m^{\omega}$ captures the cost of inverting the Hessian of the logarithmic barrier function in IPMs, while $N^{\omega}$ represents the cost associated with factorizing the slack matrix.}
\label{tab:SDP_Complexity}
\end{table}

Before proceeding further, we recall the basic representation-theoretic structure underlying permutation symmetry.
The $n$-fold tensor product $\rho^{\otimes n}$ carries a natural action of the symmetric group $\mathfrak{S}_{n}$ via permutations of tensors.
Because each copy in $\rho^{\otimes n}=\rho\otimes\cdots\otimes\rho$ is identical, $\rho^{\otimes n}$ commutes with all permutation operators $P_{\pi}$ with $\pi\in\mathfrak{S}_{n}$.
By Schur–Weyl duality, this induces the decomposition
\begin{align}\label{eq:SW_decomp}
    (\mathbb{C}^{d})^{\otimes n}
    =\bigoplus_{\substack{\lambda \vdash n \\ \ell(\lambda)\leqslant d}}\;
    V_{\lambda}\otimes W_{\lambda},
\end{align}
where $\lambda \vdash n  $ denotes that $\lambda$ is an interger partition  of $n$,  $\lambda$ labels Young diagrams with at most $d$ rows and $n$ boxes, and $V_{\lambda}$ and $W_{\lambda}$ are irreducible representations of $\mathrm{SU}(d)$ and $\mathfrak{S}_{n}$, respectively.
Here $\ell(\lambda)$ denotes the number of rows of the Young diagram $\lambda$, and the condition $\ell(\lambda)\leqslant d$ restricts to diagrams with at most $d$ rows.
For notational simplicity, we will henceforth leave the constraints on $\lambda$ implicit in the direct-sum notation.
By Schur's lemma, any operator commuting with all permutations is block-diagonal in this decomposition and takes the form
\begin{align}
    \bigoplus_{\lambda}\;A_{\lambda}\otimes \1_{W_{\lambda}}.
\end{align}
Here $A_{\lambda}\in\mathbb{C}^{m_{\lambda}\times m_{\lambda}}$, $m_{\lambda}:=\dim V_{\lambda}$, and $d_{\lambda}:=\dim W_{\lambda}$.


To place the analysis in a more general setting, consider the trace overlap between the operator $J$ and the product state $\rho^{\otimes n}\otimes \sigma$.
Here, $\rho^{\otimes n}$ acts on systems $A_1\ldots A_n$, collectively denoted by $A$, while $\sigma$ acts on system $B$; 
accordingly, $J$ acts on the joint system $AB$.
For an arbitrary permutation $P_{\pi}$, with $\pi\in\mathfrak{S}_{n}$, acting on $\rho^{\otimes n}$, one then has
\begin{align}\label{eq:J_SW}
    \Tr[J\cdot\left(\rho^{\otimes n}\otimes \sigma\right)]
    &=
    \Tr[J\cdot\left(P_{\pi}\left(\rho^{\otimes n}\right)P_{\pi}^{\dagger}\otimes \sigma\right)]\\
    &=
    \Tr[(P_{\pi}^{\dagger}\otimes\1_{B})J(P_{\pi}\otimes\1_{B})\cdot\left(\rho^{\otimes n}\otimes \sigma\right)].
\end{align}
In other words, replacing $J$ with $(P_{\pi}^{\dagger}\otimes\1_{B})J(P_{\pi}\otimes\1_{B})$ leaves the overlap invariant.
The same holds upon averaging over all such permutations, which leads to the following symmetrized operator
\begin{align}
    \bar{J}\coloneqq\frac{1}{n!}\sum_{\pi\in\mathfrak{S}_{n}}
    (P_{\pi}^{\dagger}\otimes\1_{B})
    J
    (P_{\pi}\otimes\1_{B}).
\end{align}
One then obtains
\begin{align}
    \Tr[\bar{J}\cdot\left(\rho^{\otimes n}\otimes \sigma\right)]
    =
    \Tr[J\cdot\left(\rho^{\otimes n}\otimes \sigma\right)].
\end{align}

On the other hand, since both $\rho^{\otimes n}$ and $\bar{J}$ commute with all permutation operators, they admit the block-diagonal structure
\begin{align}\label{eq:rho_n_decomp}
    \rho^{\otimes n}=
    \bigoplus_{\lambda}\;\rho_{\lambda}\otimes \1_{W_{\lambda}},
\end{align}
and
\begin{align}\label{eq:J_bar_decomp}
    \bar{J}=
    \bigoplus_{\lambda}\;J_{\lambda}\otimes \1_{W_{\lambda}}.
\end{align}
Although the decompositions in Eqs.~\eqref{eq:rho_n_decomp} and~\eqref{eq:J_bar_decomp} appear similar, they differ in their support.
In Eq.~\eqref{eq:rho_n_decomp}, $\rho_{\lambda}$ acts on the irreducible space $V_{\lambda}$ defined in Eq.~\eqref{eq:SW_decomp}, whereas in Eq.~\eqref{eq:J_bar_decomp}, $J_{\lambda}$ acts on the enlarged space $V_{\lambda}\otimes B$.
Accordingly, $\rho_{\lambda}$ is a $m_{\lambda}\times m_{\lambda}$ matrix, while $J_{\lambda}$ is an $m_{\lambda}d\times m_{\lambda}d$ matrix.
The block-diagonal structure derived above allows the overlap to be re-expressed as
\begin{align}
    \Tr[\bar{J}\cdot\left(\rho^{\otimes n}\otimes \sigma\right)]
    =
    \sum_{\lambda}d_{\lambda}\Tr[J_{\lambda}\cdot(\rho_{\lambda}\otimes\sigma)].
\end{align}

Maximizing the overlap $\Tr[J\cdot(\rho^{\otimes n}\otimes \sigma)]$ under the constraint that $J$ represents the Choi operator of a quantum channel, then leads to the following SDP 
\begin{align}\label{eq:Overlap_Original}
    F_{\ast}\coloneqq\max \quad 
    & 
    \Tr[J\cdot\left(\rho^{\otimes n}\otimes \sigma\right)]
    \\
    \text{s.t.} \quad 
    &J\geqslant0, \Tr_{B}[J]=\1_{A}.
\end{align}
In light of the symmetry analysis above, this optimization is equivalent to the following reduced formulation
\begin{align}\label{eq:Overlap_SW}
    \max \quad 
    & 
    \sum_{\lambda}d_{\lambda}\Tr[J_{\lambda}\cdot(\rho_{\lambda}\otimes\sigma)]
    \\
    \text{s.t.} \quad 
    &J_{\lambda}\geqslant0, \Tr_{B}[J_{\lambda}]=\1_{V_{\lambda}}.
\end{align}
For different values of $\lambda$, the optimization decouples into independent blocks and can therefore be carried out in parallel. Each block is much smaller than the original problem, yielding a substantial computational speedup over the original SDP.
More precisely, the problem reduces to optimizing over each block separately, leading to the following reduced formulation
\begin{align}\label{eq:Overlap_SW_block}
    F_{\lambda}\coloneqq\max \quad 
    & 
    \Tr[J_{\lambda}\cdot(\rho_{\lambda}\otimes\sigma)]
    \\
    \text{s.t.} \quad 
    &J_{\lambda}\geqslant0, \Tr_{B}[J_{\lambda}]=\1_{V_{\lambda}},
\end{align}
which satisfies
\begin{align}\label{eq:Weighted_Sum}
    F_{\ast}=\sum_{\lambda}d_{\lambda}F_{\lambda}.
\end{align}

Equipped with this general result, the $n$-copy post-processing purification setting can be revisited. 
The average fidelity, given in Eq.~\eqref{eq:PostP_Fundamental_Limit_Simplified}, is characterized by the objective function
\begin{align}
    F_{\mathrm{PostP}}
    =
    \max_{\mE} \quad 
    & 
    \Tr[J^{\theta^{\mathrm{Post}}}_{CD}\cdot \left(\frac{1}{|\mS|}\sum_{\psi\in\mS}\left(\mN(\psi)^{\T}\right)^{\otimes n}_{C}\otimes\psi_{D}\right)
    ].
\end{align}
To project the operator 
\begin{align}\label{eq:Xi}
    \Xi\coloneqq\frac{1}{|\mS|}\sum_{\psi\in\mS}\left(\mN(\psi)^{\T}\right)^{\otimes n}_{C}\otimes\psi_{D}
\end{align}
onto the subspace $V_{\lambda}\otimes W_{\lambda}$ (see Eq.~\eqref{eq:SW_decomp}), we construct the following projector
\begin{align}\label{eq:Pi_lambda}
    \Pi_{\lambda}\coloneqq
    \frac{d_{\lambda}}{n!}\sum_{\pi\in \mathfrak{S}_{n}}\chi_{\lambda}(\pi^{-1})\,P_{\pi}.
\end{align}
Here, $\chi_{\lambda}$ denotes the character of the irreducible representation $W_{\lambda}$ of $\mathfrak{S}_{n}$, while each $P_{\pi}$ represents the operator associated with a permutation, acting on $(\mathbb{C}^{d})^{\otimes n}$ and thus taking the form of a $d^n\times d^n$ matrix. 
The parameter $d_{\lambda}$ specifies the dimension of $W_{\lambda}$.
Following this projection, the SDP decomposes into independent optimizations over the blocks of $V_{\lambda}\otimes W_{\lambda}$, enabling a substantial reduction in computational cost. 
Among all irreducible sectors, the largest block has dimension 
\begin{align}
    \left(\max_{\lambda}m_{\lambda}\right)\cdot d=
    \binom{n+d-1}{d-1}\cdot d =
    \mO\left(n^{d-1}\right).
\end{align}
For $d=2$, this yields $\max_{\lambda}m_{\lambda}=n+1$. 
For $d=4$, the corresponding $\max_{\lambda}m_{\lambda}$ becomes $\binom{n+3}{3}=\mO(n^{3})$.
The matrix dimension $N_{\mathrm{SW}}$ therefore becomes 
\begin{align}
    N_{\mathrm{SW}} = \mO\left(n^{d-1}\right) \, .
\end{align}
Compared with the original SDP --- whose matrix dimension scales as $N \times N$ with $N=d^{n+1}$ --- this reduction already yields an exponential improvement in computational tractability.

Upon exploiting Schur–Weyl duality (see Eq.~\eqref{eq:SW_decomp}), the primary computational bottleneck shifts from the SDP itself to the explicit construction of the projection operator $\Pi_{\lambda}$ (see Eq.~\eqref{eq:Pi_lambda}). 
This step incurs factorial overhead --- requiring $\mO(n!\cdot d^{2n})$ in time and $\mO(d^{2n})$ in memory --- which ultimately dominates the overall complexity.
While tractable at moderate system sizes, for example, $d=4$, $n=4$ involves 24 passes over $256\times256$ matrices, the factorial scaling rapidly becomes prohibitive. 
In practice, this confines the approach to $n \lesssim 5$ for $d=4$ and $n \lesssim 8$ for $d=2$.
The distributed PPT purification bounds ($d=4$) investigated in Sec.~\ref{sec:Advantages_DQSP} are therefore limited to $n=4$ copies; 
extending to $n=5$ already lies beyond the computational reach of this construction.

The Schur-Weyl duality developed here admits a concrete algorithmic realization, distilled in Alg.~\ref{alg:char_pipeline}. 
This pipeline takes as input a target ensemble $\mS$ on $\mathbb{C}^{d}$ and returns the corresponding $n$-to-1 purification fidelity $F_{\mathrm{PostP}}$ in Eq.~\eqref{eq:PostP_Fundamental_Limit_Simplified}, for arbitrary local dimension $d\geqslant 2$.
At its core lies the routine \textsc{Char-Reduce}, which performs the block extraction for a single qudit state in three steps: 
(1) construction of the projector $\Pi_{\lambda}$ in Eq.~\eqref{eq:Pi_lambda}; 
(2) application of a Young symmetrizer to isolate a canonical copy of $V_{\lambda}$ within the $\lambda$ sector; 
and (3) a QR orthonormalization that yields an isometry $\Phi_{\lambda}\in\mathbb{C}^{d^{n}\times m_{\lambda}}$ onto this subspace (the notation $\Phi_{\lambda}$ avoids any ambiguity with permutation operators $P_{\pi}$).
Built on this primitive, the outer pipeline implements Steps (i)–(iv) of Subsec.~\ref{subsec:Optimization_Purification_Performance}: (1) block extraction applied to each $\sigma_i^{\T}=\mN(\psi_i)^{\T}$; 
(2) construction of the ensemble-averaged operators $\Xi_{\lambda}$ (see Eq.~\eqref{eq:Xi}) on $V_{\lambda}\otimes D$ (see Eq.~\eqref{eq:SW_decomp}); decomposition into independent semidefinite programming across irreducible sectors $\lambda$; and (3) a final aggregation through multiplicity-weighted summation with weights $d_{\lambda}$ (see Eq.~\eqref{eq:Overlap_SW}). 
Together, these steps translate the representation-theoretic structure into a tractable optimization procedure.

\begin{algorithm}[htbp]
    \DontPrintSemicolon
    \SetKwProg{Fn}{Function}{}{end}
    \KwIn{target-state ensemble $\mS = \{\psi_{i}\}_{i=1}^{|\mS|}$; noise channel $\mN$; copy count $n$; local dimension $d$.}
    \KwOut{optimal post-processing fidelity $F_{\mathrm{PostP}}$, the $n$-to-1 SDP optimum of Eq.~\eqref{eq:PostP_Fundamental_Limit_Simplified}.}
    \BlankLine
    
    \Fn{\normalfont\textsc{Char-Reduce}($\rho,\, n,\, d$)}{
        \tcc*[h]{Input: single-qudit $\rho\in\mathbb{C}^{d\times d}$, copy count $n$, local dimension $d$. Output: reduced blocks $\bigl\{R_{\lambda}^{(n)}\bigr\}_{\lambda}$ over partitions $\lambda\vdash n$ with $\ell(\lambda)\leqslant d$, where $R_{\lambda}^{(n)} \coloneqq \Phi_{\lambda}^{\dagger}\rho^{\otimes n}\Phi_{\lambda}\in\mathbb{C}^{m_{\lambda}\times m_{\lambda}}$ ($m_{\lambda} = \dim V_{\lambda}$) and $\Phi_{\lambda}\in\mathbb{C}^{d^{n}\times m_{\lambda}}$ is an isometry onto one fixed copy of $V_{\lambda}$ (see Eq.~\eqref{eq:SW_decomp}) inside $(\mathbb{C}^{d})^{\otimes n}$.}\;
        \ForEach{partition $\lambda \vdash n$ with $\ell(\lambda)\leqslant d$}{
            $d_{\lambda} \leftarrow \dfrac{n!}{\prod_{(i, k)\in\lambda} h_{\lambda}(i, k)}$ \tcp*{hook-length formula; $d_{\lambda} = \dim W_{\lambda}$, the number of orthogonal copies of $V_{\lambda}$ inside $(\mathbb{C}^{d})^{\otimes n}$. For $d = 2$, $\lambda = (\tfrac{n}{2}+j, \tfrac{n}{2}-j)$ recovers Eq.~\eqref{eq:dj_formula}}
            $\Pi_{\lambda} \leftarrow \dfrac{d_{\lambda}}{n!}\displaystyle\sum_{\pi\in\mathfrak{S}_{n}} \chi_{\lambda}(\pi^{-1})\,P_{\pi}$ \tcp*{isotypic projector onto $V_{\lambda}\otimes W_{\lambda}$; $P_{\pi}$ is the permutation unitary on $(\mathbb{C}^{d})^{\otimes n}$, $\chi_{\lambda}$ is the character of the irreducible $\mathfrak{S}_{n}$-representation on $W_{\lambda}$; Eq.~\eqref{eq:Pi_lambda}. For $\mathfrak{S}_{n}$, every element is conjugate to its inverse and all characters are real, so $\chi_{\lambda}(\pi^{-1}) = \chi_{\lambda}(\pi)$ numerically}
            diagonalize $\Pi_{\lambda}$; stack the $m_{\lambda} d_{\lambda}$ unit-eigenvalue eigenvectors as the columns of $Q_{\lambda}\in\mathbb{C}^{d^{n}\times m_{\lambda} d_{\lambda}}$ \tcp*{$\operatorname{image}(Q_{\lambda}) = V_{\lambda}\otimes W_{\lambda}$; Hermitian eigensolvers (\texttt{eigh}) return orthonormal eigenvectors, so $Q_{\lambda}^{\dagger}Q_{\lambda} = \1_{m_{\lambda} d_{\lambda}}$ automatically, and no Gram--Schmidt step is needed}
            fix a standard Young tableau $T$ of shape $\lambda$ and form the Young symmetrizer $Y_{T} \coloneqq a_{T}\,b_{T}$, with row symmetrizer $a_{T} \coloneqq \sum_{\pi\in R(T)} P_{\pi}$ and column antisymmetrizer $b_{T} \coloneqq \sum_{\pi\in C(T)} \operatorname{sgn}(\pi)\,P_{\pi}$ \tcp*{$R(T), C(T)\subseteq\mathfrak{S}_{n}$ are the row- and column-stabilizer subgroups of the tableau $T$}
            $\Phi_{\lambda} \leftarrow$ orthonormal basis for $\operatorname{range}(Y_{T}\,Q_{\lambda})$, obtained by rank-revealing QR or SVD truncated to rank $m_{\lambda}$ \tcp*{$Y_{T}$ is a quasi-idempotent satisfying $Y_{T}^{2} = (n!/d_{\lambda})\,Y_{T}$; applied to $Q_{\lambda}$ it projects the $\lambda$-isotypic sector onto one fixed copy of $V_{\lambda}$, so $Y_{T}\,Q_{\lambda}\in\mathbb{C}^{d^{n}\times m_{\lambda}d_{\lambda}}$ has rank exactly $m_{\lambda}$ and its $m_{\lambda}d_{\lambda} - m_{\lambda}$ spurious columns must be discarded. The resulting $\Phi_{\lambda}\in\mathbb{C}^{d^{n}\times m_{\lambda}}$ satisfies $\Phi_{\lambda}^{\dagger}\Phi_{\lambda} = \1_{m_{\lambda}}$. The choice of $T$ is immaterial: different standard tableaux select different copies of $V_{\lambda}$ (see Eq.~\eqref{eq:SW_decomp}), and Schur's lemma forces the downstream $R_{\lambda}^{(n)}$ to coincide}
            $R_{\lambda}^{(n)} \leftarrow \Phi_{\lambda}^{\dagger}\,\rho^{\otimes n}\,\Phi_{\lambda}$\;
        }
        \KwRet $\bigl\{R_{\lambda}^{(n)}\bigr\}_{\lambda}$\;
    }
    
    \BlankLine
    \tcc*[h]{\textbf{Step (i)--(ii).} Per-state noisy reduced blocks. The transpose $\sigma_{i}^{\T}$ enters because the SDP objective of Eq.~\eqref{eq:PostP_Fundamental_Limit_Simplified} carries $\bigl(\mN(\psi_{i})^{\T}\bigr)^{\otimes n}$.}\;
    \ForEach{$\psi_{i}\in\mS$}{
        $\sigma_{i} \leftarrow \mN(\psi_{i})$ \tcp*{noisy single-copy state, $\sigma_{i}\in\mathbb{C}^{d\times d}$}
        $\bigl\{\rho_{\lambda, i}^{(n)}\bigr\}_{\lambda} \leftarrow \textsc{Char-Reduce}\bigl(\sigma_{i}^{\T},\, n,\, d\bigr)$ \tcp*{$\rho_{\lambda, i}^{(n)}$: $\lambda$-reduced block of $(\sigma_{i}^{\T})^{\otimes n}$, size $m_{\lambda}\times m_{\lambda}$}
    }
    
    \BlankLine
    \tcc*[h]{\textbf{Step (iii).} Assemble the ensemble-averaged operator on $V_{\lambda}\otimes D$, where $D$ is the $d$-dimensional target system. The trivial action of $\rho^{\otimes n}$ on $W_{\lambda}$ (Eq.~\eqref{eq:rho_n_decomp}) drops $W_{\lambda}$ from the SDP variables entirely. Note: $(\sigma_{i}^{\T})^{\otimes n}$ and hence $\Omega_{\lambda}$ are generally not Hermitian; since $J_{\lambda}$ is constrained to be Hermitian, the objective $\Tr[J_{\lambda}\Omega_{\lambda}]$ depends only on the Hermitian part $(\Omega_{\lambda} + \Omega_{\lambda}^{\dagger})/2$, which standard SDP solvers enforce automatically.}\;
    \ForEach{partition $\lambda\vdash n$ with $\ell(\lambda)\leqslant d$}{
        $\Xi_{\lambda} \leftarrow \dfrac{1}{|\mS|}\displaystyle\sum_{i=1}^{|\mS|} \rho_{\lambda, i}^{(n)} \otimes \psi_{i, D}$ \tcp*{$\psi_{i, D}\equiv\ket{\psi_{i}}\!\bra{\psi_{i}}$ on $D$; $\Xi_{\lambda}$ is a $d\,m_{\lambda}\times d\,m_{\lambda}$ operator on $V_{\lambda}\otimes D$}
    }
    
    \BlankLine
    \tcc*[h]{\textbf{Step (iv).} The bundled SDP $\max_{\{J_{\lambda}\}}\sum_{\lambda} d_{\lambda}\,\Tr[J_{\lambda}\Xi_{\lambda}]$ subject to $J_{\lambda}\geqslant 0$ and $\Tr_{D}[J_{\lambda}] = \1_{V_{\lambda}}$ carries no cross-$\lambda$ constraints (Eq.~\eqref{eq:Overlap_SW_block}); it therefore separates into independent block SDPs, one per $\lambda$, and additively gives $F_{\mathrm{PostP}} = \sum_{\lambda} d_{\lambda}\,F_{\lambda}$.}\;
    \ForEach{$\lambda\vdash n$ with $\ell(\lambda)\leqslant d$ \textup{(independent; parallelisable)}}{
        $F_{\lambda} \leftarrow \displaystyle\max_{\substack{J_{\lambda}\geqslant 0 \\ \Tr_{D}[J_{\lambda}] = \1_{V_{\lambda}}}} \Tr[\,J_{\lambda}\,\Xi_{\lambda}\,]$ \tcp*{$J_{\lambda}$: $\lambda$-block of the post-processing Choi operator on $V_{\lambda}\otimes D$; the TP condition descends from $\Tr_{D}\bigl[J^{\theta^{\mathrm{Post}}}\bigr] = \1_{C}$ restricted to the $\lambda$-sector}
    }
    \KwRet $F_{\mathrm{PostP}} \leftarrow \displaystyle\sum_{\lambda} d_{\lambda}\,F_{\lambda}$ \tcp*{each sector contributes with its multiplicity $d_{\lambda}$ (Eq.~\eqref{eq:Overlap_SW_block})}
    \caption{\textsc{Schur-Weyl-Reduction}: End-to-end character-formula pipeline for evaluating the $n$-to-1 qudit purification SDP in Eq.~\eqref{eq:PostP_Fundamental_Limit_Simplified}, applicable to arbitrary local dimension $d\geqslant 2$. 
    The construction scales as $\mO(n!\,d^{2n})$ in time and $\mO(d^{2n})$ in memory per input state; 
    while the block-structured SDP decomposes into independent sectors of dimension $d\,m_{\lambda} = \mO(d\,n^{d-1})$.}
    \label{alg:char_pipeline}
\end{algorithm}


\subsection{Symmetry Collapses Complexity: Clebsch-Gordan Recursion of SDPs for qubits}\label{subsec:SDP_CG}

This subsection completes the symmetry-reduction framework in Subsec.~\ref{subsec:SDP_Symmetry} for the special case of $d=2$ (qubits) by introducing a second, complementary ingredient: Clebsch-Gordan recursion. 
While parts of the discussion inevitably overlap with the preceding Subsec.~\ref{subsec:SDP_Symmetry}, the perspective adopted here is different in emphasis. 
Rather than further exploiting permutation symmetry at the level of Schur-Weyl block decomposition, we focus on the physical structure underlying these blocks --- specifically, how irreducible sectors arise from the recursive coupling of subsystems. 
This viewpoint reveals that the Schur-Weyl sectors for $d=2$, namely $V_{\lambda}\otimes W_{\lambda}$ (see Eq.~\eqref{eq:SW_decomp}), are not merely abstract symmetry labels, but encode a hierarchical organization governed by angular-momentum addition and conserved quantities.
Building on this insight, we show that the same structure that underpins the block decomposition naturally gives rise to a recursive construction of the relevant subspaces, thereby bypassing the factorial overhead associated with explicit projector construction, i.e., $\Pi_{\lambda}$ (see Eq.~\eqref{eq:Pi_lambda}). 
In this way, Clebsch-Gordan recursion is not introduced as a separate technical tool, but emerges directly from the physical organization of the problem, providing an efficient and scalable route to implementing the symmetry-reduced SDP in the many-copy regime.

Rather than proceeding directly to the Clebsch-Gordan recursion, it is instructive to first clarify how the algebraic structure of Schur-Weyl duality maps onto its physical content. 
This correspondence provides the conceptual footing on which the recursion naturally emerges. 
Before introducing the formal machinery, we therefore adopt a viewpoint familiar to any reader of quantum mechanics: the classification of many-body states by symmetry, ranging from bosonic and fermionic sectors to the continuum of intermediate symmetry types.

A simple convention fixes the language used throughout. 
The purification SDP (see Eq.~\eqref{eq:FA_Fundamental_Limit}) is defined on $n$ qubits, that is, $n$ copies of the two-dimensional Hilbert space $\mathbb{C}^{2}$. 
It is convenient to regard each qubit as a spin-1/2
particle, with the identifications
\begin{align}
    \ket{0} \equiv \ket{\uparrow},
    \quad\text{and}\quad
    \ket{1} \equiv \ket{\downarrow}.
\end{align}
In this picture, a computational basis string $\ket{b_{1}\, b_{2}\, \cdots\, b_{n}}$ represents a configuration of $n$ spins with definite $\hat{S}_{z}$ eigenvalues assigned site by site. 
We will use the two vocabularies --- qubits (i.e., $\ket{0}$/$\ket{1}$) and spin-1/2 particles (i.e., $\ket{\uparrow}$/$\ket{\downarrow}$) --- interchangeably, selecting whichever is more natural in context.

In the quantum mechanics of $n$ identical particles, the two permutation-symmetry classes encountered most often are bosons, whose wavefunction is fully symmetric under particle exchange, and fermions, whose wavefunction is fully antisymmetric. 
In fact they represent the two endpoints of a much richer spectrum of symmetry types.
The intermediate structures already appear in the nontrivial case of $n=3$ spin-1/2 particles. Standard angular-momentum coupling decomposes the 
$2^3$-dimensional Hilbert space into a spin-3/2 quartet (one copy of a 4-dimensional multiplet) and a spin-1/2 sector consisting of two independent copies of a 2-dimensional doublet. 
The spin-3/2 sector is fully symmetric under particle exchange, analogous to the bosonic case in many-body quantum mechanics (the qubits here are distinguishable resources rather than identical bosons, but the symmetry structure under the action of $\mathfrak{S}_{n}$ is the same): 
its highest-weight state
\begin{align}
    \ket{\uparrow\uparrow\uparrow}
\end{align}
is invariant under any permutation of the three particles. 
By contrast, the spin-1/2 sector carries mixed symmetry. 
It is neither fully symmetric nor fully antisymmetric, but instead exhibits partial symmetry under some exchanges and antisymmetry under others.
A concrete example is obtained by first placing particles 1 and 2 in a singlet state, which is antisymmetric under their exchange, and then coupling to particle 3
\begin{align}
    \frac{1}{\sqrt{2}}\left(\ket{{\uparrow}{\downarrow}} - \ket{{\downarrow}{\uparrow}}\right) \otimes \ket{{\uparrow}}.
\end{align}
This state is antisymmetric under the exchange of particles 1 and 2, yet has no definite transformation property under exchange of the pair (1,2) with particle 3, illustrating the characteristic structure of mixed-symmetry sectors.

Before turning to the graphical picture, we make precise the two operations it encodes, using the case of spin-1/2 particles as a concrete illustration.
Given a subset of $n$ particles, symmetrization is implemented by the projector
\begin{align}
    \mathrm{Sym}_{k} \;\coloneqq\; \frac{1}{n!}\sum_{\pi \in \mathfrak{S}_{n}} P_{\pi},
\end{align}
where $P_{\pi}$ permutes the chosen particles within the tensor product. 
For $n=2$, this reduces to 
\begin{align}
    \mathrm{Sym}_{2} = \tfrac{1}{2}(\mathbb{1} + P_{(12)}).
\end{align}
Acting on an asymmetric state such as $\ket{{\uparrow}{\downarrow}}$, one obtains
\begin{align}
  \mathrm{Sym}_{2}\,\ket{{\uparrow}{\downarrow}}
  \;=\; \frac{1}{2}\left(\ket{{\uparrow}{\downarrow}} + P_{(12)}\ket{{\uparrow}{\downarrow}}\right)
  \;=\; \frac{1}{2}\left(\ket{{\uparrow}{\downarrow}} + \ket{{\downarrow}{\uparrow}}\right),
\end{align}
which is invariant under exchange: applying $P_{(12)}$ leaves the state unchanged. 
This invariance under all permutations within the subset characterises the symmetric subspace.

Antisymmetrization is defined analogously, with each permutation weighted by its sign,
\begin{align}
    \mathrm{AntiSym}_{n} \;\coloneqq\; \frac{1}{n!}\sum_{\pi \in \mathfrak{S}_{n}} \mathrm{sgn}(\pi)\, P_{\pi},
\end{align}
For $n=2$, one has 
\begin{align}
    \mathrm{AntiSym}_{2} = \frac{1}{2}(\mathbb{1} - P_{(12)}),
\end{align}
yielding
\begin{align}
  \mathrm{AntiSym}_{2}\,\ket{{\uparrow}{\downarrow}}
  \;=\; \frac{1}{2}\left(\ket{{\uparrow}{\downarrow}} - \ket{{\downarrow}{\uparrow}}\right)
\end{align}
the familiar singlet state. 
Under exchange, this state acquires a minus sign, i.e., 
\begin{align}
    P_{(12)}\,\mathrm{AntiSym}_{2}\ket{{\uparrow}{\downarrow}} = -\,\mathrm{AntiSym}_{2}\ket{{\uparrow}{\downarrow}},
\end{align}
reflecting its antisymmetric character.
A direct consequence appears when two particles occupy the same single-particle state. 
For instance,
\begin{align}
    \mathrm{AntiSym}_{2}\,\ket{{\uparrow}{\uparrow}}
  \;=\; \frac{1}{2}\left(\ket{{\uparrow}{\uparrow}} - P_{(12)}\ket{{\uparrow}{\uparrow}}\right)
  \;=\; \frac{1}{2}\left(\ket{{\uparrow}{\uparrow}} - \ket{{\uparrow}{\uparrow}}\right)
  \;=\; 0,
\end{align}
since the swap acts trivially and the two contributions cancel. 
The antisymmetric projection thus annihilates any configuration with repeated single-particle states. 
This provides a structural origin of Pauli exclusion: the corresponding state is forced to vanish, and hence lies outside the physical Hilbert space. 
The same cancellation mechanism persists for general $n$ whenever two particles in the subset share the same single-particle state.

With symmetrization and antisymmetrization in hand, the recipe ``symmetrise within a group, antisymmetrise between groups'' admits a compact visual encoding. 
Heuristically, boxes in the same row are symmetrised while those in the same column are antisymmetrised; strictly, Young's construction composes these two operations (symmetrise along rows, then antisymmetrise down columns) to produce a projector onto the irreducible component labelled by the diagram, so the two symmetries do not hold simultaneously. The heuristic picture is sufficient for everything that follows.
This prescription is captured by Young diagrams.
For $n=3$ and $d=2$, the admissible arrangements are
\begin{align}
  \ydiagram{3}\;,
  \qquad
  \ydiagram{2,1}\;,
  \qquad
  \ydiagram{1,1,1}\;.
\end{align}
We label each diagram by its row-length sequence, writing $\lambda=(3), (2,1)$ and $(1,1,1)$, respectively. 
The pictorial and algebraic notations are interchangeable, and we use whichever is more convenient in context.

These diagrams encode the symmetry structure of the decomposition. 
The fully symmetric diagram
\begin{align}
    \ydiagram{3}\;,
\end{align}
corresponds to the spin-3/2 sector. 
The mixed diagram
\begin{align}
    \ydiagram{2,1}\;,
\end{align}
represents the spin-1/2 sector: two particles are symmetrized along the top row, while antisymmetrization along the column enforces partial exchange structure.
By contrast, the fully antisymmetric diagram
\begin{align}
    \ydiagram{1,1,1}\;,
\end{align}
is forbidden for qubits, as it would require antisymmetrizing three particles within a two-dimensional single-particle space. 
This is the Pauli exclusion in diagrammatic form: no column may exceed height $d$.

Each identification admits a direct physical verification. 
Consider first $\lambda=(3)$. 
To make the angular-momentum structure explicit, we write the single-particle spin operators as
\begin{align}
\hat{S}_{x}^{(k)} = \frac{1}{2}\sigma_{x}^{(k)},\qquad
\hat{S}_{y}^{(k)} = \frac{1}{2}\sigma_{y}^{(k)},\qquad
\hat{S}_{z}^{(k)} = \frac{1}{2}\sigma_{z}^{(k)},
\end{align}
with $\sigma_{\alpha}^{(k)}$ acting on the $k$-th spin. 
The total angular momentum operators are their sums,
\begin{align}
    \hat{J}_{\alpha} \;\coloneqq\; \sum_{k=1}^{3} \hat{S}_{\alpha}^{(k)},
    \quad \alpha \in \{x, y, z\}.
\end{align}
We reserve $J^{\theta}$ (without hat, introduced in Eq.~\eqref{eq:Choi}) for Choi operators and $\hat{J}_{\alpha}$ (with hat) for angular-momentum operators throughout this section.

Acting on the fully aligned state yields
\begin{align}
  \hat{J}_{z}\ket{{\uparrow}{\uparrow}{\uparrow}}
  \;=\; \left(\frac{1}{2} + \frac{1}{2} + \frac{1}{2}\right)\ket{{\uparrow}{\uparrow}{\uparrow}}
  \;=\; \frac{3}{2}\ket{{\uparrow}{\uparrow}{\uparrow}},
\end{align}
the maximal eigenvalue for three spin-1/2 particles. 
The state $\ket{{\uparrow}{\uparrow}{\uparrow}}$ is invariant under all permutations, and therefore resides in the fully symmetric sector, confirming the identification with the spin-3/2 irrep.

Several standard notions used throughout are fixed here for clarity, namely the raising and lowering operators and the structure of spin multiplets.
The raising and lowering operators are defined as
\begin{align}
    \hat{J}_{+} \;\coloneqq\; \hat{J}_{x} + i\,\hat{J}_{y} \;=\; \sum_{k=1}^{3} \hat{S}_{+}^{(k)},
\end{align}
and
\begin{align}
    \hat{J}_{-} \;\coloneqq\; \hat{J}_{x} - i\,\hat{J}_{y} \;=\; \sum_{k=1}^{3} \hat{S}_{-}^{(k)},
\end{align}
where the single-particle operators act as
\begin{align}
    \hat{S}_{+}^{(k)}\ket{{\downarrow}}_{k} = \ket{{\uparrow}}_{k}, 
    \quad
    \hat{S}_{+}^{(k)}\ket{{\uparrow}}_{k} = 0,
\end{align}
and
\begin{align}
    \hat{S}_{-}^{(k)}\ket{{\uparrow}}_{k} = \ket{{\downarrow}}_{k},
    \quad
    \hat{S}_{-}^{(k)}\ket{{\downarrow}}_{k} = 0.
\end{align}
The commutation relations $[\hat{J}_{z}, \hat{J}_{\pm}] = \pm\,\hat{J}_{\pm}$ ensure that $\hat{J}_{+}$ and $\hat{J}_{-}$ act as ladder operators, shifting the $\hat{J}_{z}$ eigenvalue by $\pm1$, with annihilation at the boundaries of the spectrum.
A spin-$j$ multiplet is a set of $2j + 1$ orthogonal quantum states that
(\romannumeral1) share a common $\hat{J}^{2}$ eigenvalue $j(j+1)$ (total spin squared),
(\romannumeral2) carry the $2j + 1$ distinct $\hat{J}_{z}$ eigenvalues $\{-j, -j+1, \ldots, +j\}$, with exactly one state per eigenvalue, and
(\romannumeral3) are joined into a single ladder by $\hat{J}_{\pm}$: starting from any state of the set, repeated application of $\hat{J}_{+}$ climbs the ladder until it annihilates the top state, called the highest-weight vector, and repeated application of $\hat{J}_{-}$ descends until it annihilates the bottom.
A familiar example is the hydrogen $\ell$-shell, which realizes the $\mathrm{SU}(2)$ representation of angular momentum $\ell$: its $2\ell + 1$ states, labelled by the magnetic quantum number $m_{\ell}$, form a single multiplet whose members are ladder-connected by $\hat{L}_{\pm}$ and mapped into one another by rotations of the reference frame.
A notational remark will be useful in what follows. 
The dimension $2j + 1$ of a spin-$j$ representation is denoted $m_{\lambda}$, where $\lambda$ is the Young diagram associated with $j$; this is the size of a single multiplet, and it agrees with the definition $m_{\lambda} \coloneqq \dim V_{\lambda}$ introduced after Eq.~\eqref{eq:SW_decomp} since each irrep $V_{\lambda}$ carries one copy of the multiplet. 
It is to be distinguished from the multiplicity $d_{\lambda}$ introduced shortly below, which counts the number of independent copies of that multiplet inside $(\mathbb{C}^{d})^{\otimes n}$, and from the magnetic quantum number $m_{\ell}$ (a $\hat{J}_{z}$ eigenvalue, indexed by angular momentum label $\ell$), with which it merely shares the letter $m$.

With these notions in place, we return to the three-particle example. 
The state $\ket{{\uparrow}{\uparrow}{\uparrow}}$ carries the maximal eigenvalue $\hat{J}_{z} = 3/2$ and is annihilated by $\hat{J}_{+}$, as no spin can be raised further. 
It therefore serves as the highest-weight vector of a spin-3/2 multiplet. 
Successive applications of $\hat{J}_{-}$ generate the remaining states, leading to
\begin{align}\label{eq:Young_Diagram_3}
  \left\{\ket{{\uparrow}{\uparrow}{\uparrow}},
  \frac{1}{\sqrt{3}}\left(\ket{{\downarrow}{\uparrow}{\uparrow}}+\ket{{\uparrow}{\downarrow}{\uparrow}}+\ket{{\uparrow}{\uparrow}{\downarrow}}\right),
  \frac{1}{\sqrt{3}}\left(\ket{{\uparrow}{\downarrow}{\downarrow}}+\ket{{\downarrow}{\uparrow}{\downarrow}}+\ket{{\downarrow}{\downarrow}{\uparrow}}\right),
  \ket{{\downarrow}{\downarrow}{\downarrow}}
  \right\}.
\end{align}
Each of these states is fully symmetric under permutations of the three particles, and hence lies in the symmetric subspace. 
The resulting multiplet has dimension
\begin{align}
    2j+1=2\cdot\frac{3}{2}+1=4,
\end{align}
in agreement with $m_{(3)}=4$ for the single-row Young diagram $\lambda=(3)$.

We now turn to the Young diagram $\lambda=(2,1)$ and establish that it corresponds to total spin $j=1/2$.
To this end, consider a representative state of this symmetry type, shown below.
\begin{align}
    \ket{\psi} \;=\; \frac{1}{\sqrt{2}}\left(\ket{{\uparrow}{\downarrow}} - \ket{{\downarrow}{\uparrow}}\right) \otimes \ket{{\uparrow}},
\end{align}
in which particles 1 and 2 form a singlet and particle 3 is spin-up. 
The state $\psi$ thus lies in the $\lambda=(2,1)$ sector: it is antisymmetric under exchange of particles 1 and 2, yet neither fully antisymmetric (which is forbidden for three qubits) nor fully symmetric, and therefore cannot belong to the $\lambda=(3)$ sector (see Eq.~\eqref{eq:Young_Diagram_3}).
Its total spin can be read off by evaluating $\hat{J}^{2}\ket{\psi}$.
The calculation follows from the elementary identity
\begin{align}
    (\vec{a} + \vec{b} + \vec{c})^{2} = \vec{a}^{\,2} + \vec{b}^{\,2} + \vec{c}^{\,2} + 2(\vec{a}\!\cdot\!\vec{b} + \vec{a}\!\cdot\!\vec{c} + \vec{b}\!\cdot\!\vec{c}),
\end{align}
applied to the three spin operators
\begin{align}
  \hat{J}^{2} \;=\; \Bigl(\hat{S}^{(1)} + \hat{S}^{(2)} + \hat{S}^{(3)}\Bigr)^{\!2}
  \;=\; \sum_{k=1}^{3}\hat{S}^{(k)\,2} \;+\; 2\bigl(\hat{S}^{(1)}\!\cdot\hat{S}^{(2)} + \hat{S}^{(1)}\!\cdot\hat{S}^{(3)} + \hat{S}^{(2)}\!\cdot\hat{S}^{(3)}\bigr).
\end{align}
Here, 
\begin{align}
    \hat{S}^{(k)} = (\hat{S}_{x}^{(k)}, \hat{S}_{y}^{(k)}, \hat{S}_{z}^{(k)})
\end{align}
is the three-component spin vector on particle $k$.
Its squared magnitude
\begin{align}
    \hat{S}^{(k)\,2} = (\hat{S}_{x}^{(k)})^{2} + (\hat{S}_{y}^{(k)})^{2} + (\hat{S}_{z}^{(k)})^{2}
\end{align}
measures the intrinsic spin carried by that particle.
The cross terms 
\begin{align}
    \hat{S}^{(i)}\!\cdot\hat{S}^{(j)} = \hat{S}_{x}^{(i)} \hat{S}_{x}^{(j)} + \hat{S}_{y}^{(i)} \hat{S}_{y}^{(j)} + \hat{S}_{z}^{(i)} \hat{S}_{z}^{(j)}
\end{align}
measure the relative orientation of two spins: positive for parallel alignment, negative for antiparallel.
The usefulness of this expansion lies in the fact that each contribution can be evaluated immediately on $\psi$, avoiding a brute-force computation in the full $8 \times 8$ space. 
The three types of terms act as follows.

First, for each particle, the operator $\hat{S}^{(k)\,2}$ is the Casimir operator: 
a polynomial in the spin components that commutes with all of them and therefore takes a fixed value $j(j+1)$ on every state within a given spin-$j$ representation.
For spin-1/2, this gives
\begin{align}
    \hat{S}^{(k)\,2} = 
    \frac{1}{2}(\frac{1}{2}+1)\1 = 
    \frac{3}{4}\1
\end{align}
and hence
\begin{align}
    \sum_{k=1}^{3}\hat{S}^{(k)\,2}\ket{\psi} \;=\; 3 \cdot \tfrac{3}{4}\ket{\psi} \;=\; \tfrac{9}{4}\ket{\psi}.
\end{align}

Second, the pair 1 and 2 forms a singlet, which is an eigenstate of $(\hat{S}^{(1)}+\hat{S}^{(2)})^{2}$ with eigenvalue zero. 
Using
\begin{align}
    \hat{S}^{(1)}\!\cdot\hat{S}^{(2)} = \frac{1}{2}\left[(\hat{S}^{(1)}+\hat{S}^{(2)})^{2} - \hat{S}^{(1)\,2} - \hat{S}^{(2)\,2}\right],
\end{align}
one finds
\begin{align}
    \hat{S}^{(1)}\!\cdot\hat{S}^{(2)}\,\ket{\psi} \;=\; \frac{1}{2}\left(0 - \frac{3}{4} - \frac{3}{4}\right)\ket{\psi}
    \;=\; -\frac{3}{4}\ket{\psi}.
\end{align}
Finally, the remaining cross terms combine as
\begin{align}
    \hat{S}^{(1)}\!\cdot\hat{S}^{(3)} + \hat{S}^{(2)}\!\cdot\hat{S}^{(3)}
    \;=\; \left(\hat{S}^{(1)} + \hat{S}^{(2)}\right)\!\cdot\hat{S}^{(3)}.
\end{align}
Since the singlet is annihilated by $\hat{S}^{(1)}+\hat{S}^{(2)}$, this contribution vanishes on $\psi$.
Collecting the terms,
\begin{align}
    \hat{J}^{2}\ket{\psi}
  \;=\; \left[\frac{9}{4} \;+\; 2\cdot\left(-\frac{3}{4}\right) \;+\; 2\cdot 0\right]\ket{\psi}
  \;=\; \left(\frac{9}{4} - \frac{3}{2}\right)\ket{\psi}
  \;=\; \frac{3}{4}\ket{\psi}
  \;=\; \frac{1}{2}\left(\frac{1}{2}+1\right)\ket{\psi}.
\end{align}
Thus $j(j+1)=3/4$, giving $j=1/2$. 
The $\lambda=(2,1)$ sector therefore contains spin-1/2 multiplets. 
The question of how many independent copies occur will be addressed after resolving the forbidden case $\lambda=(1,1,1)$ in the qubit setting.

Third, we explain why the Young diagram $\lambda=(1,1,1)$ is absent in the qubit case ($d=2$). 
Equivalently, we show that the projector $\mathrm{AntiSym}_{3}$ annihilates every computational basis vector on $(\mathbb{C}^{2})^{\otimes 3}$, so that the fully antisymmetric sector has zero dimension for $d=2$.
Consider an arbitrary computational basis vector $\ket{i_{1}\, i_{2}\, i_{3}}$ with $i_{k} \in \{0, 1\}$.
Since three indices are drawn from a two-element set, at least two must coincide (pigeonhole principle); without loss of generality, take $i_1=i_2$.
Partition the six permutations in $\mathfrak{S}_{3}$ into three pairs
\begin{align}
    \left\{\pi, \pi \circ (12)\right\},
\end{align}
each related by composition with the transposition $(1 2)$.
Within each pair, the two permutations have opposite parity
\begin{align}
    \mathrm{sgn}(\pi \circ (12)) = -\mathrm{sgn}(\pi),
\end{align}
yet act identically on $\ket{i_{1}\, i_{1}\, i_{3}}$, since the transposition $(1 2)$ leaves the state unchanged.
The contributions therefore, cancel pairwise in the signed sum, giving
\begin{align}
    \mathrm{AntiSym}_{3}\,\ket{i_{1}\, i_{1}\, i_{3}}
  \;=\; \frac{1}{6}\sum_{\pi \in \mathfrak{S}_{3}} \mathrm{sgn}(\pi)\, P_{\pi}\ket{i_{1}\, i_{1}\, i_{3}}
  \;=\; 0.
\end{align}
Because the computational basis spans $(\mathbb{C}^{2})^{\otimes 3}$, it follows that
\begin{align}
    \mathrm{AntiSym}_{3}=0.
\end{align}
The fully antisymmetric subspace is therefore trivial. 
In diagrammatic terms, this is precisely the Pauli: no column may exceed height $d$.

Finally, we turn to the analysis of the multiplicity associated with each sector labelled by a Young diagram $\lambda$. 
Rather than beginning with a formal count, it is more instructive to first ask what, precisely, is being counted. 
The answer turns out to depend qualitatively on the symmetry of the sector. 
This distinction provides the physical basis for introducing a new quantity, the multiplicity, which we now define.

In the case of $\lambda=(3)$, the spin-$3/2$ multiplet is generated from $\ket{{\uparrow}{\uparrow}{\uparrow}}$ the unique configuration with all three qubits aligned. 
The construction is entirely insensitive to particle labels: each qubit contributes identically, and no distinction between particles is retained. 
There exists only a single way to realize a fully symmetric multiplet, and it exhausts the $\lambda=(3)$ sector; that is
\begin{align}\label{eq:d_3}
    d_{\;\vcenter{\hbox{\scalebox{0.4}{\ydiagram{3}}}}} = d_{(3)} = 1.
\end{align}
Particle identity is erased.

In contrast, the spin-1/2 multiplet constructed in $\lambda=(2, 1)$ requires a choice. 
Pairing qubits $(1,2)$ into a singlet and assigning the residual spin to qubit $3$ is only one possibility. 
Equivalent constructions arise by pairing $(1,3)$ or $(2,3)$ instead. Although each yields a spin-$1/2$ multiplet with identical $\hat{J}^2$ and $\hat{J}_z$ structure, their basis states are distinct superpositions.
\begin{align}
  \ket{A} \;&=\; \tfrac{1}{\sqrt{2}}(\ket{{\uparrow}{\downarrow}{\uparrow}} - \ket{{\downarrow}{\uparrow}{\uparrow}}),
  \qquad\text{(pair $(1,2)$ singlet, particle $3$ up)},\label{eq:spin_1_2_A}\\
  \ket{B} \;&=\; \tfrac{1}{\sqrt{2}}(\ket{{\uparrow}{\uparrow}{\downarrow}} - \ket{{\uparrow}{\downarrow}{\uparrow}}),
  \qquad\text{(pair $(2,3)$ singlet, particle $1$ up)},\label{eq:spin_1_2_B}\\
  \ket{C} \;&=\; \tfrac{1}{\sqrt{2}}(\ket{{\uparrow}{\uparrow}{\downarrow}} - \ket{{\downarrow}{\uparrow}{\uparrow}}).
  \qquad\text{(pair $(1,3)$ singlet, particle $2$ up)},\label{eq:spin_1_2_C}
\end{align}
At first glance, this suggests three independent multiplets. 
This intuition is misleading. 
A direct calculation reveals that only two are linearly independent. 
Writing the highest-weight states explicitly and comparing them shows that one is always a linear combination of the others. 
\begin{align}
  \ket{A} + \ket{B}
  \;=\; \tfrac{1}{\sqrt{2}}\bigl(\ket{{\uparrow}{\downarrow}{\uparrow}} - \ket{{\downarrow}{\uparrow}{\uparrow}} + \ket{{\uparrow}{\uparrow}{\downarrow}} - \ket{{\uparrow}{\downarrow}{\uparrow}}\bigr)
  \;=\; \tfrac{1}{\sqrt{2}}\bigl(\ket{{\uparrow}{\uparrow}{\downarrow}} - \ket{{\downarrow}{\uparrow}{\uparrow}}\bigr)
  \;=\; \ket{C}.
\end{align}
Thus, the three constructions span a two-dimensional subspace.
The implication is precise: in the mixed-symmetry sector $\lambda=(2, 1)$, the choice of singlet pair constitutes a genuine degree of freedom --- but only two such choices are independent. 
The number of independent realisations defines the multiplicity of the sector, 
\begin{align}\label{eq:d_2_1}
    d_{\;\vcenter{\hbox{\scalebox{0.4}{\ydiagram{2,1}}}}} = d_{(2,1)} = 2.
\end{align}

The fully symmetric, i.e., $\lambda=(3)$, sector contains a single multiplet, with no internal freedom to generate another. 
By contrast, the mixed-symmetry sector, i.e., $\lambda=(2, 1)$, consists of several formally identical multiplets, distinguished not by their total spin $j$ or magnetic quantum number $m$, but by their internal coupling pattern. 
To specify a state uniquely, one therefore needs an additional label, beyond $j$ and $m$, that identifies which copy of the multiplet one is in. 
This extra index is the multiplicity label, and the number of values it can take is the multiplicity, denoted
\begin{align}
    d_{\lambda} \;\coloneqq\; 
    \text{number of linearly independent multiplets of type $\lambda$ inside $(\mathbb{C}^{d})^{\otimes n}$}
\end{align}
In the Schur-Weyl decomposition of Eq.~\eqref{eq:SW_decomp}, this freedom is encoded precisely in $V_\lambda \otimes W_\lambda$.
The space $W_\lambda$ carries the multiplicity: its dimension is $d_\lambda$, and each of its basis vectors selects a distinct copy of the multiplet.

Let $m_\lambda$ denote the dimension of a single multiplet of type $\lambda$. 
For qubits, this is simply the spin-multiplet size $2j+1$. 
Thus,
\begin{align}\label{eq:m_3}
    m_{\;\vcenter{\hbox{\scalebox{0.4}{\ydiagram{3}}}}} = m_{(3)} = 4.
\end{align}
for the spin-3/2 quadruplet of case $\lambda=(3)$, as demonstrated in Eq.~\eqref{eq:Young_Diagram_3}. 
Meanwhile, 
\begin{align}\label{eq:m_2_1}
    m_{\;\vcenter{\hbox{\scalebox{0.4}{\ydiagram{2,1}}}}} = m_{(2,1)} = 2.
\end{align}
for the spin-1/2 doublet of case $\lambda=(2,1)$.
Here, $d_\lambda$ counts how many independent copies of that multiplet occur in the corresponding sector. 
The total dimension of the $\lambda$ sector $V_\lambda \otimes W_\lambda$ is therefore
\begin{align}
    \dim(\lambda\text{-sector}) \;=\; m_{\lambda}\, d_{\lambda}.
\end{align}
Summing over all allowed Young diagrams $\lambda$ must recover the full Hilbert-space dimension $d^{n}$:
\begin{align}\label{eq:dim_sum}
  \sum_{\lambda \vdash n,\; \ell(\lambda) \leqslant d} m_{\lambda}\, d_{\lambda} \;=\; d^{n}.
\end{align}
For the running example $n=3$ and $d=2$, we have $m_{(3)}=4$ and $m_{(2,1)}=2$, while the fully antisymmetric sector $\lambda=(1,1,1)$ is absent. Equation~\eqref{eq:dim_sum} therefore becomes
\begin{align}
    \underbrace{4\,d_{(3)}}_{\text{spin-$3/2$ sector}} \;+\; \underbrace{2\,d_{(2,1)}}_{\text{spin-$1/2$ sector}} \;+\; 0 \;=\; 8 \;=\; 2^{3}.
\end{align}
Since constructions $\lambda=(3)$ and $\lambda=(2,1)$ explicitly realize at least one multiplet in each sector, we must have $d_{(3)}\geqslant 1$ and $d_{(2,1)}\geqslant 1$. 
The only non-negative integer solution is $d_{(3)}=1$ and $d_{(2,1)}=2$.
This matches the physical picture exactly: 
one fully symmetric multiplet, and two independent mixed-symmetry multiplets distinguished by their internal coupling structure.

The argument above determines $d_{(3)}$ and $d_{(2,1)}$ for the special case $n = 3$ by combining the dimension constraint Eq.~\eqref{eq:dim_sum} with the lower bounds $d_{(3)} \geqslant 1$ and $d_{(2,1)} \geqslant 1$, established through explicit constructions of at least one multiplet in each sector. 
For general $n$ (still with $d = 2$), however, a closed-form expression for $d_{\lambda}$ is preferable --- one that avoids the need to construct multiplets sector by sector.
Such an expression can be obtained by counting computational basis states at a fixed total $\hat{J}_{z}$ eigenvalue $m$ in two complementary ways: first in the computational basis, where the counting is elementary, and then in the spin-multiplet basis, where the multiplicities $d_{\lambda}$ naturally appear. 
For qubits, namely $d = 2$, the Schur-Weyl label $\lambda = (\lambda_{1}, \lambda_{2})$, with $\lambda_{1} \geqslant \lambda_{2}$ and $\lambda_{1} + \lambda_{2} = n$, reduces to the spin quantum number
\begin{align}\label{eq:j_from_lambda_qubit}
    j = \frac{\lambda_{1} - \lambda_{2}}{2},
\end{align}
and we therefore write $d_{j}$ in place of $d_{\lambda}$ throughout.

A computational basis string $\ket{b_{1}\, b_{2}\, \cdots\, b_{n}}$, with $b_{k} \in \{0,1\}$, has total $\hat{J}_{z}$ eigenvalue
\begin{align}
    m \;=\; \frac{\#\{k : b_{k} = 0\} - \#\{k : b_{k} = 1\}}{2} \;=\; \frac{n - 2r}{2},
\end{align}
where
\begin{align}
    r \coloneqq \#\{k : b_{k} = 1\},
\end{align}
is the number of down-spins. 
Fixing $m$ is therefore equivalent to fixing $r = n/2 - m$, and the number of such strings is
\begin{align}\label{eq:dj_lhs}
  \#\{\text{computational basis strings with $\hat{J}_{z} = m$}\} \;=\; \binom{n}{r} \;=\; \binom{n}{\frac{n}{2} - m}.
\end{align}
In the spin multiplet basis, the same eigenspace is organised by total spin $j$. 
A sector contributes to the $m$-eigenspace if and only if $|m| \leqslant j$. 
Whenever this condition holds, each of the $d_{j}$ copies of the spin-$j$ multiplet contributes exactly one state with magnetic label $m$. 
Summing over all such sectors gives
\begin{align}\label{eq:dj_column_sum}
  \#\{\text{spin-multiplet basis states with $\hat{J}_{z} = m$}\} \;=\; \sum_{j \,\geqslant\, |m|} d_{j}.
\end{align}
These two counts must agree, as they enumerate the same $\hat{J}_{z}$-eigenspace in two different orthonormal bases related by a unitary transformation.
Equating Eqs.~\eqref{eq:dj_lhs} and \eqref{eq:dj_column_sum} yields the column-sum identity on the multiplicities,
\begin{align}\label{eq:dj_column_identity}
  \binom{n}{\frac{n}{2} - m} \;=\; \sum_{j \,\geqslant\, |m|} d_{j},
  \quad \forall\, m.
\end{align}

Evaluating Eq.~\eqref{eq:dj_column_identity} at $m = j$ and at $m = j + 1$ (assuming $j + 1 \leqslant n/2$ so that both are admissible) gives the adjacent identities
\begin{align}
    \binom{n}{\frac{n}{2} - j} \;=\; d_{j} + d_{j+1} + d_{j+2} + \cdots,
\end{align}
and
\begin{align}
    \binom{n}{\frac{n}{2} - j - 1} \;=\; d_{j+1} + d_{j+2} + \cdots,
\end{align}
whose difference isolates $d_{j}$
\begin{align}\label{eq:dj_formula}
  d_{j} \;=\; \binom{n}{\frac{n}{2} - j} \;-\; \binom{n}{\frac{n}{2} - j - 1},
\end{align}
with the convention $\binom{n}{k} \coloneqq 0$ for $k < 0$ or $k > n$, covering the edge cases $j = 0$ and $j = n/2$.
This expression is the $d=2$ instance of the general hook-length formula, here derived by elementary counting.

For $n = 3$, i.e., $(\mathbb{C}^{d})^{\otimes 3}$, Eq.~\eqref{eq:dj_formula} gives
\begin{align}
    d_{3/2} \;=\; \binom{3}{0} - \binom{3}{-1} \;=\; 1 - 0 \;=\; 1,
  \qquad
  d_{1/2} \;=\; \binom{3}{1} - \binom{3}{0} \;=\; 3 - 1 \;=\; 2,
\end{align}
in agreement with the values obtained earlier from explicit constructions and the dimension constraint. 
The closed form in Eq.~\eqref{eq:dj_formula} thus reproduces the result without recourse to any explicit multiplet construction, and directly yields the decomposition
\begin{align}
  (\mathbb{C}^{2})^{\otimes 3}
  \;=\; 
  \underbrace{V_{\;\vcenter{\hbox{\scalebox{0.4}{\ydiagram{3}}}}}}_{\dim = 4} 
  \otimes \underbrace{W_{\;\vcenter{\hbox{\scalebox{0.4}{\ydiagram{3}}}}}}_{\dim = 1}
  \;\oplus\; 
  \underbrace{V_{\;\vcenter{\hbox{\scalebox{0.4}{\ydiagram{2,1}}}}}}_{\dim = 2} 
  \otimes 
  \underbrace{W_{\;\vcenter{\hbox{\scalebox{0.4}{\ydiagram{2,1}}}}}}_{\dim = 2}.
\end{align}
We summarise the discussion developed thus far for $(\mathbb{C}^{2})^{\otimes 3}$ in Tab.~\ref{tab:multiplicity_n_3}.

\begin{table}[htbp]
\centering
\begin{tblr}{
  colspec = {c || c | c | c | c | l},
  row{1} = {bg=gray!50, font=\bfseries}, 
  row{2} = {bg=gray!50, font=\bfseries},
  column{1} = {bg=gray!10, font=\bfseries},
  hlines,
  vlines,
  cells = {m, c},
  row{3} = {bg=mOrange},
  row{4} = {bg=mBlue},
}
  \SetCell[c=6]{c} Decomposition of $(\mathbb{C}^{2})^{\otimes 3}$ & & & & & \\
  $j$ & $\binom{n}{n/2-j}$ & $\binom{n}{n/2-j-1}$ & $d_j$ & $m_j$ & Sector Dimension $m_j d_j$ \\
  3/2 & $\binom{3}{0} = 1$ & $\binom{3}{-1} = 0$ & 1 & 4 & $4 \cdot 1 = 4$ \\
  1/2 & $\binom{3}{1} = 3$ & $\binom{3}{0} = 1$ & 2 & 2 & $2 \cdot 2 = 4$ \\
  \SetCell[c=5]{r} Total & & & & & $8 = 2^3$ \\
\end{tblr}
\caption{
\textbf{Schur-Weyl Decomposition of $(\mathbb{C}^{2})^{\otimes 3}$}. 
Each row corresponds to an allowed total-spin value $j$ for three spin-1/2 particles. 
Columns two and three evaluate the binomial coefficients appearing in Eq.~\eqref{eq:dj_formula}, whose difference yields the multiplicity $d_j$ (column four). 
Column five lists the dimension of the irrep $m_j = 2j+1$.
The final column reports the sector dimension $m_{j} d_{j} = (2j+1)\,d_{j}$. 
The two sectors together account for the full Hilbert space, namely $\dim \bigl((\mathbb{C}^{2})^{\otimes 3}\bigr) = 8$, consistent with the Schur-Weyl decomposition of Eq.~\eqref{eq:SW_decomp}.}
\label{tab:multiplicity_n_3}
\end{table}

Building on the results established above, we extend the analysis to the case $n = 4$, namely $(\mathbb{C}^{2})^{\otimes 4}$. 
The corresponding sector structure and multiplicities are provided in Tab.~\ref{tab:multiplicity_n_4}.

\begin{table}[htbp]
\centering
\begin{tblr}{
  colspec = {c || c | c | c | c | l},
  row{1} = {bg=gray!50, font=\bfseries}, 
  row{2} = {bg=gray!50, font=\bfseries},
  column{1} = {bg=gray!10, font=\bfseries},
  hlines,
  vlines,
  cells = {m, c},
  row{3} = {bg=mOrange},
  row{4} = {bg=mBlue},
  row{5} = {bg=mPurple}
}
  \SetCell[c=6]{c} Decomposition of $(\mathbb{C}^{2})^{\otimes 4}$ & & & & & \\
  $j$ & $\binom{n}{n/2-j}$ & $\binom{n}{n/2-j-1}$ & $d_j$ & $m_j$ & Sector Dimension $m_j d_j$ \\
  2 & $\binom{4}{0} = 1$ & $\binom{4}{-1} = 0$ & 1 & 5 & $5 \cdot 1 = 5$ \\
  1 & $\binom{4}{1} = 4$ & $\binom{4}{0} = 1$ & 3 & 3 & $3 \cdot 3 = 9$ \\
  0 & $\binom{4}{2} = 6$ & $\binom{4}{1} = 4$ & 2 & 1 & $1 \cdot 2 = 2$ \\
  \SetCell[c=5]{r} Total & & & & & $16 = 2^4$ \\
\end{tblr}
\caption{\textbf{Schur-Weyl Decomposition of $(\mathbb{C}^{2})^{\otimes 4}$}. 
Each row labels an admissible total-spin value $j$ for four spin-$\tfrac{1}{2}$ particles.
Columns two and three list the binomial coefficients entering Eq.~\eqref{eq:dj_formula}; their difference gives the multiplicity $d_{j}$ shown in column four.
Column five records the dimension of the corresponding irreducible representation, $m_{j} = 2j + 1$, while the final column gives the sector dimension $m_{j} d_{j} = (2j+1)\,d_{j}$.
The three sectors together account for the full Hilbert space, $\dim\bigl((\mathbb{C}^{2})^{\otimes 4}\bigr) = 16$, in agreement with the Schur-Weyl decomposition of Eq.~\eqref{eq:SW_decomp}.}
\label{tab:multiplicity_n_4}
\end{table}

With this structural understanding of the Schur-Weyl decomposition of tensor-product systems in place, we are now equipped to proceed further.
Nevertheless, as highlighted in Subsec.~\ref{subsec:SDP_Symmetry}, a direct implementation remains computationally demanding.
To overcome this limitation, we introduce an alternative approach --- the Clebsch-Gordan recursion --- which builds the decomposition iteratively.
The key idea is to adopt a physically transparent viewpoint: rather than decomposing $(\mathbb{C}^{2})^{\otimes n}$ in a single step, we examine how the total angular momentum evolves as qubits are added one at a time. 
This incremental perspective not only reduces computational complexity but also reveals the underlying structure in a particularly transparent way.

To circumvent the exponential overhead of manipulating $\rho^{\otimes n}$ --- whether as a $d^{n+1}\times d^{n+1}$ matrix or via the $\mO(n!\cdot d^{2n})$ character-formula construction of Subsec.~\ref{subsec:SDP_Symmetry} --- we construct the $n$-qubit spin structure sequentially, adding one qubit at a time. 
The procedure follows the ladder
\begin{align}\label{eq:tensor_ladder}
    \rho \;\longrightarrow\; \rho^{\otimes 2} \;\longrightarrow\; \rho^{\otimes 3} \;\longrightarrow\; \cdots \;\longrightarrow\; \rho^{\otimes n},
\end{align}
reusing the $(k-1)$-qubit construction to generate the $k$-qubit object, and thereby never invoking the full $2^{n}$-dimensional space.

The construction is separated into three parts. 
We begin by fixing the two bases that define the representation. 
We then turn to the multiplicity index $\alpha$ in the coupled basis, showing that, despite its formal presence, it is operationally inert in the final SDP, a point that is essential for the complexity reduction. 
The essential gain then becomes immediate: in the coupled basis, permutation symmetry is resolved into explicit block structure, recasting an exponential problem into one governed by symmetry-reduced sectors.

A key observation is that $(\mathbb{C}^{2})^{\otimes n}$ admits two natural bases.
The computational basis $\ket{b_{1}\, b_{2}\, \ldots\, b_{n}}$, with $b_{k} \in \{0, 1\}$, resolves the $\hat{S}_{z}$ eigenvalue of each qubit individually.
By contrast, the coupled basis $\ket{j, m, \alpha}_{n}$, organizes states in terms of quantum numbers: 
the total spin $j$, magnetic quantum number $m$, and a multiplicity label $\alpha \in \{1, \ldots, d_{j}\}$ distinguishing the $d_{j}$ orthogonal states of the spin-$j$ multiplet (Eq.~\eqref{eq:dj_formula}).
These two descriptions are connected by a unitary transformation whose matrix elements are the Clebsch-Gordan coefficients. 
For $n = 2$, where $d_{j} = 1$ in both sectors, this transformation reduces to the familiar singlet-triplet decomposition:
\begin{align}\label{eq:n2_change_of_basis}
    \ket{1, +1}_{2} 
    &\coloneqq 
    \ket{1, +1, 1}_{2}
    =
    \ket{00}, \\
    \ket{1, 0}_{2} 
    &\coloneqq 
    \ket{1, 0, 1}_{2}
    =
    \frac{1}{\sqrt{2}}\left(\ket{01} + \ket{10}\right), \\
    \ket{1, -1}_{2} 
    &\coloneqq
    \ket{1, -1, 1}_{2}
    =
    \ket{11}, \\
    \ket{0, 0}_{2} 
    &\coloneqq
    \ket{0, 0, 1}_{2}
    =
    \frac{1}{\sqrt{2}}\left(\ket{01} - \ket{10}\right),
\end{align}
with $\pm1/\sqrt{2}$ are the first Clebsch-Gordan coefficients to arise.
The total magnetic quantum number $m$ is immediately read off from a computational basis vector: it is half the difference between the number of $0$'s and $1$'s in the bit string. 
The total spin $j$, however, is not directly accessible in this representation. 
A generic computational basis state is typically a superposition of components with different total spin values.
This is already evident for $n=2$,
\begin{align}\label{eq:uncoupled_is_mixture}
    \ket{01} \;=\; \frac{1}{\sqrt{2}}\;\underbrace{\ket{1,\, 0}_{2}}_{j = 1}
     \;+\; \frac{1}{\sqrt{2}}\;\underbrace{\ket{0,\, 0}_{2}}_{j = 0},
\end{align}
showing that $\ket{01}$ decomposes into equal contributions from the triplet and singlet sectors.
Thus, while specifying the individual spin orientations fixes $m=0$, it does not determine whether the state belongs to the symmetric triplet or the antisymmetric singlet --- two distinct $\hat{J}^{2}$ eigenspaces sharing the same $m$. 
The Clebsch-Gordan transformation resolves this ambiguity: it reorganises computational basis states into superpositions that diagonalize $\hat{J}^{2}$, thereby isolating components of total spin.

The multiplicity index $\alpha$ encodes the residual freedom that remains once the total spin quantum numbers are fixed. 
As noted in this subsection, a given $j$-sector may contain multiple orthogonal copies of the same irreducible multiplet. 
The $\lambda=(2,1)$ sector for three qubits provides the simplest instance: it hosts $d_{\;\vcenter{\hbox{\scalebox{0.4}{\ydiagram{2,1}}}}}=2$ copies (see Eq.~\eqref{eq:d_2_1}) of the spin-1/2 doublet. 
These copies are degenerate under all rotationally invariant observables --- they share identical $(\hat{J}^{2},\hat{J}_{z})$ eigenvalues --- and differ only in their internal coupling structure, namely how the constituent qubits are paired into singlets and assembled into total spin.
This freedom is formalised by the multiplicity space $W_{j}$ in Eq.~\eqref{eq:SW_decomp}, a $d_{j}$-dimensional space that parametrizes all inequivalent realizations of spin $j$. 
The index $\alpha\in\{1,\ldots,d_{j}\}$ selects an orthonormal basis vector in $W_{j}$, thereby specifying a concrete copy of the multiplet.
For instance, the states $\ket{A}$ and $\ket{B}$ defined in Eqs.~\eqref{eq:spin_1_2_A} and \eqref{eq:spin_1_2_B} provide two distinct realisations of the same multiplet; their structure can be visualised equivalently through the corresponding Young tableaux.
\begin{align}
    \begin{ytableau} 
        1 & 2 \\ 3 
    \end{ytableau}\;,
    \quad\text{and}\quad
    \begin{ytableau} 
        1 & 3 \\ 2 
    \end{ytableau}\;.
\end{align}

Recall that for $n$ copies of an initial state $\rho$, i.e., $\rho^{\otimes n}$, Schur-Weyl duality yields the decomposition in Eq.~\eqref{eq:rho_n_decomp}, with each irreducible sector $\lambda$ supporting a block $\rho_{\lambda}$. 
In the qubit case, $\lambda$ is in one-to-one correspondence with the total spin $j$, and we use the two labels interchangeably. 
We therefore denote by $\rho^{(n)}_{j}$ the block associated with total spin $j$ in the $n$-copy decomposition. 
\begin{align}\label{eq:rho_n_to_rho_j}
    \rho^{\otimes n}\longrightarrow
    \rho^{(n)}_{j}.
\end{align}
This notation anticipates the recursive construction that follows, in which $\rho^{(n)}_{j}$ is built directly from the $(n-1)$-copy structure.

The matrix dimension is set solely by the total spin $j$: 
each block $V_{j}\equiv V_{\lambda}$ has size $m_{j}\times m_{j}$ with $m_{j}=2j+1$, independent of the iteration step $n$. 
Once $j$ is fixed, the block retains the same structure throughout the construction. 
By contrast, the set of admissible $j$ values evolves with $n$: after coupling $n$ qubits, the allowed total spin sectors are given by
\begin{align}\label{eq:allowed_j_at_level_k}
    j \;\in\;
    \begin{cases}
        \{\,0,\; 1,\; \ldots,\; \frac{n}{2}\,\}, & n \text{ even}, \\[0.2em]
        \{\,\tfrac{1}{2},\; \frac{3}{2},\; \ldots,\; \frac{n}{2}\,\}, & n \text{ odd}.
    \end{cases}
\end{align}
Iterating this construction across all coupling steps yields
\begin{align}\label{eq:ladder_coupled}
    \{\rho\}
    \;\longrightarrow\;
    \bigl\{\rho_{0}^{(2)},\, \rho_{1}^{(2)}\bigr\}
    \;\longrightarrow\;
    \bigl\{\rho_{\frac{1}{2}}^{(3)},\, \rho_{\frac{3}{2}}^{(3)}\bigr\}
    \;\longrightarrow\; \cdots \;\longrightarrow\;
    \bigl\{\rho_{j}^{(n)}\bigr\}_{j}.
\end{align}
The full $2^{n}$-dimensional object never appears. The construction is instead driven by a local recursion that advances the iteration one qubit at a time. At each step, the previously obtained coupled-basis blocks $\{\rho_{j}^{(n-1)}\}_{j}$ are updated by coupling in a single qubit described by the $2\times 2$ state $\rho$ in the computational basis. 
The task is therefore to determine an update rule that produces the new block $\rho_{j'}^{(n)}$ from these ingredients. 
Schematically,
\begin{align}\label{eq:recursion_expected}
    \underbrace{\rho_{j'}^{(n)}}_{\text{coupled block after $n$ qubits}}
    \;=\;
    \mR\!\left[\,
    \underbrace{\bigl\{\rho_{j}^{(n - 1)}\bigr\}_{j}}_{\text{coupled blocks after $n - 1$ qubits}},\;\;
    \underbrace{\rho}_{\text{new qubit (comp.\ basis)}}
    \,\right].
\end{align}
The map $\mR$ will be developed explicitly soon.

By the definition of $\rho_{j'}^{(n)}$, the block recursion is inherited from a simpler recursion on the coupled-basis vectors themselves: each $n$-qubit vector $\ket{j', m', \alpha'}_{n}$ must be expressible as a linear combination of already-known $(n-1)$-qubit coupled basis vectors tensored with the new qubit's computational basis states, and substituting any such expansion back yields $\mR$. 
Pinning down such an expansion requires two pieces of data --- which $(n-1)$-qubit vectors actually appear on the right-hand side, and with what coefficients. 
These reduce to the two physical questions below:
\begin{enumerate}
    \item[(i)] \textbf{Admissible total spin.}
    Which terms appear? 
    Upon appending the $n$-th qubit to an $(n-1)$-qubit state expressed in the coupled basis, the total spin is constrained by angular momentum addition. 
    The problem reduces to identifying the admissible values of $j'$ that can arise from coupling the existing total spin with the added qubit.
    \item[(ii)] \textbf{Clebsch-Gordan coefficients.}
    With what coefficients? 
    How is each newly formed coupled-basis state $\ket{j', m', \alpha'}_{n}$ decomposed into a linear combination of the known $(n-1)$-qubit basis states tensored with the appended qubit, and what coefficients govern this expansion?
\end{enumerate}
Both are questions of angular-momentum composition rather than matrix algebra, and together they fix the structure of the recursion. 
The first is resolved by the branching rule derived below, while the second is addressed in the subsequent analysis of the resulting expansion coefficients.

Introducing the $n$-th qubit into an already-coupled $(n-1)$-qubit system of total spin $j$ is, physically, adding a new angular momentum vector of magnitude 1/2 to an existing one of magnitude $j$.
The combined magnitude can take only two values: $j + 1/2$ (the new qubit lengthens the total-spin vector), or $j - 1/2$ (the new qubit locks into a singlet with one of the existing spin components, shortening it).
This dichotomy is not a continuous interpolation: angular-momentum addition in $\mathrm{SU}(2)$ is discrete, and the tensor product $j \otimes 1/2$ decomposes into precisely two irreducible representations, $j + 1/2$ and $j - 1/2$, with nothing in between.
The triangle inequality $|j_{1} - j_{2}| \leqslant J \leqslant j_{1} + j_{2}$ leaves only these two values when $j_{2} = 1/2$.
In symbols, the two outcomes are (with the boundary case $j = 0$ keeping only the aligned branch, since total spin cannot be negative: the branching requires $j \geqslant 1/2$ for both branches to exist):
\begin{align}\label{eq:branching_values}
    j' \;\in\; \left\{\,j + \frac{1}{2},\;\; j - \frac{1}{2}\,\right\}.
\end{align}
The physical picture delivers these two values but does not rule out a third; a $\hat{J}_{z}$-eigenvalue count, carried out immediately below, makes the dichotomy rigorous. Before running the count, we state the underlying Hilbert-space identity: Eq.~\eqref{eq:branching_values} is the basis-vector statement of the representation-theoretic decomposition
\begin{align}\label{eq:branching_rule}
    j \;\otimes\; \frac{1}{2}
    \;=\; \left(j + \frac{1}{2}\right) \;\oplus\; \left(j - \frac{1}{2}\right),
\end{align}
valid for $j \geqslant 1/2$ (the edge case $j = 0$ gives only the $+$ branch). On the left, $j \otimes 1/2$ is the tensor product $V_{j} \otimes V_{1/2}$ of the spin-$j$ multiplet with the single-qubit space $V_{1/2} = \mathbb{C}^{2}$, of dimension $(2j + 1) \times 2 = 4j + 2$, with product-state basis $\{\ket{j, m} \otimes \ket{s}\}_{m, s}$. On the right, $(j + 1/2) \oplus (j - 1/2)$ is the orthogonal direct sum $V_{j + 1/2} \oplus V_{j - 1/2}$ of dimension 
\begin{align}
    \left(2j + 2\right) + 2j = 4j + 2.
\end{align}
Each summand carries a sharp total spin, whereas the left-hand product vectors generally do not (see Eq.~\eqref{eq:uncoupled_is_mixture}). 
The two sides are the same Hilbert space in two bases, with the change-of-basis coefficients, i.e., the Clebsch-Gordan coefficients, to be computed explicitly soon.
A product vector $\ket{j, m} \otimes \ket{s}$ carries total $\hat{J}_{z}$-eigenvalue 
\begin{align}
    M = m + \left(\frac{1}{2} - s\right),
\end{align}
so enumerating all $4j + 2$ product vectors gives the following Tab.~\ref{tab:Magnetic_Quantum_Number}
\begin{table}[htbp]
\centering
\begin{tblr}{
  colspec = {c || c | c | c | c | c},
  row{1} = {bg=gray!50, font=\bfseries},
  row{2} = {bg=gray!50, font=\bfseries},
  column{1} = {bg=gray!10, font=\bfseries},
  hlines,
  vlines,
  cells = {m, c},
}
  \SetCell[c=6]{c} Magnetic Quantum Number & & & & & \\
  $M$ & $+\left(j + \frac{1}{2}\right)$ & $+\left(j - \frac{1}{2}\right)$ & $\cdots$ & $-\left(j - \frac{1}{2}\right)$ & $-\left(j + \frac{1}{2}\right)$ \\
  \# product vectors & 1 & 2 & $\cdots$ & 2 & 1 \\
\end{tblr}
\caption{\textbf{Product Vector Distribution}. The top row shows the total magnetic quantum number $M$, and the bottom row indicates the corresponding number of product vectors.}
\label{tab:Magnetic_Quantum_Number}
\end{table}

The two extremes reached by a single product vector each, every interior $M$ is reached by two. A spin-$J$ multiplet $V_{J}$ contributes one state at every $M \in [-J, +J]$; matching this count uniquely requires one $V_{j + 1/2}$ (covering every $M$) plus one $V_{j - 1/2}$ (covering the interior only), and admits nothing else. The same conclusion follows from the triangle inequality $|j_{1} - j_{2}| \leqslant J \leqslant j_{1} + j_{2}$: substituting $j_{1} = j, j_{2} = 1/2$ gives $J \in \{j - 1/2,\, j + 1/2\}$, a window of width $1$. Both arguments rely on the added factor being spin-1/2: for a larger local dimension $d > 2$, the recursion would open up into $d$ simultaneous branches at every step, and the coupled basis would no longer admit a two-child recursive construction.

Iterating the two-branch step from $n = 1$ ($j = 1/2$) up to the target $n$ enumerates all $d_{j}$ spin-$j$ multiplets: every sequence of choices (each step picking $+1/2$ or $-1/2$) that terminates at the desired total spin labels one orthogonal copy, and the number of such choice-sequences equals $d_{j}$.
For $n = 3$, we obtain Tab.~\ref{tab:spin_addition_paths}.

\begin{table}[htbp]
\centering
\begin{tblr}{
  colspec = {c || c | l},
  row{1} = {bg=gray!50, font=\bfseries}, 
  row{2} = {bg=gray!50, font=\bfseries},
  column{1} = {bg=gray!10, font=\bfseries},
  hlines,
  vlines,
  cells = {m, l}, 
  row{3} = {bg=mOrange}, 
  row{4,5} = {bg=mBlue}, 
  row{6,7} = {bg=mPurple}, 
}
  \SetCell[c=3]{c} Recursive Addition of Spin-$1/2$ Particles & & \\
  $\mathbf{n}$ & Total Spin $\mathbf{j}$ & Derivation \\
  1 & $j = 1/2$ & Initial state \\
  2 & $j = 1$   & via $+$ from $j=1/2$ \\
    & $j = 0$   & via $-$ from $j=1/2$ \\
  3 & $j = 3/2$ & via $+$ from $j=1$ \\
    & $j = 1/2$ & via $-$ from $j=1$, or via $+$ from $j=0$ \\
\end{tblr}
\caption{\textbf{Angular Momentum Addition Path}. This table tracks the evolution of total angular momentum $j$ as additional spin-$1/2$ particles are coupled to the system.}
\label{tab:spin_addition_paths}
\end{table}

The two paths reaching $j = 1/2$ at $n = 3$ are the $d_{1/2} = 2$ orthogonal multiplets of the sector, distinguished by whether qubits $1, 2$ were coupled into a triplet and reduced by qubit $3$, or coupled into a singlet and extended by qubit $3$.
The same branching tree appears in atomic physics when constructing LS-coupled multielectron states by adding one electron at a time; the expansion coefficients are the same $\mathrm{SU}(2)$ Clebsch-Gordan coefficients we derive below. For atomic electrons, Pauli exclusion further prunes the tree --- an additional constraint absent for the distinguishable qubit copies considered here.

Question (i) stated earlier is now answered: 
only two parent spins $j' \pm 1/2$ contribute. 
We turn to (ii), the form of the linear combination itself. Fix a target $n$-qubit coupled-basis vector $\ket{j', m'}_{n}$ and a choice of parent spin $j \in \{j' - 1/2,\, j' + 1/2\}$.

A conservation law narrows the candidates further. Because $\hat{J}_{z}$ is additive across the coupling step, the magnetic label $m'$ of the $n$-qubit state must equal the sum of the magnetic label $m$ of the $(n-1)$-qubit state and the $n$-th qubit's $\hat{S}_{z}^{(n)}$-eigenvalue. 
Writing $s \in \{0, 1\}$ for the $n$-th qubit's computational basis label, i.e.,
\begin{align}
    \hat{S}_{z}^{(n)}\ket{s}_{n} = \left(\frac{1}{2} - s\right)\ket{s}_{n},
\end{align}
the conservation reads
\begin{align}\label{eq:Jz_conservation}
    m' \;=\; m + \frac{1}{2} - s
    \qquad\Longleftrightarrow\qquad
    m \;=\; m' - \frac{1}{2} + s.
\end{align}
For each of the two qubit states $s = 0, 1$, the parent's magnetic label $m$ is thus pinned to a single value, namely $m' - 1/2$ and $m' + 1/2$ respectively.

The $n$-qubit coupled-basis vector is therefore pinned to be a linear combination of at most two product states:
\begin{align}\label{eq:recursion_form_preview}
    \ket{j', m'}_{n}
    \;=\;
    c_{0}\,\ket{j,\, m' - \tfrac{1}{2}}_{n - 1} \otimes \ket{0}_{n}
    \;+\;
    c_{1}\,\ket{j,\, m' + \tfrac{1}{2}}_{n - 1} \otimes \ket{1}_{n}.
\end{align}
This is the shape of the recursion promised at the start of the subsubsection: the $n$-qubit basis vector is written out in terms of two $(n-1)$-qubit basis vectors, each extended by one of the two computational basis states of the added qubit, with weights $c_{0}$ and $c_{1}$ to be determined. Nothing more appears on the right-hand side; the parent spin $j$ is chosen once, and the multiplicity index $\alpha$ (when $d_{j} > 1$) propagates unchanged through the coupling and is therefore suppressed.

The two weights are inner products, namely overlaps, of the left-hand side with each of the two product states on the right,
\begin{align}\label{eq:recursion_overlap}
    c_{s}
    \;=\;
    \bigl(\bra{j,\, m' - \tfrac{1}{2} + s}_{n - 1} \otimes \bra{s}_{n}\bigr)\,\ket{j', m'}_{n},
\end{align}
and they carry a clean physical reading: $c_{s}$ is the amplitude for the process that takes the parent state $\ket{j,\, m' - 1/2 + s}_{n - 1}$, appends the $n$-th qubit in $\ket{s}_{n}$, and projects onto the $n$-qubit total-spin sector $(j', m')$. 
These two numbers are the only non-trivial data introduced at the $n$-th recursion step; everything else is linear bookkeeping on top of them.

The overlaps $c_{s}$ are the standard Clebsch-Gordan coefficients of angular momentum coupling, usually written with the bra-ket notation $\langle j_{1}, m_{1};\, j_{2}, m_{2} \mid J, M\rangle$; the amplitude for the uncoupled product $\ket{j_{1}, m_{1}} \otimes \ket{j_{2}, m_{2}}$ to be found in the coupled eigenstate $\ket{J, M}$. In our setting 
\begin{align}
    (j_{1}, j_{2}) &= (j, \frac{1}{2}),\\
    (J, M) &= (j', m'),\\
    m_{1} &= m' - \frac{1}{2} + s,\\
    m_{2} &= \frac{1}{2} - s,
\end{align}
and thus
\begin{align}\label{eq:CG_def}
c_{s}
\;\equiv\;
\langle\, j,\; m' - \tfrac{1}{2} + s ;\;\; \tfrac{1}{2},\; \tfrac{1}{2} - s \,\mid\, j',\, m' \,\rangle,
\end{align}
with the convention $s \in \{0, 1\}$: $s = 0$ corresponds to the new qubit in $\ket{\uparrow}_{n} \equiv \ket{0}_{n}$ (so $\hat{S}_{z}^{(n)} = +1/2$) and $s = 1$ to $\ket{\downarrow}_{n} \equiv \ket{1}_{n}$ (so $\hat{S}_{z}^{(n)} = -1/2$).

To make the construction concrete, consider the case of $\rho^{\otimes 3}$ and suppose that
\begin{align}\label{eq:rho_diag_example}
    \rho \;=\; p_{0}\ketbra{0}{0} + p_{1} \ketbra{1}{1},
    \qquad
    p_{0} + p_{1} = 1,\quad p_{0}, p_{1} \geqslant 0.
\end{align}
Then we have
\begin{align}
    \rho^{\otimes 3}
    \;=\; \mathrm{diag}\big(
        p_{0}^{3},\,
        p_{0}^{2}p_{1},\,
        p_{0}^{2}p_{1},\,
        p_{0} p_{1}^{2},\,
        p_{0}^{2} p_{1},\,
        p_{0} p_{1}^{2},\,
        p_{0} p_{1}^{2},\,
        p_{1}^{3}
    \big)
\end{align}
in the order $\ket{000}, \ket{001}, \ket{010}, \ket{011}, \ket{100}, \ket{101}, \ket{110}, \ket{111}$.
Reducing this $8 \times 8$ array to the two spin sectors $j = 3/2$ (one $4 \times 4$ multiplet, $d_{3/2} = 1$) and $j = 1/2$ (two copies of a $2 \times 2$ multiplet, $d_{1/2} = 2$) cannot be carried out yet --- it requires the explicit Clebsch-Gordan coefficients, which will be derived soon.

Here, we formulate the required closed forms in four steps, assemble them into the Clebsch-Gordan recursion at the level of matrix entries, and verify the result on the $n = 3$ example set up at the end.
For notational convenience within this subsubsection we promote the weights $c_{s}$ of Eq.~\eqref{eq:recursion_form_preview} to $c_{s}^{j' \leftarrow j}(m')$, making the parent spin $j$ and the target pair $(j', m')$ explicit as labels (the numerical value is unchanged):
\begin{align}\label{eq:CG_expansion}
    \ket{j', m'}_{n}
    \;=\;
    \sum_{s \in \{0, 1\}}
    c_{s}^{j' \leftarrow j}(m')\,
    \ket{j, m' - \tfrac{1}{2} + s}_{n-1} \otimes \ket{s}_{n},
    \qquad j = j' \mp \tfrac{1}{2},
\end{align}
with $c_{s}^{j' \leftarrow j}(m') \equiv \langle j, m' - 1/2 + s;\, 1/2, 1/2 - s \mid j', m' \rangle$ as per Eq.~\eqref{eq:CG_def}, and $\ket{s}_{n}$ the computational basis state of the $n$-th qubit ($\ket{0}$ for $s = 0$, $\ket{1}$ for $s = 1$).

The closed-form expressions for the CG coefficients can be found in most quantum mechanics textbooks. Here we derive them from first principles
for self-consistency. Here we derive it from first principles for self consistence.
The closed-form evaluation follows from a short lowering-operator calculation, sketched in four labelled steps below so that the final expressions can be reproduced without recourse to a CG table.
The coefficients for the $j = j' - 1/2$ branch (parent spin smaller than child) are derived first in Steps 1-3, after which Step 4 fixes the coefficients for the $j = j' + 1/2$ branch (parent spin larger than child) by orthogonality within the two-dimensional $\hat{J}_{z} = m'$ subspace.

\textbf{Step 1: Starting from the highest weight state.}
Take the $j = j' - 1/2$ branch first --- the aligned case of Eq.~\eqref{eq:branching_rule}, where appending the $n$-th qubit raises the total spin from $j$ to $j + 1/2 = j'$. The $j = j' + 1/2$ branch is deferred to Step 4 because its candidate highest-weight vector is not uniquely pinned by $\hat{J}_{z}$ alone --- Step 4 makes this explicit.
At the top of the target spin-$j'$ multiplet ($m' = j'$), the only product state in $V_{j} \otimes V_{1/2}$ with matching $\hat{J}_{z}$-eigenvalue is $\ket{j, j}_{n-1} \otimes \ket{0}_{n}$, so up to an overall phase (fixed to $+1$ by Condon and Shortley~\cite{condon1935theory}),
\begin{align}\label{eq:CG_top_vector}
    \ket{j', j'}_{n} \;=\; \ket{j, j}_{n-1} \otimes \ket{0}_{n}.
\end{align}
Three observations make explicit why Eq.~\eqref{eq:CG_top_vector} has no unknown CG coefficient to solve for.

\begin{itemize}
  \item \textbf{Branch at maximal magnetic quantum number.}
  Only one $s$-branch survives. 
  At $m' = j'$, the general 2-term expansion Eq.~\eqref{eq:CG_expansion} would have an $s = 1$ term proportional to $\ket{j, j + 1}_{n-1} \otimes \ket{1}_{n}$, but $\ket{j, j + 1}_{n-1}$ does not exist (the spin-$j$ magnetic range tops out at $m = j$), so that term vanishes. Only the $s = 0$ term remains.
  \item \textbf{Unit coefficient.}
  The surviving coefficient has modulus $1$.
  Both $\ket{j', j'}_{n}$ and $\ket{j, j}_{n-1} \otimes \ket{0}_{n}$ are unit vectors, so the coefficient relating them has absolute value $1$.
  \item \textbf{Condon-Shortley phase fixing.}
  The phase is $+1$ by convention, and it fixes the phases of every descendant state. 
  A modulus-$1$ coefficient is still free up to $e^{i\theta}$; 
  rule below of the Condon-Shortley convention pins $\theta = 0$. That single phase choice then propagates down every $\hat{J}_{-}$-descendant of $\ket{j', j'}_{n}$ in Step 2 --- replacing the $+1$ by $-1$ here would flip every Clebsch-Gordan coefficient at every lower value of $m'$.
\end{itemize}

The Condon-Shortley convention~\cite{condon1935theory} is adopted throughout this work, and constitutes the \emph{de facto} phase convention in Clebsch-Gordan tables and symbolic implementations. For clarity and later reference, we state it in the form of two rules
\begin{itemize}
  \item \textbf{Phases within one multiplet.}
  $\bra{j, m - 1}\hat{J}_{-}\ket{j, m} = +\sqrt{(j + m)(j - m + 1)}$: the lowering matrix element is real and positive. This pins all phases within one multiplet --- once $\ket{j, j}$ is chosen, each successive $\hat{J}_{-}$-application to produce the next lower-$m$ state inherits a positive coefficient, and it is exactly this rule that Step 2 invokes when selecting the $+$ root in Eq.~\eqref{eq:lowering_coefficient}.
  \item \textbf{Phases across irreducible representations.}
  The overlap $\langle j_{1}, j_{1};\, j_{2}, J - j_{1} \mid J, J \rangle$ between the stretched product state (both factors at maximum weight) and the highest weight state of the coupled multiplet is positive. At $(j_{1}, j_{2}, J) = (j, 1/2, j')$, this pins $\langle j, j;\, 1/2, 1/2 \mid j', j'\rangle = +1$ --- the coefficient on the RHS of Eq.~\eqref{eq:CG_top_vector}.
\end{itemize}

Every subsequent sign in this subsubsection (the $+$ root in Eq.~\eqref{eq:lowering_coefficient}, the $+$ roots in Eq.~\eqref{eq:CG_coeffs_up}, the $-$ in $c_{0}$ of Eq.~\eqref{eq:CG_coeffs_down}) follows mechanically from rules above. 
An overall phase on any multiplet is unobservable (it cancels in all $|\braket{\psi}{\psi}|^{2}$), so the convention is pure bookkeeping: it exists to make Clebsch-Gordan tables from different sources agree.

\textbf{Step 2: Lowering along the ladder with $\hat{J}_{-}^{k}$.}
Act on both sides of Eq.~\eqref{eq:CG_top_vector} with $\hat{J}_{-}^{k}$ for $k \coloneqq j' - m'$, and expand the $n$-qubit lowering operator using the binomial identity
\begin{align*}
    \hat{J}_{-}^{k} \;=\; \left(\hat{J}_{-}^{(A)} + \hat{J}_{-}^{(B)}\right)^{k} \;=\; \sum_{r=0}^{k} \binom{k}{r}\,\left(\hat{J}_{-}^{(A)}\right)^{k - r}\,\left(\hat{J}_{-}^{(B)}\right)^{r},
\end{align*}
where, throughout Step 2, the superscripts $A$ and $B$ label subsystems: $\hat{J}_{-}^{(A)}$ acts on the first $n - 1$ qubits (the parent spin-$j$ multiplet), and $\hat{J}_{-}^{(B)}$ acts on the newly appended $n$-th qubit. Only the terms $r = 0$ and $r = 1$ survive because $\hat{J}_{-}^{(B)}$ is nilpotent of degree two on the spin-1/2 factor: applying it twice to any single-qubit state returns zero, $(\hat{J}_{-}^{(B)})^{2}\ket{\uparrow} = \hat{J}_{-}^{(B)}\ket{\downarrow} = 0$, so any term with $r \geqslant 2$ in the binomial expansion vanishes.
This truncation from $k + 1$ candidate terms to just $2$ is specific to $d = 2$ and is what keeps the recursion polynomial in $n$: for local dimension $d > 2$, the analogous nilpotency of $\hat{J}_{-}^{(B)}$ sets in only at higher powers, and the binomial sum retains more terms.

We now derive a closed-form expression for the action of $\hat{J}_{-}^{p}$ on a highest-weight vector, which will be reused repeatedly in what follows. 
Each surviving binomial term calls for the action of $\hat{J}_{-}^{p}$ on the top of some spin-$j$ multiplet---on the parent via $\hat{J}_{-}^{(A)}$ with $p = k$ or $p = k - 1$, and on the new qubit via $\hat{J}_{-}^{(B)}$ with $p \in \{0, 1\}$. We therefore pause the main derivation to justify, once and for all,
\begin{align}\label{eq:single_multiplet_lowering}
    \hat{J}_{-}^{p}\ket{j, j} \;=\; \sqrt{\frac{(2j)!\,p!}{(2j - p)!}}\;\ket{j, j - p}
    \qquad (0 \leqslant p \leqslant 2j).
\end{align}

The derivation of Eq.~\eqref{eq:single_multiplet_lowering} below runs entirely inside a single spin-$j$ irreducible representation, treated abstractly: the symbol $\ket{j, m}$ here is any state satisfying $\hat{J}^{2}\ket{j, m} = j(j + 1)\ket{j, m}$ and $\hat{J}_{z}\ket{j, m} = m\ket{j, m}$, with no tensor-product structure invoked, and the single-step coefficient $C_{j, m} \coloneqq \bra{j, m - 1}\hat{J}_{-}\ket{j, m}$ is the matrix element of $\hat{J}_-$ taken within a single multiplet --- bra and ket both in the same spin-$j$ irreducible representation. $C_{j, m}$ is not a Clebsch-Gordan coefficient: a Clebsch-Gordan coefficient couples two different irreducible representations (e.g.,\ $V_{j} \otimes V_{1/2} \to V_{j'}$) and is what we read off at the end of Step 3. The argument below uses only (a) the $\mathrm{SU}(2)$ commutation relations $[\hat{J}_{z}, \hat{J}_{\pm}] = \pm\hat{J}_{\pm}$, $[\hat{J}_{+}, \hat{J}_{-}] = 2\hat{J}_{z}$, (b) the eigenvalue definition of $\ket{j, m}$ stated above, (c) Hermiticity $\hat{J}_{+} = \hat{J}_{-}^{\dagger}$, and (d) first rule of the Condon-Shortley convention~\cite{condon1935theory}; the result therefore applies uniformly to every realisation of the spin-$j$ irreducible representation (a single spin-$j$ particle, any copy of the $j$-sector inside $(\mathbb{C}^{2})^{\otimes n}$, any abstract $\mathrm{SU}(2)$ module), and in particular does not require knowing the $n$-qubit coupled basis in the computational basis; there is no circularity with the Clebsch-Gordan coefficients we are after.

From $[\hat{J}_{z}, \hat{J}_{-}] = -\hat{J}_{-}$, one finds 
\begin{align}
    \hat{J}_{z}(\hat{J}_{-}\ket{j, m}) = (m - 1)\,\hat{J}_{-}\ket{j, m},
\end{align}
so $\hat{J}_{-}\ket{j, m}$ is a $\hat{J}_{z}$-eigenvector with eigenvalue $m - 1$; and $[\hat{J}^{2}, \hat{J}_{-}] = 0$ keeps it in the same spin-$j$ multiplet. Hence
\begin{align}\label{eq:single_step_lowering}
    \hat{J}_{-}\ket{j, m} \;=\; C_{j, m}\,\ket{j, m - 1},
\end{align}
for some coefficient $C_{j, m}$ to be determined. Taking the norm-squared and using $\hat{J}_{+} = \hat{J}_{-}^{\dagger}$ together with the operator identity $\hat{J}_{+}\hat{J}_{-} = \hat{J}^{2} - \hat{J}_{z}^{2} + \hat{J}_{z}$ (which follows from $[\hat{J}_{+}, \hat{J}_{-}] = 2\hat{J}_{z}$ combined with $\hat{J}^{2} = \hat{J}_{z}^{2} + 1/2(\hat{J}_{+}\hat{J}_{-} + \hat{J}_{-}\hat{J}_{+})$, itself obtained by writing $\hat{J}_{x}, \hat{J}_{y}$ in terms of $\hat{J}_{\pm}$ and expanding $\hat{J}_{x}^{2} + \hat{J}_{y}^{2}$):
\begin{align}
    |C_{j, m}|^{2}
    &\;=\; \bra{j, m}\hat{J}_{+}\hat{J}_{-}\ket{j, m}\\
    &\;=\; \bra{j, m}\hat{J}^{2} - \hat{J}_{z}^{2} + \hat{J}_{z}\ket{j, m} \\
    &\;=\; j(j + 1) - m^{2} + m\\
    &\;=\; (j + m)(j - m + 1).
\end{align}
Condon-Shortley convention~\cite{condon1935theory} selects the positive root,
\begin{align}\label{eq:lowering_coefficient}
    C_{j, m} \;=\; +\sqrt{(j + m)(j - m + 1)}.
\end{align}
Iterating Eq.~\eqref{eq:single_step_lowering} successively from $m = j$ down to $m = j - p + 1$,
\begin{align}
    \hat{J}_{-}^{p}\ket{j, j}
    \;=\; \left(\prod_{q = 0}^{p - 1} C_{j, j - q}\right)\,\ket{j, j - p}.
\end{align}
Substituting Eq.~\eqref{eq:lowering_coefficient} at $m = j - q$ gives
\begin{align}
    C_{j, j - q} = \sqrt{(2j - q)(q + 1)}.
\end{align}
So the product separates into two falling factorials:
\begin{align*}
    \prod_{q = 0}^{p - 1} (2j - q) &\;=\; (2j)(2j - 1) \cdots (2j - p + 1) \;=\; \frac{(2j)!}{(2j - p)!}, 
\end{align*}
and
\begin{align}
    \prod_{q = 0}^{p - 1} (q + 1) &\;=\; 1 \cdot 2 \cdots p \;=\; p!.
\end{align}
Combining, we have
\begin{align}
    \prod_{q = 0}^{p - 1} C_{j, j - q} = \sqrt{\frac{(2j)!\,p!}{(2j - p)!}},
\end{align}
which is Eq.~\eqref{eq:single_multiplet_lowering}. The derivation has used only the four ingredients listed in the scope warning: $\mathrm{SU}(2)$ commutators, the eigenvalue definition of $\ket{j, m}$, Hermiticity, and the Condon-Shortley convention. 
It does not invoke any Clebsch-Gordan tables or $n$-qubit structure.

Four limiting cases fix the formula. 
In particular, at $p = 0$, it reduces to 
\begin{align}
    \hat{J}_{-}^{0}\ket{j, j} = \ket{j, j},
\end{align}
as it must, with the factorials collapsing to
\begin{align}
    \sqrt{\frac{(2j)!\cdot 0!}{(2j)!}} = 1.
\end{align}
At $p = 1$, it gives
\begin{align}
    \sqrt{\frac{(2j)!}{(2j-1)!}}\,\ket{j, j - 1} = \sqrt{2j}\,\ket{j, j - 1},
\end{align}
matching 
\begin{align}
    C_{j, j} = \sqrt{2j \cdot 1},
\end{align}
from Eq.~\eqref{eq:lowering_coefficient}. 
At the opposite extreme, i.e., $p = 2j$, the formula reaches the bottom of the multiplet $\ket{j, -j}$ with the total ladder weight 
\begin{align}
    \sqrt{\frac{(2j)!\cdot(2j)!}{0!}} = (2j)!\;.
\end{align}
One step further, $p = 2j + 1$, lies outside the stated range: the factorial $(2j - p)! = (-1)!$ is undefined. 
Physically, the ladder has already terminated at the lower boundary, where the state is annihilated by the lowering operator.
\begin{align}
    C_{j, -j} = \sqrt{0 \cdot (2j + 1)} = 0,
\end{align}
so the $(2j + 1)$-th application returns the zero vector --- the algebraic origin of the range constraint $0 \leqslant p \leqslant 2j$.

We now leave the abstract single-multiplet setting and specialise Eq.~\eqref{eq:single_multiplet_lowering} to the concrete $n$-qubit decomposition $\ket{j', j'}_{n} = \ket{j, j}_{n - 1} \otimes \ket{0}_{n}$ delivered by Step 1 (obtained from $\hat{J}_{z}$-counting alone, with no Clebsch-Gordan input). 
On the parent subsystem $A$ the operator $\hat{J}_{-}^{(A)}$ acts within the spin-$j$ multiplet on $\ket{j, j}_{n - 1}$, so we apply Eq.~\eqref{eq:single_multiplet_lowering} with $p = k - r$; on the new-qubit subsystem $B$ the operator $\hat{J}_{-}^{(B)}$ acts within the spin-$\tfrac{1}{2}$ multiplet on $\ket{0}_{n}$ (which is itself the highest weight vector of a spin-1/2 multiplet), so we apply Eq.~\eqref{eq:single_multiplet_lowering} with $p = r \in \{0, 1\}$. Collecting the two surviving $r = 0$ and $r = 1$ terms of the binomial,
\begin{align}\label{eq:CG_lowered}
    \hat{J}_{-}^{k}\ket{j', j'}_{n}
    \;=\; \sqrt{\frac{(2j)!\,k!}{(2j - k)!}}\;\ket{j, j - k}_{n-1} \otimes \ket{0}_{n} + 
    k\,\sqrt{\frac{(2j)!\,(k-1)!}{(2j - k + 1)!}}\;\ket{j, j - k + 1}_{n-1} \otimes \ket{1}_{n}.
\end{align}

\textbf{Step 3: Determining the coefficients for $j = j' - \tfrac{1}{2}$ branch.}
The left-hand side of Eq.~\eqref{eq:CG_lowered} is itself one more application of Eq.~\eqref{eq:single_multiplet_lowering} within a single multiplet, this time inside the child spin-$j'$ sector: 
\begin{align}
    \hat{J}_{-}^{k}\ket{j', j'}_{n} = \sqrt{\frac{(2j')!\,k!}{(2j' - k)!}}\,\ket{j', j' - k}_{n} = 
    \sqrt{\frac{(2j')!\,k!}{(2j' - k)!}}\,\ket{j', m'}_{n},
\end{align}
where we have used $m' = j' - k$.
Dividing Eq.~\eqref{eq:CG_lowered} by 
\begin{align}
    L \coloneqq \sqrt{\frac{(2j')!\,k!}{(2j' - k)!}},
\end{align}
converts it into the Clebsch-Gordan expansion of $\ket{j', m'}_{n}$; matching the right-hand side against Eq.~\eqref{eq:CG_expansion} reads off
\begin{align}
    c_{0}^{2} \;=\; \frac{(2j)!}{(2j - k)!}\cdot\frac{(2j' - k)!}{(2j')!}, 
    \quad\text{and}\quad 
    c_{1}^{2} \;=\; k\,\frac{(2j)!}{(2j - k + 1)!}\cdot\frac{(2j' - k)!}{(2j')!},
\end{align}
where $k!$ has cancelled from $c_{0}$, and the $r = 1$ term's $k^{2}\,(k-1)!/k! = k$ has been folded in for $c_{1}$. 
Three substitutions now collapse the factorials.

\begin{itemize}
  \item \textbf{First substitution.}
  Branch relation $2j = 2j' - 1$.
  The numerator and denominator factorials now differ by exactly one step, so each factorial ratio collapses to a single factor:
\begin{align}
    \frac{(2j)!}{(2j')!} \;=\; \frac{(2j' - 1)!}{(2j')!} \;=\; \frac{1}{2j'}, \quad\text{and}\quad
    \frac{(2j' - k)!}{(2j - k)!} \;=\; \frac{(2j' - k)!}{(2j' - k - 1)!} \;=\; 2j' - k.
\end{align}

\item \textbf{Second substitution.}
Cancellation of the $c_{1}$ denominator.
For $c_{1}$, the combination $2j - k + 1 = 2j' - k$ turns its remaining factorial ratio into an exact cancellation:
\begin{align}
    \frac{(2j' - k)!}{(2j - k + 1)!} \;=\; \frac{(2j' - k)!}{(2j' - k)!} \;=\; 1.
\end{align}
Collecting these two substitutions gives
\begin{align}
    c_{0}^{2} \;=\; \frac{2j' - k}{2j'}, \quad\text{and}\quad 
    c_{1}^{2} \;=\; \frac{k}{2j'}.
\end{align}

\item \textbf{Third substitution.}
$k = j' - m'$. 
From Step 2, $k \coloneqq j' - m'$, so $2j' - k = j' + m'$, giving
\begin{align}
    c_{0}^{2} \;=\; \frac{j' + m'}{2j'}, \quad\text{and}\quad c_{1}^{2} \;=\; \frac{j' - m'}{2j'}.
\end{align}
Taking the positive square root (the Condon-Shortley convention selected the $+$ root at every step of Eq.~\eqref{eq:single_multiplet_lowering}, so the signs inherit positively here),
\begin{align}
    c_{0}^{j' \leftarrow j}(m') \;=\; \sqrt{\frac{j' + m'}{2j'}}, \quad\text{and}\quad c_{1}^{j' \leftarrow j}(m') \;=\; \sqrt{\frac{j' - m'}{2j'}},
\end{align}
which is Eq.~\eqref{eq:CG_coeffs_up} for the branch with parent spin $j = j' - 1/2$.
\end{itemize}

\textbf{Step 4: Fixing the $j = j' + 1/2$ branch by orthogonality.}
This branch has parent spin $j = j' + 1/2$ (parent larger than child), and its candidate highest-weight vector is not uniquely pinned by $\hat{J}_{z}$ alone: at highest weight $\hat{J}_{z} = j - 1/2$, two product states share the eigenvalue; that are 
\begin{align}
    \ket{j, j - 1}_{n - 1} \otimes \ket{0}_{n},
    \quad\text{and}\quad
    \ket{j, j}_{n - 1} \otimes \ket{1}_{n}.
\end{align}
So no direct highest weight construction fixes its phase. 
We instead fix this branch by orthogonality to the branch already derived in Step 3.
The fixed-$\hat{J}_{z}$ subspace at generic $m'$. Holding the parent spin at $j$ and the total $\hat{J}_{z}$ at $m'$, only two product states remain available:
\begin{align*}
    \mathcal{V}_{m'} \;\coloneqq\; \operatorname{span}\bigl\{\,\ket{j, m' - \tfrac{1}{2}}_{n-1} \otimes \ket{0}_{n},\ \ \ket{j, m' + \tfrac{1}{2}}_{n-1} \otimes \ket{1}_{n}\,\bigr\}.
\end{align*}
Both coupled states $\ket{j + 1/2, m'}_{n}$ (child of the $j' = j + 1/2$ branch) and $\ket{j - 1/2, m'}_{n}$ (child of the $j' = j - 1/2$ branch) lie in $\mathcal{V}_{m'}$, and they are orthogonal because they belong to distinct $\hat{J}^{2}$-eigenvalues $(j + 1/2)(j + 3/2)$ versus $(j - 1/2)(j + 1/2)$. Since $\dim\mathcal{V}_{m'} = 2$, they span an orthonormal basis of $\mathcal{V}_{m'}$; the state $\ket{j - 1/2, m'}_{n}$ is therefore the \emph{unique} unit vector orthogonal to $\ket{j + 1/2, m'}_{n}$ (up to an overall sign).
The state $\ket{j + \tfrac{1}{2}, m'}_{n}$, re-expressed in terms of parent-$j$. 
Specialise Step 3 to parent $j$ and child $j' = j + 1/2$ (substitute $j' \to j + 1/2$ in Eq.~\eqref{eq:CG_coeffs_up}, so $2j' = 2j + 1$):
\begin{align}
    \ket{j + \tfrac{1}{2}, m'}_{n}
    \;=\; \sqrt{\frac{j + \tfrac{1}{2} + m'}{2j + 1}}\;\ket{j, m' - \tfrac{1}{2}}_{n-1} \otimes \ket{0}_{n} 
    + \sqrt{\frac{j + \frac{1}{2} - m'}{2j + 1}}\;\ket{j, m' + \tfrac{1}{2}}_{n-1} \otimes \ket{1}_{n}.
\end{align}
Write the state $\ket{j - 1/2, m'}_{n}$ in the same $\mathcal{V}_{m'}$ basis,
\begin{align}
    \ket{j - \tfrac{1}{2}, m'}_{n} \;=\; \alpha\,\ket{j, m' - \tfrac{1}{2}}_{n-1} \otimes \ket{0}_{n} + \beta\,\ket{j, m' + \tfrac{1}{2}}_{n-1} \otimes \ket{1}_{n}.
\end{align}
Orthogonality to $\ket{j + 1/2, m'}_{n}$ gives
\begin{align}
    \alpha\,\sqrt{\frac{j + \frac{1}{2} + m'}{2j + 1}} + \beta\,\sqrt{\frac{j + \frac{1}{2} - m'}{2j + 1}} \;=\; 0,
\end{align}
and combining with unit norm $|\alpha|^{2} + |\beta|^{2} = 1$ yields, up to an overall sign $t = \pm 1$,
\begin{align}
    \alpha \;=\; -\,t\,\sqrt{\frac{j + \frac{1}{2} - m'}{2j + 1}}, 
    \quad\text{and}\quad 
    \beta \;=\; t\,\sqrt{\frac{j + \frac{1}{2} + m'}{2j + 1}}.
\end{align}

The Condon-Shortley convention requires the overlap $\langle j_{1}, j_{1};\, j_{2}, J - j_{1} \mid J, J\rangle$ between the stretched product state (both factors at maximum weight) and the highest weight state of the coupled multiplet to be positive. 
Specialised to the $j = j' - 1/2$ branch at its own highest weight $m' = j - 1/2$, this selects $\beta > 0$ there, hence $t = +1$, so $\alpha$ is negative; the minus sign in $c_{0}$ of Eq.~\eqref{eq:CG_coeffs_down}.
Substitute the parent variable $j = j' + 1/2$ into the expressions above ($2j + 1 = 2j' + 2$; $j + 1/2 \pm m' = j' + 1 \pm m'$):
\begin{align}
    c_{0}^{j' \leftarrow j}(m') \;=\; \alpha \;=\; -\sqrt{\frac{j' - m' + 1}{2j' + 2}}, \quad\text{and}\quad 
    c_{1}^{j' \leftarrow j}(m') \;=\; \beta \;=\; \sqrt{\frac{j' + m' + 1}{2j' + 2}},
\end{align}
which is exactly Eq.~\eqref{eq:CG_coeffs_down}.

Collecting Steps 3 and 4, the Clebsch-Gordan weights for the two branches $j = j' - 1/2$ (parent spin smaller than child) and $j = j' + 1/2$ (parent spin larger than child) take the closed forms
\begin{align}\label{eq:CG_coeffs_up}
    j = j' - \frac{1}{2}:\quad
    c_{0}^{j' \leftarrow j}(m')
    = \sqrt{\frac{j' + m'}{2 j'}},
    \quad
    c_{1}^{j' \leftarrow j}(m')
    = \sqrt{\frac{j' - m'}{2 j'}},
\end{align}
and
\begin{align}\label{eq:CG_coeffs_down}
    j = j' + \frac{1}{2}:\quad
    c_{0}^{j' \leftarrow j}(m')
    = -\sqrt{\frac{j' - m' + 1}{2 j' + 2}},
    \quad
    c_{1}^{j' \leftarrow j}(m')
    = \sqrt{\frac{j' + m' + 1}{2 j' + 2}},
\end{align}
with the signs fixed by the Condon-Shortley convention.

Closed forms above reconstruct the singlet-triplet decomposition of Eq.~\eqref{eq:n2_change_of_basis}.
Setting $(j, j', m') = (1/2, 1, 0)$ in the $j = j' - 1/2$ branch, Eq.~\eqref{eq:CG_coeffs_up} gives $c_{0} = c_{1} = 1/2$, and Eq.~\eqref{eq:CG_expansion} reads
\begin{align}
    \ket{1, 0}_{2}
    \;=\;\frac{1}{\sqrt{2}}\,\ket{\tfrac{1}{2}, -\tfrac{1}{2}}_{1} \otimes \ket{0}_{2}
    \;+\;\frac{1}{\sqrt{2}}\,\ket{\tfrac{1}{2}, +\tfrac{1}{2}}_{1} \otimes \ket{1}_{2}
    \;=\;\frac{1}{\sqrt{2}}\left(\ket{10} + \ket{01}\right),
\end{align}
using the single-qubit identifications $\ket{1/2, +1/2}_{1} = \ket{0}$ and $\ket{1/2, -1/2}_{1} = \ket{1}$.
For the $j = j' + 1/2$ branch, $(j, j', m') = (1/2, 0, 0)$ in Eq.~\eqref{eq:CG_coeffs_down} gives $c_{0} = -\sqrt{1/2}$ and $c_{1} = +\sqrt{1/2}$, producing $\ket{0, 0}_{2} = 1/\sqrt{2}(\ket{01} - \ket{10})$.
Both match Eq.~\eqref{eq:n2_change_of_basis}, confirming that the closed forms reduce to the familiar $n = 2$ coefficients $\pm 1/\sqrt{2}$ without any external Clebsch-Gordan table.
The two branches are orthogonal as vectors, and the full set $\{\ket{j', m'}_{n}\}$ obtained by running Eq.~\eqref{eq:CG_expansion} from $j' = 1/2$ (the single-qubit starting case) up to $j' = n/2$ is an orthonormal basis of the coupled space. With closed-form Clebsch-Gordan coefficients now in hand, we will substitute them into Eq.~\eqref{eq:rho_n_to_rho_j} to obtain an explicit Clebsch-Gordan recursion at the level of matrix entries of the block $\rho_{j}^{(n)}$, and then plug the result back into the SDP formulated in Eq.~\ref{eq:PostP_Fundamental_Limit_Simplified}.

The $n$-qubit sector $V_{j'}$ picks up contributions from two parent sectors: the $+$ branch of parent $j = j' - 1/2$ and the $-$ branch of parent $j = j' + 1/2$, with multiplicities adding as 
\begin{align}
    d_{j'}^{(n)} = d_{j' - \frac{1}{2}}^{(n-1)} + d_{j' + \frac{1}{2}}^{(n-1)}.
\end{align}
The Schur-Weyl duality Eq.~\eqref{eq:rho_n_decomp} guarantees that the extracted block $\rho_{j'}^{(n)}$ is the same matrix for any choice of parent, so any single pathway suffices to compute it.
Writing 
\begin{align}
    R_{j'}^{(n)}[m_{1}', m_{2}'] \coloneqq (\rho_{j'}^{(n)})_{m_{1}', m_{2}'},
\end{align}
and fixing any single parent $j$ whose block $R_{j}^{(n-1)}$ is already available, the recursion reads
\begin{align}\label{eq:CG_recursion}
    R_{j'}^{(n)}[m_{1}', m_{2}']
    \;=\;
    \sum_{s_{1}, s_{2} \in \{0, 1\}}
    c_{s_{1}}^{j' \leftarrow j}(m_{1}')\,
    c_{s_{2}}^{j' \leftarrow j}(m_{2}')\;
    R_{j}^{(n-1)}\left[m_{1}' - \frac{1}{2} + s_{1},\; m_{2}' - \frac{1}{2} + s_{2}\right]\;
    \rho[s_{1}, s_{2}],
\end{align}
with $\rho[s_{1}, s_{2}] \coloneqq \langle s_{1} | \rho | s_{2} \rangle$ and the convention $R_{j}^{(n-1)} \equiv 0$ whenever an $m$-index falls outside $[-j, j]$.
The base case is $R_{1/2}^{(1)}[m_{1}, m_{2}] = \rho[s_{1}, s_{2}]$ with $s_{k} = 1/2 - m_{k}$.
All arithmetic is among $(2j + 1) \times (2j + 1)$ matrices with $j \leqslant n/2$; the ambient $(\mathbb{C}^{2})^{\otimes n}$ is never constructed.

Finally, we assess the computational complexity of the Clebsch-Gordan recursion. 
At level $n$, there are $\mO(n)$ admissible values of $j'$, each associated with a block of size $\mathcal{O}(n^{2})$, and each matrix element is obtained from a four-term summation. 
Accumulating the cost over all levels up to $n$ yields polynomial scaling in both time and memory for constructing the blocks from a single input state $\rho \in \mS$, denoted by $T_{\mathrm{CG}}(n)$ and $M_{\mathrm{CG}}(n)$, respectively.
\begin{align}\label{eq:CG_cost_per_state}
    T_{\mathrm{CG}}(n) \;=\; \mO(n^{4}),
    \qquad
    M_{\mathrm{CG}}(n) \;=\; \mO(n^{3}),
\end{align}
For an ensemble of $|\mS|$ input states, the total construction cost scales as $\mO(|\mS|\cdot n^{4})$, in sharp contrast to the $\mO(n!\cdot 4^{n})$ of the method based on Schur-Weyl duality introduced in Subsec.~\ref{subsec:SDP_Symmetry}.

We now show how the Clebsch-Gordan recursion solves the SDP formulated in Eq.~\eqref{eq:PostP_Fundamental_Limit_Simplified}, proceeding in four steps.
\begin{enumerate}
    \item[(i)] \textbf{Formulation of the noisy state.}
    For each $\psi_{i} \in \mS$, construct the corresponding noisy state $\sigma_{i} \coloneqq \mN(\psi_{i})$, where $\mN$ denotes the noise channel.
    \item[(ii)] \textbf{Clebsch-Gordan recursion.}
    Apply the Clebsch–Gordan recursion in Eq.~\eqref{eq:CG_recursion} to $\sigma_{i}^{\T}$, building up from $n=1$ using the closed-form coefficients in Eqs.~\eqref{eq:CG_coeffs_up}–\eqref{eq:CG_coeffs_down}. This yields, for each input state $\psi_{i}$ and each total spin sector $j$, the $j$-th block of $(\sigma_{i}^{\T})^{\otimes n}$.
    Note that in the qubit case, label $\lambda$ in Schur-Weyl duality (see Eq.~\eqref{eq:SW_decomp}) is in one-to-one correspondence with the total spin $j$.
    \item[(iii)] \textbf{Parallelization of the SDP.}
    Solving block SDPs in parallel across $\lambda$ sectors.
    \item[(iv)] \textbf{Assembling the blocks into the final solution.}
    Combine the blocks to form the ensemble-averaged objective for each $\lambda$, as prescribed by Eq.~\eqref{eq:Overlap_SW_block}.
\end{enumerate}
The construction scales as $\mO(|\mS|\cdot n^{4})$, while the SDP itself is solved with cost polynomial in the block dimension $m_{\lambda}d=\mO(n)$. 
As summarised in Tab.~\ref{tab:SDP_complexity}, this reduction enables simulations at system sizes of several tens of qubits, even on standard hardware.

To illustrate the recursion, consider the case $n=3$ with $\rho=\mathrm{diag}(p_{0},p_{1})$ as in Eq.~\eqref{eq:rho_diag_example}. 
The $n=2$ parent blocks are 
\begin{align}
    R_{0}^{(2)}=(p_{0}p_{1}),
\end{align}
for the singlet and 
\begin{align}
    R_1^{(2)} = 
    \begin{pmatrix}
        p_0^2 & 0 & 0 \\
        0 & p_0 p_1 & 0 \\
        0 & 0 & p_1^2
    \end{pmatrix},
\end{align}
for the triplet.
We first construct the $j'=1/2$ block at $n=3$. Starting from the parent sector $j=0$, Eq.~\eqref{eq:CG_coeffs_up} yields 
\begin{align}
    c_{0}(m'=+\frac{1}{2})=1,
\end{align}
and
\begin{align}
    c_{1}(m'=+\frac{1}{2})=0,
\end{align}
so that only the contribution with $s_{1}=s_{2}=0$ survives; that is
\begin{align}
    R_{\frac{1}{2}}^{(3)}[+\tfrac{1}{2}, +\tfrac{1}{2}]
    \;=\; R_{0}^{(2)}[0, 0] \cdot \rho[0, 0]
    \;=\; (p_{0} p_{1}) \cdot p_{0}
    \;=\; p_{0}^{2} p_{1},
\end{align}
and analogously 
\begin{align}
    R_{\frac{1}{2}}^{(3)}[-\tfrac{1}{2}, -\tfrac{1}{2}] = p_{0} p_{1}^{2},
\end{align}
giving 
\begin{align}
    \rho_{\frac{1}{2}}^{(3)} = 
    \begin{pmatrix}
        p_0^2 p_1 & 0 \\
        0 & p_0 p_1^2
    \end{pmatrix}.
\end{align}
As a cross-check, parent $j = 1$ gives 
\begin{align}
    c_{0} = -\sqrt{\frac{1}{3}},
\end{align}
and
\begin{align}
    c_{1} = \sqrt{\frac{2}{3}},
\end{align}
and the surviving $s_{1} = s_{2}$ terms sum to
\begin{align}
    \frac{1}{3}(p_{0} p_{1}) p_{0} + \frac{2}{3} p_{0}^{2} p_{1} = p_{0}^{2} p_{1},
\end{align}
matching the parent-$(j = 0)$ result as required by Schur's lemma.

The disappearance of the factorial overhead reflects a shift in perspective rather than a simplification of the combinatorics. 
In the method based on Schur-Weyl duality (see Eq.~\eqref{eq:SW_decomp}), one resolves $\Pi_{j}\equiv\Pi_{\lambda}$ (see Eq.~\eqref{eq:Pi_lambda}) by explicitly summing over all permutations, incurring an $\mO(n!)$ cost. 
The Clebsch-Gordan recursion reaches the same Schur-Weyl blocks $V_{j}\otimes W_{j}\equiv V_{\lambda}\otimes W_{\lambda}$ without ever enumerating $\mathfrak{S}_{n}$: it proceeds inductively, adding one qubit at a time and invoking only the local branching rule in Eq.~\eqref{eq:branching_rule}.
This efficiency is a direct manifestation of Schur-Weyl duality. 
The $\mathfrak{S}_{n}$-invariant and $\mathrm{SU}(2)$-covariant viewpoints encode an identical block structure, yet expose it in fundamentally different ways. 
The former resolves the symmetry globally through permutation sums; the latter accesses it locally through angular momentum coupling. 
By following the $\mathrm{SU}(2)$ route, the Clebsch-Gordan recursion reconstructs the same $\mathfrak{S}_{n}$ projectors while bypassing the factorial complexity entirely.

This subsection introduces the Clebsch-Gordan recursion as the key reduction that converts an intractable optimization into a practically solvable one. 
A direct treatment of the purification SDP in Eq.~\eqref{eq:PostP_Fundamental_Limit_Simplified} entails a memory cost 
\begin{align}
    \mO(N^{4})=\mO((d^{n+1})^{4}),
\end{align}
which is already prohibitive at modest system sizes. 
In the qubit setting, i.e., $d=2$, this scaling becomes
\begin{align}
    \mO(4^{2(n+1)}),
\end{align}
and rapidly becomes prohibitive, confining brute-force approaches to at most $n \approx 7$ copies. 
In the distributed setting with $d=4$, the constraint is even more severe, with practical limits around $n \approx 4$.
Exploiting the permutation invariance of $\rho^{\otimes n}$ reorganizes this exponential scale problem into $\mO(n)$ independent block SDPs, one for each Young diagram label $\lambda$, as prescribed by the decomposition in Eq.~\eqref{eq:Overlap_SW_block}. 
Building the per-state blocks $\rho_{j}^{(n)}$ --- the construction bottleneck of the Schur-Weyl duality --- was itself trimmed from 
\begin{align}
    \mO(n!\,4^{n})
\end{align}
via Eq.~\eqref{eq:Pi_lambda} of Subsec.~\ref{subsec:SDP_Symmetry} to 
\begin{align}
    \mO(n^{4})
\end{align}
through the Clebsch-Gordan recursion of Eq.~\eqref{eq:CG_recursion}.
On the other hand, the progression of memory complexity reduction achieved by each method is summarised as follows.
\begin{align}
    \underbrace{\mO(4^{n+1})}_{\text{Brute-Force Evaluation}}
    \longrightarrow
    \underbrace{\mO(4^{n})}_{\text{Schur-Weyl Duality}}
    \longrightarrow
    \underbrace{\mO(n^{3})}_{\text{Clebsch-Gordan Recursion}}.
\end{align}
In this form, the purification problem becomes accessible on standard hardware, extending the reachable regime to system sizes well into the tens, as quantified by the benchmarks in Tab.~\ref{tab:SDP_complexity}.

A close inspection may suggest that the size $|\mS|$ of the initial state enters the time complexity of the Clebsch-Gordan recursion, yet is absent from the Schur-Weyl-duality-based approach. 
This distinction is only apparent. 
When the scaling is expressed explicitly, the dependence on $|\mS|$ is retained. 
For the Schur-Weyl duality method introduced in Subsec.~\ref{subsec:SDP_Symmetry}, the time complexity takes the form
\begin{align}
    \mO\left(4^{n}\left(n!+4|\mS|+4n^2\right)\right).
\end{align}
For large $n$, the factorial term $n!$ dominates the scaling. 
Subleading contributions proportional to $|\mS|$ and $n^2$ are therefore neglected, yielding the expression $\mO(n!\, 4^{n})$ shown in Tab.~\ref{tab:SDP_complexity}.

\begin{table}[htbp]
    \centering
    \begin{tblr}{
      colspec = {l || c | c | c},
      row{1} = {bg=gray!50, font=\bfseries}, 
      row{2} = {bg=gray!50, font=\bfseries},
      column{1} = {bg=gray!10, font=\bfseries},
      hlines,
      vlines,
      cells = {m, c},
      cell{3}{2-4} = {bg=mOrange},
      cell{4}{2-4} = {bg=mBlue},
      cell{5}{2-4} = {bg=mGreen},
      cell{6}{2} = {bg=mRed}, cell{6}{3-4} = {bg=mGreen},
      cell{7}{2-3} = {bg=mRed}, cell{7}{4} = {bg=mGreen},
      cell{8}{2-3} = {bg=mRed}, cell{8}{4} = {bg=mGreen},
    }
      \SetCell[c=4]{c} Scaling of Computational Complexity & & & \\
      Metrics & Brute-Force (Tab~\ref{tab:SDP_Complexity}) & Schur-Weyl (Subsec.~\eqref{subsec:SDP_Symmetry}) & Clebsch-Gordan (Subsec.~\eqref{subsec:SDP_CG})  \\
      Time & $\mO(\mathrm{poly}(2^{n+1}))$ & $\mO(n!\, 4^{n})$ & $\mO(|\mS|\,n^{4})$ \\
      Memory & $\mO(4^{n+1})$ & $\mO(4^{n})$ & $\mO(n^{3})$ \\
      $n=6$ & 55.8 GB, 419 s & 0.4 GB, <1 s & 0.2 GB, <1 s \\
      $n=20$ & infeasible & 3.4 GB, 241 s & 1.9 GB, 15 s \\
      $n=25$ & infeasible & >124 GB (OOM) & 4.9 GB, 38 s \\
      $n=30$ & infeasible & infeasible & 11.3 GB, 86 s \\
    \end{tblr}
    \caption{\textbf{Symmetry-Reduced Scaling of Purification SDP in Eq.~\eqref{eq:PostP_Fundamental_Limit_Simplified}}. 
    Scaling of computational cost for the $n$-to-1 qubit purification SDP ($d=2$, $|\mS|=2$, $p{=}0.5$, single core). 
    The naive formulation becomes rapidly infeasible due to exponential memory growth. 
    Schur-Weyl reduction compresses the problem into symmetry-adapted blocks, substantially lowering resource demands, while Clebsch-Gordan recursion further reduces the construction cost to polynomial scaling. 
    Benchmarks demonstrate tractable performance up to $n=30$.
    Here, the abbreviation OOM denotes an out-of-memory error.
    Time and memory refer to time complexity and memory complexity, respectively.
    }
    \label{tab:SDP_complexity}
\end{table}

In this work, computational tasks were performed on a high-performance workstation equipped with an AMD Ryzen Threadripper PRO 7975WX processor, featuring 32 physical cores and 64 logical threads based on the Zen 4 architecture. 
The system operates at a maximum boost frequency of 5.36 GHz and utilizes a multi-tier cache hierarchy, including 128 MiB of L3 cache shared across CCX groups. 
To accelerate specialized workloads, the CPU supports advanced instruction sets such as AVX-512 (including VNNI for AI acceleration) and BF16. 
The hardware configuration is supported by 128 GB of physical RAM (124 GiB usable), ensuring sufficient memory bandwidth and capacity for data-intensive processing.

Finally, the Clebsch-Gordan (CG) recursion is distilled into an explicit algorithm, presented at the end of this subsection. 
In conjunction with the four-step pipeline (i)-(iv), Eq.~\eqref{eq:CG_recursion} gives rise to Alg.~\ref{alg:CG_pipeline}, which realizes the $d=2$ counterpart of the general-$d$ character-formula pipeline in Alg.~\ref{alg:char_pipeline}. 
This qubit structure admits a further simplification of the associated SDP, both in formulation and computational cost.
At its core, \textsc{CG-Recursion} replaces \textsc{Char-Reduce}. 
Rather than evaluating the $\mO(n!\,d^{2n})$ character formula, it exploits angular-momentum addition to decompose $\rho^{\otimes n}$ directly into its spin-$j$ sectors, generating all blocks at a cost $\mO(n^{4})$ per state, without constructing the $2^{n}\times m_{j}$ isometry $\Phi_{j}^{(n)}$ (see Alg.~\ref{alg:char_pipeline}).
The outer routine, \textsc{Clebsch-Gordan-Reduction}, retains the structure of Alg.~\ref{alg:char_pipeline} specialised to qubits:
it assembles the sector-resolved operators $\Xi_{j}\equiv\Xi_{\lambda}$, solves the resulting block SDPs independently, and aggregates the solution through the multiplicity-weighted sum $F_{\mathrm{PostP}}=\sum_{j} d_{j} F_{j}$ (see Eq.~\eqref{eq:Weighted_Sum}).

\begin{algorithm}[htbp]
\DontPrintSemicolon
\SetKwProg{Fn}{Function}{}{end}
\KwIn{target-state ensemble $\mS = \{\psi_{i}\}_{i=1}^{|\mS|}$ of pure qubit states; noise channel $\mN$; copy count $n$.}
\KwOut{optimal post-processing fidelity $F_{\mathrm{PostP}}$, the $n$-to-1 SDP optimum of Eq.~\eqref{eq:PostP_Fundamental_Limit_Simplified}.}
\BlankLine

\Fn{\normalfont\textsc{CG-Recursion}($\rho,\, n$)}{
    \tcc*[h]{Input: single-qubit $\rho\in\mathbb{C}^{2\times 2}$ and level $n$. Output: reduced blocks $\bigl\{R_{j}^{(n)}\bigr\}_{j}$, where $R_{j}^{(k)} \coloneqq \Phi_{j}^{(k)\,\dagger}\,\rho^{\otimes k}\,\Phi_{j}^{(k)} \in \mathbb{C}^{m_{j}\times m_{j}}$ (with $m_{j} = 2j + 1$) is the spin-$j$ block of $\rho^{\otimes k}$ in the coupled basis, and $\Phi_{j}^{(k)}$ is the $k$-qubit counterpart of the isometry constructed in Alg.~\ref{alg:char_pipeline}. Magnetic labels run over $m_{1}, m_{2} \in \{-j, -j+1, \ldots, +j\}$, and $\rho[s_{1}, s_{2}] \coloneqq \langle s_{1}|\rho|s_{2}\rangle$ for $s_{i} \in \{0, 1\}$. The recursion evolves only the blocks $R_{j}^{(k)}$; the isometry $\Phi_{j}^{(k)}$ is never materialised.}\;
    $R_{1/2}^{(1)}[m_{1}, m_{2}] \leftarrow \rho\bigl[\tfrac{1}{2} - m_{1},\; \tfrac{1}{2} - m_{2}\bigr]$ for $m_{1}, m_{2} \in \bigl\{-\tfrac{1}{2}, +\tfrac{1}{2}\bigr\}$ \tcp*{base case ($k = 1$), using the identification $\ket{0} \leftrightarrow \ket{\tfrac{1}{2}, +\tfrac{1}{2}}$ and $\ket{1} \leftrightarrow \ket{\tfrac{1}{2}, -\tfrac{1}{2}}$, i.e.\ $s = \tfrac{1}{2} - m$}
    \For{$k = 2$ \KwTo $n$}{
        \ForEach{admissible total spin $j'$ at level $k$}{
            pick $j = j' + \tfrac{1}{2}$ if $j' = 0$; otherwise either parent $j \in \{j' - \tfrac{1}{2},\, j' + \tfrac{1}{2}\}\cap\mathbb{R}_{\geqslant 0}$ works \tcp*{both choices yield the same $R_{j'}^{(k)}$ by Schur's lemma; at $j' = 0$ only $j = \tfrac{1}{2}$ exists, and Eq.~\eqref{eq:CG_coeffs_up} carries a formal $1/(2j')$ singularity that must be avoided---use Eq.~\eqref{eq:CG_coeffs_down}}
            \ForEach{magnetic pair $(m_{1}', m_{2}')$ with $m_{i}'\in\{-j', -j'+1, \ldots, j'\}$}{
                $R_{j'}^{(k)}[m_{1}', m_{2}'] \leftarrow \displaystyle\sum_{s_{1}, s_{2} \in \{0,1\}} c_{s_{1}}^{j' \leftarrow j}(m_{1}')\, c_{s_{2}}^{j' \leftarrow j}(m_{2}')\, R_{j}^{(k-1)}\bigl[m_{1}' - \tfrac{1}{2} + s_{1},\; m_{2}' - \tfrac{1}{2} + s_{2}\bigr]\, \rho[s_{1}, s_{2}]$ \tcp*{Clebsch--Gordan amplitudes $c_{s}^{j'\leftarrow j}(m')$ from Eqs.~\eqref{eq:CG_coeffs_up}--\eqref{eq:CG_coeffs_down} (real in Condon--Shortley convention); any $R_{j}^{(k-1)}$ index falling outside $[-j, j]$ contributes zero; Eq.~\eqref{eq:CG_recursion}}
            }
        }
    }
    \KwRet $\bigl\{R_{j}^{(n)}\bigr\}_{j}$\;
}

\BlankLine
\tcc*[h]{\textbf{Step (i)--(ii).} Per-state noisy reduced blocks. The transpose $\sigma_{i}^{\T}$ enters because the SDP objective of Eq.~\eqref{eq:PostP_Fundamental_Limit_Simplified} carries $\bigl(\mN(\psi_{i})^{\T}\bigr)^{\otimes n}$.}\;
\ForEach{$\psi_{i} \in \mS$}{
    $\sigma_{i} \leftarrow \mN(\psi_{i})$ \tcp*{noisy single-copy state, $\sigma_{i}\in\mathbb{C}^{2\times 2}$}
    $\bigl\{\rho_{j, i}^{(n)}\bigr\}_{j} \leftarrow \textsc{CG-Recursion}\bigl(\sigma_{i}^{\T},\, n\bigr)$ \tcp*{$\rho_{j, i}^{(n)}$: spin-$j$ reduced block of $(\sigma_{i}^{\T})^{\otimes n}$, of size $m_{j}\times m_{j}$}
}

\BlankLine
\tcc*[h]{\textbf{Step (iii).} Assemble the ensemble-averaged operator on $V_{j}\otimes D$, where $D$ is the single-qubit target ($d = 2$). The trivial action of $\rho^{\otimes n}$ on the multiplicity factor $W_{j}$ (Eq.~\eqref{eq:rho_n_decomp}) drops $W_{j}$ from the SDP variables entirely.}\;
\ForEach{admissible $j$ at level $n$}{
    $\Xi_{j} \leftarrow \dfrac{1}{|\mS|}\displaystyle\sum_{i = 1}^{|\mS|} \rho_{j, i}^{(n)} \otimes \psi_{i, D}$ \tcp*{$\psi_{i, D} \equiv \ket{\psi_{i}}\!\bra{\psi_{i}}$ on $D$; $\Xi_{j}$ is a $2 m_{j} \times 2 m_{j}$ Hermitian operator on $V_{j}\otimes D$}
}

\BlankLine
\tcc*[h]{\textbf{Step (iv).} The bundled SDP $\max_{\{J_{j}\}}\sum_{j} d_{j}\,\Tr[J_{j}\Xi_{j}]$ subject to $J_{j}\geqslant 0$ and $\Tr_{D}[J_{j}] = \1_{V_{j}}$ carries no cross-$j$ constraints (Eq.~\eqref{eq:Overlap_SW_block}); it therefore separates into $\mO(n)$ independent block SDPs of size $\mO(n)$ and additively gives $F_{\mathrm{PostP}} = \sum_{j} d_{j}\,F_{j}$.}\;
\ForEach{admissible $j$ \textup{(independent; parallelisable)}}{
    $F_{j} \leftarrow \displaystyle\max_{\substack{J_{j} \geqslant 0 \\ \Tr_{D}[J_{j}] = \1_{V_{j}}}} \Tr[\,J_{j}\,\Omega_{j}\,]$ \tcp*{$J_{j}$: spin-$j$ block of the post-processing Choi operator on $V_{j}\otimes D$; the TP condition descends from $\Tr_{D}\bigl[J^{\theta^{\mathrm{Post}}}\bigr] = \1_{C}$ restricted to the $j$-sector}
}
\KwRet $F_{\mathrm{PostP}} \leftarrow \displaystyle\sum_{j} d_{j}\, F_{j}$ \tcp*{each sector contributes with its multiplicity $d_{j}$ (Eq.~\eqref{eq:Overlap_SW_block})}
\caption{\textsc{Clebsch-Gordan-Reduction}: 
End-to-end Clebsch-Gordan pipeline for evaluating the $n$-to-1 qubit purification SDP in Eq.~\eqref{eq:PostP_Fundamental_Limit_Simplified}.
The construction scales as $\mO(|\mS|\,n^{4})$ in time and $\mO(n^{3})$ in memory per input state. 
The resulting block-structured SDP decomposes into independent spin-$j$ sectors of dimension $m_{j}=\mO(n)$, each solvable in polynomial time.}
\label{alg:CG_pipeline}
\end{algorithm}


\section{Advantages in Distributed Quantum State Purification}\label{sec:Advantages_DQSP}

The preceding analysis established that the spatiotemporal framework yields a clear advantage in global purification, improving both attainable fidelity and sample efficiency. 
Here, the framework is extended to distributed quantum state purification, in which bipartite states are shared between remote agents. 
This setting is foundational: entanglement distillation is recovered as a special case, linking the present approach directly to central tasks in quantum communication and computation.
Subsection~\ref{subsec:Where_PreP_Matters} delineates the operational regimes in which pre-processing enhances purification, and those in which it becomes ineffective under amplitude damping noise. Subsection~\ref{subsec:When_Noise_Matters} turns to depolarizing noise, showing that pre-processing inserted between the entangling operation and the noise offers no improvement when restricted to local unital maps.
These limitations sharpen the role of genuinely spatiotemporal strategies. Subsection~\ref{subsec:LU_Pre_Processing} introduces forward-assisted (FA) protocols tailored to the distributed setting and demonstrates a clear performance separation from conventional schemes.
The advantage extends beyond fidelity: pre-processing also reduces resource requirements. In some regimes, single-copy pre-processing outperforms multi-copy post-processing, including 4-to-1 protocols, indicating a substantial improvement in sample efficiency that mirrors the global case.
The present analysis reaches up to 4-to-1 distributed purification under current computational constraints; this bound reflects numerical tractability rather than a limitation of the FA framework itself.
A further step is taken in Subsec.~\ref{subsec:Beyond_NP_Theorem}, where FA purification, through the inclusion of pre-processing, circumvents established no-purification constraints, including those for Bell states. 
Detailed numerical analyses throughout this section substantiate these results and delineate the regimes in which the advantage emerges.


\subsection{Where Pre-Processing Matters: Activation of Advantage}\label{subsec:Where_PreP_Matters}

This work focuses on a restricted class of pre-processing (PreP) protocols, namely those consisting solely of local operations, which are readily implementable in experimental settings~\cite{glc7-xy8t}. 
Two central questions will be addressed: whether pre-processing can provide an advantage, and whether its placement influences the achievable performance. 
Both are answered in the affirmative. 
To elucidate the underlying mechanism, the discussion is specialized to single-shot entanglement distillation --- a particular instance of distributed purification --- which provides a minimal setting for isolating the role of PreP and sets the stage for the subsequent analysis.

As a concrete setting, consider the preparation of a maximally entangled state $\phi^{+}$, followed by degradation under amplitude damping (AD) noise (see Tab.~\ref{tab:Quantum_Noise_Models}). 
Two natural placements of pre-processing arise: it may be applied prior to the entangling operation (see Fig.~\ref{fig:PreP_Locations}(a)), as in~\cite{glc7-xy8t}, or inserted immediately before the noisy channels (see Fig.~\ref{fig:PreP_Locations}(b)), as illustrated in Fig.~\ref{fig:PreP_Locations}.

\begin{figure}[htbp]
    \centering   
    \includegraphics[width=1\textwidth]{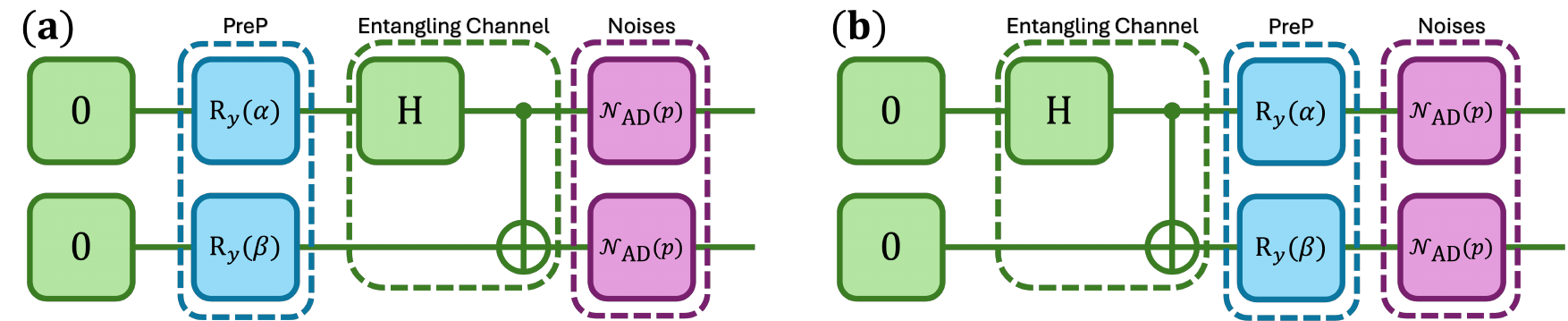}
    \caption{\textbf{Pre-Processing-Augmented Entanglement Distillation}. 
        In (a), pre-processing consists of local rotations about the $y$-axis applied prior to the entangling operation $\mathrm{CNOT}\circ(\mathrm{H}\otimes\id)H$. 
        In (b), the same local rotations are applied after the entangling operation $\mathrm{CNOT}\circ(\mathrm{H}\otimes\id)$ but immediately before the noisy channels.
    }
    \label{fig:PreP_Locations}
\end{figure}

A key distinction between the two PreP placements in Fig.~\ref{fig:PreP_Locations} is their operational relevance to purification. 
When PreP is applied as in Fig.~\ref{fig:PreP_Locations}(b), it is ineffective: despite being present, it does not improve the purification task, or equivalently, entanglement distillation. 
By contrast, the placement in Fig.~\ref{fig:PreP_Locations}(a) can yield a genuine advantage. 
To understand this difference, we now examine the two cases in detail.

For clarity, attention is restricted to the case where the local noise consists of identical AD channels with parameter $p$. 
To establish that, in Fig.~\ref{fig:PreP_Locations}(b), pre-processing via $y$-axis rotations does not improve upon the case without pre-processing, we evaluate the fidelity of the noisy Bell states shown below 

\begin{figure}[htbp]
    \centering   
    \includegraphics[width=1\textwidth]{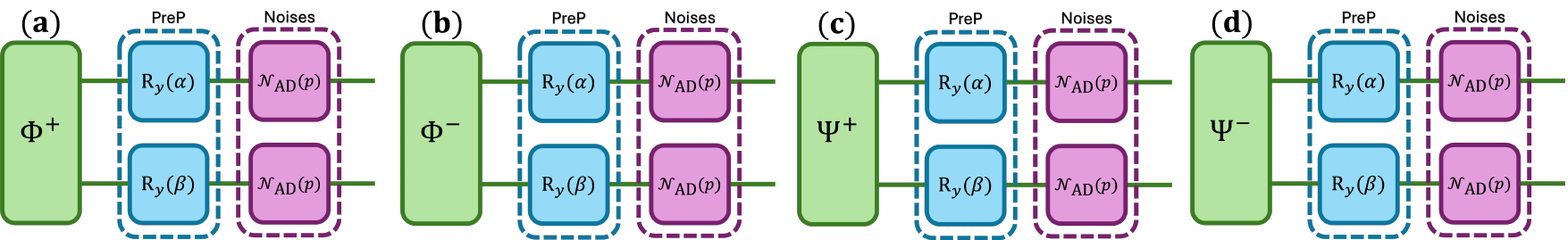}
    \caption{\textbf{Noisy Bell States}. 
        Bell states defined in Subsec.~\ref{subsec:Bell_States} are subjected to amplitude damping (AD) noise. 
        A pre-processing (PreP) operation is applied prior to the noise, which can be interpreted as an encoding step in quantum error correction (QEC). 
        The resulting fidelities corresponding to the configurations in (a)–(d) are denoted by $F_{\mathrm{Bell}, a}, \ldots, F_{\mathrm{Bell}, d}$.
    }
    \label{fig:Noisy_Bell_States}
\end{figure}

When a single Bell state $\Phi^{+}$ is considered, it is denoted by $\phi^{+}$ throughout this work to represent the maximally entangled state (MES).
The fidelities $F_{\mathrm{Bell},a}, F_{\mathrm{Bell},b}, F_{\mathrm{Bell},c}, F_{\mathrm{Bell},d}$ of protocols shown in Fig.~\ref{fig:Noisy_Bell_States} are given in the following lemma.

\begin{mylem}{Fidelities of Noisy Bell States in Fig.~\ref{fig:Noisy_Bell_States}}{Fidelities_Noisy_Bell_States}
    For pre-processing–augmented (PreP) generation of Bell states under amplitude damping (AD) noise, as illustrated in Fig.~\ref{fig:Noisy_Bell_States}, the corresponding fidelities are given by
    \begin{align}
        F_{\mathrm{Bell},a}(\alpha,\beta)=
        &\frac{1}{2}
        \left[
            p + (1-p)(2-p)\cos^{2}\!\left(\frac{\beta-\alpha}{2}\right)
        \right],\label{eq:Fidelity_Noisy_Bell_1}\\
        F_{\mathrm{Bell},b}(\alpha,\beta)=
        &\frac{1}{2}
        \left[
            p + (1-p)(2-p)\cos^{2}\!\left(\frac{\alpha+\beta}{2}\right)
        \right],\label{eq:Fidelity_Noisy_Bell_2}\\
        F_{\mathrm{Bell},c}(\alpha,\beta)=
        &\frac{1-p}{2}
        \left[
            p + (2-p)\cos^{2}\!\left(\frac{\alpha+\beta}{2}\right)
        \right],\label{eq:Fidelity_Noisy_Bell_3}\\
        F_{\mathrm{Bell},d}(\alpha,\beta)=
        &\frac{1-p}{2}
        \left[
            p + (2-p)\cos^{2}\!\left(\frac{\beta-\alpha}{2}\right)
        \right],\label{eq:Fidelity_Noisy_Bell_4}
    \end{align}
    where $\alpha$ and $\beta$ parameterize the local rotations in the pre-processing stage, and $p$ denotes the noise strength.
\end{mylem}

The results follow from direct calculation and are therefore omitted. 
Based on these expressions in Lem.~\ref{lem:Fidelities_Noisy_Bell_States}, it is straightforward to verify that all the noisy Bell states satisfy
\begin{align}
    F_{\mathrm{Bell},i}(\alpha,\beta)\leqslant
    F_{\mathrm{Bell},i}(0,0),
    \quad\forall\,\,\,
    i\in\{a, b, c, d\}.
\end{align}
This yields the following theorem.

\begin{mythm}{Ineffectiveness of Local-Unitary Pre-Processing}{Ineffectiveness_LU_PreP}
    For single-shot pre-processing (PreP) purification, or equivalently entanglement distillation, as considered in Fig.~\ref{fig:Noisy_Bell_States}, when PreP consists solely of local $y$-axis rotations, i.e., $\mathrm{R}_{y}(\alpha)\otimes\mathrm{R}_{y}(\beta)$, applied after the preparation of Bell states (see Eq.~\eqref{eq:Bell_Set}) but prior to the noise, $\mN_{\mathrm{AD}}(p)\otimes\mN_{\mathrm{AD}}(p)$, no improvement in the final fidelity can be achieved. 
    In other words, local-unitary (LU) PreP is ineffective in this setting.
\end{mythm}

As a direct consequence of Thm.~\ref{thm:Ineffectiveness_LU_PreP}, we obtain the following corollary.

\begin{mycor}{Ineffectiveness of Pre-Processing in Fig.~\ref{fig:PreP_Locations}(b)}{Ineffectiveness_LU_PreP_4b}
    The pre-processing stage in the protocol of Fig.~\ref{fig:PreP_Locations}(b) is ineffective: it does not improve the fidelity and offers no advantage for purification.
\end{mycor}

It is then natural to ask whether pre-processing-augmented purification in Fig.~\ref{fig:PreP_Locations}(a) can be effective. 
To address this question, the analysis is restricted to a minimal setting by setting $\beta = 0$, such that pre-processing is applied only on system $A$ (Alice's side). 
Under this choice, the states transform as follows
\begin{align}
    \ket{00}
    \xrightarrow{\mathrm{R}_{y}(\alpha)\otimes\mathrm{R}_{y}(0)}
    \frac{1}{\sqrt{2}}\left(\cos{\frac{\alpha}{2}}+\sin{\frac{\alpha}{2}}\right)\ket{00}
    +
    \frac{1}{\sqrt{2}}\left(\cos{\frac{\alpha}{2}}-\sin{\frac{\alpha}{2}}\right)\ket{11}.
\end{align}
Following the subsequent entangling operation $\mathrm{CNOT}\circ(\mathrm{H}\otimes\id)$ and the action of noise $\mathrm{R}_{y}(\alpha)\otimes\mathrm{R}_{y}(\beta)$, the fidelity of the final state is characterized by

\begin{mylem}{Fidelity of Protocol in Fig.~\ref{fig:PreP_Locations}(a)}{Fidelity_LU_PreP_4a}
    For pre-processing-augmented (PreP) generation of a maximally entangled state $\phi^{+}$ under amplitude damping (AD) noise $\mN_{\mathrm{AD}}(p)\otimes\mN_{\mathrm{AD}}(p)$, as illustrated in Fig.~\ref{fig:PreP_Locations}(a) with $\beta = 0$, the resulting fidelity $F_{\ref{fig:PreP_Locations}(a)}(\alpha)$ is given by
    \begin{align}
        F_{\ref{fig:PreP_Locations}(a)}(\alpha)=
        \frac{1 - p + p^{2}}{2}
        +
        \frac{(1 - p)\sqrt{1 + p^{2}}}{2}
        \cos\!\left(\alpha - \alpha^{*}\right),
    \end{align}
    where the quantity $\alpha^{*}$ is defined as $\alpha^{*} = \arctan(p)$.
\end{mylem}

The maximal performance achievable with PreP in Fig.~\ref{fig:PreP_Locations}(a) is given by
\begin{align}
    \max_{\alpha} F_{\ref{fig:PreP_Locations}(a)}(\alpha)=
    \frac{1}{2}
    \left(
        1 - p + p^{2} + (1 - p)\sqrt{1 + p^{2}}
    \right).
\end{align}
In the absence of PreP, i.e., for $\alpha = 0$, the fidelity reduces to
\begin{align}
    F_{\ref{fig:PreP_Locations}(a)}(0)=
    1 - p + \frac{p^{2}}{2}.
\end{align}
Their difference
\begin{align}
    \max_{\alpha} F_{\ref{fig:PreP_Locations}(a)}(\alpha)-F_{\ref{fig:PreP_Locations}(a)}(0)
    =
    \frac{1 - p}{2}
    \left(
        \sqrt{1 + p^{2}} - 1
    \right)
    > 0,
    \quad\forall\,\,\, 
    p \in (0,1),
\end{align}
shows that a local $y$-axis rotation $\mathrm{R}_{y}(\alpha)\otimes\1$ on a single subsystem enhances the fidelity.

\begin{mycor}{Effectiveness of Pre-Processing in Fig.~\ref{fig:PreP_Locations}(a)}{Effectiveness_LU_PreP_4a}
    The pre-processing stage in the protocol of Fig.~\ref{fig:PreP_Locations}(a) is effective: there exist local unitary (LU) PreP operations that enhance purification, more precisely, single-shot entanglement distillation in this setting.
\end{mycor}

Corollary~\ref{cor:Effectiveness_LU_PreP_4a} captures the mechanism underlying the advantage demonstrated in recent quantum optical experiments on pre-processing–augmented single-shot entanglement distillation~\cite{glc7-xy8t}.
The preceding analysis in this subsection resolves the questions posed at the outset. 
Pre-processing can yield a genuine advantage in purification, as demonstrated by the protocol in Fig.~\ref{fig:PreP_Locations}(a). 
Crucially, its placement matters: 
although Fig.~\ref{fig:PreP_Locations}(a) is effective, the protocol in Fig.~\ref{fig:PreP_Locations}(b) is not.

A clear operational picture now emerges. 
Local-unitary pre-processing is not intrinsically beneficial; its impact is entirely governed by where it intervenes in the protocol. 
When inserted after entanglement generation, immediately before the noisy channels (see Fig.~\ref{fig:PreP_Locations}(b)), it is functionally redundant: 
the fidelity expressions of Lem.~\ref{lem:Fidelities_Noisy_Bell_States} certify that no improvement is possible, and Thm.~\ref{thm:Ineffectiveness_LU_PreP} together with Cor.~\ref{cor:Ineffectiveness_LU_PreP_4b} formalize this as a strict no-advantage statement. 
In contrast, when applied prior to the entangling operation (see Fig.~\ref{fig:PreP_Locations}(a)), pre-processing reshapes the state in a manner that is not erased by the subsequent dynamics. 
Lemma~\ref{lem:Fidelity_LU_PreP_4a} and Cor.~\ref{cor:Effectiveness_LU_PreP_4a} show that this early intervention produces a strictly positive gain for nontrivial noise strengths, with a well-defined optimal rotation that aligns the state against amplitude damping.
Pre-processing is effective only when it precedes the stage at which correlations are established, thereby modifying how entanglement is created and subsequently degraded; once entanglement is fixed, local rotations cannot recover what the noise removes. 
This identifies pre-processing as a stage-sensitive resource whose utility derives from its ability to bias the formation of entanglement, rather than to repair it after the fact.
These results also provide a concrete design principle for forward-assisted protocols --- namely, that useful interventions must be positioned upstream of both entanglement generation and noise --- and set the conceptual foundation for the more general spatiotemporal strategies developed in the following sections.


\subsection{When Noise Matters: Loss of Advantage}\label{subsec:When_Noise_Matters}

The previous subsection focused on amplitude damping (AD) noise. 
An immediate question is whether qualitatively different behaviour arises under other noise models, such as depolarizing noise, and whether the limitations of local-unitary (LU) pre-processing persist. 
The analysis below shows that LU pre-processing applied after the entangling gate and before the depolarizing channel does not confer any advantage for purification.

A key property underlying the analysis is the commutation of a quantum channel with unitary operations, known as unitary covariance (UC), defined as follows

\begin{mydef}{Unitary Covariance (UC)}{Unitary_Covariance}
    A quantum channel $\mE$ is said to be unitary covariant (UC) if it commutes with an arbitrary unitary operation $U$, that is,
    \begin{align}
    [\mE, \mU] = 0,
    \end{align}
    where the unitary channel $\mU$ acts as $\mU(X) = U X U^\dagger$.
\end{mydef}

Unitary covariance captures a fundamental symmetry of quantum dynamics, reflecting invariance under arbitrary changes of basis. 
It has been extensively studied in quantum information processing, where it often underlies the structure of optimal protocols and simplifies analytical characterizations. 
In particular, it plays a central role in recent developments on virtual broadcasting~\cite{PhysRevLett.132.110203,z2pr-zbwl,8g6j-w7ld,okada2025virtualphasecovariantquantumbroadcasting,wang2026practicalquantumbroadcasting}, where symmetry constraints critically shape the form and performance of optimal transformations.

Another useful concept is that of a unital channel.

\begin{mydef}{Unital}{Unital}
    A quantum channel $\mE$ is called unital if it preserves the identity operator, namely
    \begin{align}
    \mE(\1)=\1.
    \end{align}
    In terms of its Kraus operators $\{K_i\}_{i}$, where $\mE(\rho)=\sum_i K_i \rho K_i^{\dagger}$, this is equivalent to
    \begin{align}
    \sum_i K_i K_i^{\dagger}=\1.
    \end{align}
\end{mydef}

Unital channels play a foundational role in quantum information theory because they represent a specific class of physical processes that preserve ``maximal disorder''. 
Their importance spans from quantum thermodynamics to coherence.
Examples of unital channels include unitary evolutions, Pauli channels --- such as bit flip $\mN_{\mathrm{BF}}$, phase flip $\mN_{\mathrm{PF}}$, and depolarizing channels $\mN_{\mathrm{D}}$ --- as well as, more generally, mixed unitary channels. 
In contrast, the amplitude damping channel $\mN_{\mathrm{AD}}$ provides a canonical example of a non-unital channel.

For the depolarizing channel $\mN_{\mathrm{D}}$ (see Tab.~\ref{tab:Quantum_Noise_Models}), one has

\begin{mylem}{Commutation with Unital Channels}{Unital}
    The depolarizing channel commutes with any unital channel, i.e.,
    \begin{align}
        [\mN_{\mathrm{D}}(p), \mE]=0,
    \end{align}
    holds for any unital channel $\mE$.
\end{mylem}

\begin{proof}
    For any input state $\rho$, it holds that
    \begin{align}
        \mN_{\mathrm{D}}(p) \circ \mE(\rho)=
        p\,\mE(\rho)+(1-p)\frac{\1}{d}=
        \mE\circ\mN_{\mathrm{D}}(p)(\rho).
    \end{align}
    This concludes the proof.
\end{proof}

As a direct consequence, it follows that

\begin{mycor}{Unitary Covariance of $\mN_{\mathrm{D}}$}{UC_Depolarizing}
    The depolarizing channel is unitary covariant.
\end{mycor}

Equipped with these results, we now state a central theorem of this subsection: under depolarizing noise, pre-processing restricted to local unital channels cannot improve purification performance in the distributed setting.

\begin{figure}[htbp]
    \centering   
    \includegraphics[width=1\textwidth]{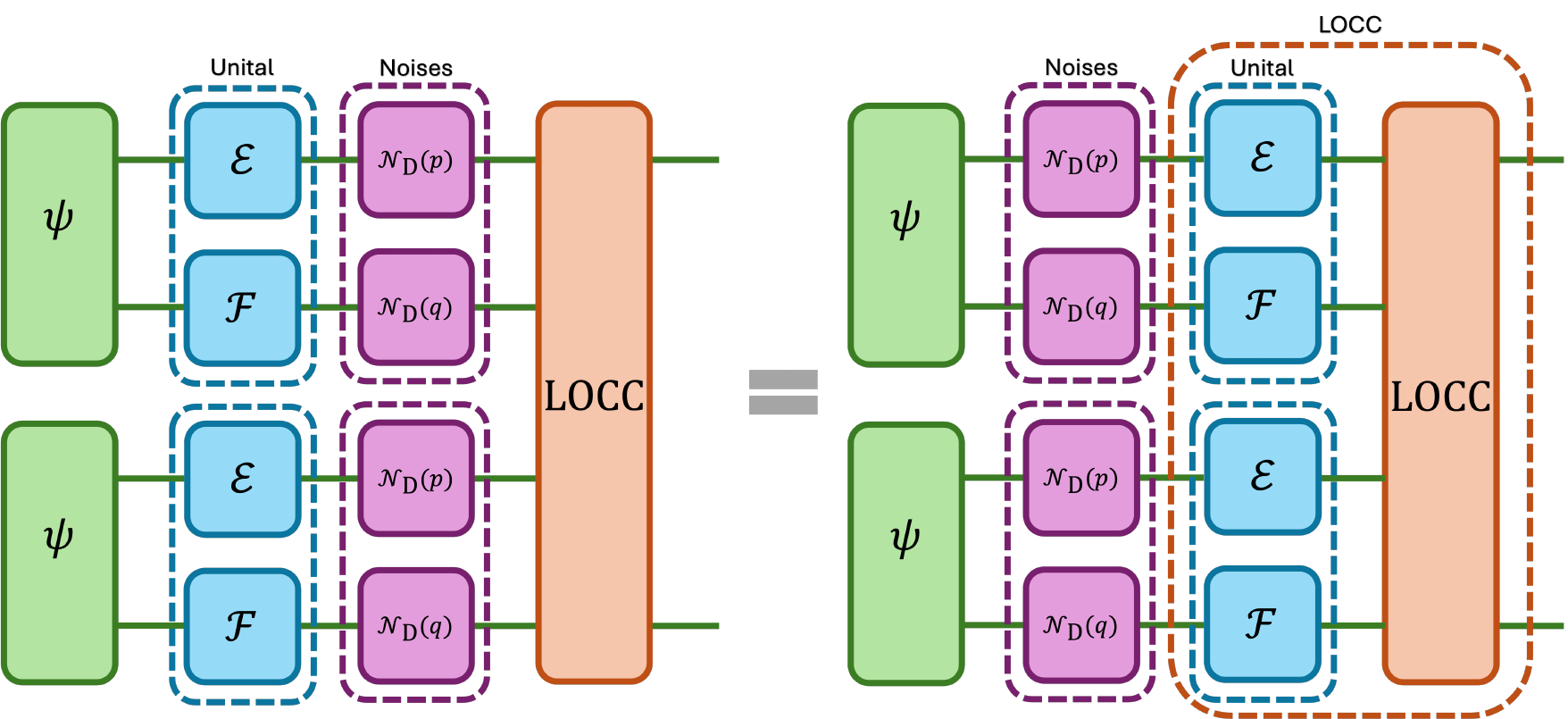}
    \caption{\textbf{2-to-1 LOCC Purification}. 
        In the left panel, local unital pre-processing (PreP) is applied after state preparation and prior to the noise. 
        Following the action of the noise, LOCC post-processing (PostP) is performed to enhance purification performance. 
        Owing to the commutation between the depolarizing channel and unital channels (Lem.~\ref{lem:Unital}), the purple and blue boxes can be interchanged, yielding the right panel. 
        In this representation, the unital operations can be absorbed into the LOCC PostP, implying that PreP provides no advantage in this protocol. 
        Although illustrated for the 2-to-1 setting, the same argument extends to arbitrary n-to-1 purification.
    }
    \label{fig:D_LOCC}
\end{figure}

\begin{mythm}{Ineffectiveness of Local-Unital Pre-Processing}{Ineffectiveness_LUnital_PreP}
    For n-to-1 unassisted (UA) purification, as considered in Fig.~\ref{fig:D_LOCC}, suppose the pre-processing (PreP) consists solely of local unital channels, e.g., $\mE\otimes\mF$ with $\mE$ and $\mF$ unital, applied after the state preparation but before the noise $\mN_{\mathrm{D}}(p)\otimes\mN_{\mathrm{D}}(q)$, while the post-processing (PostP) is restricted to local operations and classical communication (LOCC). Then, no improvement in the final fidelity can be achieved. 
    In other words, local-unital (LU) PreP is ineffective in this setting.
\end{mythm}

A graphical proof of the above Thm.~\ref{thm:Ineffectiveness_LUnital_PreP}, which also serves as a schematic illustration, is provided in Fig.~\ref{fig:D_LOCC}.

The results of this subsection reveal a qualitatively different mechanism governing purification under depolarizing noise. 
Owing to Lem.~\ref{lem:Unital}, local-unital pre-processing commutes with the noise and can therefore be freely rearranged within the protocol, as demonstrated in Fig.~\ref{fig:D_LOCC}. 
In particular, any such operation can be absorbed into the subsequent LOCC post-processing without affecting the overall purification. 
The apparent freedom to pre-process thus carries no operational consequence.
The material developed in this and the preceding subsection lays the groundwork for distributed quantum state purification, to be explored in the subsections that follow.


\subsection{Local Unitary Pre-Processing Purification}\label{subsec:LU_Pre_Processing}

Section~\ref{sec:Advantages_GQSP} developed the spatiotemporal framework in the context of global quantum state purification (see Fig.~\ref{fig:Conventional_Purification}), demonstrating that forward-assisted (FA) protocols can outperform conventional approaches limited to post-processing when arbitrary CPTP maps are available. 
This viewpoint, while comprehensive, leaves open an equally fundamental regime: 
distributed purification. 
Here, each copy is prepared as a bipartite state shared across spatially separated agents forming a network, and the admissible operations are restricted to local operations and classical communication (LOCC). 
This setting naturally encompasses entanglement distillation as a central special case --- a cornerstone of quantum information processing and quantum computing.

Within this constrained operational landscape, FA protocols exhibit two advantages. They deliver higher output fidelity than conventional post-processing-only schemes, and exhibit a clear separation in sample efficiency: a single-copy protocol augmented with pre-processing can already surpass conventional distributed purification that consumes 4 copies. 
The present demonstration is limited to the 4-to-1 regime by computational constraints rather than by any intrinsic limitation of the protocol. 
Notably, in this distributed setting, the advantage arises from pre-processing alone, without the need for subsequent post-processing, revealing an operational role that is qualitatively distinct from that observed in global quantum state purification.

To expose this advantage in its clearest form, we focus on the minimal setting of bipartite qubit systems and the purification of noisy maximally entangled states --- a canonical primitive underlying entanglement-based quantum technology. 
To situate our protocol within established approaches, we begin from the standard scenario in which a maximally entangled state is prepared and subsequently subjected to noise within a quantum circuit.
Concretely, we first initialize the system in the computational basis state $\ket{00}_{AB}$, shared between Alice and Bob.
Entanglement is then generated by applying a Hadamard gate to Alice's qubit, followed by a global controlled-NOT (CNOT) operation. 
In the absence of noise, this procedure prepares the maximally entangled Bell state $\phi^{+}$, or simply maximally entangled state (MES), shared between the two parties:
\begin{align}
    \ket{\phi^{+}}_{AB}:=\mathrm{CNOT}\circ(\mathrm{H}\otimes\id)\ket{00}=\frac{1}{\sqrt{2}}\left(\ket{00}+\ket{11}\right).
\end{align}
Subsystem labels (e.g., 
$A$ and $B$) are suppressed whenever they are unambiguous from context, in order to streamline notation without loss of clarity.
In practice, however, imperfections in control and unavoidable environmental interactions introduce errors, so that the ideal state is replaced by a noisy output arising from the action of a quantum channel on the entangled pair (see Fig.~\ref{fig:Distributed_Purifications}(a)).
\begin{align}\label{eq:psi_a}
    \psi_{a} \coloneqq \mN_1(p)\otimes\mN_2(q)(\phi^{+}).
\end{align}
We denote the maximally entangled state by $\phi^{+}=\ketbra{\phi^{+}}{\phi^{+}}$ for brevity, and write $\mN_1(p)$ and $\mN_2(q)$ in Eq.~\eqref{eq:psi_a} for the local noisy quantum channels acting on systems $A$ and $B$, respectively. 
Noise is assumed to act after the entangling gate. 
If it occurred beforehand, the state could be discarded and reset to $\ket{00}$, which can be implemented easily in experiments.
Here, the subscript $a$ labels the protocol rather than the underlying physical system, a convention adopted for clarity in subsequent comparisons.

\begin{figure}[htbp]
    \centering   
    \includegraphics[width=0.6\textwidth]{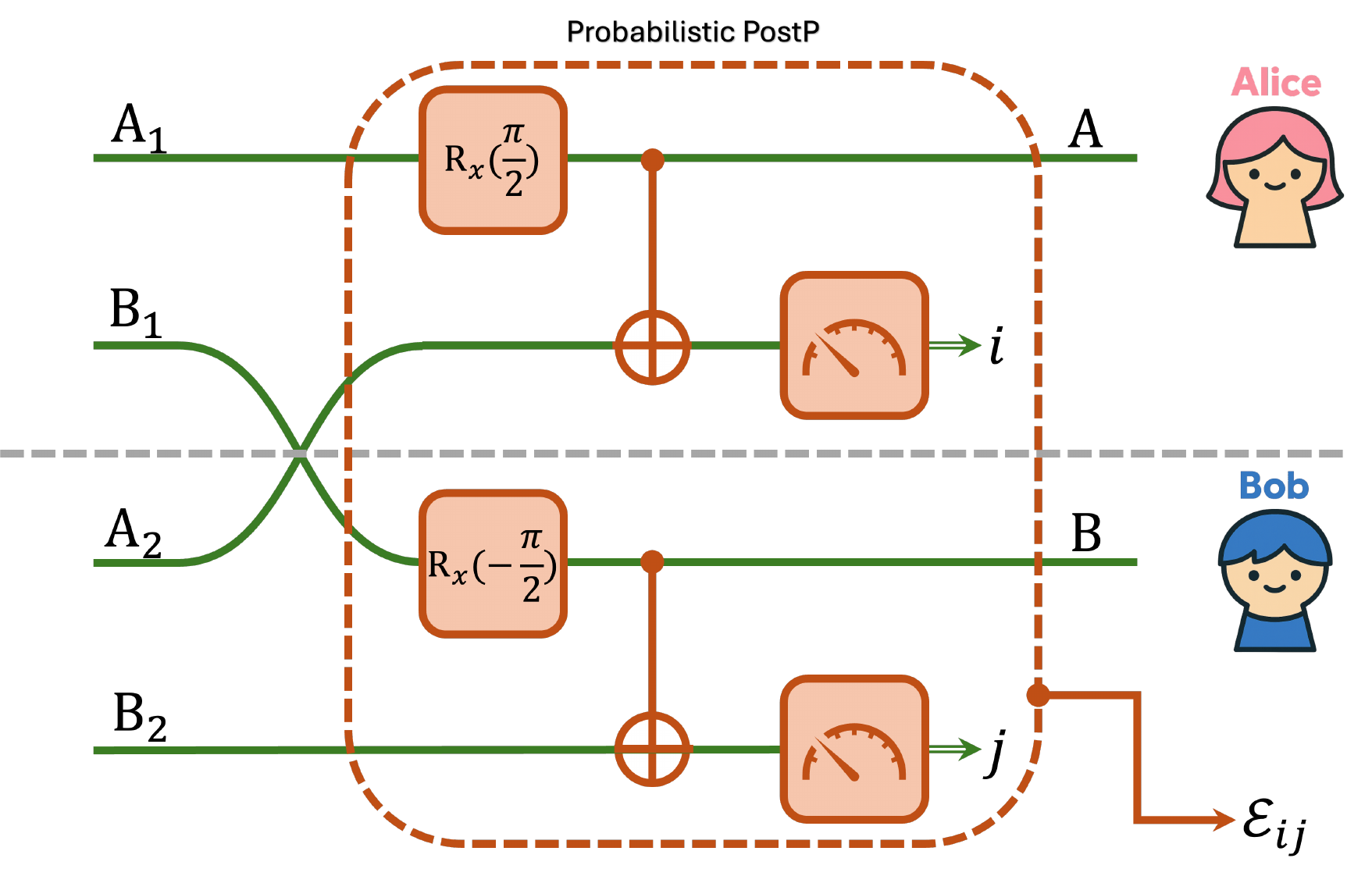}
    \caption{\textbf{Probabilistic Post-Processing}. 
        Given two copies of a bipartite input state, a probabilistic LOCC post-processing protocol (shown above) is applied to purify the noisy states. 
        Depending on the measurement outcomes $(i,j)$ on Alice's and Bob's sides, the resulting quantum operation is described by a map $\mE_{ij}$. Particular attention is given to the outcome $(0,0)$, which heralds successful purification.
    }
    \label{fig:Probabilistic_PostP}
\end{figure}

When multiple copies of the noisy maximally entangled state $\psi_{a}$ are available, one may employ the probabilistic purification protocol of~\cite{3bb1-pmtp}, as illustrated in Fig.~\ref{fig:Probabilistic_PostP}. 
The purification map is denoted by $\mE_{ij}$, where $i$ and $j$ correspond to the measurement outcomes on Alice's and Bob's sides. 
The focus is on the case $i=j=0$, with the resulting state denoted by $\psi_b$ (see Fig.~\ref{fig:Distributed_Purifications}(b)), namely,
\begin{align}\label{eq:psi_b}
    \psi_b \coloneqq \mE_{00}(\psi_a\otimes\psi_a).
\end{align}

\begin{figure}[htbp]
    \centering   
    \includegraphics[width=1\textwidth]{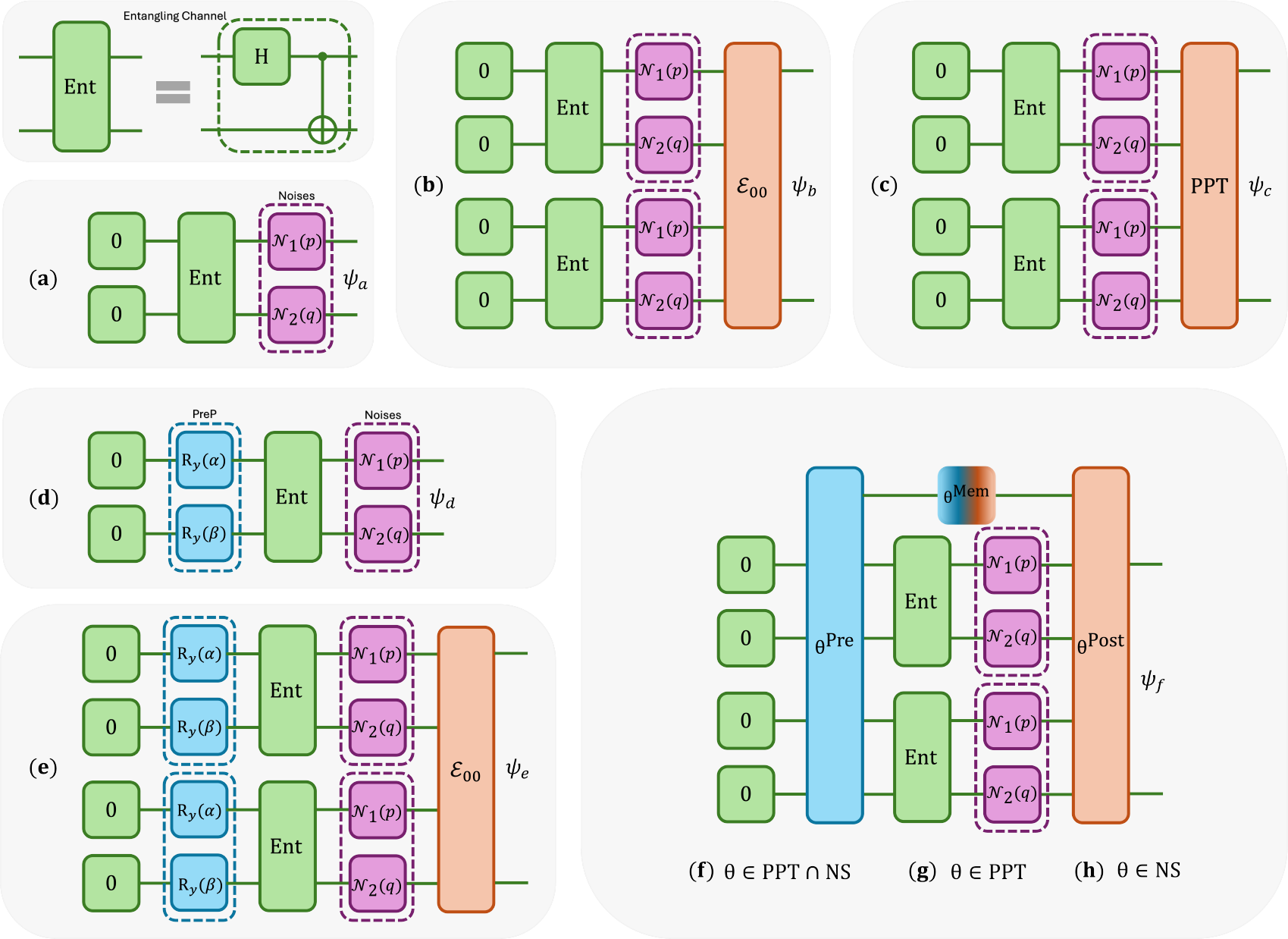}
    \caption{\textbf{Forward-Assisted Distributed Purifications}. 
        (a) Preparation of a single-copy noisy maximally entangled state $\psi_a$ (see Eq.~\eqref{eq:psi_a}) via an entangling channel (Ent), followed by noise $\mN_1(p)\otimes\mN_2(q)$.
        (b) Two-copy probabilistic post-processing (PostP): independent realizations of $\psi_a$ are combined and processed via a LOCC map $\mE_{00}$ (see Fig.~\ref{fig:Probabilistic_PostP}), yielding $\psi_b$ (see Eq.~\eqref{eq:psi_b}).
        (c) Deterministic purification under PPT PostP, optimized over all PPT-preserving maps, producing $\psi_c$ (see Eq.~\eqref{eq:psi_c}).
        (d) Local-unitary (LU) pre-processing (PreP), consisting of $y$-axis rotations $\mathrm{R}_{y}(\alpha)\otimes\mathrm{R}_{y}(\beta)$ applied prior to Ent, followed by noise, resulting in $\psi_d$ (see Eq.~\eqref{eq:psi_d}).
        (e) Hybrid protocol: two copies of the pre-processed states $\psi_d$ are used as inputs to probabilistic PostP $\mE_{00}$, leading to $\psi_e$ (see Eq.~\eqref{eq:psi_e}).
        (f)-(h) General forward-assisted (FA) purification: PreP and PostP are connected via a quantum memory channel, forming a superchannel $\theta$ (see Fig.~\ref{fig:Superchannel}), which may satisfy PPT (g), NS (h), or both constraints (f).
    }
    \label{fig:Distributed_Purifications}
\end{figure}

Any probabilistic purification protocol can be extended to a deterministic one by assigning failure outcomes to 0, leaving the expectation value unchanged. 
It is therefore natural to examine the corresponding limitations of deterministic protocols. 
In distributed settings, purification is typically restricted to LOCC operations, whose full characterization remains intractable. 
The analysis is thus extended to positive partial transpose preserving (PPTp or simply PPT) operations (see Def.~\ref{def:PPT_p}), an efficiently computable superset that provides an upper bound on all LOCC protocols via semidefinite programming (SDP). 
The focus is then placed on the output state of the purification map $\mE: A_1B_1A_2B_2\to AB$ (see Fig.~\ref{fig:Distributed_Purifications}(c)).
\begin{align}\label{eq:psi_c}
    \psi_c(\mE):=\mE(\psi_a\otimes\psi_a), \quad\text{with}\,\,\,
    \mE\in\mathrm{PPT}.
\end{align}
It should be noted that the fidelity of $\psi_c(\mE)$ with $\mE\in\mathrm{PPT}$ is not necessarily higher than that of $\psi_b$ in Eq.~\eqref{eq:psi_b}, since $\psi_b$ is obtained via post-selection. 
This will be confirmed by our numerical analysis.

Noise processes are inherently dynamical and are described by quantum channels. 
The most general manipulations of such channels are characterized by superchannels~\cite{Chiribella_2008,8678741,xiao2025superchanneltearsgeneralizedoccams}, which consist of both pre-processing and post-processing stages. 
Pre-processing corresponds to operations applied prior to the action of noise, while post-processing captures operations applied afterwards. 
In quantum error correction, these stages are analogous to encoding and decoding. 
A key distinction, however, is that superchannels may incorporate quantum memory linking the pre- and post-processing stages. 
This framework provides the natural operational setting for forward-assisted (FA) purification (see Subsec.~\ref{subsec:Forward_Assisted_Purification_Protocols}).
To highlight the advantage of forward-assisted (FA) purification, attention is restricted to the minimal setting of local unitary (LU) pre-processing, enabling direct comparison with existing protocols and their fundamental limits. 
The resulting state takes the form (see Fig.~\ref{fig:Distributed_Purifications}(d))
\begin{align}\label{eq:psi_d}
    \psi_d(\alpha,\beta):=
    \underbrace{\left(\mN_1(p)\otimes\mN_2(q)\right)}_{\text{Noise}}
    \circ \
    \underbrace{\mathrm{CNOT}\circ(\mathrm{H}\otimes\id)}_{\text{Entangling Gate}}
    \circ
    \overbrace{\left(\mathrm{R}_{y}(\alpha)\otimes
    \mathrm{R}_{y}(\beta)\right)}^{\text{Pre-Processing}}
    \left(\ketbra{00}{00}\right),
\end{align}
where $\mathrm{R}_{y}(\alpha)$ denotes a rotation about the $y$-axis by an angle $\alpha$. 
This choice offers improved purification performance while remaining experimentally accessible.
Equation~\eqref{eq:psi_d} describes a single-copy, or single-shot, purification scenario. 
In contrast, $\psi_b$ and $\psi_c$ in Eqs.~\eqref{eq:psi_b} and~\eqref{eq:psi_c}, respectively, correspond to two-copy purification protocols.

Note that the pre-processing stage is applied prior to the entangling gate.
This choice is motivated by both physical and mathematical perspectives.
Physically, within the framework of dynamical entanglement~\cite{bauml2019resourcetheoryentanglementbipartite,PhysRevLett.125.180505,xing2023fundamentallimitationscommunicationquantum,glc7-xy8t}, the relevant resource for distributed purification is the entangling quantum dynamics itself. 
Noise acts by degrading this entangling channel, making it natural to treat the noisy entangling operation as the fundamental dynamical resource and to implement pre-processing before it.
Mathematically, for the amplitude damping and depolarizing channels commonly encountered in quantum information processing, inserting pre-processing between an ideal entangling gate and subsequent noise offers no advantage, as established in Cor.~\ref{cor:Ineffectiveness_LU_PreP_4b} and Thm.~\ref{thm:Ineffectiveness_LUnital_PreP}.

Combining LU pre-processing with the post-processing purification protocol $\mE_{00}$ in Fig.~\ref{fig:Probabilistic_PostP} leads to the following probabilistic purified state (see Fig.~\ref{fig:Distributed_Purifications}(e)),
\begin{align}\label{eq:psi_e}
    \psi_e(\alpha,\beta):=\mE_{00}(\psi_d(\alpha,\beta)\otimes\psi_d(\alpha,\beta)),
\end{align}
which defines a probabilistic forward-assisted purification protocol.

So far, the protocols introduced do not involve any quantum memory between the pre-processing (PreP) and post-processing (PostP) stages. 
More generally, one can consider additional classes of purification protocols that incorporate quantum memory. 
In particular, three such classes arise: one that is both PPT (with respect to post-processing) and non-signalling (NS), and two others that satisfy PPT (with respect to post-processing) and NS constraints separately. 
Formally, the FA purification protocol that is both PPT and NS gives rise to the following state (see Fig.~\ref{fig:Distributed_Purifications}(f))
\begin{align}\label{eq:psi_f}
    \psi_f(\theta):=\overbrace{\theta^{\mathrm{Post}}}^{\text{Post-Processing}}
    \circ \
    \left(\overbrace{\theta^{\mathrm{Mem}}}^{\text{Memory}}\otimes\left(\underbrace{\left(\mN_1(p)\otimes\mN_2(q)\right)}_{\text{Noise}}
    \circ 
    \underbrace{\mathrm{CNOT}\circ(\mathrm{H}\otimes\id)}_{\text{Entangling Gate}}\right)^{\otimes2}\right)
    \circ
    \overbrace{\theta^{\mathrm{Pre}}}^{\text{Pre-Processing}}
    \left(\ketbra{00}{00}^{\otimes2}\right),
\end{align}
with
\begin{align}
    \theta:=\theta^{\mathrm{Post}}\circ\theta^{\mathrm{Mem}}\circ\theta^{\mathrm{Pre}}\in\mathrm{PPT}\cap\mathrm{NS}.
\end{align}
The conditions of PPT and NS can be formulated as SDPs. 
To illustrate, consider the superchannel shown in Fig.~\ref{fig:Superchannel}. 
The PPT constraint with respect to the post-processing stage is expressed as
\begin{align}
    (J^{\theta})^{\T_{\mathrm{Post}}}
    =
    (J^{\theta})^{\T_{CD}}
    \geqslant0,
\end{align}
where $J^{\theta}$ denotes the Choi operator of the superchannel $\theta$.
Note that, in this setting, the memory system $F$ is not regarded as part of the post-processing systems.
Again, taking the superchannel $\theta$ in Fig.~\ref{fig:Superchannel} as an example, the NS condition reads
\begin{align}
    \Tr_{D}[J^{\theta}]&=J^{\theta}_{AB}\otimes\frac{\1_{C}}{d_C},\\
    \Tr_{B}[J^{\theta}]&=\frac{\1_{A}}{d_A}\otimes J^{\theta}_{CD}.
\end{align}

\begin{table}[htbp]
    \centering
    \begin{tblr}{
      colspec = {l || c | c | c},
      row{1,2} = {bg=gray!50, font=\bfseries}, 
      column{1} = {bg=gray!10},                
      hlines,                                  
      vlines,                                  
      cells = {m},                             
      row{1} = {c},                            
    }
      \SetCell[c=4]{c} Forward-Assisted Distributed Purification Protocols & & & \\
      Distributed Purifications & Performance & Number of Copies & Constraints \\
      Protocol a (see Eq.~\eqref{eq:psi_a}) & $F_a:=F(\psi_{a}, \phi^{+})$ & single-copy & without any purification \\
      Protocol b (see Eq.~\eqref{eq:psi_b}) & $F_b:=F(\psi_{b}, \phi^{+})$ & two-copy & {with probabilistic post-processing\\ via $\mE_{00}$ (see Fig.~\ref{fig:Probabilistic_PostP})} \\
      Protocol c (see Eq.~\eqref{eq:psi_c}) & {$F_c:=\max F(\psi_c(\mE), \phi^{+})$\\ \,\,\,\, s.t. $\mE\in\mathrm{PPT}$} & two-copy & with PPT post-processing \\
      Protocol d (see Eq.~\eqref{eq:psi_d}) & {$F_d:=\max F(\psi_d(\alpha,\beta), \phi^{+})$ \\
      \,\,\,\, s.t. $\alpha,\beta\in[0,2\pi)$} & 
      single-copy &
      {with LU pre-processing\\ via $\mathrm{R}_{y}(\alpha)\otimes\mathrm{R}_{y}(\beta)$}\\
      Protocol e (see Eq.~\eqref{eq:psi_e}) & {$F_e:=\max F(\psi_e(\alpha,\beta), \phi^{+})$ \\
      \,\,\,\, s.t. $\alpha,\beta\in[0,2\pi)$} & 
      two-copy & 
      {with LU pre-processing implemented\\ by $\mathrm{R}_{y}(\alpha)\otimes\mathrm{R}_{y}(\beta)$\\ and probabilistic post-processing\\ via $\mE_{00}$ (see Fig.~\ref{fig:Probabilistic_PostP})} \\
      Protocol f (see Eq.~\eqref{eq:psi_f}) & {$F_f:=\max F(\psi_f(\theta), \phi^{+})$ \\
      \,\,\,\, s.t. $\theta\in\mathrm{PPT}\cap\mathrm{NS}$} & two-copy & 
      {with FA purification implemented\\ via a superchannel satisfying both\\ PPT (with respect to the\\ post-processing systems) and NS} \\
      Protocol g (see Eq.~\eqref{eq:psi_f}) & {$F_g:=\max F(\psi_f(\theta), \phi^{+})$ \\
      \,\,\,\, s.t. $\theta\in\mathrm{PPT}$} & two-copy & 
      {with FA purification implemented\\ via a superchannel satisfying PPT\\ (with respect to the\\ post-processing systems)} \\
      Protocol h (see Eq.~\eqref{eq:psi_f}) & {$F_h:=\max F(\psi_f(\theta), \phi^{+})$ \\
      \,\,\,\, s.t. $\theta\in\mathrm{NS}$} & two-copy & 
      {with FA purification implemented\\ via a superchannel satisfying NS} \\
    \end{tblr}
    \caption{\textbf{Forward-Assisted Distributed Purification Protocols}.
        Summary of forward-assisted (FA) distributed 2-to-1 purification protocols, including their performance measures, number of input copies, and operational constraints. 
        Protocols (a)–(h) span conventional schemes without purification, probabilistic and PPT post-processing, LU pre-processing, and general FA protocols implemented via superchannels subject to PPT and/or NS constraints.
    }
    \label{tab:Purification_Comparison_1}
\end{table}

All purification protocols, including the conventional one, are summarized in Tab.~\ref{tab:Purification_Comparison_1} and illustrated in Fig.~\ref{fig:Distributed_Purifications}.
Building on the hierarchy established in Fig.~\ref{fig:Venn}, we arrive at a theorem that characterizes the structural relations among forward-assisted purification protocols in the distributed setting.

\begin{mythm}{Hierarchy of Forward-Assisted Purification in the Distributed Setting}{Hierarchy_Distributed}
For forward-assisted purification protocols in the distributed setting (see Fig.~\ref{fig:Distributed_Purifications}), as summarized in Tab.~\ref{tab:Purification_Comparison_1}, the achievable performance is ordered by the following inequality chain.
\begin{align}
    F_a\leqslant F_c\leqslant F_f\leqslant F_g,
\end{align}
and, under an alternative admissible constraint, the ordering extends to
\begin{align}
    F_a\leqslant F_c\leqslant F_f\leqslant F_h.
\end{align}
Moreover, when probabilistic post-processing is permitted, the achievable performances satisfy
\begin{align}
    F_b\leqslant F_e.
\end{align}
\end{mythm}

The hierarchy established in Thm.~\ref{thm:Hierarchy_Distributed} reveals a monotonic enhancement in performance as the accessible operational resources are enlarged.
The role of $F_d$ (see Tab.~\ref{tab:Purification_Comparison_1}), however, remains unclear, and its relation to the other protocol classes is not yet resolved. 
This gap will be addressed through the numerical analysis that follows.

We now proceed to a systematic numerical investigation of the distributed purification protocols introduced above (see Fig.~\ref{fig:Distributed_Purifications}).
The analysis addresses four aspects.

\begin{itemize}
    \item First, we compare the performance of the purification protocols summarized in Tab.~\ref{tab:Purification_Comparison_1}. 
    \item Second, for a fixed noise model and parameter, we vary the pre-processing angles $\alpha$ and $\beta$ in $\mathrm{R}_{y}(\alpha)\otimes\mathrm{R}_{y}(\beta)$ (see Eq.~\eqref{eq:psi_d}) to probe the underlying mechanism and identify the origin of the observed advantage.
    \item Third, we consider asymmetric local noise with distinct parameters and quantify the performance gain achievable using only LU pre-processing on a single copy, in comparison with two-copy PPT post-processing, thereby highlighting the role of pre-processing in distributed purification.
    \item Fourth, the numerical results reveal regimes in which local-unitary (LU) pre-processing on a single copy already outperforms two-copy PPT post-processing, which serves as a computable upper bound for conventional protocols.
    This naturally raises the question of whether such an advantage can persist beyond the two-copy setting, for example, in 3-to-1 purification or higher, which we explore in detail below.
\end{itemize}

\begin{figure}[htbp]
    \centering   
    \includegraphics[width=1\textwidth]{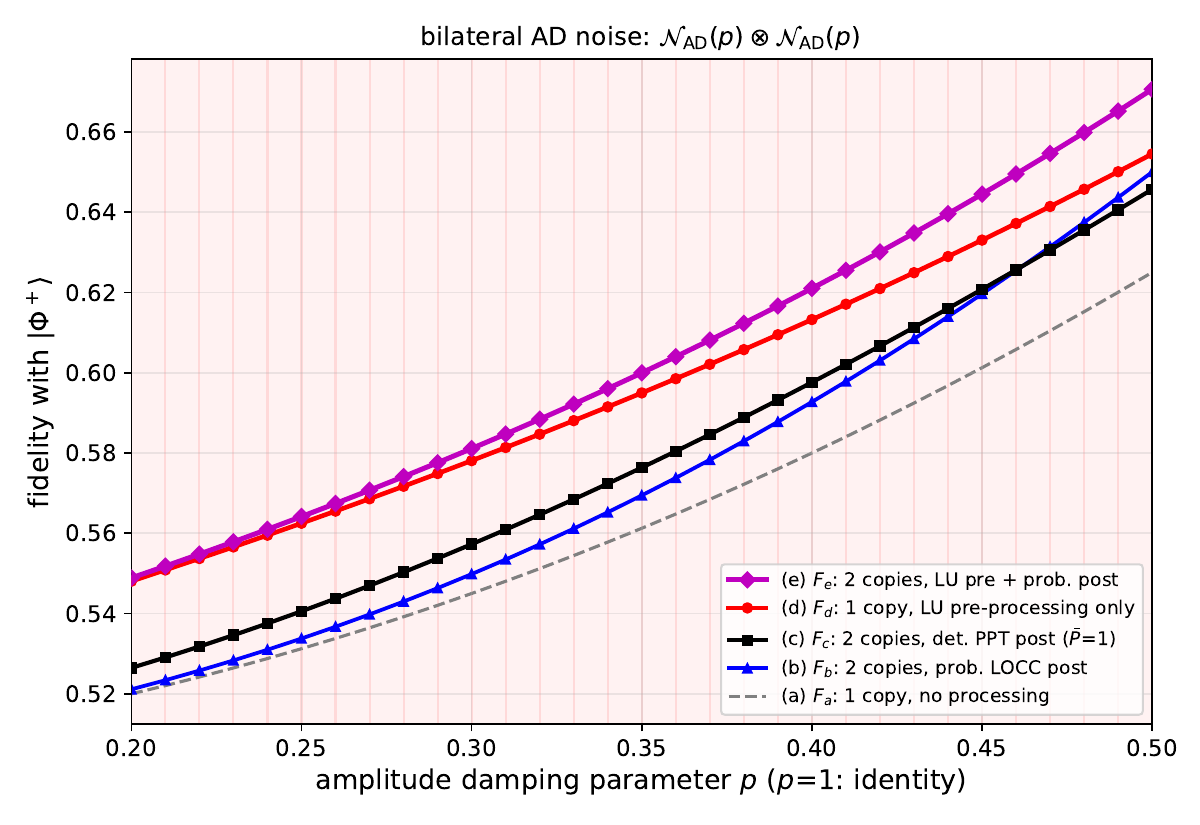}
    \caption{\textbf{Comparison Across Distributed Purification Protocols}. 
        Identical amplitude damping noise $\mN_{\mathrm{AD}}(p)$ is applied to all systems. 
        The performance of the protocols summarized in Tab.~\ref{tab:Purification_Comparison_1} is compared as a function of the noise parameter $p$. 
        Probabilistic post-processing via $\mE_{00}$ (see Fig.~\ref{fig:Probabilistic_PostP}) achieves higher performance than PPT post-processing in certain regimes, as $\mE_{00}$ implements a probabilistic purification, whereas Protocol c (see Eq.~\eqref{eq:psi_c}) corresponds to deterministic (i.e., purificiation success probability $\bar{P}=1$) PPT post-processing.
    }
    \label{fig:Distributed_Num_Point1_1}
\end{figure}

For the first numerical analysis, attention is restricted to the symmetric setting in which both local noise channels are AD channels characterized by the same parameter $p$, namely 
\begin{align}
    \mN_1(p)=\mN_2(q)=\mN_{\mathrm{AD}}(p).
\end{align}
The resulting performance, as a function of $p$, is shown in Fig.~\ref{fig:Distributed_Num_Point1_1} for the protocols in Fig.~\ref{fig:Distributed_Purifications}.

The numerical results reveal that, for the noise model considered, even single-copy local-unitary pre-processing (see Fig.~\ref{fig:Distributed_Purifications}(d)) can outperform two-copy PPT post-processing (see Fig.~\ref{fig:Distributed_Purifications}(c)). 
Notably, in contrast to the global setting discussed in Subsec.~\ref{subsec:PreP_Post}, the protocol in Fig.~\ref{fig:Distributed_Purifications}(d) involves no post-processing following the pre-processing stage. 
The observed advantage, therefore, arises entirely from pre-processing, highlighting its independent operational power in the distributed regime.
Beyond achieving a performance advantage, this highlights an additional feature: a reduction in the number of copies required for purification. 
Such resource savings are particularly important for current quantum technologies, where preparing multiple copies is experimentally demanding.

Moreover, using two copies of the LU pre-processing-augmented state $\psi_d(\alpha,\beta)$ of Eq.~\eqref{eq:psi_d} (see Fig.~\ref{fig:Distributed_Purifications}(e)), followed by probabilistic LOCC post-processing $\mE_{00}$ (see Fig.~\ref{fig:Probabilistic_PostP}), leads to a further enhancement in purification performance (see Fig.~\ref{fig:Distributed_Num_Point1_1}). 
Notably, this protocol is readily implementable with existing experimental techniques.
To conclude, for the noise model considered in Fig.~\ref{fig:Distributed_Num_Point1_1} and the protocols analyzed in Tab.~\ref{tab:Purification_Comparison_1}, we establish the following inequality chains
\begin{align}
    F_a\leqslant F_b\leqslant F_d\leqslant F_e,
\end{align}
and
\begin{align}
    F_a\leqslant F_c\leqslant F_d\leqslant F_e.
\end{align}
This ordering makes explicit that pre-processing alone can exceed the fundamental limits of conventional purification schemes, isolating it as an independent operational resource rather than a mere auxiliary step.

Beyond amplitude damping noise $\mN_{\mathrm{AD}}$ (see Tab.~\ref{tab:Quantum_Noise_Models}), we consider the canonical noise models listed in Tab.~\ref{tab:Quantum_Noise_Models} to examine the role of pre-processing. 
Comparing the protocols in Fig.~\ref{fig:Distributed_Purifications}(e) and Fig.~\ref{fig:Distributed_Purifications}(b), corresponding to cases with and without pre-processing, reveals a clear advantage, as shown in Fig.~\ref{fig:Distributed_Num_Point1_2}.

\begin{figure}[htbp]
    \centering   
    \includegraphics[width=1\textwidth]{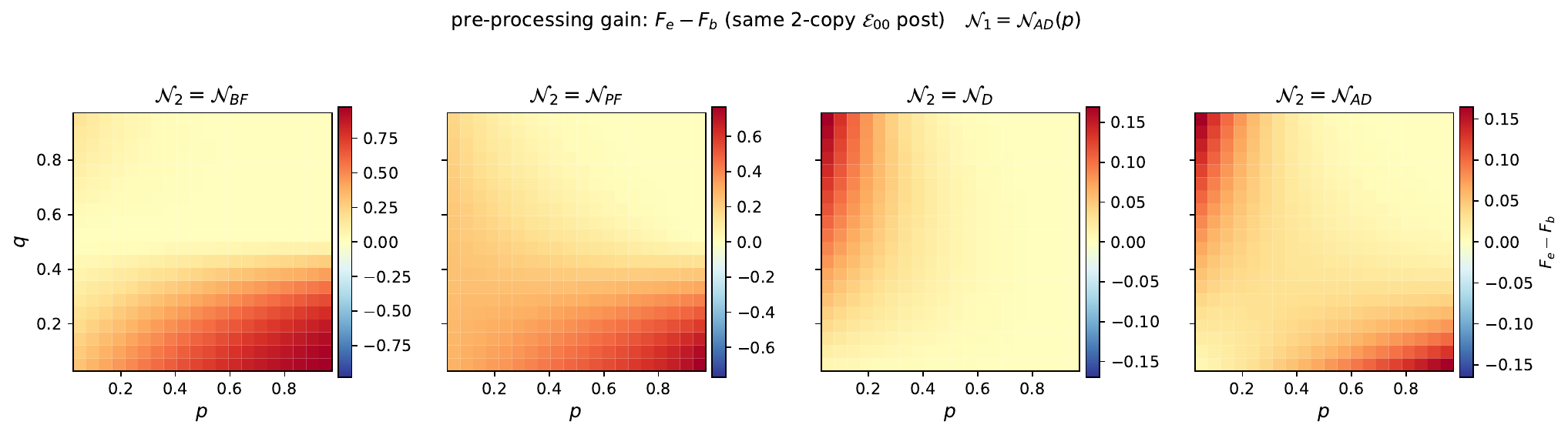}
    \caption{\textbf{With vs Without Pre-Processing}. 
        The composite noise $\mN_1(p)\otimes\mN_2(q)$ combines different noise models, with $\mN_1$ fixed as amplitude damping and $\mN_2$ varied over bit flip, phase flip, depolarizing, and amplitude damping channels. 
        In all regimes, pre-processing (see Fig.~\ref{fig:Distributed_Purifications}(e)) outperforms the no-pre-processing case (see Fig.~\ref{fig:Distributed_Purifications}(b)) for two-copy protocols followed by post-processing $\mE_{00}$ (see Fig.~\ref{fig:Probabilistic_PostP}).
    }
    \label{fig:Distributed_Num_Point1_2}
\end{figure}

We next investigate the effect of varying the pre-processing parameters $\alpha$ and $\beta$, with the aim of identifying regimes in which single-copy LU pre-processing (see Fig.~\ref{fig:Distributed_Purifications}(d)) outperforms two-copy PPT post-processing (see Fig.~\ref{fig:Distributed_Purifications}(c)). 
The corresponding numerical results are shown in Fig.~\ref{fig:Distributed_Num_Point2}.

\begin{figure}[htbp]
    \centering   
    \includegraphics[width=1\textwidth]{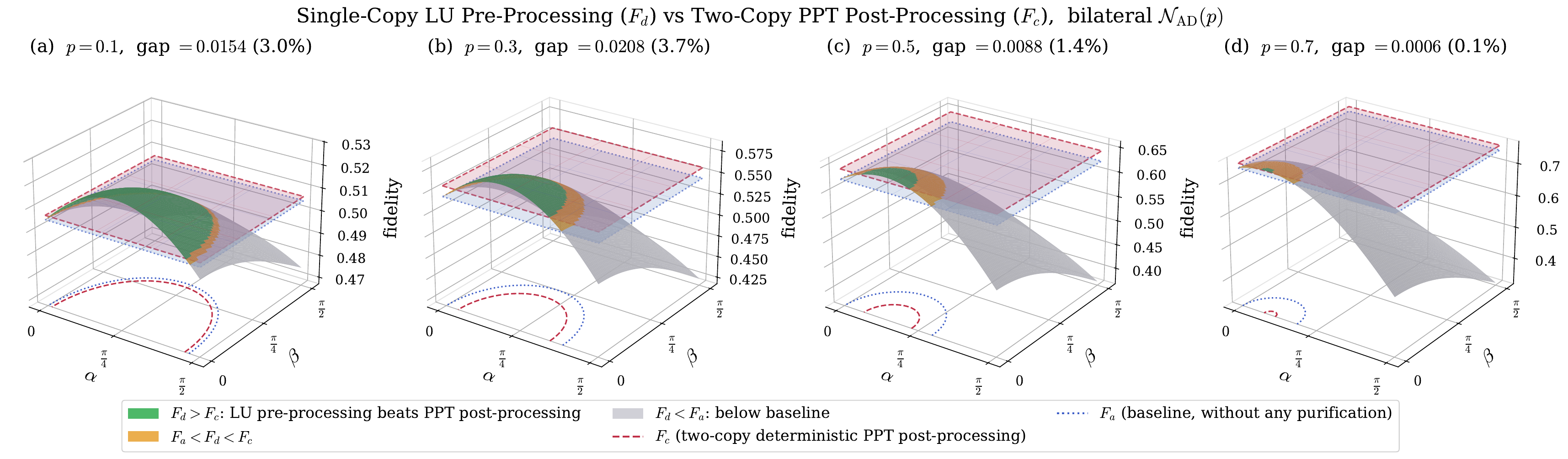}
    \caption{\textbf{Single-Copy LU Pre-Processing vs Two-Copy PPT Post-Processing}. 
        Surface plots of $F_d$ (see Tab.~\ref{tab:Purification_Comparison_1}) as a function of local unitary angles $\alpha$ and $\beta$ in $R_y(\alpha)\otimes R_y(\beta)$, under identical amplitude damping noise $\mN_{\mathrm{AD}}(p)$ on both subsystems, with $p=0.1$ (a), $0.3$ (b), $0.5$ (c), and $0.7$ (d).
        The red dashed plane denotes $F_c$ (see Tab.~\ref{tab:Purification_Comparison_1}), the optimal two-copy deterministic PPT post-processing fidelity (see Fig.~\ref{fig:Distributed_Purifications}(c)), while the blue dotted plane marks the baseline $F_a$ without purification (see Fig.~\ref{fig:Distributed_Purifications}(a)). Green regions ($F_d>F_c$) indicate that single-copy LU pre-processing (see Fig.~\ref{fig:Distributed_Purifications}(d)) outperforms two-copy PPT post-processing; 
        orange regions ($F_a<F_d<F_c$) indicate intermediate performance; 
        gray regions ($F_d<F_a$) fall below the baseline.
        Floor contours show the projected boundaries $F_d=F_c$ and $F_d=F_a$ in the $(\alpha,\beta)$ plane.
    }
    \label{fig:Distributed_Num_Point2}
\end{figure}

For the third aspect, we examine the improvement ratio of $F_d$ over $F_c$, corresponding to single-copy LU pre-processing (see Fig.~\ref{fig:Distributed_Purifications}(d)) relative to two-copy PPT post-processing (see Fig.~\ref{fig:Distributed_Purifications}(c)). 
In this setting, asymmetric local noise is considered, with distinct parameters $p$ and $q$, i.e., the noise is given by 
\begin{align}
    \mN_1(p)\otimes\mN_2(q)=\mN_{\mathrm{AD}}(p)\otimes\mN_{\mathrm{AD}}(q).
\end{align}
The corresponding numerical results are shown in Fig.~\ref{fig:Distributed_Num_Point3}, highlighting the advantage of protocol d (see Fig.~\ref{fig:Distributed_Purifications}(d)), using single-copy LU pre-processing, over protocol c (see Fig.~\ref{fig:Distributed_Purifications}(c)), which employs two-copy PPT post-processing.

\begin{figure}[htbp]
    \centering   
    \includegraphics[width=0.43\textwidth]{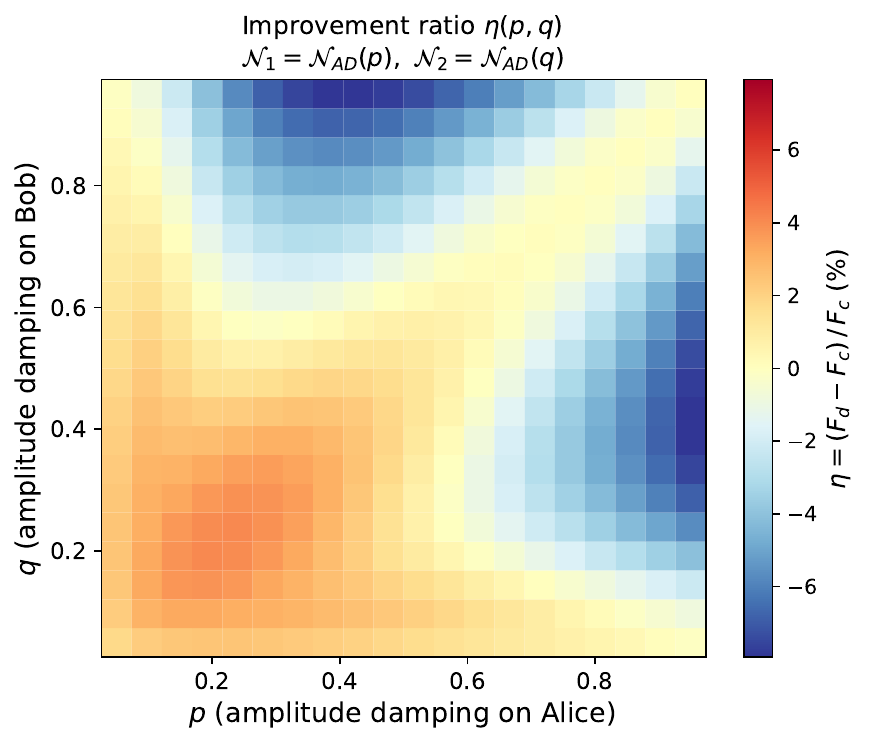}
    \caption{\textbf{Improvement Ratio of LU Pre-Processing over PPT Post-Processing}. 
        Heat map of the improvement ratio $\eta(p,q) \coloneqq (F_d - F_c)/F_c$ under asymmetric amplitude damping noise, where $\mN_1(p)=\mN_{\mathrm{AD}}(p)$ acts on Alice's system and $\mN_2(q)=\mN_{\mathrm{AD}}(q)$ acts on Bob's system. 
        Positive values indicate regimes in which single-copy LU pre-processing (see Fig.~\ref{fig:Distributed_Purifications}(d)) outperforms two-copy PPT post-processing (see Fig.~\ref{fig:Distributed_Purifications}(c)), while negative values indicate the opposite. 
        A broad parameter region with $\eta(p,q)>0$ demonstrates the advantage of LU pre-processing in distributed purification.
    }
    \label{fig:Distributed_Num_Point3}
\end{figure}

We finally turn to the question of how far pre-processing can outperform PPT post-processing. 
As already observed, a single copy with LU pre-processing (see Fig.~\ref{fig:Distributed_Purifications}(d)), yielding $\psi_d$ (see Eq.~\eqref{eq:psi_d}), can surpass two-copy PPT post-processing, yielding $\psi_c$ (see Eq.~\eqref{eq:psi_c}). 
A natural question is whether increasing the number of copies in PPT post-processing --- the conventional approach --- can eventually recover the advantage and outperform the single-copy LU pre-processing scheme.

While additional copies generally improve the performance of PPT post-processing, this is not universally sufficient. 
Remarkably, we identify parameter regimes in which single-copy LU pre-processing continues to outperform multi-copy PPT post-processing, demonstrating a persistent advantage beyond the two-copy setting.

In particular, we consider 3-to-1 and 4-to-1 purification protocols, as illustrated in Fig.~\ref{fig:Distributed_Multiple_Copies}. 
More precisely, given $n$ noisy copies of $\psi_a$ (see Eq.~\eqref{eq:psi_a}), PPT post-processing yields the following output state
\begin{align}\label{eq:psi_c_n_to_1}
    \psi_c^{n\to 1}(\mE) \coloneqq \mE(\psi_a^{\otimes n}), \quad \text{with}\,\,\,
    \mE\in\mathrm{PPT}.
\end{align}
The optimal performance over all PPT post-processing protocols is then given by
\begin{align}\label{eq:PPT_n_to_1}
    F_c^{n\to 1} \coloneqq
    \max \quad 
    & F(\psi_c^{n\to 1}(\mE), \phi^{+})\\
    \text{s.t.} \quad 
    &\mE\in\mathrm{PPT}.
\end{align}

\begin{figure}[htbp]
    \centering   
    \includegraphics[width=1\textwidth]{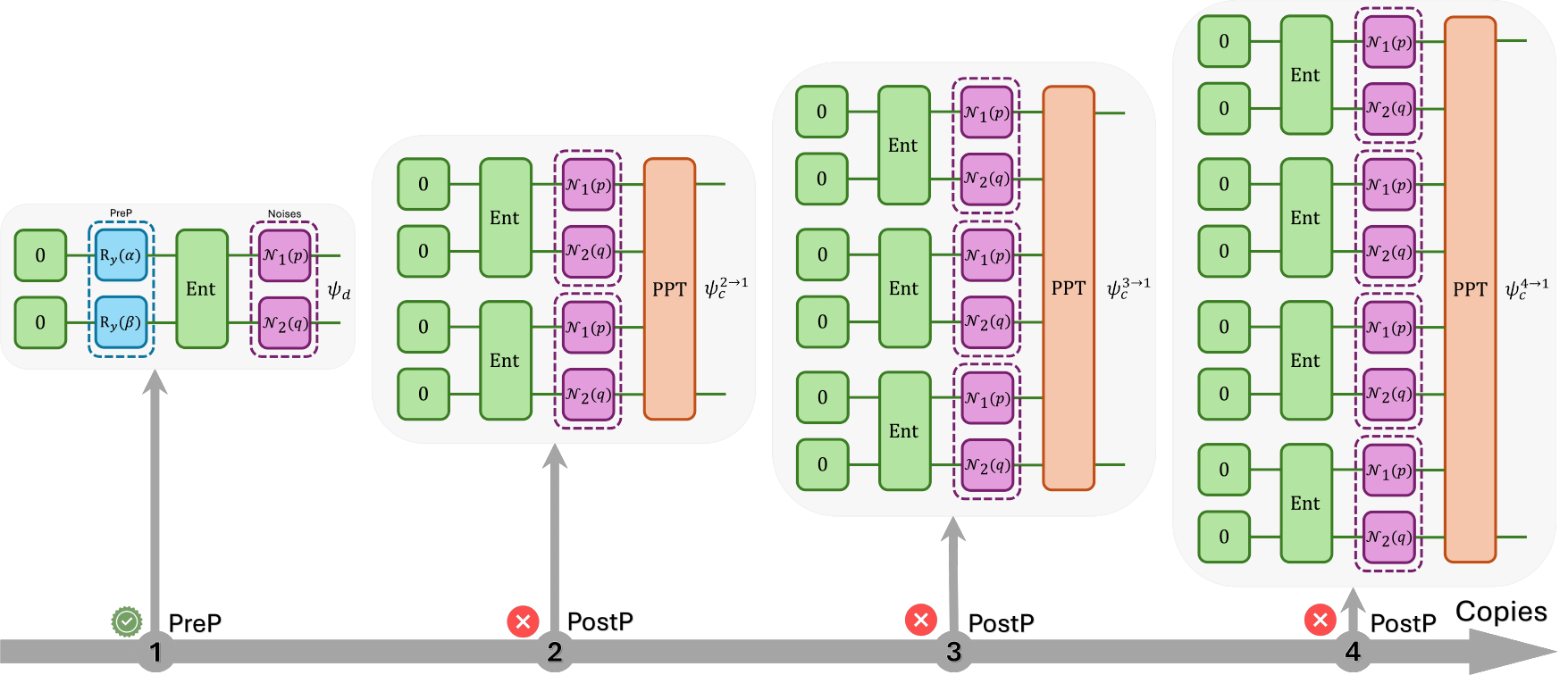}
    \caption{\textbf{Multi-Copy Distributed Purification}. 
        Comparison of single-copy LU PreP with $n$-copy PPT PostP ($n=2,3,4$). 
        Comparison between single-copy LU pre-processing (PreP) and $n$-copy PPT post-processing (PostP) for $n=2,3,4$. 
        While increasing the number of copies enhances the performance of PPT-based post-processing, there remain regimes in which single-copy pre-processing yields higher fidelity. 
        This demonstrates a persistent advantage of pre-processing, even against multi-copy purification strategies.
    }
    \label{fig:Distributed_Multiple_Copies}
\end{figure}

Here, we again consider the symmetric amplitude damping setting, where both local noise channels are amplitude damping channels characterized by the same parameter $p$. 
As noise parameter $p$ varies, the comparison between single-copy LU pre-processing and multi-copy ($n=2,3,4$) PPT post-processing (see Fig.~\ref{fig:Distributed_Multiple_Copies}) is summarized in Tab.~\ref{tab:Distributed_Multiple_Copies}.

\begin{table}[htbp]
    \centering
    \begin{tblr}{
      colspec = {c || c | c | c | c | c | c | c | c},
      row{1,2} = {bg=gray!50, font=\bfseries}, 
      column{1} = {bg=gray!10, font=\bfseries}, 
      hlines,
      vlines,
      cells = {m, c},
      cell{3}{7} = {bg=mRed}, 
      cell{4-13}{7} = {bg=mGreen},
      cell{3-9}{8} = {bg=mRed},
      cell{10-13}{8} = {bg=mGreen},
      cell{3-10}{9} = {bg=mRed},
      cell{11-13}{9} = {bg=mGreen},
    }
      \SetCell[c=9]{c} Single-Copy LU Pre-Processing vs Multi-Copy PPT Post-Processing & & & & & & & & \\
      Noise Parameter $p$ & $F_a$ & $F_d$ & $F_c^{2\to 1}$ & $F_c^{3\to 1}$ & $F_c^{4\to 1}$ & $F_d>F_c^{2\to 1}$ & $F_d>F_c^{3\to 1}$ & $F_d>F_c^{4\to 1}$ \\
      0.95 & 0.95125 & 0.95184 & 0.95184 & 0.99257 & 0.99326 & $\times$    & $\times$    & $\times$    \\
      0.9  & 0.90500 & 0.90724 & 0.90724 & 0.97286 & 0.97578 & $\checkmark$ & $\times$    & $\times$    \\
      0.8  & 0.82000 & 0.82792 & 0.82785 & 0.90996 & 0.92135 & $\checkmark$ & $\times$    & $\times$    \\
      0.7  & 0.74500 & 0.76041 & 0.75975 & 0.83279 & 0.85471 & $\checkmark$ & $\times$    & $\times$    \\
      0.6  & 0.68000 & 0.70311 & 0.69990 & 0.75514 & 0.78317 & $\checkmark$ & $\times$    & $\times$    \\
      0.5  & 0.62500 & 0.65451 & 0.64569 & 0.68407 & 0.71071 & $\checkmark$ & $\times$    & $\times$    \\
      0.4  & 0.58000 & 0.61324 & 0.59760 & 0.62218 & 0.64282 & $\checkmark$ & $\times$    & $\times$    \\
      0.3  & 0.54500 & 0.57810 & 0.55728 & 0.57110 & 0.58402 & $\checkmark$ & $\checkmark$ & $\times$    \\
      0.2  & 0.52000 & 0.54806 & 0.52647 & 0.53263 & 0.53858 & $\checkmark$ & $\checkmark$ & $\checkmark$ \\
      0.1  & 0.50500 & 0.52227 & 0.50685 & 0.50842 & 0.50986 & $\checkmark$ & $\checkmark$ & $\checkmark$ \\
      0.05 & 0.50125 & 0.51073 & 0.50174 & 0.50214 & 0.50247 & $\checkmark$ & $\checkmark$ & $\checkmark$ \\
    \end{tblr}
    \caption{\textbf{Scaling Comparison}. 
        Performance comparison under symmetric amplitude damping noise, where both channels share the same noise parameter $p$. 
        Results are reported for single-copy LU pre-processing and multi-copy ($n=2,3,4$) PPT post-processing as $p$ varies, highlighting regimes in which single-copy pre-processing outperforms multi-copy post-processing.
    }
    \label{tab:Distributed_Multiple_Copies}
\end{table}

The advantage persists even in the 4-to-1 purification setting when only PPT post-processing is employed. 
A clearer illustration of the relation between the number of noisy copies and the achieved performance is provided in Fig.~\ref{fig:Disributed_Num_Scaling}. 
We do not extend the comparison to higher-copy regimes due to the rapidly increasing computational cost. 
Nevertheless, the observed trends suggest a strong conjecture: there exist instances in which even arbitrarily many noisy copies, processed solely via PPT post-processing, fail to outperform single-copy LU pre-processing.

\begin{figure}[htbp]
    \centering   
    \includegraphics[width=0.6\textwidth]{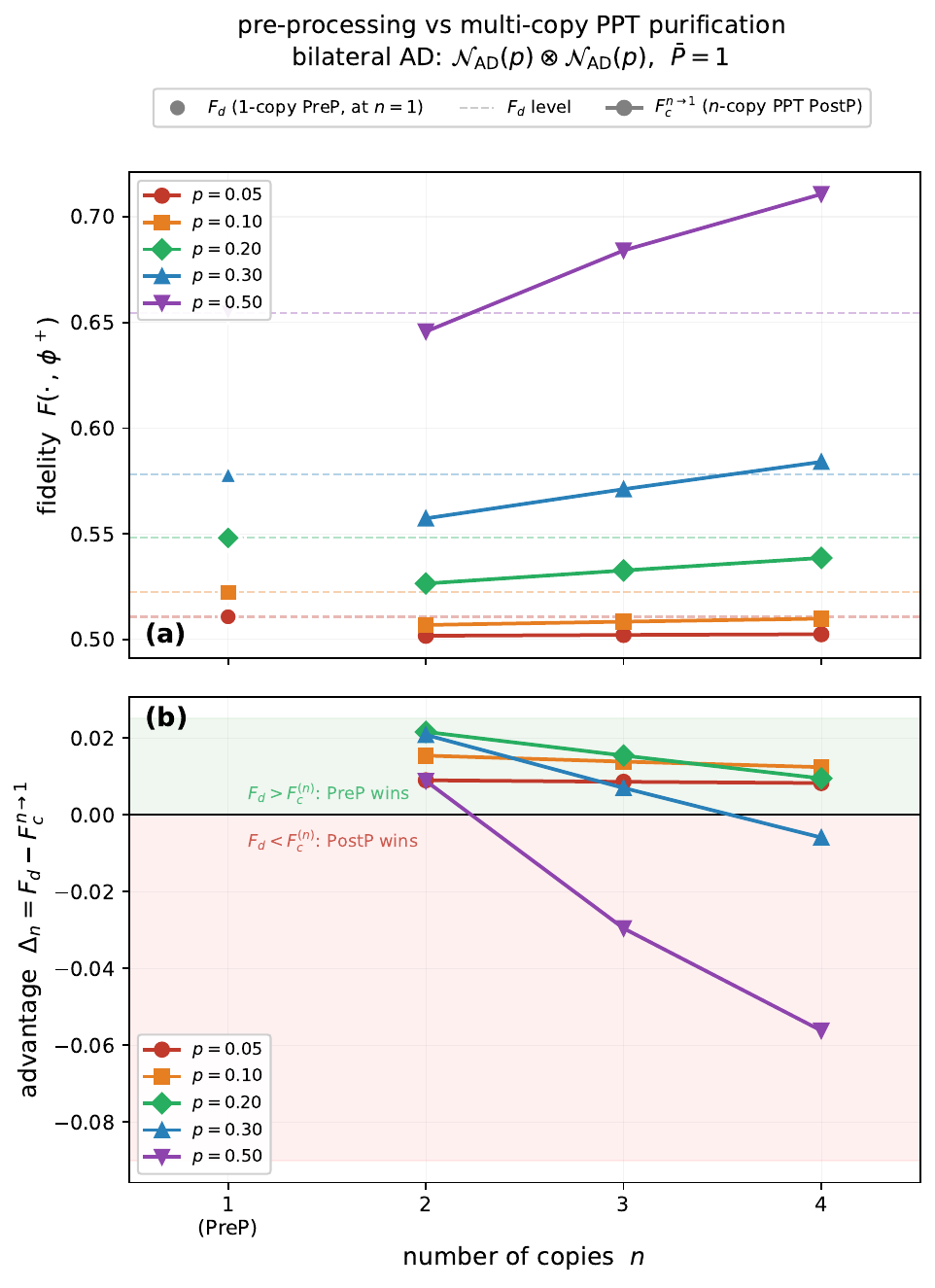}
    \caption{\textbf{Single-Copy LU Pre-Processing vs Multi-Copy PPT Post-Processing}. 
        (a) Fidelity as a function of the number of copies $n$ under symmetric amplitude damping noise $\mN_{\mathrm{AD}}(p)\otimes\mN_{\mathrm{AD}}(p)$.
        The horizontal dashed lines indicate the performance $F_d$ (see Tab.~\ref{tab:Purification_Comparison_1}) of single-copy LU pre-processing (see Fig.~\ref{fig:Distributed_Purifications}(d)) with different $p$, while the solid lines show the performance $F_c^{n\to 1}$ (see Eq.~\eqref{eq:PPT_n_to_1}) of $n$-copy PPT post-processing.
        (b) Advantage $\Delta_n:= F_d - F_c^{n\to 1}$ as a function of $n$. 
        Positive values, i.e., $\Delta_n>0$, indicate regimes where single-copy LU pre-processing outperforms multi-copy PPT post-processing, while negative values indicate the opposite.
        Across a broad range of noise parameters $p$, single-copy LU pre-processing maintains an advantage even against multi-copy PPT post-processing, highlighting a persistent performance gain beyond increasing copy number.
    }
    \label{fig:Disributed_Num_Scaling}
\end{figure}

In this subsection, we have introduced forward-assisted (FA) distributed purification protocols (see Tab.~\ref{tab:Purification_Comparison_1}) and compared their performance in terms of fidelity (see Fig.~\ref{fig:Distributed_Num_Point1_1}). 
From a theoretical perspective, FA protocols equipped with quantum memory can achieve higher performance limits, surpassing conventional approaches. 
However, such memory-assisted protocols are challenging to implement with current quantum technologies.
To address this, we further considered pre-processing–augmented (PreP) purification, which is operationally simpler. 
In particular, we focus on the minimal setting of local-unitary (LU) pre-processing (see Fig.~\ref{fig:Distributed_Purifications}(d)), avoiding the need for complex post-processing operations such as LOCC or PPT operations, whose optimal performance remains experimentally inaccessible.

Despite this simplicity, LU pre-processing can already deliver advantages (see Fig.~\ref{fig:Distributed_Num_Point2}). 
In certain regimes, pre-processing-augmented purification achieves performance unattainable by any PPT post-processing protocol (see Fig.~\ref{fig:Distributed_Num_Point3}). 
Moreover, beyond improving fidelity, pre-processing can reduce the number of copies required: remarkably, even single-copy LU pre-processing can outperform four-copy PPT post-processing for a broad range of noise parameters (see Fig.~\ref{fig:Disributed_Num_Scaling}).
These results highlight both the fundamental role of pre-processing in distributed purification and its practical relevance for near-term quantum technologies.


\subsection{Beyond the No-Purification Theorem}\label{subsec:Beyond_NP_Theorem}

Recent work on distributed quantum state purification has established that, for the set of four Bell states $\mS_{\mathrm{Bell}}$ (see Eq.~\eqref{eq:Bell_Set}), no non-trivial 2-to-1 LOCC purification protocol exists, even probabilistically (see Thm.~2 of~\cite{3bb1-pmtp}). 
This result is derived within the conventional framework, where only post-processing is considered. 
It is therefore natural to ask how this limitation changes when forward-assisted (FA) purification (see Subsec.~\ref{subsec:Forward_Assisted_Purification_Protocols}) is permitted, even in its simplest form involving pre-processing. 
In this subsection, it is shown that incorporating pre-processing circumvents the no-go theorem, enabling efficient purification even in the deterministic regime.

We now recall the no-purification theorem for the set of Bell states established in~\cite{3bb1-pmtp}:

\begin{mylem}
{No Purification of Bell States~\cite{3bb1-pmtp}}{NP_Bell}
For the bilocal depolarizing channel 
\begin{align}
    \mN_1(p)\otimes\mN_2(q)=
    \mN_{\mathrm{D}}(p)\otimes\mN_{\mathrm{D}}(q),
\end{align}
with noise parameters $p$ and $q$, no non-trivial 2-to-1 purification protocol based on PPT post-processing exists for the Bell-state set $\mS_{\mathrm{Bell}}$ (see Eq.~\eqref{eq:Bell_Set}), for any success probability.
\end{mylem}

Note that, in this setting, the noisy state takes the form 
\begin{align}\label{eq:psi_pq}
    \psi_{(p,q)}:=
    \mN_{\mathrm{D}}(p)\otimes\mN_{\mathrm{D}}(q)(\psi), 
    \quad\text{with}\,\,\,
    \psi\in\mS_{\mathrm{Bell}}.
\end{align}
Rather than adopting the original statement of Thm.~2 in Ref.~\cite{3bb1-pmtp}, we use its strongest formulation (see Lem.~\ref{lem:NP_Bell}), as presented in Supplementary Material G of~\cite{3bb1-pmtp}.
This version applies not only to LOCC operations but also to the strictly larger class of PPT operations, and holds for arbitrary success probability.

\begin{figure}[htbp]
    \centering   
    \includegraphics[width=1\textwidth]{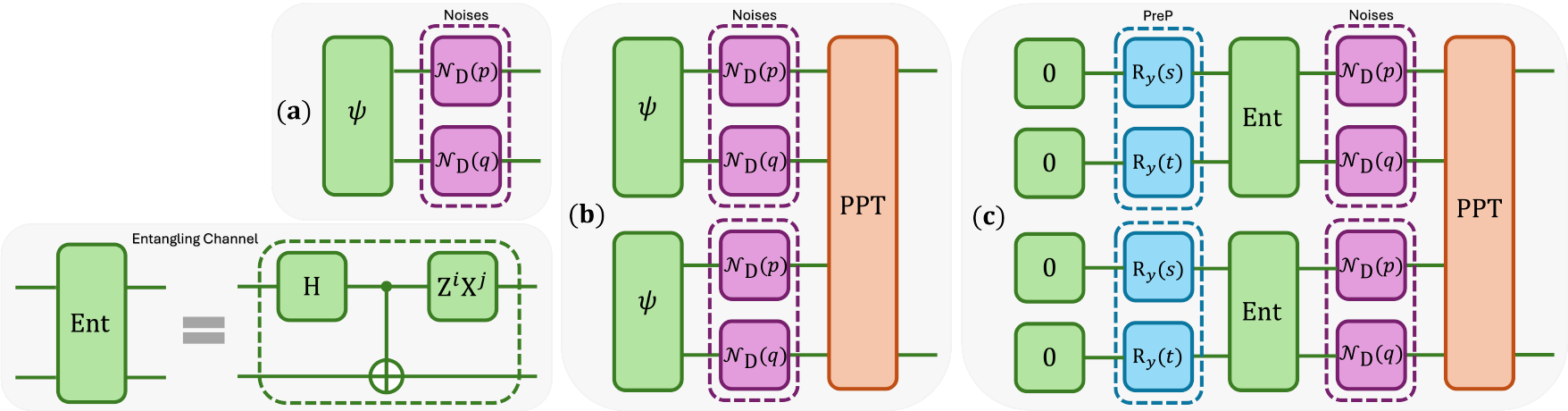}
    \caption{\textbf{Forward-Assisted Bell-State Purification}. 
        The input state is one of the four Bell states (see Eq.~\eqref{eq:Bell_Set}), and performance is quantified by the average fidelity over this set.
        (a) Benchmarking: noisy Bell states generated via an entangling channel (Ent), followed by local depolarizing noise $\mN_{\mathrm{D}}(p)\otimes\mN_{\mathrm{D}}(q)$.
        (b) Conventional Purification: two noisy copies are processed via PPT post-processing.
        (c) PreP Purification: local amplitude damping pre-processing is applied prior to Ent, followed by depolarizing noise and PPT post-processing.
    }
    \label{fig:Nogo_Protocols}
\end{figure}

In the deterministic setting, the average fidelity under the bilocal depolarizing channel is then characterized as (see Fig.~\ref{fig:Nogo_Protocols}(a))
\begin{align}\label{eq:Bell_ave}
    F_{\mathrm{ave}}(p,q):=
    \frac{1}{4}
    \sum_{\psi\in\mS_{\mathrm{Bell}}}
    F(\psi_{(p,q)},\psi),
\end{align}
and serves as a baseline benchmark for purification protocols. 
By contrast, one may employ a 2-to-1 PPT purification map $\mE\in \mathrm{PPT}$, leading to the following performance (see Fig.~\ref{fig:Nogo_Protocols}(b))
\begin{align}\label{eq:Bell_ave_PPT_pq}
    F_{\mathrm{ave,\ PPT}}(p,q):=
    \frac{1}{4}\max_{\mE\in\mathrm{PPT}}
    \sum_{\psi\in\mS_{\mathrm{Bell}}}
    F(\mE\left(\psi_{(p,q)}\otimes\psi_{(p,q)}\right),\psi),
\end{align}
where $\psi_{(p,q)}$ is defined in Eq.~\eqref{eq:psi_pq}.
Since Lem.~\ref{lem:NP_Bell} holds for arbitrary success probability, it also applies to the deterministic setting considered here, implying
\begin{align}
    F_{\mathrm{ave,\ PPT}}(p,q)=
    F_{\mathrm{ave}}(p,q).
\end{align}

We proceed to examine the role of pre-processing in Bell-state purification.
In contrast to the protocols considered in Subsec.~\ref{subsec:LU_Pre_Processing}, where pre-processing is treated in isolation, it is here integrated with PPT post-processing to probe the ultimate limits of UA purification (see Tab.~\ref{tab:Purification_Comparison_1}). 
The analysis begins by applying pre-processing to states drawn from the set $\mS_{\mathrm{Bell}}$ (see Eq.~\eqref{eq:Bell_Set}).
By treating the noisy entangling gate as a dynamical resource, we can act on the input before the gate is applied. 
Taking this pre-processing to be a local unitary, specifically $R_y(s)\otimes R_y(t)$, gives the state
\begin{align}\label{eq:psi_st}
    \psi^{(s,t)}(i,j) \coloneqq
    \underbrace{(\mathrm{Z}^{i}\mathrm{X}^{j}\otimes\id)\circ\mathrm{CNOT}\circ(\mathrm{H}\otimes\id)}_{\text{Entangling Gate}}
    \circ
    \overbrace{\left(R_y(s)\otimes R_y(t)\right)}^{\text{Pre-Processing}}(\ketbra{00}{00}), 
    \quad\text{with}\,\,\,
    i, j\in\{0,1\}.
\end{align}
To benchmark the performance, the four Bell states are expressed as (see Fig.~\ref{fig:Nogo_Protocols})
\begin{align}\label{eq:psi_ij}
    \psi(i,j) \coloneqq
    \underbrace{(\mathrm{Z}^{i}\mathrm{X}^{j}\otimes\id)\circ\mathrm{CNOT}\circ(\mathrm{H}\otimes\id)}_{\text{Entangling Gate}}
    (\ketbra{00}{00}), 
    \quad\text{with}\,\,\,
    i, j\in\{0,1\}.
\end{align}
After the action of noise, specifically, depolarizing noise $\mN_{\mathrm{D}}$ on both subsystems, the state $\psi^{(s,t)}(i,j)$ becomes
\begin{align}\label{eq:psi_stpq}
    \psi^{(s,t)}_{(p,q)}(i,j) \coloneqq
    \underbrace{\left(\mN_{\mathrm{D}}(p)\otimes\mN_{\mathrm{D}}(q)\right)}_{\text{Noise}}
    \left(\psi^{(s,t)}(i,j)\right), 
    \quad\text{with}\,\,\,
    i, j\in\{0,1\}.
\end{align}
Here, $\mN_{\mathrm{D}}(p)\otimes\mN_{\mathrm{D}}(q)$ represents the noise (see Fig.~\ref{fig:Nogo_Protocols}).

In the final stage, a PPT operation is applied to purify the noisy state, with the resulting 2-to-1 purification performance quantified by (see Fig.~\ref{fig:Nogo_Protocols}(c))
\begin{align}\label{eq:Bell_ave_Pre_PPT_stpq}
    F_{\mathrm{ave,\ PPT}}^{\mathrm{Pre}}(s,t,p,q) \coloneqq
    \frac{1}{4}\max_{\mE\in\mathrm{PPT}}
    \sum_{i,j}
    F(\mE\left(\psi^{(s,t)}_{(p,q)}(i,j)\otimes\psi^{(s,t)}_{(p,q)}(i,j)\right),\psi(i,j)).
\end{align}
where $\psi^{(s,t)}_{(p,q)}(i,j)$ is defined in Eq.~\eqref{eq:psi_stpq}.

To assess the performance over all pre-processing operations, one may further optimize over the parameters $s$ and $t$, which ultimately yields the following average fidelity
\begin{align}\label{eq:Bell_ave_Pre_PPT_pq}
    F_{\mathrm{ave,\ PPT}}^{\mathrm{Pre}}(p,q) \coloneqq
    \frac{1}{4}\max_{s,t}\max_{\mE\in\mathrm{PPT}}
    \sum_{i,j}
    F(\mE\left(\psi^{(s,t)}_{(p,q)}(i,j)\otimes\psi^{(s,t)}_{(p,q)}(i,j)\right),\psi(i,j)).
\end{align}
A comparison between $F_{\mathrm{ave,\ PPT}}^{\mathrm{Pre}}(s,t,p,q)$ (or even $F_{\mathrm{ave,\ PPT}}^{\mathrm{Pre}}(p,q)$) and $F_{\mathrm{ave,\ PPT}}(p,q)$ is well justified, as the set of free operations, namely, PPT operations, is identical in both cases. 
In particular, although the UA protocol incorporates an additional pre-processing stage, it does not introduce any extra resource cost. 
The comparison is therefore fair from the perspective of dynamical entanglement~\cite{bauml2019resourcetheoryentanglementbipartite,PhysRevLett.125.180505,xing2023fundamentallimitationscommunicationquantum,glc7-xy8t}.

Having introduced the conventional purification protocol based on PPT post-processing (see Eq.~\eqref{eq:Bell_ave_PPT_pq}), and the unassisted (UA) protocol combining local unitary (LU) pre-processing $R_y(s)\otimes R_y(t)$ with PPT post-processing (see Eq.~\eqref{eq:Bell_ave_Pre_PPT_stpq}), their performance is compared using average fidelities. 
The protocols considered here are summarized in Tab.~\ref{tab:Bell_Purification_Comparison}.

\begin{table}[htbp]
    \centering
    \begin{tblr}{
      colspec = {l || c | c | c},
      row{1,2} = {bg=gray!50, font=\bfseries}, 
      column{1} = {bg=gray!10},                
      hlines,                                  
      vlines,                                  
      cells = {m},                             
      row{1} = {c},                            
    }
      \SetCell[c=4]{c} Forward-Assisted Distributed Bell-State Purification Protocols (2-to-1) & & & \\
      Distributed Purifications & Performance & Number of Copies & Constraints \\
      Bell Protocol a (see Eq.~\eqref{eq:Bell_ave}) & $B_a:=F_{\mathrm{ave}}(p,q)$ & single-copy & without any purification \\
      Bell Protocol b (see Eq.~\eqref{eq:Bell_ave_PPT_pq}) & $B_b:=F_{\mathrm{ave,\ PPT}}(p,q)$ & two-copy & with PPT post-processing \\
      Bell Protocol c (see Eq.~\eqref{eq:Bell_ave_Pre_PPT_stpq}) & $B_c:=F_{\mathrm{ave,\ PPT}}^{\mathrm{Pre}}(s,t,p,q)$ & two-copy & {with pre-processing implemented\\ by $R_y(s)\otimes R_y(t)$\\ and PPT post-processing} \\
    \end{tblr}
    \caption{\textbf{Forward-Assisted Distributed Bell-State Purification Protocols}.
        Summary of forward-assisted (FA) distributed 2-to-1 Bell-state purification protocols, specifying their performance metrics, number of input copies, and operational constraints.
        The protocols include the baseline without purification, PPT post-processing, and FA schemes combining local unitary pre-processing with PPT post-processing.
    }
    \label{tab:Bell_Purification_Comparison}
\end{table}

It is worth noting that inserting pre-processing consisting solely of local unital channels after the entangling channel in Fig.~\ref{fig:Nogo_Protocols}, and before the depolarizing noise $\mN_{\mathrm{D}}(p)\otimes\mN_{\mathrm{D}}(q)$, does not improve purification performance, as established in Thm.~\ref{thm:Ineffectiveness_LUnital_PreP}.

As shown in Lem.~\ref{lem:NP_Bell}, Bell protocols a and b in Tab.~\ref{tab:Bell_Purification_Comparison} achieve identical performance, namely
\begin{align}
    B_b = B_a,
\end{align}
indicating that PPT post-processing alone provides no improvement. 
We therefore examine the difference
\begin{align}\label{eq:Delta_B_c_B_a}
    \Delta \coloneqq B_c - B_a.
\end{align}
A strictly positive value signals a violation of the no-purification limitation enabled by pre-processing, while zero indicates no improvement.

The numerical results for depolarizing noise $\mN_{\mathrm{D}}(p)\otimes\mN_{\mathrm{D}}(p)$ with strengths $p\in\{0.1, 0.3, 0.5, 0.7\}$ are shown in Fig.~\ref{fig:Nogo_Num_AD_Det}(a)-(d), where a strictly positive value is clearly observed. 
This demonstrates that, once pre-processing is introduced, even simple local rotations about the $y$-axis are sufficient to overcome the no-go result of Lem.~\ref{lem:NP_Bell}, thereby enabling genuine purification of Bell states.
Notably, this advantage is achieved without additional resources, as local rotations $R_y(s)\otimes R_y(t)$ belong to LOCC, and the setting considered here is fully deterministic.
These findings point to a practically accessible route towards more efficient purification protocols.
The result is formalized in the following theorem.

\begin{figure}[htbp]
    \centering   
    \includegraphics[width=1\textwidth]{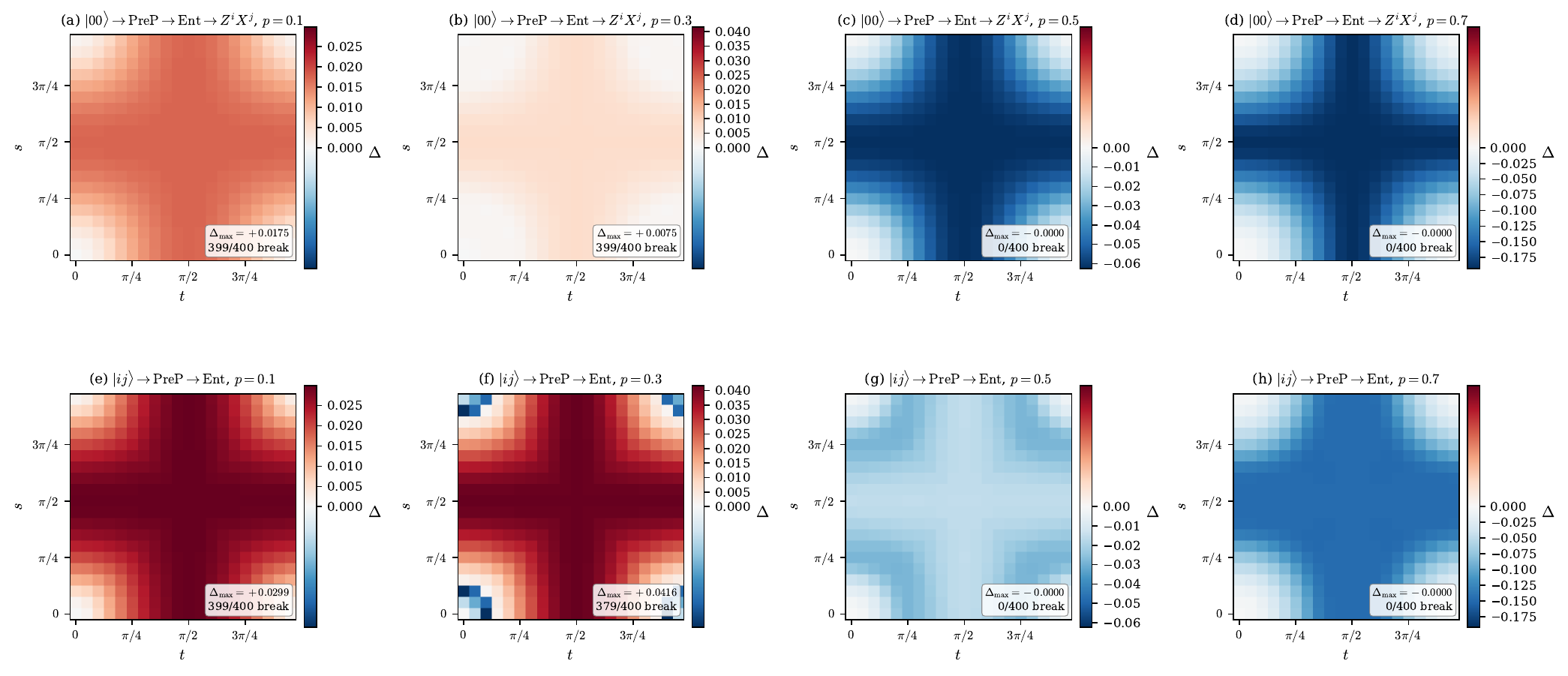}
    \caption{\textbf{Breaking No-Purification via PreP}. 
        Heat maps of $\Delta$ (see Eq.~\eqref{eq:Delta_B_c_B_a}) over $(s,t)$ (see Eq.~\eqref{eq:psi_st}) for depolarizing noise $\mN_{\mathrm{D}}(p)\otimes\mN_{\mathrm{D}}(p)$ with (a) $p=0.1$, (b) $p=0.3$, (c) $p=0.5$, and (d) $p=0.7$, based on the protocol in Fig.~\ref{fig:Nogo_Protocols}. 
        Here $\Delta$ denotes the performance difference between protocols a and c in Tab.~\ref{tab:Bell_Purification_Comparison}, corresponding to schemes without and with local amplitude damping pre-processing, respectively. 
        Positive values ($\Delta>0$, red regions) indicate a violation of Lem.~\ref{lem:NP_Bell}, showing that local pre-processing enables genuine Bell-state purification within LOCC and in a deterministic setting, whereas negative values ($\Delta<0$, blue regions) indicate degraded performance. 
        Panels (e)-(h) show the corresponding results for the alternative Bell-state generation scheme illustrated on the right-hand side of Fig.~\ref{fig:Bell}, with (e) $p=0.1$, (f) $p=0.3$, (g) $p=0.5$, and (h) $p=0.7$.
        The full purification protocol for this alternative construction is detailed in Fig.~\ref{fig:Nogo_Protocols_Alternative}.
    }
    \label{fig:Nogo_Num_AD_Det}
\end{figure}

\begin{mythm}{Efficient Purification of Bell States}{EP_Bell}
    For the bilocal depolarizing channel 
    \begin{align}
        \mN_1(p)\otimes\mN_2(q)=
        \mN_{\mathrm{D}}(p)\otimes\mN_{\mathrm{D}}(q),
    \end{align}
    there exist efficient 2-to-1 UA purification protocols for the Bell-state set $\mS_{\mathrm{Bell}}$ (see Eq.~\eqref{eq:Bell_Set}), in which pre-processing is implemented by $R_y(s)\otimes R_y(t)$ and post-processing is optimized over all PPT operations (see Fig.~\ref{fig:Nogo_Protocols}(c)). 
    In particular, such a protocol achieves efficient purification even in the deterministic regime, as shown in Fig.~\ref{fig:Nogo_Num_AD_Det}(a)-(d).
\end{mythm}

Remark that the preparation of Bell states (see Eq.~\eqref{eq:Bell_Set}) is not unique. 
In the analysis so far, we adopt the approach of first generating a maximally entangled state and then obtaining the remaining Bell states via local Pauli operations $\mathrm{Z}^{i}\mathrm{X}^{j}$. 
Alternatively, one may fix the entangling circuit to be H followed by CNOT, and generate the four Bell states by varying the computational basis inputs $(i,j)$, i.e., $\ket{ij}$. 
The correspondence between $(i,j)$ and the resulting Bell states is illustrated in Fig.~\ref{fig:Bell}.

\begin{figure}[htbp]
    \centering   
    \includegraphics[width=0.5\textwidth]{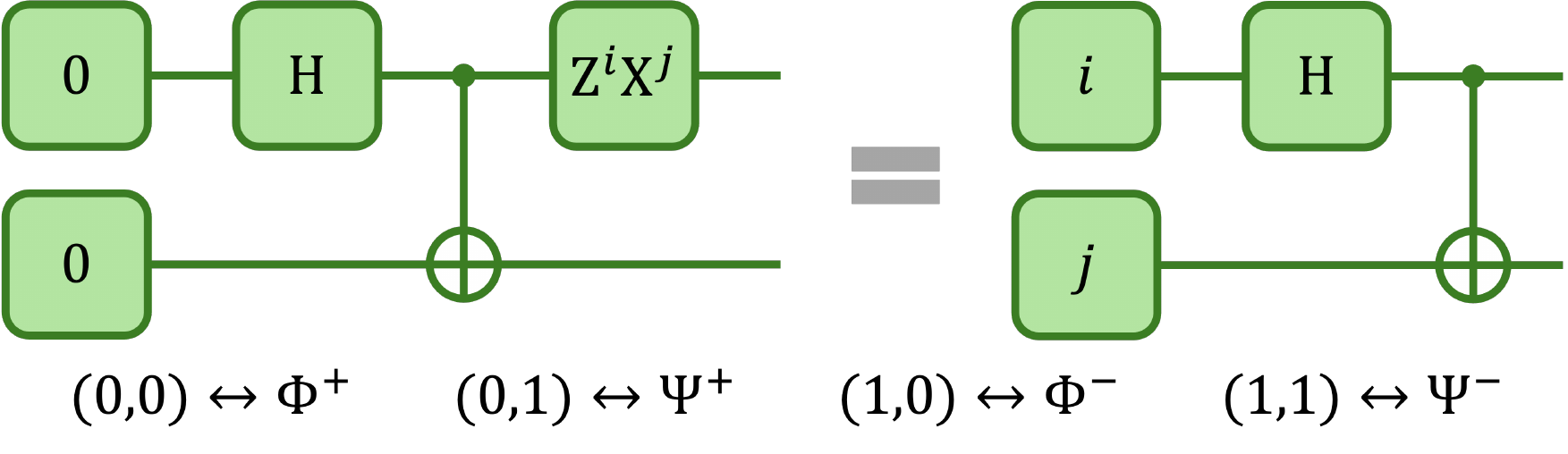}
    \caption{\textbf{Equivalent Constructions of Bell States}. 
        Two equivalent methods for generating the four Bell states (see Subsec.~\ref{subsec:Bell_States}) are illustrated. 
        Left: a fixed input $\ket{00}$ is first entangled via a Hadamard and CNOT gate, followed by local Pauli operations $\mathrm{Z}^{i}\mathrm{X}^{j}$ to obtain all Bell states. 
        Right: the entangling circuit (Hadamard + CNOT) is fixed, while the computational basis inputs $\ket{ij}$ are varied. 
        The correspondence between $(i,j)$ and the resulting Bell states is shown below, demonstrating the equivalence of the two constructions.
    }
    \label{fig:Bell}
\end{figure}

This naturally raises the question of whether pre-processing remains effective when Bell states are generated using the alternative construction shown on the right-hand side of Fig.~\ref{fig:Bell}. 
The corresponding protocol is illustrated in Fig.~\ref{fig:Nogo_Protocols_Alternative}.

\begin{figure}[htbp]
    \centering   
    \includegraphics[width=0.75\textwidth]{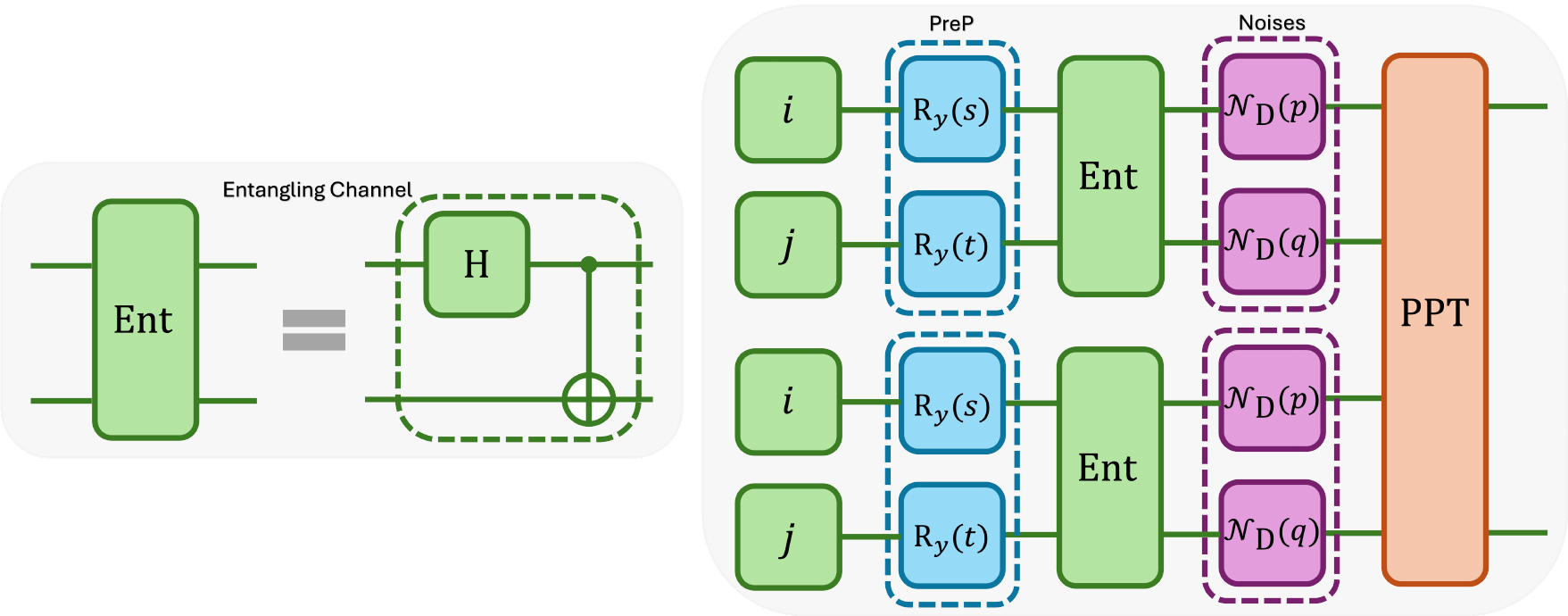}
    \caption{\textbf{Forward-Assisted Bell-State Purification with Alternative Bell-State Generation}. 
        PreP purification protocol based on the alternative Bell-state generation scheme shown on the right-hand side of Fig.~\ref{fig:Bell}.
    }
    \label{fig:Nogo_Protocols_Alternative}
\end{figure}

Despite variations in numerical outcomes across different PreP schemes, a key feature remains unchanged: the 2-to-1 no-purification limitation for Bell states is consistently violated. 
This further underscores the essential role of pre-processing in enabling enhanced purification. 
Detailed numerical results are presented in Fig.~\ref{fig:Nogo_Num_AD_Det}(e)-(h).

Up to this point, the discussion has focused on deterministic purification of Bell states. This restriction is not essential: the breakdown of the no-purification limitation extends beyond the deterministic regime. 
Numerical experiments for probabilistic purification have also been performed, with representative results at success probability $\bar{P}=0.1$ shown in Fig.~\ref{fig:Nogo_Num_AD_Prob}. 
These results confirm that genuine Bell-state purification persists in the probabilistic setting once pre-processing is incorporated.

\begin{figure}[htbp]
    \centering   
    \includegraphics[width=1\textwidth]{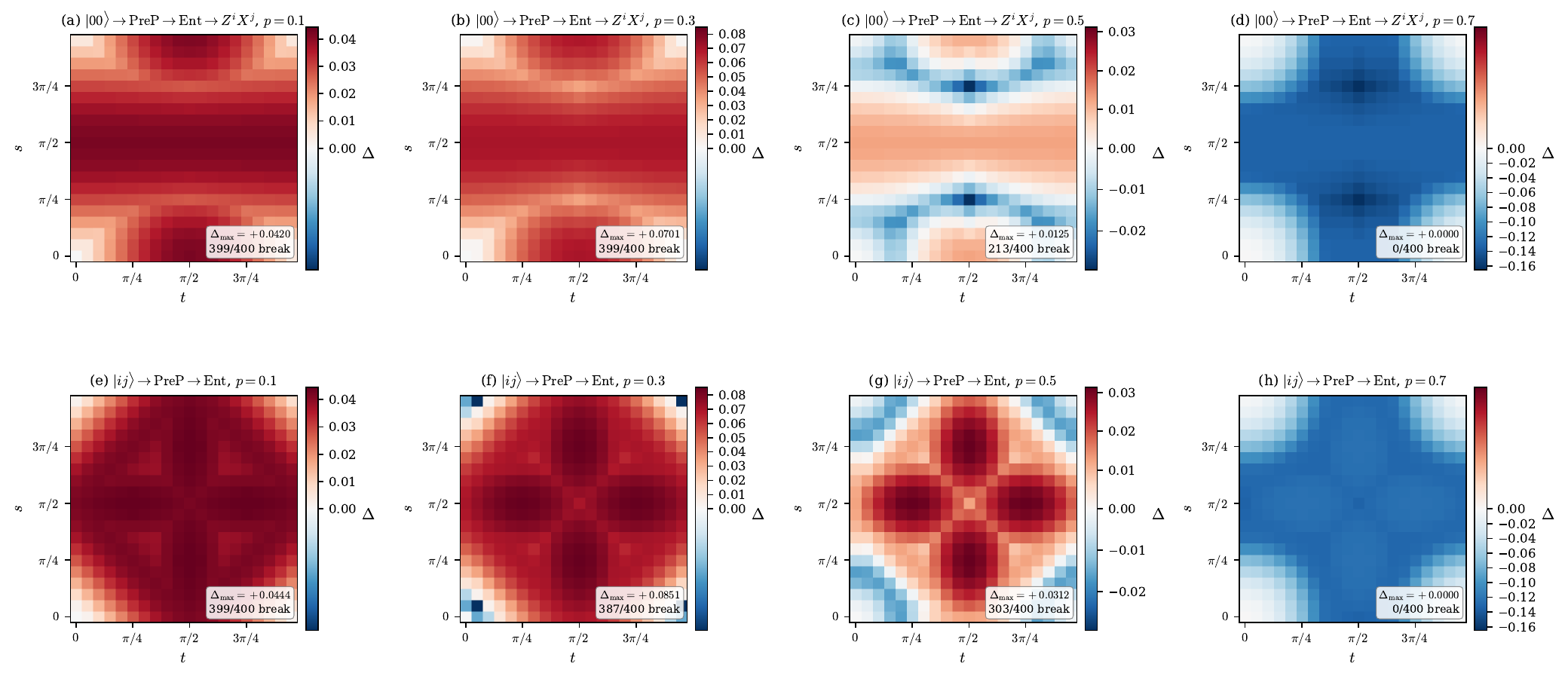}
    \caption{\textbf{Breaking Probabilistic No-Purification via PreP}. 
        Heat maps of $\Delta$ in the probabilistic regime. 
        Distributions of $\Delta$ (Eq.~\eqref{eq:Delta_B_c_B_a}) over $(s,t)$ (Eq.~\eqref{eq:psi_st}) under depolarizing noise $\mN_{\mathrm{D}}(p)\otimes\mN_{\mathrm{D}}(p)$ with (a) $p=0.1$, (b) $p=0.3$, (c) $p=0.5$, and (d) $p=0.7$, based on Fig.~\ref{fig:Nogo_Protocols}. 
        Here $\Delta$ measures the performance difference between protocols (a) and (c) in Tab.~\ref{tab:Bell_Purification_Comparison}, without and with local amplitude-damping pre-processing. 
        Positive values ($\Delta>0$, red) indicate violation of Lem.~\ref{lem:NP_Bell}, while negative values ($\Delta<0$, blue) indicate degraded performance. 
        Panels (e)–(h) show the corresponding results for the alternative Bell-state generation scheme (Fig.~\ref{fig:Bell}) with the same ordering of $p$. 
        The full protocol is given in Fig.~\ref{fig:Nogo_Protocols_Alternative}. 
        All panels correspond to the probabilistic regime with success probability $\bar{P}=0.1$.
    }
    \label{fig:Nogo_Num_AD_Prob}
\end{figure}

Theorem~\ref{thm:EP_Bell} shows that pre-processing breaks the no-purification limitation for 2-to-1 Bell-state purification. 
This prompts a natural question: what mechanism enables this advantage, and why does pre-processing overcome the no-go theorem of Lem.~\ref{lem:NP_Bell}? 
To investigate this, we replace the LU pre-processing $R_y(s)\otimes R_y(t)$ with randomly sampled quantum channels generated via Monte Carlo methods.
In each trial, a random single-qubit CPTP map $\mE_A$ (and independently $\mF_B$) is constructed by sampling a $4\times4$ Choi operator from a Wishart distribution and projecting onto the trace-preserving (TP) constraint; 
Kraus operators are then obtained via eigendecomposition of the Choi operator.

\begin{figure}[htbp]
    \centering   
    \includegraphics[width=0.85\textwidth]{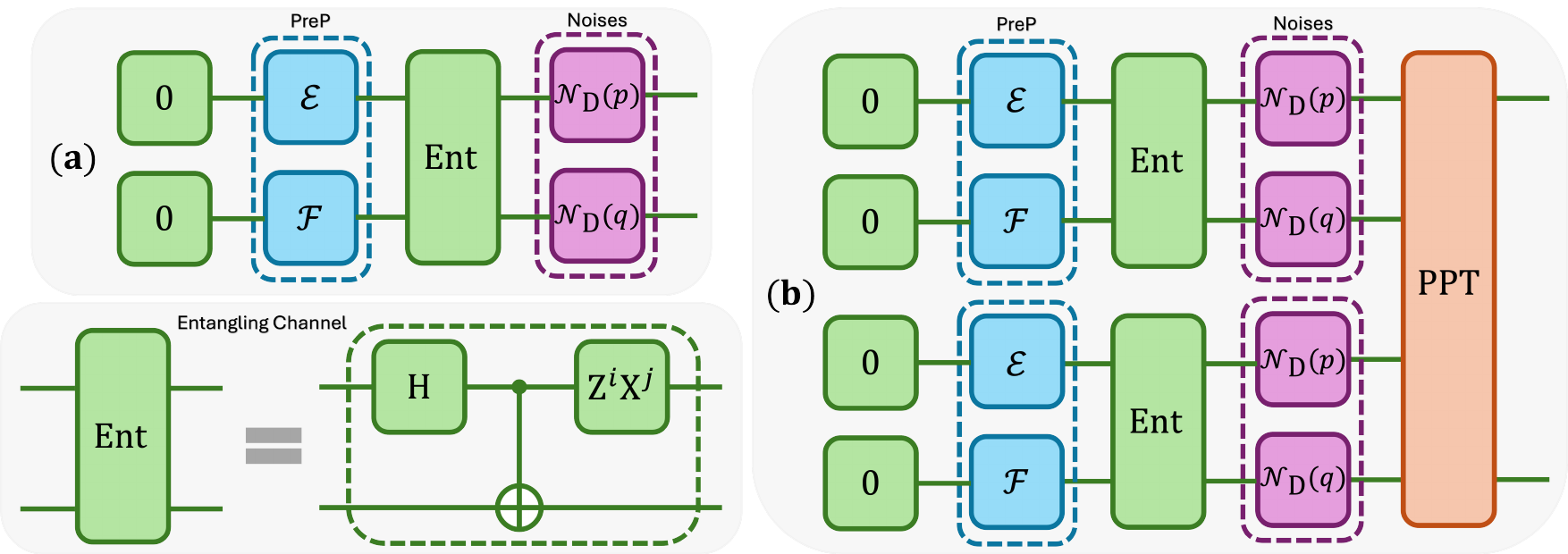}
    \caption{\textbf{Forward-Assisted Bell-State Purification Using Random Quantum Channels}. 
        The input state is one of the four Bell states (see Eq.~\eqref{eq:Bell_Set}), and performance is quantified by the average fidelity over this set.
        (a) Randomized PreP: local pre-processing channels $\mE$ and $\mF$ are applied prior to the entangling channel (Ent), followed by local depolarizing noise $\mN_{\mathrm{D}}(p)\otimes\mN_{\mathrm{D}}(q)$. 
        In each trial, $\mE$ and $\mF$ are independently sampled single-qubit CPTP maps generated via Monte Carlo methods. 
        (b) Randomized PreP with PostP: two noisy copies produced as in (a) are further processed via PPT post-processing. 
        This setup is used to investigate the mechanism by which pre-processing overcomes the no-purification limitation of Lem.~\ref{lem:NP_Bell}.
    }
    \label{fig:Nogo_Protocols_Random}
\end{figure}

We now consider the protocols in Fig.~\ref{fig:Nogo_Protocols_Random} and evaluate their performance using PreP of randomly sampled channels $\mE \otimes \mF$, with the results shown in Tab.~\ref{tab:Nogo_Protocols_Random}.

\begin{table}[htbp]
\centering
\begin{tblr}{
  colspec = {c || c | c | c | c},
  row{1,2} = {bg=gray!50, font=\bfseries}, 
  column{1} = {bg=gray!10, font=\bfseries}, 
  hlines, 
  vlines,
  cells = {m, c},
  cell{3,4,5,6,7}{4} = {bg=mRed},
  cell{5,6}{5} = {bg=mRed},
  cell{3,4,7}{5} = {bg=mGreen},
}
  \SetCell[c=5]{c} Forward-Assisted Distributed Purification: Individual Bell States & & & & \\
  Bell State & {Prot. of Fig.~\ref{fig:Nogo_Protocols}(a) \\ (1-Copy) No Pur.} & {Prot. of Fig.~\ref{fig:Nogo_Protocols}(b)\\ (2-to-1) PPT PostP} & {Prot. of Fig.~\ref{fig:Nogo_Protocols_Random}(a)\\ (1-Copy) PreP} & {Prot. of Fig.~\ref{fig:Nogo_Protocols_Random}(b)\\ (2-to-1) PreP + PPT PostP} \\
  
  $\Phi^+$ & 0.2575 & 0.2575 & 0.2505 & 0.2706 \\
  
  $\Phi^-$ & 0.2575 & 0.2575 & 0.2495 & 0.3112 \\
  
  $\Psi^+$ & 0.2575 & 0.2575 & 0.2507 & 0.2507 \\
  
  $\Psi^-$ & 0.2575 & 0.2575 & 0.2497 & 0.2497 \\
  
  Average & \textbf{0.2575} & \textbf{0.2575} & \textbf{0.2501} & \textbf{0.2706} \\
\end{tblr}
\caption{\textbf{Comparison of State Fidelity across Different Protocols}. 
The action of the optimal PPT post-processing --- maximizing the average fidelity --- is evaluated on each Bell state individually. 
Both subsystems are subject to depolarizing noise with parameter $p=0.1$. 
The reported data correspond to a fixed random seed (seed = 138) used to initialize the pseudo-random number generator.
Green shading \colorbox{mGreen}{\phantom{x}} indicates an improvement in fidelity relative to the baseline (see Fig.~\ref{fig:Nogo_Protocols}), while red shading \colorbox{mRed}{\phantom{x}} denotes a decrease.}
\label{tab:Nogo_Protocols_Random}
\end{table}

The numerical data in Tab.~\ref{tab:Nogo_Protocols_Random} reveal that the role of pre-processing here is fundamentally different from that in Subsec.~\ref{subsec:LU_Pre_Processing}, where it directly enhances purification performance. 
In the present setting, random channel pre-processing $\mE\otimes\mF$ initially reduces the average fidelity, appearing detrimental at first glance. 
However, this reduction plays a crucial role: it breaks the symmetry among the four Bell states, as evidenced in the third column of Tab.~\ref{tab:Nogo_Protocols_Random}.
Once this symmetry is lifted, PPT post-processing becomes effective and can subsequently improve the purification performance. 
Pre-processing is therefore essential in this distributed setting, not because it enhances fidelity directly, but because it induces the asymmetry required to unlock further gains.

Till now, the analysis has focused on deterministic protocols and the corresponding violation of the no-purification limitation. 
This mechanism, however, is not restricted to deterministic settings. 
It persists in the probabilistic regime, where the success probability is less than one. 
Additional numerical experiments, including probabilistic purification with randomly sampled channels, is presented in Fig.~\ref{fig:Nogo_Protocols_Random_Prob}.

\begin{figure}[htbp]
    \centering   
    \includegraphics[width=1\textwidth]{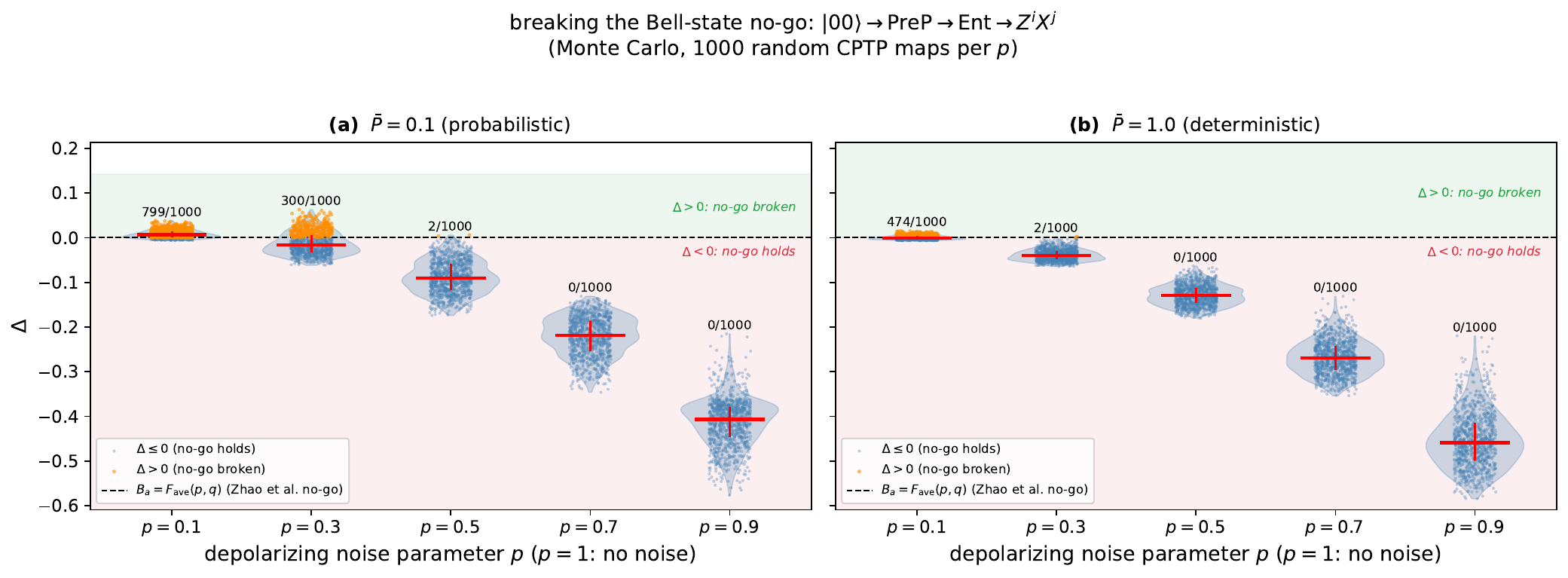}
    \caption{\textbf{Breaking No-Purification with Random PreP}. 
        Distribution of $\Delta$ obtained from Monte Carlo sampling of random CPTP pre-processing channels $\mE \otimes \mF$ (1000 samples per $p$). 
        (a) Probabilistic setting with $\bar{P}=0.1$; 
        (b) deterministic setting with $\bar{P}=1$. 
        The horizontal axis denotes the depolarizing noise parameter $p$, and the dashed line marks $\Delta=0$, separating regimes where the no-go theorem holds ($\Delta \leqslant 0$) and is violated ($\Delta > 0$). 
        Positive values indicate that pre-processing enables purification beyond the PPT limit of Lem.~\ref{lem:NP_Bell}. 
        Counts above each cluster denote the number of samples with $\Delta>0$. 
        Across both probabilistic and deterministic regimes, strictly positive instances are observed, confirming that random pre-processing can break the no-purification limitation.
    }
    \label{fig:Nogo_Protocols_Random_Prob}
\end{figure}

In summary, the results of this subsection establish that forward-assisted protocols fundamentally alter the landscape of Bell-state purification. 
While conventional PPT post-processing alone is constrained by the no-purification theorem, the inclusion of pre-processing --- particularly beyond structured or symmetric operations --- enables one to circumvent this limitation. 
Crucially, the role of pre-processing is not merely to enhance fidelity directly, but to reshape the structure of the input ensemble, breaking symmetries that otherwise inhibit purification. 
This mechanism persists across both deterministic and probabilistic regimes, indicating that the observed advantage is robust and not tied to a specific operational setting. 
Taken together, these findings reveal that pre-processing acts as a catalytic resource in distributed purification, unlocking capabilities that are fundamentally inaccessible within the conventional purification framework.


\end{document}